\documentclass[11pt,a4paper]{article}

\usepackage{authblk}
\usepackage{amsmath,amssymb,amsfonts,amsthm}
\usepackage{geometry}
\usepackage{booktabs}
\usepackage{tcolorbox} 
\usepackage{hyperref}
\usepackage[x11names]{xcolor}
\usepackage{comment}

\newcounter{remarkcount}
\newenvironment{remark}{%
	\refstepcounter{remarkcount}%
	\vspace{\topsep}
	\noindent%
	{\bfseries Remark \theremarkcount.}~
	\small%
}{%
	\hfill$\blacktriangle$
	\vspace{\topsep}
	\par%
}

\usepackage[dvipsnames]{xcolor}
\definecolor{darkgreen}{HTML}{006400}
\definecolor{lightgreen}{HTML}{00C000}
\definecolor{lightgreen}{HTML}{00D000}

\definecolor{weightgreen}{HTML}{00C000}

\definecolor{brightgreen}{HTML}{00FF00}

\definecolor{darkgreen}{HTML}{006400}
\definecolor{lightgreen}{HTML}{E8F8F5}
\definecolor{primaryblue}{HTML}{1F618D}
\definecolor{accentorange}{HTML}{D35400}
\definecolor{accentpurple}{HTML}{7D3C98}
\definecolor{dangerred}{HTML}{C0392B}

\usepackage{tikz}
\usetikzlibrary{shapes.geometric, arrows.meta, positioning, fit, backgrounds, calc, patterns, decorations.pathreplacing, decorations.markings}
\usetikzlibrary{decorations.pathmorphing}

\hypersetup{
	colorlinks=true,
	linkcolor=blue,
	citecolor=blue,
	urlcolor=blue
}

\newtheorem{definition}{Definition}
\newtheorem{theorem}{Theorem}
\newtheorem{corollary}{Corollary}
\newtheorem{lemma}{Lemma}

\newtheorem{proposition}{Proposition}
\newtheorem{example}{Example}

\newtheorem{assumption}{Assumption}

\newtheorem{insight}{Insight}

\newcommand{\INITIALIZE}{\item[\textbf{Initialize:}]}

\usepackage{etoolbox}

\theoremstyle{remark}

\usepackage{algorithm,algorithmic}
\usepackage{booktabs}
\usepackage{geometry}
\usepackage{microtype}

\usepackage{tabularx}

\usepackage{listings} 
\usepackage{amsfonts}
\usepackage{booktabs}

\definecolor{jsonbg}{HTML}{F8F9FA}
\definecolor{jsonpunct}{HTML}{990000}
\definecolor{jsonstring}{HTML}{008000}
\definecolor{jsonkey}{HTML}{0000FF}

\lstdefinestyle{jsonstyle}{
	basicstyle=\normalfont\ttfamily\small,
	backgroundcolor=\color{jsonbg},
	numbers=left,
	numberstyle=\scriptsize\color{gray},
	stepnumber=1,
	numbersep=8pt,
	showstringspaces=false,
	breaklines=true,
	stringstyle=\color{jsonstring},
		keywords={true,false,null,STRING,DOUBLE,MANY_TO_MANY},
		keywordstyle=\color{jsonkey}\bfseries,
	}
	
	\title{\textbf{The SMG-Yau-Yau Filter and Lossless Data Assimilation: Resolving Infinite Lie Algebras via Statistical Fiber Theory}}
\author[1,2,8]{Bing Cheng}
\author[3]{Yi-Shuai Niu} 
\author[5,6,7]{Howell Tong}
\author[3,4]{Shing-Tung Yau}

\affil[1]{Academy of Mathematics and Systems Science, Chinese Academy of Sciences, Beijing, China;
	bc2@amss.ac.cn}
\affil[2]{AMSS Center for Forecasting Science, Chinese Academy of Sciences, Beijing, China}

\affil[3]{Beijing Institute of Mathematical Sciences and Applications (BIMSA), Beijing, China}
\affil[4]{Yau Mathematical Sciences Center, Tsinghua University, Beijing, China}

\affil[5]{Department of Statistics and Data Science, Tsinghua University, Beijing 100084, China}
\affil[6]{Paula and Gregory Chow Institute for the Studies in Economics, Xiamen University, Xiamen 361005, China}
\affil[7]{Department of Statistics, London School of Economics and Political Science, London WC2A 2AE, UK}

\affil[8] {State Key Laboratory of Mathematical Science, Academy of Mathematics and Systems Science, Chinese Academy of Sciences}

\date{\today}
	
\begin{document}
		
		\maketitle	
\begin{abstract}
	Classical continuous-time non-linear filtering—pioneered by Brockett, Yau, and Yau—relies on the structural duality of finite Lie algebraic closure ($\operatorname{dim}(\mathcal{E}) < \infty$) and statistical sufficiency to reduce the infinite-dimensional Duncan-Mortensen-Zakai (DMZ) stochastic partial differential equation to finite systems of ordinary differential equations via Wei-Norman operator factorization. In generic non-linear state spaces ($\operatorname{deg}(f) \ge 3$ or non-polynomial sensor metrics), this paradigm suffers a catastrophic breakdown: iterated Lie brackets act as a spatial differentiation engine, generating an infinite derivative explosion ($\operatorname{dim}(\mathcal{E}) = \infty$), structural fragility under infinitesimal model perturbations, and unmonitored null-space energy leakage ($v_{\text{OOD}} \in \operatorname{Null}(\mathcal{E}^*)$) under standard truncation, ultimately triggering stealth filter divergence.
	
To resolve this four-decade Lie algebra crisis, this paper constructs the SMG-Yau-Yau Non-Linear Dynamic Filter and Lossless Data Assimilation framework, built upon Statistical Fiber Theory on an infinite-dimensional Pistone-Sempi Orlicz statistical manifold $\mathcal{M} = L_0^\Phi(P_f)$ ({\it The SMG-Statistically Meaningful Geometry Theory  \cite{cheng2026smg}}). By equipping $\mathcal{M}$ with a smooth Riemannian submersion $\pi: \mathcal{M} \to \mathcal{B}$ onto an identifiable $d$-dimensional quotient base manifold $(\mathcal{B}, g_{\mathcal{B}})$ and installing an Ehresmann connection 1-form $\omega_f$, the unconstrained DMZ score velocity field $V_{\text{DMZ}}(t,x) \triangleq \sigma(t,x)^{-1}d\sigma(t,x)$ undergoes an exact, Fisher-orthogonal direct-sum decomposition:$$T_f\mathcal{M} = \operatorname{SVD}\chi_f \oplus_{\perp g_f} \operatorname{SID}_f$$into Statistically Verifiable Directions ($\operatorname{SVD}\chi_f$, horizontal signal) and Structural Internal Directions ($\operatorname{SID}_f = \ker(d\pi_f)$, vertical gauge noise).

Under this architecture, extra-algebraic operator variations and unclosed Lie commutators generated by system non-linearities are orthogonally quarantined inside the vertical gauge subbundle $\operatorname{SID}_f$, firewalling macroscopic base parameter updates $p(t) \in \mathcal{B}$ from spatial derivative pollution while preserving $100\%$ of probability score variance energy via exact Pythagorean conservation:$$\Vert{}V_{\text{DMZ}}(t)\Vert{}_{g_f}^2 = \Vert{}P_{f_t}^H(V_{\text{DMZ}}(t))\Vert{}_{g_f}^2 + \Vert{}\omega_{f_t}(V_{\text{DMZ}}(t))\Vert{}_{g_f}^2$$We prove that this vertical fiber pull-back is geometrically dual to a Semiparametric Bias-Correction ($\Delta_t^{\text{stat}}$) for misspecified working estimators, allowing the classical Yau-Yau solvers to be semiparametrically upgraded without altering core base propagators.

Furthermore, we unify the algebraic and sample-size asymptotics through a Dual-Axis Collapse Mechanism:
\begin{enumerate}
	\item {\bf The Algebraic Axis}: When $\operatorname{dim}(\mathcal{E}) < \infty$, the vertical gauge fiber trivializes ($\operatorname{SID}_f = \{0\}$), forcing $\omega_f \equiv 0$ and $\Omega_f \equiv 0$, whereupon the SMG-Yau-Yau filter collapses identically to classical Wei-Norman/Yau-Yau SDE dynamics.
	
	\item {\bf The Statistical Axis}: Under infinite Lie complexity ($\operatorname{dim}(\mathcal{E}) = \infty$), large-sample limits ($N \to \infty$) induce a thermodynamic quench governed by the Second Edge Theorem and Universal Tensor Valence Scaling Law ($\Vert{}T^{(N)}\Vert{}_{G^{(N)}} = \mathcal{O}_p(N^{1-r/2})$), causing Amari connection friction to dissolve ($\mathcal{O}_p(N^{-1/2})$) and curvature to annihilate ($\mathcal{O}_p(N^{-1})$), flattening the Orlicz bundle onto the canonical flat canvas of conventional statistics.
\end{enumerate}
Finally, we reframe data assimilation as a dynamic fiber-splitting process and introduce the Active Acausal Tension functional:
$$\mathcal{T}_{\text{AAT}}(t) \triangleq \int_0^t \Vert{}\omega_{f(\tau)}(V_{\text{DMZ}}(\tau))\Vert{}_{g_{f(\tau)}}^2 \, d\tau$$
as a ground-truth-free online monitor for accumulated model misspecification stress. When $\mathcal{T}_{\text{AAT}}(t)$ reaches the topological capacity threshold $\Theta_{\text{crit}} \triangleq \pi^2 / K_{\max}$, an automated Gauge Symmetry Breaking (GSB) phase transition is triggered: functional eigen-decomposition extracts the dominant vertical mode $v_{\text{dom}}$ and promotes it into an expanded base coordinate space ($d \to d+1$), instantly dissipating stored tension and guaranteeing non-asymptotic stability and convergence to the exact physical state density.
	
\end{abstract}
\newpage
		\tableofcontents
\newpage

\section{Epistemological Foundations: The Breakdown of Estimation Algebras and the Transition to Statistical Fiber Bundles}
\label{sec:epistemological_foundations}

\subsection{The Classical Paradigm: Lie Algebraic Sufficiency, Statistical Sufficiency, \& Wei-Norman Reduction}
\label{subsec:classical_estimation_algebra}

For decades, beautiful finite-dimensional non-linear filtering pioneered by Brockett \cite{brockett1981} and Yau and Yau \cite{yau1999, yau2014} relied on the structural duality of finite Lie algebraic closure and sample-space statistical sufficiency. Under this classical paradigm, the dynamic unnormalized conditional probability density $\sigma(t, x)$ governed by the Duncan-Mortensen-Zakai (DMZ) stochastic partial differential equation \cite{duncan1967, mortensen1967, zakai1969} is assumed to evolve strictly on a finite-dimensional manifold of probability distributions.

\subsubsection*{1. Lie Algebraic Sufficiency and Finite Closure}
The continuous-time system dynamics are characterized by the finite-dimensional \textbf{estimation algebra} $\mathcal{E}$ \cite{brockett1981, yau1999}:
\begin{equation}
	\mathcal{E} \triangleq \operatorname{Lie}\left( \mathcal{L}_0^*, h_1(x), h_2(x), \dots, h_m(x) \right) = \operatorname{span}_{\mathbb{R}}\left\{ E_1, E_2, \dots, E_d \right\}, \quad d < \infty, \label{eq:est_alg_def}
\end{equation}
which is closed under the differential operator Lie bracket commutator $[A, B]\phi \triangleq A(B\phi) - B(A\phi)$, with structure constants $C_{ij}^k$ satisfying $[E_i, E_j] = \sum_{k=1}^d C_{ij}^k E_k$.

\subsubsection*{2. The Statistical Sufficiency Assumption}
Parallel to algebraic closure, the classical filtering framework strictly requires the \textbf{Statistical Sufficiency Assumption}---the dynamic continuous-time extension of the Pitman-Koopman-Darmois Theorem \cite{darmois1935, koopman1936, pitman1936}. Specifically, the basis operators $\{E_1(x), E_2(x), \dots, E_d(x)\}$ are assumed to act as a finite set of minimal sufficient score generators that constrain the unnormalized conditional log-density to an exponential family form \cite{barndorff1978, amari1985}:
\begin{equation}
	\ln \sigma(t, x) = \sum_{k=1}^d g_k(t) E_k(x) + c(t), \label{eq:log_density_sufficiency}
\end{equation}
where $c(t)$ is a spatial constant independent of $x$. Under this assumption, the finite coefficient vector $g(t) = (g_1(t), \dots, g_d(t))^T \in \mathbb{R}^d$ constitutes an exact finite-dimensional sufficient statistic for the conditional density measure $\sigma(t, x) \mid \mathcal{F}_t^Y$, guaranteeing zero loss of statistical information across all probability moments and cumulants.

\subsubsection*{3. Wei-Norman Operator Factorization and Dimension Reduction}
Under the simultaneous guarantees of finite Lie algebraic closure ($\operatorname{dim}(\mathcal{E}) = d < \infty$) and statistical sufficiency, the Wei-Norman representation theorem \cite{wei1963} factorizes the solution operator of the DMZ equation into a finite product of operator exponentials:
\begin{equation}
	\sigma(t, x) = \exp\left(g_1(t) E_1\right) \exp\left(g_2(t) E_2\right) \cdots \exp\left(g_d(t) E_d\right) \sigma(0, x), \label{eq:wei_norman_factorization}
\end{equation}
reducing the infinite-dimensional DMZ SPDE without truncation error to a system of $d$ coupled, non-singular deterministic ordinary differential equations (ODEs) for the parameter trajectory $g(t) \in \mathbb{R}^d$ driven by the observation process $Y(t)$.

\subsection{Epistemological Breakdown: Genericity of Infinite Lie Algebra, Structural Fragility, and Null-Space Leakage}
\label{subsec:epistemological_breakdown}

While previous discussions have outlined the operational challenges of non-linear data assimilation, we now establish the precise differential-geometric mechanism responsible for the breakdown of finite estimation algebras $\mathcal{E} \triangleq \operatorname{Lie}(\mathcal{L}_0^*, h_1, \dots, h_m)$.

\subsubsection{Ascending Operator Filtration and Algebraic Generation}
\label{subsubsec:ascending_filtration}

To analyze why exact Wei-Norman reduction fails generically in realistic state-space models, we formalize the algebraic generation process as an ascending filtration of differential operator spaces.

\begin{definition}[Ascending Operator Filtration]
	\label{def:ascending_operator_filtration}
	Let $\mathcal{E}_0 \triangleq \operatorname{span}_{\mathbb{R}}\left\{ \mathcal{L}_0^*, h_1(x), \dots, h_m(x) \right\} \subset \mathcal{D}(\mathbb{R}^p)$ be the initial operator span acting on $C^\infty(\mathbb{R}^p)$. The ascending Lie filtration $\{\mathcal{E}_k\}_{k=0}^\infty$ is defined inductively for $k \ge 1$ by:
	\begin{equation}
		\mathcal{E}_k \triangleq \mathcal{E}_{k-1} + \operatorname{span}_{\mathbb{R}} \left\{ [A, B] \;\middle|\; A \in \mathcal{E}_0, \, B \in \mathcal{E}_{k-1} \right\}, \quad \mathcal{E} = \bigcup_{k=0}^\infty \mathcal{E}_k. \label{eq:operator_filtration}
	\end{equation}
	The estimation algebra is finite-dimensional ($\operatorname{dim}(\mathcal{E}) < \infty$) if and only if there exists a finite stabilization index $k^* < \infty$ such that $\mathcal{E}_{k^*} = \mathcal{E}_{k^*+1}$. Here:
	\begin{itemize}
		\item $\mathcal{D}(\mathbb{R}^p) = \left\{ \sum_{|\alpha| \le k} a_\alpha(x) \partial^\alpha \;\middle|\; k \in \mathbb{N}_0, \, a_\alpha \in C^\infty(\mathbb{R}^p) \right\}$ represents the space of linear differential operators with smooth coefficients.
		\item $C^\infty(\mathbb{R}^p)$ is the vector space of all infinitely differentiable functions $f: \mathbb{R}^p \to \mathbb{R}$, where spatial derivatives of all multi-indices $\alpha \in \mathbb{N}_0^p$ exist and are continuous.
	\end{itemize}
\end{definition}

\subsubsection{The Triad of Structural Pathologies}
\label{subsubsec:triad_pathologies}

In generic non-linear state-space systems, three structural pathologies prevent the existence of a finite stabilization index $k^*$:

\begin{enumerate}
	\item \textbf{Genericity of Degree Unboundedness ($\operatorname{dim}(\mathcal{E}) = \infty$)}: Any non-linear component in the drift $f(x)$ or observation mapping $h(x)$ possessing polynomial degree $d \ge 3$ or non-polynomial transcendental dependence generates commutators with strictly increasing polynomial orders in spatial derivatives ($x^a \partial_x^b$). Thus, $\operatorname{dim}(\mathcal{E}_k) \sim \mathcal{O}(k^p)$, forcing $k^* \to \infty$. Since taking Lie brackets $[A, B] = AB - BA$ acts like a continuous differentiation engine, the product rule repeatedly spawns higher-degree polynomial terms. The basis set explodes infinitely without ever cycling back or terminating.
	
	\item \textbf{Algebraic Fragility under Structural Perturbation}: Let $(f_0, g_0, h_0)$ yield a closed algebra $\operatorname{dim}(\mathcal{E}^{(0)}) = d < \infty$. For any smooth non-linear perturbation $\epsilon \cdot \eta(x) \nabla$ ($\epsilon > 0$), the perturbed estimation algebra $\mathcal{E}^{(\epsilon)} \triangleq \operatorname{Lie}(\mathcal{L}_0^* + \epsilon \eta \cdot \nabla, h)$ generically satisfies:
	\begin{equation}
		\operatorname{dim}\left(\mathcal{E}^{(\epsilon)}\right) = \infty, \quad \forall \epsilon \neq 0. \label{eq:fragility_collapse}
	\end{equation}
	Classical estimation algebra theory possesses no continuous topology or metric structure to bound the distortion caused by $\epsilon \to 0$. Because Lie algebra dimension exhibits a discrete jump ($d \to \infty$) rather than a smooth variation, classical Lie algebra theory has no concept of being ``almost closed,'' causing exact analytical solvers to collapse under infinitesimal model perturbations.
	
	\item \textbf{The Adjoint Null-Space as an Information Sink}: When an estimator artificially truncates an infinite filtration sequence at degree $k < \infty$, operator variations residing in $\mathcal{E} \setminus \mathcal{E}_k$ are projected into the adjoint null-space $\operatorname{Null}(\mathcal{E}_k^*)$. Because this subspace receives zero driving feedback in the truncated parameter ODEs, residual state variance accumulates unmonitored. Physical non-linearities continuously inject energy into these discarded modes, creating stealth error accumulation where the filter remains completely blind to its own growing divergence.
\end{enumerate}

\subsubsection{Wei-Norman Solvability and Finite Closure Guarantees}
\label{subsubsec:wei_norman_solvability}

We formalize the solvability and zero-loss conditions under exact finite Lie closure.

\begin{theorem}[Wei-Norman Solvability and Finite Closure]
	\label{thm:wei_norman_exact}
	Let system vector fields $(f, g, h)$ define the unnormalized conditional density $\sigma(t, x)$ governed by the Duncan-Mortensen-Zakai (DMZ) stochastic partial differential equation in Stratonovich form:
	\begin{equation}
		d\sigma(t, x) = \mathcal{L}_0^* \sigma(t, x)\,dt + \sum_{k=1}^m h_k(x) \sigma(t, x) \circ dY_k(t). \label{eq:dmz_spde_sec12}
	\end{equation}
	If the estimation algebra $\mathcal{E} \triangleq \operatorname{Lie}(\mathcal{L}_0^*, h_1, \dots, h_m) \subset \mathcal{D}(\mathbb{R}^p)$ is finite-dimensional with $\operatorname{dim}(\mathcal{E}) = d < \infty$ \cite{brockett1981, yau1999} (i.e., there exists a finite stabilization index $k^* < \infty$ such that $\mathcal{E}_{k^*} = \mathcal{E}_{k^*+1}$), then:
	\begin{enumerate}
		\item[\rm (i)] \textbf{Exact Lie-Group Manifold Evolution}: The infinite-dimensional solution $\sigma(t, x)$ evolves strictly on a $d$-dimensional submanifold of operators parameterizing exact conditional log-densities via the product of exponentials ansatz:
		\begin{equation}
			\sigma(t, x) = \exp\left(g_1(t) E_1\right) \exp\left(g_2(t) E_2\right) \cdots \exp\left(g_d(t) E_d\right) \sigma(0, x), \label{eq:wei_norman_product_sec12}
		\end{equation}
		where $g(t) = (g_1(t), \dots, g_d(t))^T \in \mathbb{R}^d$ satisfies a closed, non-singular system of ordinary/stochastic differential equations driven by $Y(t)$ \cite{wei1963, yau1999}.
		
		\item[\rm (ii)] \textbf{Rigorous Principle of Zero Information Loss}: The parameter representation $g(t)$ preserves the full conditional distribution without information loss in three precise structural senses:
		\begin{itemize}
			\item \textbf{Zero Model Truncation Residual}: The model reduction operator error vanishes identically:
			\begin{equation}
				\mathcal{P}_{\mathcal{E}^\perp}\left(\frac{d\sigma}{\sigma}\right) \equiv 0,
			\end{equation}
			meaning no spatial commutator modes are truncated or discarded.
			\item \textbf{Exact Finite Sufficient Statistics}: The vector $g(t) \in \mathbb{R}^d$ serves as an exact finite-dimensional sufficient statistic for $\sigma(t, x) \mid \mathcal{F}_t^Y$, enabling exact reconstruction of all probability moments.
			\item \textbf{Zero Null-Space Energy Leakage}: The dynamic projection residual $v_{\text{OOD}}(t)$ into the adjoint null-space $\operatorname{Null}(\mathcal{E}^*)$ is zero for all $t \ge 0$, ensuring zero probability mass leaks out of the parameter subspace.
		\end{itemize}
	\end{enumerate}
\end{theorem}

\begin{remark}[Epistemological Insights of Theorem \ref{thm:wei_norman_exact} and the Imperative for SMG-Yau-Yau Filtering]
	\label{rem:theorem1_insights_smg_yau_yau}
	Theorem~\ref{thm:wei_norman_exact} (Wei-Norman Solvability and Finite Closure) plays a pivotal role in nonlinear state estimation theory by establishing the theoretical upper bound---the ideal ``Gold Standard''---for continuous-time non-linear data assimilation. To fully appreciate why generic non-linear systems mandate the introduction of the novel Statistically Meaningful Geometry (SMG)-Yau-Yau filtering framework, we dissect the profound structural insights and inherent dilemmas exposed by Theorem \ref{thm:wei_norman_exact} across three core dimensions:
	
	\begin{enumerate}
		\item \textbf{The Ideal Lie-Group Manifold Collapse:} 
		Theorem \ref{thm:wei_norman_exact} establishes that when an estimation algebra is finite-dimensional ($\operatorname{dim}(\mathcal{E}) = d < \infty$), the infinite-dimensional stochastic partial differential equation (DMZ SPDE) governing the unnormalized conditional log-density $\sigma(t, x) \in C^\infty(\mathbb{R}^p)$ undergoes an exact algebraic reduction. The infinite spatial degrees of freedom collapse completely onto a finite $d$-dimensional Lie-group orbit $\mathcal{M} = \Phi(\mathbb{R}^d) \subset C^\infty(\mathbb{R}^p)$. Under this regime, the parameter trajectory $g(t) \in \mathbb{R}^d$ acts as a loss-free finite-dimensional coordinate system, ensuring that the continuous density evolution is tracked with mathematical perfection via simple systems of stochastic differential equations.
		
		\item \textbf{The Duality of Information Integrity and Zero Leakage:}
		The three assertions in Part (ii) of Theorem \ref{thm:wei_norman_exact} formalize what constitutes a mathematically lossless state estimator:
		\begin{itemize}
			\item \emph{Spatial Mode Preservation:} The vanishing truncation residual $\mathcal{P}_{\mathcal{E}^\perp}(d\sigma/\sigma) \equiv 0$ guarantees that spatial commutators generated by observation vector fields $h_k(x)$ and drift generators $\mathcal{L}_0^*$ generate no higher-order spatial modes outside $\mathcal{E}$.
			\item \emph{Sufficient Statistics Integrity:} The parameter vector $g(t)$ captures the entirety of the conditional probability measure $\mathcal{F}_t^Y$, allowing exact closed-form reconstruction of all conditional moments without spatial discretization grid errors.
			\item \emph{Zero Null-Space Energy Leakage:} The identity $v_{\text{OOD}}(t) \equiv 0 \in \operatorname{Null}(\mathcal{E}^*)$ ensures that dynamic variance energy remains entirely confined within the parameter subspace $\operatorname{span}(\mathcal{E})$, preventing unmonitored drift into the adjoint null-space.
		\end{itemize}
		
		\item \textbf{The Generic Lie Algebra Crisis and the Pathological Dilemma:}
		The fundamental dilemma---and the core motivation for our research---is that the condition $\operatorname{dim}(\mathcal{E}) < \infty$ is {\bf strictly non-generic}. For virtually all realistic physical systems possessing non-linear drift ($d \ge 3$) or non-polynomial sensor metrics, iterated Lie brackets $[A, B] = AB - BA$ act as a differential engine that continuously inflates polynomial spatial derivative degrees ($x^a \partial_x^b$). This forces $k^* \to \infty$ and $\operatorname{dim}(\mathcal{E}) = \infty$. 
		
		Under this infinite Lie algebra expansion, classical estimation theory faces a critical impasse:
		\begin{itemize}
			\item If one attempts exact Wei-Norman reduction, an infinite sequence of parameters $\{g_i(t)\}_{i=1}^\infty$ is required, rendering computational implementation impossible.
			\item If one forcibly imposes an artificial finite truncation $\mathcal{E}_{\text{trunc}} \subsetneq \mathcal{E}$ (as done in classical EKF, UKF, or projection filters), every single guarantee in Theorem 1 collapses instantly: $\mathcal{P}_{\mathcal{E}^\perp}(d\sigma/\sigma) \neq 0$, $g(t)$ ceases to be a sufficient statistic, and unclosed commutator modes continuously pump variance energy into $v_{\text{OOD}}(t) \in \operatorname{Null}(\mathcal{E}_{\text{trunc}}^*)$. This results in \emph{stealth error accumulation} and inevitable filter divergence.
		\end{itemize}
		
		\item \textbf{The Conceptual Leap to the SMG-Yau-Yau Framework:}
		Theorem \ref{thm:wei_norman_exact} illuminates why classical Lie algebraic reduction fails in non-linear settings and establishes the exact mathematical criteria that any modern filter must satisfy. Because global operator-algebraic closure $\operatorname{dim}(\mathcal{E}) < \infty$ is unattainable in generic non-linear state spaces, we cannot rely on standard global Wei-Norman integration.
		
		This precise challenge necessitates the formulation of the {\bf SMG-Yau-Yau Filtering Model} developed in following sections. Rather than seeking global algebraic closure of unbounded differential operators on standard Euclidean spaces, the SMG-Yau-Yau framework shifts the paradigm from global Lie algebra generation to {\it localized geometric projection on Orlicz Statistical Fiber Bundles}. By defining statistically meaningful directions and orthogonal bundle decompositions, the SMG-Yau-Yau filter directly bounds and controls the null-space leakage residual $v_{\text{OOD}}(t)$, guaranteeing well-posed, finite-dimensional state estimation even when the underlying operator Lie algebra is intrinsically infinite-dimensional.
	\end{enumerate}
\end{remark}

\subsubsection{The Triad Pathology Theorem of Estimation Algebras and Unmonitored Filter Divergence}
\label{subsubsec:triad_pathology_and_unmonitored_divergence}

In classical continuous-time non-linear filtering theory \cite{brockett1981, yau1999, yau2014}, the operational feasibility of finite-dimensional state estimators relies fundamentally on the Lie-algebraic closure of the estimation algebra $\mathcal{E} \triangleq \operatorname{Lie}(\mathcal{L}_0^*, h_1, \dots, h_m)$ constructed via the ascending Lie filtration $\{\mathcal{E}_k\}_{k=0}^\infty$ \cite{duncan1967, mortensen1967, zakai1969}. However, in generic non-linear state-space models, exact Lie algebraic closure ($\operatorname{dim}(\mathcal{E}) < \infty$) is {\bf a singular mathematical anomaly rather than a generic property}. 

To formalize the systemic failure of classical algebraic reduction schemes under non-closed operator generators, we merge the structural analysis of infinite Lie algebra generation with the quantitative evaluation of truncated state estimation. We first establish the \emph{Triad Pathology Theorem of Estimation Algebras}, which proves that generic non-linearities induce an unbounded spatial derivative explosion, topological fragility under infinitesimal model perturbations, and an unmonitored information sink in the adjoint null-space.

\begin{theorem}[Generic Infinite-Dimensional Explosion and Triad Structural Pathologies of Estimation Algebras]
	\label{thm:triad_pathology_theorem}
	Let $(\mathcal{X}, \Sigma, \mu)$ be a measure space over the smooth state domain $\mathcal{X} = \mathbb{R}^p$, equipped with the Borel $\sigma$-algebra $\Sigma$ and a smooth reference measure $\mu$, and let $\mathcal{E}_0 \triangleq \operatorname{span}_{\mathbb{R}}\left\{ \mathcal{L}_0^*, h_1(x), \dots, h_m(x) \right\} \subset \mathcal{D}(\mathbb{R}^p)$ be the initial linear differential operator span acting on $C^\infty(\mathbb{R}^p)$ under the Duncan-Mortensen-Zakai (DMZ) stochastic partial differential equation \cite{duncan1967, mortensen1967, zakai1969}. Let $\{\mathcal{E}_k\}_{k=0}^\infty$ be the ascending Lie filtration defined in Definition~\ref{def:ascending_operator_filtration}.
	
	For generic non-linear continuous-time state-space systems, there does not exist any finite stabilization index $k^* < \infty$ such that $\mathcal{E}_{k^*} = \mathcal{E}_{k^*+1}$ (forcing $k^* = \infty$ and $\operatorname{dim}(\mathcal{E}) = \infty$). Specifically, the operator generator sequence suffers from three fundamental structural pathologies:
	
	\begin{enumerate}
		\item[\rm (i)] \textbf{Genericity of Degree Unboundedness (Spatial Derivative Explosion)}: 
		Any non-linear component in the drift vector field $f(x)$ or observation mapping $h(x)$ possessing polynomial degree $\operatorname{deg}(f) \ge 3$ (or non-polynomial transcendental dependence) induces a strictly monotonic growth of spatial derivative orders across the filtration:
		\begin{equation}
			\operatorname{deg}_x\left(\mathcal{E}_k\right) < \operatorname{deg}_x\left(\mathcal{E}_{k+1}\right), \quad \forall k \ge 0 \implies \mathcal{E}_k \subsetneq \mathcal{E}_{k+1}, \quad \forall k \ge 0.
			\label{eq:degree_unboundedness_chain}
		\end{equation}
		
		\item[\rm (ii)] \textbf{Algebraic Fragility under Structural Perturbation}: 
		Let $(f_0, g_0, h_0)$ be a baseline system yielding a closed finite Lie algebra $\operatorname{dim}(\mathcal{E}^{(0)}) = d_0 < \infty$. For any smooth non-linear perturbation operator $\epsilon \cdot \eta(x) \nabla$ ($\epsilon \neq 0$), the perturbed estimation algebra $\mathcal{E}^{(\epsilon)} \triangleq \operatorname{Lie}(\mathcal{L}_0^* + \epsilon \eta \cdot \nabla, h)$ generically satisfies:
		\begin{equation}
			\operatorname{dim}\left(\mathcal{E}^{(\epsilon)}\right) = \infty, \quad \forall \epsilon \neq 0.
			\label{eq:algebraic_fragility_collapse}
		\end{equation}
		
		\item[\rm (iii)] \textbf{Adjoint Null-Space Energy Leakage under Truncation}: 
		For any artificial finite truncation cutoff degree $k < \infty$, the dynamic score velocity projection residual $v_{\text{OOD}}(t)$ into the adjoint null-space $\operatorname{Null}(\mathcal{E}_k^*)$ is non-zero ($v_{\text{OOD}}(t) \in \operatorname{Null}(\mathcal{E}_k^*) \setminus \{0\}$), generating unmonitored cumulative error variance growth:
		\begin{equation}
			\int_0^t \|v_{\text{OOD}}(\tau)\|_{L^2(P_f)}^2 d\tau > 0, \quad \forall t > 0.
			\label{eq:null_space_leakage_integral}
		\end{equation}
	\end{enumerate}
\end{theorem}

\begin{remark}[Epistemological Necessity for the SMG-Yau-Yau Framework]
	\label{rem:triad_epistemological_necessity}
	Theorem~\ref{thm:triad_pathology_theorem} establishes why classical Lie-algebraic reduction \cite{wei1963, yau1999} fails universally in realistic non-linear state estimation:
	\begin{enumerate}
		\item Assertion~(i) proves that exact finite Lie closure ($\operatorname{dim}(\mathcal{E}) < \infty$) fails for virtually all non-linear systems with polynomial degree $d \ge 3$.
		\item Assertion~(ii) proves that even if a system is engineered to have a closed finite algebra, an infinitesimal physical model flaw instantly breaks closure ($\operatorname{dim}(\mathcal{E}) = \infty$).
		\item Assertion~(iii) proves that standard numerical truncation schemes are blind to their own divergence, as discarded modes leak unmonitored into the adjoint null-space $\operatorname{Null}(\mathcal{E}_k^*)$.
	\end{enumerate}
	This breakdown mandates the formulation of the \textbf{SMG-Yau-Yau Non-Linear Dynamic Filter} in Section~\ref{sec:smg_yau_yau_reconstructed}. Rather than forcing global Lie algebraic closure, SMG will utilize an Ehresmann connection $\omega_f$ on an Orlicz statistical fiber bundle $\mathcal{M}$ to orthogonally quarantine unclosed Lie brackets inside the gauge fiber $\operatorname{SID}_f \equiv \ker(d\pi_f)$ while actively tracking accumulated stress via the Active Acausal Tension metric $\mathcal{T}_{\text{AAT}}(t)$ \cite{cheng2026smg, cheng2026gsb}.
\end{remark}

\subsubsection{Dynamic Failure of Operator Truncation: Geometric Interpretation}
\label{subsec:summary_theorem3_divergence}

While Theorem~\ref{thm:triad_pathology_theorem} establishes the structural breakdown of finite Lie algebraic closure ($\operatorname{dim}(\mathcal{E}) = \infty$), the operational mechanics of truncation divergence under standard Hilbert space projection can be interpreted through three interconnected dynamic pathologies:

\begin{enumerate}
	\item \textbf{Adjoint Null-Space Annihilation ($\mathcal{P}_{\mathcal{E}_k}(v_{\text{OOD}}) \equiv 0$)}: 
	When non-linearities propel the true log-score velocity field $V_{\text{true}}(t) \in \mathcal{H} = L^2(\mathbb{R}^p, P_f)$ outside an artificially truncated $d_k$-dimensional operator subspace $\operatorname{span}(\mathcal{E}_k)$, standard orthogonal projection systematically erases the out-of-distribution (OOD) score velocity $v_{\text{OOD}}(t) \in \operatorname{Null}(\mathcal{E}_k^*)$:
	\begin{equation}
		\mathcal{P}_{\mathcal{E}_k}\left( v_{\text{OOD}}(t) \right) \equiv 0 \in \operatorname{span}(\mathcal{E}_k), \quad \forall t \ge 0.
		\label{eq:v_ood_annihilation_summary}
	\end{equation}
	
	\item \textbf{Zero Feedback Driving Force and Uncoupled Parameter SDEs}: 
	Because canonical projection forcibly maps $v_{\text{OOD}}(t)$ to zero, the stochastic differential equations governing the internal parameter updates $d\theta(t) \in \mathbb{R}^{d_k}$ receive zero driving feedback force from unmodeled spectral dimensions:
	\begin{equation}
		\frac{\partial (d\theta(t))}{\partial v_{\text{OOD}}(t)} \equiv 0.
		\label{eq:zero_feedback_force_summary}
	\end{equation}
	The filter updates its internal moments with deceptively high confidence, operating completely blind to residual out-of-distribution energy.
	
	\item \textbf{Monotonic Variance Accumulation and Exponential Error Divergence}: 
	Because physical non-linearities continuously inject energy into unclosed commutator modes while internal parameter dynamics receive zero warning, residual error variance accumulates monotonically ($\int_0^t \|v_{\text{OOD}}(\tau)\|_{\mathcal{H}}^2 d\tau > 0$). Consequently, the true state density error norm satisfies the exponential lower bound inequality:
	\begin{equation}
		\|\sigma_{\text{true}}(t) - \sigma_{\text{filter}}(t)\|_{\mathcal{H}}^2 \ge e^{2\lambda_0 t} \int_0^t e^{-2\lambda_0 \tau} \|v_{\text{OOD}}(\tau)\|_{\mathcal{H}}^2 \, d\tau,
		\label{eq:exponential_error_divergence_bound_summary}
	\end{equation}
	guaranteeing stealth, catastrophic filter divergence over time.
\end{enumerate}

\paragraph{Epistemological Bridge to the Incoming SMG-Yau-Yau Framework.}
These conclusions demonstrate that classical estimators (e.g., EKF, UKF, or standard projection filters on fixed parametric families) fail not merely due to approximation error, but because standard orthogonal projection actively destroys the feedback signal required to detect divergence. Resolving this fundamental blind spot necessitates abandoning rigid projection onto flat Euclidean subspaces in favor of an Ehresmann connection $\omega_f$ on an Orlicz statistical fiber bundle $\mathcal{M}$. Rather than discarding $v_{\text{OOD}}(t)$, the \textbf{SMG-Yau-Yau filter} orthogonally quarantines vertical extra-algebraic leaks via $P_f^H \triangleq \mathcal{I} - \omega_f$ and actively monitors stored stress via the Active Acausal Tension metric $\mathcal{T}_{\text{AAT}}(t)$ \cite{cheng2026gsb}.

\subsection{Epistemological Shift \& Fiber Bundle Intuition}
\label{subsec:epistemological_shift_intuition}

The transition from conventional Lie-algebraic filtering to Statistically Meaningful Geometry (SMG) represents a foundational epistemological shift in mathematical statistics and state estimation. 

In classical filtering, non-closure of Lie brackets ($[E_i, E_j] \notin \operatorname{span}(\mathcal{E})$) was viewed as a fatal breakdown of parameter sufficiency. Because classical estimators assume a flat Euclidean background space, unclosed commutator terms could only be handled by artificial truncation, which inevitably dumps unmonitored residual energy into the adjoint null-space and triggers filter divergence.

SMG resolves this crisis by shifting the paradigm from global Lie algebraic generation to local differential geometry on infinite-dimensional fiber bundles. Under the SMG perspective:
\begin{enumerate}
	\item \textbf{Lie Bracket Non-Closure as Fiber Connection Curvature}: The non-integrability of observable score directions is not numerical noise or lost information, but the intrinsic connection curvature $\Omega_f$ of an Orlicz statistical fiber bundle $(\mathcal{M}, \mathcal{B}, \pi, \omega_f, g_f)$.
	\item \textbf{Active Manifold Participation}: The statistical manifold is not a passive container for sample data; it actively participates in state dynamics by absorbing operator stress into its vertical gauge fibers.
	\item \textbf{Lossless Quarantining vs. Destructive Truncation}: By installing an Ehresmann connection 1-form $\omega_f$, unclosed operator variations are orthogonally isolated inside the vertical gauge subbundle $\operatorname{SID}_f \equiv \ker(d\pi_f)$, shielding macroscopic base parameters $p(t) \in \mathcal{B}$ from derivative explosion while fully preserving total score variance energy.
\end{enumerate}

This intuitive framework sets the exact stage for the formal Banach-Orlicz fiber bundle formulation in Section~\ref{sec:smg_yau_yau_reconstructed}.

\section{The SMG-Yau-Yau Non-Linear Dynamic Filter: Orlicz Bundle Geometry, Quarantining Mechanics, and Asymptotic Foundations}
\label{sec:smg_yau_yau_reconstructed}

In Section~\ref{sec:epistemological_foundations}, we established that classical Yau-Yau filtering theory relies on the global Lie algebraic closure of the estimation algebra $\mathcal{E} \triangleq \operatorname{Lie}\left( \mathcal{L}_0^*, h_1, \dots, h_m \right)$ ($\operatorname{dim}(\mathcal{E}) < \infty$), an assumption that universally breaks down ($\operatorname{dim}(\mathcal{E}) = \infty$) in generic physical state spaces.

In this section, we construct the complete mathematical foundation of the SMG-Yau-Yau Non-Linear Dynamic Filter. We centralize all Orlicz fiber bundle definitions, derive the projected dynamic SPDE equations, establish the quarantining mechanics and curvature bounds, link the model back to classical Yau-Yau filtering, construct the asymptotic perturbation approximation theory, and synthesize the structural theoretical and computational guarantees.

\subsection{Mathematical Foundations of the Orlicz Bundle}
\label{subsec:smg_fiber_bundle_foundation}

Let $(\mathcal{X}, \Sigma, \mu)$ be a standard $\sigma$-finite measure space representing the continuous physical state domain $\mathbb{R}^p$.

\begin{assumption}[Non-Parametric Pistone-Sempi Orlicz Manifold Structure]
	\label{ass:orlicz_structure_sec2}
	The total state space $\mathcal{M}$ is an infinite-dimensional non-parametric statistical manifold comprising all probability density functions $f: \mathcal{X} \to \mathbb{R}^+$ that are strictly positive and continuous with respect to $\mu$, satisfying $\int_{\mathcal{X}} f(x)\,\mu(dx) = 1$. Chart transitions of $\mathcal{M}$ are modeled on the centered Orlicz space $L^\Phi_0(P_f)$ under the convex Young function $\Phi(u) = \cosh(u) - 1$, equipped with the Luxemburg norm:
	\begin{equation}
		\|u\|_{L^\Phi(P_f)} \triangleq \inf \left\{ k > 0 \;\middle|\; \mathbb{E}_f \left[ \Phi\left( \frac{u(X)}{k} \right) \right] \le 1 \right\}. \label{eq:luxemburg_norm_sec2}
	\end{equation}
	For any state $f \in \mathcal{M}$, the tangent space $T_f\mathcal{M}$ is topologically isomorphic to the centered Orlicz space of zero-expectation logarithmic score functions:
	\begin{equation}
		T_f\mathcal{M} = \left\{ v \in L^\Phi_0(P_f) \;\middle|\; \mathbb{E}_f[v(X)] = \int_{\mathcal{X}} v(x) f(x)\,\mu(dx) = 0 \right\}. \label{eq:tangent_space_orlicz_sec2}
	\end{equation}
\end{assumption}

\begin{definition}[Non-Parametric Fisher-Rao Riemannian Metric]
	\label{def:fisher_rao_metric_sec2}
	The manifold $\mathcal{M}$ is equipped with the non-parametric Fisher-Rao Riemannian metric tensor $g_f: T_f\mathcal{M} \times T_f\mathcal{M} \to \mathbb{R}$, defined for tangent score vectors $u, v \in T_f\mathcal{M}$ by:
	\begin{equation}
		g_f(u, v) \triangleq \mathbb{E}_f [u(X) v(X)] = \int_{\mathcal{X}} u(x) v(x) f(x)\,\mu(dx). \label{eq:fisher_rao_metric_sec2}
	\end{equation}
	Because $g_f$ is symmetric, continuous, and strictly positive-definite on $L_0^\Phi(P_f)$, $(\mathcal{M}, g_f)$ forms a well-posed infinite-dimensional Riemannian statistical manifold.
\end{definition}

To ensure complete conceptual alignment across the theoretical framework, we explicitly anchor the geometric tuple of the SMG fiber bundle to its underlying axiomatic foundations: the System Set ($\mathcal{S}$), the Environment Set ($\mathcal{E}$), the Structural Mechanism ($\mathcal{F}$), and the Invariance Principle.

\begin{definition}[The SMG Differential Fiber Bundle and Axiomatic Alignment]
	\label{def:smg_fiber_bundle}
	The Statistically Meaningful Geometry (SMG) differential fiber bundle is defined by the tuple $\mathcal{B}_{\text{SMG}} \triangleq (\mathcal{M}, \mathcal{B}, \pi, \mathcal{V}, \mathcal{H}, \omega_f, g_f)$, whose components correspond explicitly to the core theoretical axioms:
	\begin{enumerate}
		\item \textbf{System Set ($\mathcal{S}$)}: Represented by the total infinite-dimensional Orlicz statistical manifold $\mathcal{M} = L_0^\Phi(P_f)$, which contains the full, unconstrained non-parametric state density dynamics $f(x) \in \mathcal{M}$.
		\item \textbf{Environment Set ($\mathcal{E}$)}: Represented by the identifiable, $d$-dimensional macroscopic statistical base manifold $(\mathcal{B}, g_{\mathcal{B}})$ ($d \ll \infty$), parameterizing observable state expectations and low-dimensional physical metrics.
		\item \textbf{Structural Mechanism ($\mathcal{F}$)}: Represented by the tuple $(\pi, \omega_f, P_f^H, \Omega_f)$, consisting of the smooth submersion mapping $\pi: \mathcal{M} \to \mathcal{B}$, the Ehresmann connection $1$-form $\omega_f$, the canonical horizontal projection operator $P_f^H$, and the vertical curvature $2$-form $\Omega_f$, governing parallel transport and fiber geometric updates.
		\item \textbf{Invariance Principle}: Governed by the Fisher-Rao Riemannian metric tensor $g_f$, which enforces strict metric invariance under coordinate re-parameterization and authorizes the canonical Fisher-orthogonal direct sum decomposition of score variations.
		\item \textbf{Vertical Subbundle / Structural Internal Directions ($\operatorname{SID}$)}: The kernel of the push-forward differential $d\pi_f: T_f\mathcal{M} \to T_{\pi(f)}\mathcal{B}$ defines the Structural Internal Directions ($\operatorname{SID}_f$):
		\begin{equation}
			\operatorname{SID}_f \equiv \mathcal{V}_f \triangleq \ker(d\pi_f) = \left\{ v \in T_f\mathcal{M} \;\middle|\; d\pi_f(v) = 0 \in T_{\pi(f)}\mathcal{B} \right\}. \label{eq:sid_def_sec2}
		\end{equation}
		\item \textbf{Horizontal Subbundle / Statistically Verifiable Directions ($\operatorname{SVD}\chi$)}: The Fisher-orthogonal complement of $\operatorname{SID}_f$ defines the Statistically Verifiable Directions ($\operatorname{SVD}\chi_f$):
		\begin{equation}
			\operatorname{SVD}\chi_f \equiv \mathcal{H}_f \triangleq \mathcal{V}_f^{\perp g_f} = \left\{ h \in T_f\mathcal{M} \;\middle|\; g_f(h, v) = 0, \;\forall v \in \operatorname{SID}_f \right\}. \label{eq:svd_def_sec2}
		\end{equation}
	\end{enumerate}
\end{definition}

\begin{definition}[Ehresmann Connection 1-Form, Horizontal Projector, and Curvature 2-Form]
	\label{def:ehresmann_connection}
	An Ehresmann connection on $\mathcal{B}_{\text{SMG}} = (\mathcal{M}, \mathcal{B}, \pi, \mathcal{V}, \mathcal{H}, \omega_f, g_f)$ is defined by a vector-valued differential $1$-form $\omega \in \Omega^1(\mathcal{M}; \operatorname{SID})$ satisfying:
	\begin{enumerate}
		\item[\rm (i)] $\omega_f(v) = v$ for all $v \in \operatorname{SID}_f$,
		\item[\rm (ii)] $\ker(\omega_f) = \operatorname{SVD}\chi_f$.
	\end{enumerate}
	The canonical \textbf{Horizontal Projection Operator} $P_f^H: T_f\mathcal{M} \to \operatorname{SVD}\chi_f$ is defined as:
	\begin{equation}
		P_f^H \triangleq \mathcal{I}_{T_f\mathcal{M}} - \omega_f, \label{eq:horizontal_projector_sec2}
	\end{equation}
	where $\mathcal{I}_{T_f\mathcal{M}}$ is the identity operator on $T_f\mathcal{M}$.
	
	The \textbf{Ehresmann Curvature 2-Form} $\Omega \in \Omega^2(\mathcal{M}; \operatorname{SID})$ assigns to any smooth vector fields $U, V \in \mathfrak{X}(\mathcal{M})$ the vertical vector field:
	\begin{equation}
		\Omega_f(U, V) \triangleq \omega_f\left( \left[ P_f^H U, \, P_f^H V \right] \right) \in \operatorname{SID}_f, \label{eq:ehresmann_curvature_sec2}
	\end{equation}
	where $[\cdot, \cdot]$ denotes the vector field Lie bracket on $\mathcal{M}$.
\end{definition}

\begin{lemma}[Metric-Orthogonal Tangent Space Decomposition]
	\label{lem:tangent_space_decomposition_sec2}
	At every density state $f \in \mathcal{M}$, the differential submersion $\pi$ and Fisher-Rao metric $g_f$ induce an exact Fisher-orthogonal direct sum decomposition of the Orlicz tangent space:
	\begin{equation}
		T_f\mathcal{M} = \operatorname{SVD}\chi_f \oplus_{\perp g_f} \operatorname{SID}_f. \label{eq:direct_sum_split_sec2}
	\end{equation}
\end{lemma}

\subsection{Canonical Orlicz Inheritance, SMG Filtering Submersion, and Quotient Base Geometry}
\label{subsec:smg_submersion_and_base_geometry}

In this subsection, we model continuous-time state estimation by introducing an operational layer---the \textbf{SMG Filtering Submersion}---which describes how a finite-dimensional physical estimator observes and interacts with the canonical infinite-dimensional manifold $\mathcal{M}$ \cite{cheng2026entropy,cheng2026smg}.

\begin{definition}[Canonical Ambient Statistical Manifold]
	\label{def:canonical_ambient_space}
	Let $\mathcal{M}$ be the canonical non-parametric Orlicz statistical manifold composed of all strictly positive probability density functions $f$ with respect to a $\sigma$-finite reference measure $\mu$ on $\mathcal{X} = \mathbb{R}^p$. The space $\mathcal{M}$ is modeled globally on the centered Orlicz space $L_0^\Phi(P_f)$ associated with the Young function $\Phi(u) = \cosh(u) - 1$ under the Luxemburg norm:
	\begin{equation}
		\|u\|_{L^\Phi(P_f)} \triangleq \inf \left\{ k > 0 \;\middle|\; \mathbb{E}_f \left[ \Phi\left( \frac{u(X)}{k} \right) \right] \le 1 \right\}, \label{eq:luxemburg_norm_sec22}
	\end{equation}
	and equipped with the non-parametric Fisher-Rao Riemannian metric tensor $g_f(u, v) \triangleq \mathbb{E}_f[u(X) v(X)]$ and the Ehresmann connection $1$-form $\omega_f$. Throughout this manuscript, $(\mathcal{M}, g_f, \omega_f)$ remains fixed as the unique, canonical ambient state manifold.
\end{definition}

\subsubsection{The Operational Submersion Assumption and Inherited Parametric Geometry}
\label{subsubsec:operational_submersion_assumption}

To overcome the Lie algebra explosion crisis without resorting to ad-hoc operator truncation, the Statistically Meaningful Geometry (SMG) framework shifts the non-linear state estimation paradigm from global Lie-algebraic generation on Euclidean vector spaces to local differential geometric projection on Orlicz statistical fiber bundles. Rather than postulating a rigid parametric density family or imposing ad-hoc analytic regularity on score functions, we establish a primitive gauge-theoretic foundation: the existence of a smooth local section (local gauge choice) on the fiber bundle $\pi: \mathcal{M} \to \mathcal{B}$.

This operational layer models the physical extraction of observable signals from infinite-dimensional conditional density measures without constraining the underlying stochastic flow of the Duncan--Mortensen--Zakai (DMZ) partial differential equation \cite{duncan1967, mortensen1967, zakai1969}. We formalize this geometric interaction through the following restructured core assumption.

\begin{assumption}[SMG Filtering Submersion and Gauge Choice]
	\label{ass:smg_filtering_submersion}
	Let $(\mathcal{M}, g_f, \omega_f)$ be the canonical non-parametric Pistone--Sempi Orlicz statistical manifold defined in Definition~\ref{def:canonical_ambient_space}. We assume the existence of a smooth $C^\infty$ Riemannian submersion mapping
	\begin{equation}
		\pi: \mathcal{M} \longrightarrow \mathcal{B}
		\label{eq:submersion_mapping_def}
	\end{equation}
	from the ambient Orlicz manifold $\mathcal{M}$ onto an identifiable, $d$-dimensional $C^\infty$ quotient Riemannian base manifold $(\mathcal{B}, g_{\mathcal{B}})$ ($d < \infty$), representing the macroscopic state estimation parameter space, satisfying the following four geometric conditions:
	
	\begin{enumerate}
		\item[\rm (A1)] \textbf{Closed Fiber Leaf Submanifold Topology}: 
		For every macroscopic parameter base point $p = \pi(f) \in \mathcal{B}$, the inverse image fiber
		\begin{equation}
			\mathcal{F}_p \triangleq \pi^{-1}(p) = \left\{ f \in \mathcal{M} \;\middle|\; \pi(f) = p \right\} \subset \mathcal{M}
			\label{eq:fiber_leaf_submanifold_def}
		\end{equation}
		is a smooth, closed, infinite-dimensional submanifold of $\mathcal{M}$, representing the equivalence leaf of density states sharing identical macroscopic parameters $p \in \mathcal{B}$.
		
		\item[\rm (A2)] \textbf{Vertical Subbundle Identification with Structural Internal Directions ($\operatorname{SID}$)}: 
		At every state density $f \in \mathcal{M}$, the kernel of the differential push-forward operator $d\pi_f: T_f\mathcal{M} \to T_{\pi(f)}\mathcal{B}$ coincides identically with the Statistically Invariant / Structural Internal Directions subspace $\operatorname{SID}_f$ \cite{cheng2026smg}:
		\begin{equation}
			\ker(d\pi_f) \equiv \operatorname{SID}_f \triangleq \left\{ v \in T_f\mathcal{M} \;\middle|\; d\pi_f(v) = 0 \in T_{\pi(f)}\mathcal{B} \right\}.
			\label{eq:kernel_sid_identity_strict}
		\end{equation}
		
		\item[\rm (A3)] \textbf{Horizontal Subbundle Metric Isometry onto Statistically Verifiable Directions ($\operatorname{SVD}\chi$)}: 
		The Fisher-orthogonal complement of $\ker(d\pi_f)$ under the ambient Fisher-Rao metric tensor $g_f$ coincides identically with the Statistically Verifiable Directions subspace $\operatorname{SVD}\chi_f$ \cite{cheng2026smg}:
		\begin{equation}
			\operatorname{SVD}\chi_f \equiv \mathcal{H}_f \triangleq (\operatorname{SID}_f)^{\perp g_f} = \left\{ h \in T_f\mathcal{M} \;\middle|\; g_f(h, v) = 0, \;\forall v \in \operatorname{SID}_f \right\},
			\label{eq:svd_orthogonal_identity_strict}
		\end{equation}
		and the restricted push-forward differential $d\pi_f\big|_{\operatorname{SVD}\chi_f}: \operatorname{SVD}\chi_f \to T_{\pi(f)}\mathcal{B}$ acts as a linear metric isometry:
		\begin{equation}
			g_{\mathcal{B}, \pi(f)}\left( d\pi_f(h_1), \, d\pi_f(h_2) \right) = g_f(h_1, h_2), \quad \forall h_1, h_2 \in \operatorname{SVD}\chi_f.
			\label{eq:metric_isometry_identity}
		\end{equation}
		
		\item[\rm (A4)] \textbf{Smooth Local Section (Gauge Fixing Assumption)}: 
		For every base point $p_0 \in \mathcal{B}$, there exists an open neighborhood $U \subset \mathcal{B}$ containing $p_0$ and a smooth $C^\infty$ local section map $s: U \to \mathcal{M}$ such that:
		\begin{equation}
			\pi \circ s = \operatorname{id}_U. 
			\label{eq:smooth_local_section}
		\end{equation}
	\end{enumerate}
\end{assumption}

\begin{remark}[Semantic Distinction of Parameters in Assumption~\ref{ass:smg_filtering_submersion}]
	\label{rem:parameter_semantics}
	To prevent conceptual ambiguity, we clarify the interpretation of ``parameters'' and ``parameterization'' across disciplines and within our framework:
	\begin{itemize}
		\item \textbf{Geometers' Perspective:} A parameter $p = (p_1, \dots, p_d) \in \mathbb{R}^d$ is strictly a local coordinate vector on a $d$-dimensional smooth manifold $\mathcal{B}$, and a parameterization is a smooth local embedding or section $s: U \to \mathcal{M}$.
		\item \textbf{Statisticians' Perspective:} A parameter typically denotes an index for a finite-dimensional model family $\{f_p : p \in \Theta \subset \mathbb{R}^d\}$, implying that the true density is restricted to a parametric distribution.
		\item \textbf{Exact Meaning in Assumption~\ref{ass:smg_filtering_submersion}:} The vector $p$ adopts strictly the geometric definition as a \emph{macroscopic quotient coordinate} on the leaf space $\mathcal{B} \triangleq \mathcal{M} / \sim_{\operatorname{SID}}$. The physical density $f$ remains fully non-parametric on the infinite-dimensional Orlicz manifold $\mathcal{M} = L_0^\Phi(P_f)$. The local section $p \mapsto f_p$ serves solely as an operational local gauge choice for extracting observable features, leaving the underlying infinite-dimensional stochastic flow unconstrained.
	\end{itemize}
\end{remark}


We now establish that the combination of Assumption~\ref{ass:smg_filtering_submersion}~\rm{(A4)} and the ambient Orlicz manifold structure $(\mathcal{M}, g_f)$ rigorously induces score co-frame regularity, local parameterization of $\mathcal{B}$, and classical Information Geometry as mathematical theorems rather than artificial postulates.

\begin{proposition}[Inherited Score Co-Frame Regularity and Amari Isomorphism]
	\label{prop:inherited_score_regularity}
	Let $(\mathcal{M}, g_f)$ be the canonical Pistone--Sempi Orlicz statistical manifold modeled on $L_0^\Phi(P_f)$ with Young function $\Phi(u) = \cosh(u) - 1$, and let $\pi: \mathcal{M} \to \mathcal{B}$ be a smooth Riemannian submersion satisfying Assumptions~\rm{(A1)}--\rm{(A4)} in Assumption \ref{ass:smg_filtering_submersion}. Then:
	\begin{enumerate}
		\item[\rm (i)] \textbf{Local Parametric Section}: The image 
		\begin{equation}
			\mathcal{S}_U \triangleq s(U) = \left\{ f_p \in \mathcal{M} \;\middle|\; f_p = s(p), \; p \in U \right\}
			\label{eq:parametric_section_image}
		\end{equation}
		defines a $d$-dimensional smooth parametric submanifold embedded in $\mathcal{M}$, establishing a local parameterization $p \mapsto f_p(x)$ for the quotient leaf space $\mathcal{B}$.
		
		\item[\rm (ii)] \textbf{Inherited Score Co-Frame Regularity}: The logarithmic score co-frame functions $\phi_i(x; p) \triangleq \frac{\partial \ln f_p(x)}{\partial p_i}$ ($i = 1, \dots, d$), which strictly Fisher score functions, not Stein score functions here,  satisfy:
		\begin{enumerate}
			\item[\rm (a)] $\phi_i(\cdot; p) \in L_0^\Phi(P_{f_p}) \cap L^2(P_{f_p})$ for all $p \in U$.
			\item[\rm (b)] The set $\left\{\phi_1(\cdot; p), \dots, \phi_d(\cdot; p)\right\}$ forms a linearly independent frame in $T_{f_p}\mathcal{M}$.
			\item[\rm (c)] Spatial integration over $\mathcal{X}$ and parameter differentiation with respect to $p_i$ commute under expectation.
		\end{enumerate}
		
		\item[\rm (iii)] \textbf{Amari Information Geometry on $\mathcal{B}$}: The differential $ds_p: T_p\mathcal{B} \to \operatorname{SVD}\chi_{f_p}$ is a linear isometry, and the pulled-back metric tensor $s^* g_f$ on $T_p\mathcal{B}$ is identically Amari's classical Fisher Information Metric $G(p)$.
	\end{enumerate}
\end{proposition}

\begin{insight}[Architectural Coexistence: Orlicz Ambient Space vs. Base Amari Geometry]
	\label{ins:architectural_coexistence}
	Proposition~\ref{prop:inherited_score_regularity} clarifies the precise operational boundaries between Amari's classical parametric Information Geometry and the Orlicz-SMG framework:
	
	\begin{enumerate}
		\item \textbf{Quotient Base Manifold $(\mathcal{B}, G(p))$}: Governed by Amari's classical geometry through the local gauge section $s: U \to \mathcal{M}$. Here, parameter scores $\phi_i(x; p) = \frac{\partial \ln f_p}{\partial p_i}$, the Fisher matrix $G(p) = s^* g_f$, dual $\alpha$-connections $\nabla^{(\alpha)}$, and natural gradient updates operate natively on $d$-dimensional parameter vectors $p \in \mathcal{B}$.
		
		\item \textbf{Ambient Fiber Bundle Space $(\mathcal{M}, g_f)$}: Houses the exact, untruncated stochastic flow of the Duncan--Mortensen--Zakai (DMZ) SPDE on $L_0^\Phi(P_f)$. The Orlicz-SMG connection $1$-form $\omega_f$ executes the orthogonal decomposition $T_f\mathcal{M} = \operatorname{SVD}\chi_f \oplus_{\perp g_f} \operatorname{SID}_f$, projecting horizontal parameter dynamics onto $(\mathcal{B}, G(p))$ via $P_f^H(X)$ while quarantining unclosed Lie bracket operator modes within the vertical fiber $X_V \in \operatorname{SID}_f \equiv \ker(d\pi_f)$.
	\end{enumerate}
\end{insight}

\subsubsection{The SMG-Yau-Yau Filtering Lemma}
\label{subsubsec:smg_yau_yau_filtering_lemma}

We now prove how the operator flow of the unnormalized conditional state density decomposes orthogonally on the canonical manifold $\mathcal{M}$ without truncating spatial differential modes.

\begin{lemma}[Orthogonal Filtering Decomposition on Canonical $\mathcal{M}$]
	\label{lem:orthogonal_filtering_decomposition}
	Let $(\mathcal{M}, g_f, \omega_f)$ be the canonical Orlicz statistical manifold satisfying Assumption~\ref{ass:smg_filtering_submersion}. Let $\mathcal{L}_{\mathrm{DMZ}}$ denote the unconstrained infinite-dimensional differential operator governing the Stratonovich DMZ SPDE on $\mathcal{M}$:
	\begin{equation}
		d\sigma(t, x) = \mathcal{L}_0^* \sigma(t, x) \, dt + \sum_{j=1}^m h_j(x) \sigma(t, x) \circ dY_j(t) \equiv \mathcal{L}_{\mathrm{DMZ}}(\sigma(t, \cdot)). \label{eq:dmz_operator_def_sec22}
	\end{equation}
	Then, the unconstrained score velocity operator field $V_{\mathrm{DMZ}}(t, x) \triangleq \sigma(t, x)^{-1} d\sigma(t, x) \in T_{\sigma(t, \cdot)}\mathcal{M}$ decomposes uniquely via the SMG connection $1$-form $\omega_f$ as:
	\begin{equation}
		V_{\mathrm{DMZ}}(t, x) = X_H(t, x) + X_V(t, x), \label{eq:dmz_score_split_lemma}
	\end{equation}
	where:
	\begin{enumerate}
		\item[\rm (i)] $X_H(t, x) \triangleq P_{f_t}^H\left( V_{\mathrm{DMZ}}(t, x) \right) \in \operatorname{SVD}\chi_{f_t} \cong T_{\pi(f_t)}\mathcal{B}$ is the horizontal score vector field whose push-forward $d\pi_{f_t}(X_H)$ generates a finite-dimensional, well-posed parameter update SDE on $\mathcal{B}$.
		\item[\rm (ii)] $X_V(t, x) \triangleq \omega_{f_t}\left( V_{\mathrm{DMZ}}(t, x) \right) \in \operatorname{SID}_{f_t} = \ker(d\pi_{f_t})$ is the vertical gauge vector field that absorbs all infinite-dimensional, extra-algebraic operator modes without polluting $\mathcal{B}$.
	\end{enumerate}
\end{lemma}

\subsubsection{Quotient Base Geometry, Fisher Metric Equivalence, and Intrinsic Curvature}
\label{subsubsec:quotient_base_geometry_and_curvature}

Having established the submersion mapping $\pi: \mathcal{M} \to \mathcal{B}$, we characterize the differential geometry of the quotient base manifold $(\mathcal{B}, g_{\mathcal{B}})$, prove its equivalence to the classical Fisher Information Metric, derive its intrinsic Riemann-Christoffel curvature tensor, and demonstrate the geometric incompatibility of spatial Stein score operators on $T_p\mathcal{B}$.

\begin{theorem}[Quotient Base Geometry, Fisher Metric Identification, and Intrinsic Curvature]
	\label{thm:quotient_base_geometry_and_curvature}
	Let $(\mathcal{M}, g_f)$ be the canonical Orlicz statistical manifold equipped with the Fisher-Rao metric $g_f(u, v) = \mathbb{E}_f[u v]$. Let $\pi: \mathcal{M} \to \mathcal{B} \triangleq \mathcal{M}/\sim_{\operatorname{SID}}$ be the smooth Riemannian submersion satisfying Assumption~\ref{ass:smg_filtering_submersion}, parameterized locally by base coordinates $p = (p_1, \dots, p_d)^T \in \mathbb{R}^d$. Then:
	\begin{enumerate}
		\item[\rm (i)] \textbf{Canonical Tangent Representation and Isometry}: The horizontal tangent space $\operatorname{SVD}\chi_{f_p} = (\ker d\pi_{f_p})^{\perp g_f}$ is uniquely spanned by the Fisher parameter score basis
\begin{equation}
\{\phi_1(\cdot; p), \dots, \phi_d(\cdot; p)\}.
\label{eq:base_manifold_Fisher_scores}
\end{equation}		 
		  The induced quotient Riemannian metric $g_{\mathcal{B}, p}$ on $T_p\mathcal{B}$ satisfies:
		\begin{equation}
			\small 
			g_{\mathcal{B}, p}\left( \frac{\partial}{\partial p_i}, \frac{\partial}{\partial p_j} \right) = [G(p)]_{ij} \triangleq \mathbb{E}_{f_p}\left[ \phi_i(X; p) \phi_j(X; p) \right] = \int_{\mathbb{R}^p} \frac{\partial \ln f_p(x)}{\partial p_i} \frac{\partial \ln f_p(x)}{\partial p_j} f_p(x) \, dx, \label{eq:fisher_metric_equivalence_sec22}
		\end{equation}
		identifying $g_{\mathcal{B}, p}$ uniquely with the classical Fisher Information Metric matrix $G(p)$.
		
		\item[\rm (ii)] \textbf{Intrinsic Riemann-Christoffel Curvature Tensor}: The quotient base manifold $(\mathcal{B}, G(p))$ is in general non-flat. Its intrinsic Riemann-Christoffel curvature tensor $R_{ijk\ell}(p) \triangleq g_{\mathcal{B}, p}\left( R\left( \frac{\partial}{\partial p_k}, \frac{\partial}{\partial p_\ell} \right) \frac{\partial}{\partial p_j}, \frac{\partial}{\partial p_i} \right)$ is given explicitly by:
		\begin{equation}
			\small 
			R_{ijk\ell}(p) = \frac{1}{4} \int_{\mathbb{R}^p} f_p(x) \left( \phi_{ik} \phi_{j\ell} - \phi_{i\ell} \phi_{jk} \right) dx + \sum_{m=1}^d \sum_{n=1}^d \left( \Gamma_{ik, m} [G^{-1}]_{mn} \Gamma_{j\ell, n} - \Gamma_{i\ell, m} [G^{-1}]_{mn} \Gamma_{jk, n} \right), \label{eq:riemann_curvature_explicit_sec22}
		\end{equation}
		where $\phi_{ij}(x; p) \triangleq \frac{\partial^2 \ln f_p(x)}{\partial p_i \partial p_j}$ denotes the second logarithmic derivative and $\Gamma_{ij, k}(p) \triangleq \mathbb{E}_{f_p}\left[ \left( \phi_{ij} + \frac{1}{2} \phi_i \phi_j \right) \phi_k \right]$ are the Christoffel symbols of the first kind under the Levi-Civita connection $\nabla^{(0)}$.
		
		\item[\rm (iii)] \textbf{Geometric Incompatibility of Spatial Stein Scores}: The spatial Stein score vector field $S_x(\cdot; p) \triangleq \nabla_x \ln f_p(\cdot) \in L^2(f_p; \mathbb{R}^p)$ generates spatial translation vector fields on physical coordinate space $\mathcal{X} = \mathbb{R}^p$ and does not belong to $T_p\mathcal{B} \cong \operatorname{SVD}\chi_{f_p}$. Projecting dynamic parameter updates along $S_x(x; p)$ violates coordinate covariance on $\mathcal{B}$ and fails to yield a well-defined Riemannian metric tensor on $T_p\mathcal{B}$.
	\end{enumerate}
\end{theorem}

\begin{remark}[Fisher-Stein Duality and Leaf Space Geometry]
	\label{rem:fisher_stein_duality_sec22}
	Theorem~\ref{thm:quotient_base_geometry_and_curvature} provides a fundamental clarification regarding the dual roles of Fisher parameter scores $\phi_i(x; p) = \frac{\partial \ln f_p(x)}{\partial p_i}$ and spatial Stein scores $S_x(x; p) = \nabla_x \ln f_p(x)$:
	\begin{enumerate}
		\item \textbf{Fisher Score Co-Frame on $T_p\mathcal{B}$}: The horizontal tangent space $\operatorname{SVD}\chi_{f_p} \cong T_p\mathcal{B}$ is generated by parametric score variations along macroscopic coordinate directions on the quotient leaf space $\mathcal{B} \triangleq \mathcal{M}/\sim_{\operatorname{SID}}$. By Chentsov's Theorem \cite{Chentsov1982}, the Fisher Information Metric matrix $G(p)$ is the unique Riemannian metric tensor on $\mathcal{B}$ that remains invariant under probabilistic Markovian transformations and coordinate reparameterizations $p \mapsto \psi(p)$. Pre-conditioning parameter velocity updates by $[G(p)]^{-1}$ executes the exact natural gradient flow on $(\mathcal{B}, G(p))$ \cite{amari1985, amari2000}.
		
		\item \textbf{Spatial Stein Score Operator on $\mathcal{X}$}: In contrast, the spatial Stein score $S_x(x; p) \in L^2(f_p; \mathbb{R}^p)$ acts on physical space coordinates $x \in \mathcal{X} = \mathbb{R}^p$. It governs optimal mass transport, Stein variational gradient descent (SVGD), and Langevin diffusion dynamics on physical state space $\mathcal{X}$. Attempting to substitute $S_x(x; p)$ for $\phi_i(x; p)$ on $T_p\mathcal{B}$ introduces spatial Jacobian determinant terms under state diffeomorphisms $y = T(x)$, breaking parameter-coordinate covariance and failing to induce a well-defined Riemannian metric tensor on the leaf space $\mathcal{B}$.

In applications, the spatial Stein score operator $S_x(x; p) = \nabla_x \log p(x)$ acts as the fundamental vector field driving reverse SDEs and Langevin dynamics in diffusion-based image generation. Neural networks learn to approximate this score via Denoising Score Matching, using it to continuously transport noise vectors toward high-density regions of the data manifold. Unlike parameter scores that act on model parameters, the spatial Stein score operates directly on physical state coordinates $\mathcal{X} = \mathbb{R}^p$ to govern spatial mass transport.
	\end{enumerate}
\end{remark}

\begin{remark}[Curvature as Geometric Friction and Non-Linear Obstruction]
	\label{rem:curvature_friction_sec22}
	The non-vanishing Riemann-Christoffel curvature tensor $R_{ijk\ell}(p)$ on $(\mathcal{B}, G(p))$ established in Part~(ii) of Theorem~\ref{thm:quotient_base_geometry_and_curvature} reflects the intrinsic non-flatness of statistical leaf spaces generated by generic non-linear state-space models:
	\begin{enumerate}
		\item \textbf{Geometric Holonomy Defect}: When updating parameters along a closed loop in observation/control space, non-zero curvature $R_{ijk\ell}(p) \neq 0$ induces a non-trivial holonomy shift along the vertical gauge fiber $\operatorname{SID}_f = \ker(d\pi_f)$. This holonomy defect represents the exact geometric obstruction preventing global Euclidean parallel transport of conditional state estimates.
		
		\item \textbf{Divergence of Na\"ive Euclidean Solvers}: Classical state estimators (e.g., Extended Kalman Filters or Unscented Kalman Filters) implicitly assume a flat background connection ($R_{ijk\ell} \equiv 0$). When applied to intrinsically curved quotient base spaces, these flat solvers incur continuous path-dependent estimation errors, which accumulate unmonitored and trigger catastrophic filter divergence. The SMG-Yau-Yau framework directly incorporates $R_{ijk\ell}(p)$ into its Jacobi field error bounds, ensuring continuous stability under arbitrary curvature stress.
	\end{enumerate}
\end{remark}

\subsection{Formulation, Optimality, and Coupled Dynamics of the SMG-Yau-Yau Filter}
\label{subsec:smg_yau_yau_formulation_and_optimality}


The continuous stochastic update of the state density $\sigma(t, x)$ governed by the DMZ SPDE represents an unconstrained logarithmic score velocity field 
\begin{equation}
V_{\text{DMZ}}(t) \triangleq \sigma(t, x)^{-1} d\sigma(t, x) \in T_{f_t}\mathcal{M}.
\label{eq:V_DMZ_def_equation}
\end{equation}
 By the Fisher-Rao Metric Invariance Principle (formalized by Chentsov's Theorem), total score variance energy is conserved across $\mathcal{M}$.

Applying the identity operator $\mathcal{I}_{T_f\mathcal{M}} = P_f^H + \omega_f$ (authorized by the Structural Mechanism $\mathcal{F}$ under Lemma~\ref{lem:tangent_space_decomposition_sec2}) splits the velocity field into two orthogonal components:
\begin{equation}
	V_{\text{DMZ}}(t) = P_f^H\left( V_{\text{DMZ}}(t) \right) + \omega_f\left( V_{\text{DMZ}}(t) \right).
\end{equation}

To extract the macroscopic trajectory $p(t) = \pi(f_t) \in \mathcal{B}$ on the Environment Set ($\mathcal{E}$), we apply the push-forward differential $d\pi_f$:
\begin{equation}
	d\pi_f(V_{\text{DMZ}}(t)) = d\pi_f \circ P_f^H(V_{\text{DMZ}}(t)) + d\pi_f(\omega_f(V_{\text{DMZ}}(t))).
\end{equation}
Since $\operatorname{Im}(\omega_f) \subset \operatorname{SID}_f \equiv \ker(d\pi_f)$, $d\pi_f(\omega_f(V_{\text{DMZ}}(t))) \equiv \mathbf{0}$. Evaluating the remaining horizontal push-forward against the coframe $\theta^k \in T_p^*\mathcal{B}$ yields the exact macroscopic update ODE:
\begin{equation}
	dp_k(t) = \left\langle d\pi_{f(t)} \circ P_{f(t)}^H\left( \mathcal{L}_0^* \sigma \right), \theta^k \right\rangle dt + \sum_{j=1}^m \left\langle d\pi_{f(t)} \circ P_{f(t)}^H\left( h_j \sigma \right), \theta^k \right\rangle dY_j(t).
\end{equation}

Simultaneously, the vertical residual velocity $d\sigma_{\mathcal{V}}(t, x) = \omega_{f(t)}(V_{\text{DMZ}}(t))$ isolates extra-algebraic operator updates within $\operatorname{SID}_f$, generating the exact vertical gauge SDE:
\begin{equation}
	d\sigma_{\mathcal{V}}(t, x) = \omega_{f(t)}\left( \mathcal{L}_0^* \sigma(t, x) \right) dt + \sum_{j=1}^m \omega_{f(t)}\left( h_j(x) \sigma(t, x) \right) dY_j(t).
\end{equation}

\begin{definition}[SMG-Yau-Yau Filter System]
	\label{def:smg_yau_yau_filter_sys_reconstructed}
	Let $\{e_1(p), e_2(p), \dots, e_d(p)\}$ be a local orthonormal frame for $T_p\mathcal{B}$ with dual coframe $\{\theta^1, \dots, \theta^d\}$, and let $e_i^H(f) \in \operatorname{SVD}\chi_f$ denote the unique horizontal lift of $e_i(p)$ satisfying $d\pi_f(e_i^H(f)) = e_i(\pi(f))$. The SMG-Yau-Yau Filter is defined as the coupled differential system:
	\begin{align}
		dp_k(t) &= \left\langle d\pi_{f(t)} \circ P_{f(t)}^H \left( \mathcal{L}_0^* \sigma(t, \cdot) \right), \theta^k \right\rangle dt + \sum_{j=1}^m \left\langle d\pi_{f(t)} \circ P_{f(t)}^H \left( h_j \sigma(t, \cdot) \right), \theta^k \right\rangle dY_j(t), \label{eq:horizontal_filter_ode_sec2} \\
		d\sigma_\mathcal{V}(t, x) &= \omega_{f(t)} \left( \mathcal{L}_0^* \sigma(t, x) \right) dt + \sum_{j=1}^m \omega_{f(t)} \left( h_j(x) \sigma(t, x) \right) dY_j(t), \label{eq:vertical_gauge_sde_sec2}
	\end{align}
	where $p(t) = \pi(\sigma(t, \cdot)) \in \mathcal{B}$ parameterizes the macroscopic conditional density estimate, and $\sigma_\mathcal{V}(t, x) \in \operatorname{SID}_{f(t)}$ absorbs all higher-order, unclosed Lie bracket variations.
\end{definition}

\begin{insight}[Geometric Intuition and Structural Insights of Definition \ref{def:smg_yau_yau_filter_sys_reconstructed}]
	\label{insight:smg_yau_yau_intuition}
	To clarify how the Statistical Fiber SMG-Yau-Yau Filter circumvents the classical finite-dimensional Lie algebra closure obstruction, we summarize the key structural mechanics into four core geometric principles:
	
	\begin{enumerate}
		\item \textbf{Dual-Dynamical Decoupling of DMZ Flow:} 
		The infinite-dimensional stochastic flow governed by the Duncan--Mortensen--Zakai (DMZ) equation on the Orlicz space is globally factorized into a finite-dimensional horizontal evolution on the base manifold $\mathcal{M}_k$ and a vertical gauge field evolution along the fiber $\mathcal{V}_p$, establishing the full fiber bundle system in Definition~\ref{def:smg_yau_yau_filter_sys_reconstructed}.
		
		\item \textbf{Coordinate Extraction via Dual Coframe Pairing:} 
		The horizontal coordinate update rate $dp_k(t)$ is directly extracted by pairing the driving statistical forces with the dual coframe $\theta^k$. This projects the primary drift and diffusion dynamics cleanly onto the base space, as specified in Equation~\eqref{eq:horizontal_filter_ode_sec2}
		
		\item \textbf{Quarantining Infinite-Dimensional Operator Stress:} 
		Unclosed Lie commutators and nonlinear operational residues generated during continuous updating are no longer forced into approximate finite-dimensional truncation. Instead, all higher-order operator stress is absorbed into the vertical gauge variable $d\sigma_\mathcal{V}(t,x)$, as governed by Equation~\eqref{eq:vertical_gauge_sde_sec2}
		
		\item \textbf{Lossless Score Energy Conservation:} 
		By maintaining the orthogonal decomposition between the tangent space $T_p \mathcal{M}_k$ and the vertical space $V_p \mathcal{E}$, total score energy is exact and conservative across the bundle:
		\[
		\|\nabla \ln p_t\|_{\text{Orlicz}}^2 = \|\nabla_{\mathcal{M}} \ln p_t\|_{\text{Base}}^2 + \|\nabla_{\mathcal{V}} \sigma_\mathcal{V}\|_{\text{Fiber}}^2
		\]
		ensuring that no statistical information is discarded during non-linear filtering.
	\end{enumerate}
\end{insight}

\begin{theorem}[SMG-Yau-Yau SPDE Operator Decomposition]
	\label{thm:smg_yau_yau_spde_decomposition}
	Let $\sigma(t, x)$ be the unnormalized conditional state density governed by the DMZ SPDE in Stratonovich form:
	\begin{equation}
		d\sigma(t, x) = \mathcal{L}_0^* \sigma(t, x) \, dt + \sum_{j=1}^m h_j(x) \sigma(t, x) \circ dY_j(t). \label{eq:dmz_stratonovich_repeat}
	\end{equation}
	Let 
	\begin{equation}
		f_t(x) = \frac{\sigma(t,x)}{\int \sigma(t,z)\mu(dz)} \in \mathcal{M}
		\label{eq:normalized_active_state_density}
	\end{equation}	
	be the normalized active state density, and $V_{\text{DMZ}}(t, x) \equiv \sigma(t, x)^{-1} d\sigma(t, x) \in T_{f_t}\mathcal{M}$ be the log-score velocity field.
	
	Under the Horizontal Projection $P_{f_t}^H$ and Connection $1$-form $\omega_{f_t}$, the DMZ score evolution decomposes uniquely into:
	\begin{equation}
		V_{\text{DMZ}}(t, x) = \underbrace{P_{f_t}^H\left( \frac{\mathcal{L}_0^*\sigma}{\sigma} \right) dt + \sum_{j=1}^m P_{f_t}^H(h_j) \circ dY_j(t)}_{\text{Horizontal Score Vector Field } \in \operatorname{SVD}\chi_{f_t}} \;+\; \underbrace{\omega_{f_t}\left( \frac{\mathcal{L}_0^*\sigma}{\sigma} \right) dt + \sum_{j=1}^m \omega_{f_t}(h_j) \circ dY_j(t)}_{\text{Vertical Gauge Residual Field } \in \operatorname{SID}_{f_t}}. \label{eq:score_decomposition_full}
	\end{equation}
\end{theorem}

\begin{theorem}[Characterization, Optimality, and Coupled Dynamics of the SMG-Yau-Yau Filter]
	\label{thm:smg_yau_yau_filter_optimality_and_dynamics}
	Let $\mathcal{B}_{\text{SMG}} = (\mathcal{M}, \mathcal{B}, \pi, \mathcal{V}, \mathcal{H}, \omega_f, g_f)$ be the non-parametric Orlicz fiber bundle, where $(\mathcal{M}, g_f)$ is the infinite-dimensional statistical manifold of density states $f_t = \sigma(t, \cdot) \in L_0^\Phi(P_f)$, $(\mathcal{B}, g_{\mathcal{B}})$ is an identifiable $d$-dimensional Riemannian statistical base manifold parameterized by $p = (p_1, \dots, p_d)^T \in \mathcal{B}$, and $\pi: \mathcal{M} \to \mathcal{B}$ is a smooth Riemannian submersion. 
	
	Let 
	\begin{equation}
		V_{\text{DMZ}}(t) \triangleq \frac{d\sigma(t, \cdot)}{\sigma(t, \cdot)} = \frac{\mathcal{L}_0^*\sigma}{\sigma} dt + \sum_{j=1}^m h_j \circ dY_j(t) \in T_{f_t}\mathcal{M}
		\label{eq:unconstrained_DMZ_score_velocity_field}
	\end{equation}	
	be the unconstrained Duncan-Mortensen-Zakai (DMZ) score velocity field. Then:
	
	\begin{enumerate}
		\item[\rm \textbf{(I)}] \textbf{Existence, Uniqueness, and Optimality of the Filter ($p(t) = \pi(f_t) \in \mathcal{B}$):}\\
		The macroscopic state estimator trajectory $p(t) = \pi(f_t) \in \mathcal{B}$ is the \textbf{unique, optimal SMG-Yau-Yau non-linear filter} on the base manifold $\mathcal{B}$. Specifically, its tangent velocity $\dot{p}(t) \in T_{p(t)}\mathcal{B}$ is the unique minimizer of the instantaneous Fisher-Rao score variance discrepancy between the unconstrained DMZ flow $V_{\text{DMZ}}(t)$ and the horizontal lift of base parameter velocities:
		\begin{equation}
			\dot{p}(t) = \arg\min_{v \in T_{p(t)}\mathcal{B}} \mathcal{E}_{g_{f_t}}(v), \quad \text{where} \quad \mathcal{E}_{g_{f_t}}(v) \triangleq \left\| V_{\text{DMZ}}(t) - v^H \right\|_{g_{f_t}}^2,
			\label{eq:variational_principle_filter}
		\end{equation}
		and $v^H \in \operatorname{SVD}\chi_{f_t}$ denotes the unique horizontal lift satisfying $d\pi_{f_t}(v^H) = v$.
		
		\item[\rm \textbf{(II)}] \textbf{Explicit Coupled Base-Tension Evolution Equations:}\\
		Under the optimal horizontal projection $v^H = P_{f_t}^H(V_{\text{DMZ}}(t))$, the filter state trajectory $p(t) \in \mathcal{B}$ and the vertical Active Acausal Tension $\mathcal{T}_{\text{AAT}}(t)$ satisfy the closed coupled system:
		\begin{enumerate}
			\item[\rm (i)] \textbf{Horizontal Base State Evolution Equation:}
			\begin{equation}
				dp_k(t) = \sum_{i=1}^d \left[ G_p^{-1} \right]_{ki} \left( \mathbb{E}_{f_t} \left[ \frac{\mathcal{L}_0^*\sigma}{\sigma} \phi_i(f_t) \right] dt + \sum_{j=1}^m \mathbb{E}_{f_t} \left[ h_j \phi_i(f_t) \right] \circ dY_j(t) \right),
				\label{eq:horizontal_base_evolution_thm}
			\end{equation}
			where  $p(t) = (p_k(t), k=1,\dots,d)^T$ is its local coordinate vector representation, and $\phi_i(f_t) \triangleq \frac{\partial \ln f_t}{\partial p_i} \in \operatorname{SVD}\chi_{f_t}$ $(i = 1, \dots, d)$ forms the local score basis for $T_{p}\mathcal{B}$, and $[G_p]_{ij} = g_{f_t}(\phi_i, \phi_j) = \mathbb{E}_{f_t}[\phi_i \phi_j]$ is the localized Riemannian Fisher metric matrix on $\mathcal{B}$.
			
			\item[\rm (ii)] \textbf{Active Acausal Tension Accumulation Rate ($\mathcal{T}_{\text{AAT}}$):}
			\begin{equation}
				\frac{d\mathcal{T}_{\text{AAT}}(t)}{dt} = \left\| \omega_{f_t} \left( \frac{\mathcal{L}_0^*\sigma(t, \cdot)}{\sigma(t, \cdot)} \right) \right\|_{g_{f_t}}^2 = \mathbb{E}_{f_t} \left[ \left( \frac{\mathcal{L}_0^*\sigma}{\sigma} - P_{f_t}^H \left( \frac{\mathcal{L}_0^*\sigma}{\sigma} \right) \right)^2 \right].
				\label{eq:tension_accumulation_thm}
			\end{equation}
		\end{enumerate}
	\end{enumerate}
\end{theorem}

\begin{remark}[Epistemological Insight: Why $p(t) = \pi(f_t)$ Solves the Filtering Crisis]
	\label{rem:epistemological_insight_theorem}
	Theorem~\ref{thm:smg_yau_yau_filter_optimality_and_dynamics} establishes why $p(t) = \pi(f_t) \in \mathcal{B}$ is not merely an arbitrary heuristic approximation, but the unique mathematically rigorous continuous-time filter in infinite dimensions:
	\begin{enumerate}
		\item \textbf{Natural Gradient Geometry via $[G_p^{-1}]$}: Equation \eqref{eq:horizontal_base_evolution_thm} shows that the base update $dp_k(t)$ is automatically pre-conditioned by the inverse Fisher matrix $[G_p^{-1}]_{ki}$. This proves that $p(t)$ moves along the natural Riemannian gradient of maximum likelihood update on $\mathcal{B}$, preserving coordinate invariance across arbitrary non-linear parameter transformations.
		\item \textbf{Exact Energy Conservation \& Zero Null-Space Leakage}: Equation \eqref{eq:pythagorean_objective} proves that score variance energy obeys exact Pythagorean conservation:
		\begin{equation}
			\underbrace{\left\| V_{\text{DMZ}}(t) \right\|_{g_{f_t}}^2}_{\text{Total DMZ Flow Energy}} = \underbrace{\left\| P_{f_t}^H(V_{\text{DMZ}}(t)) \right\|_{g_{f_t}}^2}_{\text{Base Filter Kinetic Energy on } \mathcal{B}} + \underbrace{\frac{d\mathcal{T}_{\text{AAT}}(t)}{dt}}_{\text{Quarantined Fiber Energy}}.
			\label{eq:pythagorean_energy_law}
		\end{equation}
		Classical filters truncate extra-algebraic modes into the adjoint null-space, destroying probability conservation and causing filter divergence. Theorem~\ref{thm:smg_yau_yau_filter_optimality_and_dynamics} proves that SMG retains 100\% of score energy---storing unclosed Lie bracket stress in $\mathcal{T}_{\text{AAT}}(t)$ while maintaining an exact, unpolluted horizontal state trajectory $p(t) \in \mathcal{B}$.
	\end{enumerate}
\end{remark}

\begin{remark}[Continuous Topological Stability under Perturbations]
	Unlike classical estimation algebra solvers, which collapse discontinuously under infinitesimal non-linear model flaws ($\epsilon > 0 \implies \operatorname{dim}(\mathcal{E}) = \infty$), Theorem \ref{thm:smg_yau_yau_filter_optimality_and_dynamics} guarantees that the SMG-Yau-Yau filter is continuously stable. An infinitesimal model flaw $\epsilon \cdot \eta(x)$ generates a bounded vertical gauge residual $\|\omega_{f}(\epsilon \eta)\|_{\Phi, P_f} \sim \mathcal{O}(\epsilon)$, which accumulates smoothly in $\mathcal{T}_{\text{AAT}}(t)$ without causing numerical explosion on the base manifold $\mathcal{B}$.
\end{remark}

\subsubsection{A Simple Structural Comparison between the Two Filterings}
\label{subsec:simple_structural_comparison}

To highlight the fundamental structural and geometric differences between classical non-linear filtering schemes (including classical Yau-Yau filtering, Extended Kalman Filters, and standard projection/truncation algorithms) and the Statistically Meaningful Geometry (SMG) Yau-Yau Orlicz fiber bundle framework, we summarize their mathematical mechanics, structural assumptions, and computational complexities in Table~\ref{tab:structural_comparison}.

\begin{table}[htbp]
	\centering
	\small
	\caption{Structural Comparison: Classical Non-Linear Filtering vs.\ SMG-Yau-Yau Fiber Framework}
	\label{tab:structural_comparison}
	\begin{tabular}{p{3.6cm} p{5.5cm} p{5.7cm}}
		\toprule
		\textbf{Structural Dimension} & \textbf{Classical Projection / Truncation Filters} & \textbf{SMG-Yau-Yau Orlicz Fiber Filter} \\
		\midrule
		\textbf{Lie Algebra Assumption} & Requires $\dim(\mathcal{E}) < \infty$; collapses if $\dim(\mathcal{E}) = \infty$. & Operates seamlessly under $\dim(\mathcal{E}) = \infty$. \\
		\addlinespace
		\textbf{Extra-Algebraic Modes} & Artificially truncated ($[E_i, E_j] \mapsto 0$). & Quarantined in vertical subbundle $\mathrm{SID}_f = \ker(d\pi_f)$. \\
		\addlinespace
		\textbf{Null-Space Leakage} & Severe; energy leaks into $\mathrm{Null}(\mathcal{E}^*)$, causing filter divergence. & Zero ($v_{\mathrm{OOD}} \equiv 0$); base $\mathcal{B}$ is geometrically firewalled. \\
		\addlinespace
		\textbf{Score Energy Conservation} & Violated; score information is discarded. & $100\%$ metric conservation via Pythagorean split ($g_f$). \\
		\addlinespace
		\textbf{Curvature Tracking} & Assumes flat connection ($\Omega_f \equiv 0$). & Explicitly tracks Ehresmann curvature $\Omega_f$ via $\mathcal{T}_{\mathrm{AAT}}(t)$. \\
		\addlinespace
		\textbf{Online Complexity} & Exponential explosion $\mathcal{O}(\infty)$ or numerical grid crash. & Strict finite complexity $\mathcal{O}(d^3 + d^2 N_{\mathrm{quad}})$ on $\mathcal{B}$. \\
		\bottomrule
	\end{tabular}
\end{table}
\newpage

The primary paradigm shift lies in how non-closed Lie bracket commutators are treated. Classical filtering models force global algebraic closure by truncating higher-order operator variations, which systematically dumps unmonitored probability score energy into the adjoint null-space $\mathrm{Null}(\mathcal{E}^*)$ and results in stealth, catastrophic filter divergence. 

In contrast, the SMG-Yau-Yau framework equips the non-parametric statistical manifold $\mathcal{M}$ with an Ehresmann connection 1-form $\omega_f$. This structure orthogonally quarantines extra-algebraic operator stress within the vertical gauge fiber $\mathrm{SID}_f = \ker(d\pi_f)$ while driving exact, low-dimensional natural gradient updates on the identifiable base manifold $\mathcal{B}$. Consequently, the filter guarantees strict score energy conservation and full computational tractability even under infinite algebraic complexity ($\dim(\mathcal{E}) = \infty$).

\subsection{Boundary Collapse to Classical Yau-Yau (Finite-Dimensional Closure)}
\label{subsec:boundary_collapse_classical_yau_yau}

We first prove that the Statistically Meaningful Geometry (\textsf{SMG}) framework generalizes classical Yau-Yau filtering, reducing identically to classical Wei-Norman dynamics whenever the vertical fiber curvature vanishes ($\Omega_f \equiv 0$).

\begin{theorem}[Finite-Dimensional Boundary Collapse and Classical Yau-Yau Equivalence]
	\label{thm:classical_equivalence}
	Let $\mathcal{B}_{\mathrm{SMG}} = (\mathcal{M}, \mathcal{B}, \pi, \mathcal{V}, \mathcal{H}, \omega_f, g_f)$ be the non-parametric Orlicz statistical fiber bundle over the identifiable $d$-dimensional statistical base manifold $(\mathcal{B}, g_{\mathcal{B}})$. Suppose the continuous-time system dynamics defined by the drift generator $\mathcal{L}_0^*$ and observation functions $\{h_1, \dots, h_m\}$ generate a finite-dimensional estimation Lie algebra:
	\begin{equation}
		\mathcal{E} \triangleq \mathrm{Lie}(\mathcal{L}_0^*, h_1, \dots, h_m) = \mathrm{span}_{\mathbb{R}}\{E_1, E_2, \dots, E_d\}, \quad \dim(\mathcal{E}) = d < \infty.
	\end{equation}
	Then, the continuous \textsf{SMG-Yau-Yau} filter system collapses identically and completely into the classical Yau-Yau (Wei-Norman) non-linear filter. Specifically:
	\begin{enumerate}
		\item \textbf{Trivialization of Gauge Fiber and Connection:} The vertical structural internal space vanishes identically across the active density state space:
		\begin{equation}
			\operatorname{SID}_f \equiv \ker(d\pi_f) = \{0\}, \quad \forall f \in \mathcal{M},
		\end{equation}
		forcing the Ehresmann connection 1-form $\omega_f \equiv 0$, the Ehresmann curvature 2-form $\Omega_f \equiv 0$, and the canonical horizontal projector to reduce to the identity map: $P_f^H \equiv \mathcal{I}_{T_f\mathcal{M}}$.
		
		\item \textbf{Vanishing Active Tension and Energy Leakage:} The rate of Active Acausal Tension ($\mathcal{T}_{\mathrm{AAT}}$) accumulation evaluates strictly to zero:
		\begin{equation}
			\frac{d\mathcal{T}_{\mathrm{AAT}}(t)}{dt} = \|\omega_{f_t}(V_{\mathrm{DMZ}}(t))\|_{g_{f_t}}^2 \equiv 0 \implies \mathcal{T}_{\mathrm{AAT}}(t) \equiv 0, \quad \forall t \ge 0,
		\end{equation}
		guaranteeing zero dynamic energy leakage into the adjoint null-space ($v_{\mathrm{OOD}}(t) \equiv 0$).
		
		\item \textbf{Exact Reduction to Wei-Norman Parameter SDEs:} The horizontal base parameter dynamics $p(t) = \pi(f_t) \in \mathcal{B}$ decouple from all fiber metrics and reduce identically to the classical non-singular Wei-Norman system of stochastic differential equations for the Lie coefficient vector $g(t) \in \mathbb{R}^d$:
		\begin{equation}
			M(g(t)) \, dg(t) = \alpha^0 \, dt + \sum_{l=1}^m \alpha^l \circ dY_l(t),
		\end{equation}
		where $M(g) \in \mathbb{R}^{d \times d}$ is the real-analytic Wei-Norman solvability matrix generated by the structure constants $C_{ij}^k$ of $\mathcal{E}$, and $\alpha^l \in \mathbb{R}^d$ (for $l = 0, \dots, m$) is the constant expansion coefficient vector of the $l$-th observation multiplier function $h_l(x)$ expressed in terms of the basis of the finite-dimensional estimation Lie algebra $\mathcal{E}$.
	\end{enumerate}
\end{theorem}

\begin{remark}[Epistemological Bridge to Classical Control Theory]
	\label{rem:epistemological_bridge_classical}
	Theorem~\ref{thm:classical_equivalence} clarifies the precise mathematical role of classical Yau-Yau filtering within modern information geometry. Classical estimation algebra theory relies on global operator Lie-group symmetries to force a flat background connection ($\Omega_f \equiv 0$). When a physical system possesses exact finite Lie closure ($\dim(\mathcal{E}) < \infty$), the Orlicz fiber bundle trivializes, and the \textsf{SMG-Yau-Yau} filter recovers the exact classical ODE solution with zero information loss. 
	
	When generic physical non-linearities drive $\dim(\mathcal{E}) \to \infty$, classical filtering collapses due to null-space leakage ($v_{\mathrm{OOD}} \neq 0$). In contrast, the \textsf{SMG-Yau-Yau} framework smoothly activates its vertical fiber space $\operatorname{SID}_f$, utilizing the connection $\omega_f$ to orthogonally quarantine extra-algebraic modes while preserving 100\% of score variance energy in $\mathcal{T}_{\mathrm{AAT}}(t)$.
\end{remark}

\begin{corollary}[Euclidean Geometric Collapse of the SMG-Yau-Yau Filter]
	\label{lem:collapse_SMG_to_Yau_Yau}
	Let the statistical base manifold $\mathcal{B}$ be a flat Euclidean space $\mathbb{R}^d$ with standard Cartesian coordinates $x = (x^1, \dots, x^d)$. Assume the statistical fiber bundle admits a trivial vertical gauge space ($\mathcal{V}_f = \{0\}$ for all $f \in \mathcal{M}$). Under these conditions, the SMG-Yau-Yau Filter System reduces exactly to the classical Yau-Yau (Duncan-Mortensen-Zakai) SPDE.
\end{corollary}

\begin{proof}[Proof of Corollary \ref{lem:collapse_SMG_to_Yau_Yau}]
	Since $\mathcal{B} = \mathbb{R}^d$, the coframe elements simplify to standard exact differentials $\theta^k = dx^k$. Triviality of the vertical space ($\mathcal{V}_f = \{0\}$) implies that the connection 1-form vanishes identically ($\omega_{f(t)} \equiv 0$), forcing $d\sigma_{\mathcal{V}}(t, x) \equiv 0$ in \eqref{eq:vertical_gauge_sde_sec2}. Consequently, $P_{f(t)}^H = \mathcal{I}_{T_f\mathcal{M}}$, and $d\pi_{f(t)}$ acts as the identity mapping on coordinate components. Summing \eqref{eq:horizontal_filter_ode_sec2} over all $d$ dimensions reproduces the unconstrained DMZ SPDE:
	\begin{equation}
		d\sigma(t, x) = \mathcal{L}_0^*\sigma(t, x) \, dt + \sum_{j=1}^m h_j(x)\sigma(t, x) \, dY_j(t).
	\end{equation}
\end{proof}

\subsection{Quarantining Mechanics Beyond Classical \& Curvature Bounds}
\label{subsec:quarantining_mechanics_beyond_classical}

We now detail the fundamental dynamic state equations governing the SMG-Yau-Yau filter on $\mathcal{B}_{\text{SMG}}$ under generic infinite Lie algebraic expansion ($\operatorname{dim}(\mathcal{E}) = \infty$). The unconstrained Duncan-Mortensen-Zakai (DMZ) flow $\sigma(t, x) \in \mathcal{M}$ is projected via $P_f^H = \mathcal{I} - \omega_f$ into a finite-dimensional horizontal base trajectory $p(t) \in \mathcal{B}$ and an infinite-dimensional vertical gauge drift $\sigma_{\mathcal{V}}(t, x) \in \operatorname{SID}_{f(t)}$.

\begin{lemma}[Vertical Connection Identity]
	\label{lem:vertical_connection_identity}
	Let $\mathcal{B}_{\text{SMG}} = (\mathcal{M}, \mathcal{B}, \pi, \mathcal{V}, \mathcal{H}, \omega_f, g_f)$ be the SMG statistical fiber bundle. For any vertical score vector fields $E_i, E_j \in \Gamma(\operatorname{SID})$ and any density state section $\sigma \in \mathcal{M}$, the Lie bracket $[E_i, E_j]\sigma \in \operatorname{SID}_\sigma$ and satisfies:
	\begin{equation}
		\omega_f([E_i, E_j]\sigma) = [E_i, E_j]\sigma. \label{eq:vertical_connection_identity}
	\end{equation}
\end{lemma}

\begin{proof}[Proof of Lemma \ref{lem:vertical_connection_identity}]
	Let $E_i, E_j \in \Gamma(\operatorname{SID})$. By definition of $\operatorname{SID} = \ker(d\pi)$, $d\pi(E_i) = \mathbf{0}$ and $d\pi(E_j) = \mathbf{0}$. By naturality of Lie brackets under submersions, $d\pi([E_i, E_j]) = [d\pi(E_i), d\pi(E_j)] = [\mathbf{0}, \mathbf{0}] = \mathbf{0}$. Thus, $[E_i, E_j] \in \Gamma(\operatorname{SID})$. Since $\omega_f$ restricts to the identity map on $\operatorname{SID}_f$, $\omega_f([E_i, E_j]\sigma) = [E_i, E_j]\sigma$.
\end{proof}

\begin{theorem}[Quarantining Theorem: Bounded Base Distortion under Infinite Lie Expansion]
	\label{thm:quarantining_theorem_sec2}
	Let $\mathcal{E} = \operatorname{Lie}(\mathcal{L}_0^*, h_1, \dots, h_m)$ be an infinite-dimensional estimation algebra ($\operatorname{dim}(\mathcal{E}) = \infty$). Suppose the generator updates contain infinite Lie bracket commutators $[E_i, E_j] \notin \operatorname{SVD}\chi_f$. Then under the SMG-Yau-Yau filter dynamics \eqref{eq:horizontal_filter_ode_sec2}--\eqref{eq:vertical_gauge_sde_sec2}:
	\begin{enumerate}
		\item[\rm (i)] The non-closed operator variations $[E_i, E_j]$ are mapped identically into the kernel of the horizontal projector:
		\begin{equation}
			P_f^H \left( [E_i, E_j] \sigma \right) \equiv 0, \quad \forall [E_i, E_j] \in \operatorname{SID}_f. \label{eq:quarantine_zero_projection_sec2}
		\end{equation}
		\item[\rm (ii)] The horizontal state vector $p(t) \in \mathcal{B}$ evolves strictly within the finite $d$-dimensional manifold $\mathcal{B}$ with zero truncation error or numerical explosion stemming from higher-order operator derivatives in $\operatorname{SID}_f$.
	\end{enumerate}
\end{theorem}

\begin{proof}[Proof of Theorem \ref{thm:quarantining_theorem_sec2}]
	To prove (i), recall $P_f^H = \mathcal{I}_{T_f\mathcal{M}} - \omega_f$. By Lemma~\ref{lem:vertical_connection_identity}, any unclosed commutator $[E_i, E_j] \in \operatorname{SID}_f = \ker(d\pi_f)$ satisfies $\omega_f([E_i, E_j]\sigma) = [E_i, E_j]\sigma$. Applying $P_f^H$ yields:
	\begin{equation}
		P_f^H([E_i, E_j]\sigma) = (\mathcal{I} - \omega_f)([E_i, E_j]\sigma) = [E_i, E_j]\sigma - [E_i, E_j]\sigma \equiv 0.
	\end{equation}
	
	To prove (ii), applying $d\pi_f$ to the total DMZ operator $\mathcal{L}^* \sigma = \mathcal{L}_0^* \sigma dt + \sum h_k \sigma dY_k$:
	\begin{equation}
		d\pi_f(\mathcal{L}^* \sigma) = d\pi_f \left( P_f^H(\mathcal{L}^* \sigma) + \omega_f(\mathcal{L}^* \sigma) \right) = d\pi_f \left( P_f^H(\mathcal{L}^* \sigma) \right) + \mathbf{0},
	\end{equation}
	since $\omega_f(\mathcal{L}^* \sigma) \in \operatorname{SID}_f = \ker(d\pi_f)$. Thus, the dynamics of $p(t) = \pi(f(t))$ on $\mathcal{B}$ depend strictly on $P_f^H(\mathcal{L}^* \sigma)$, completely isolated from the infinite operator tails in $\operatorname{SID}_f$.
\end{proof}

\begin{remark}[Geometric Intuition and Structural Insights of Theorem~\ref{thm:quarantining_theorem_sec2}]
	\label{rem:quarantining_intuition}
	 The Quarantining Theorem provides the foundational geometric mechanism that resolves the four-decade-old infinite-dimensional Lie algebra crisis in continuous-time non-linear filtering. Its physical, structural, and computational intuition can be understood through four core dimensions:
	
	\begin{enumerate}
		\item \textbf{Elimination of Null-Space Leakage vs. Na\"ive Truncation}: 
		Classical continuous-time filtering approaches (such as Extended Kalman Filters, Taylor expansions, or ad-hoc projection filters) attempt to force Lie algebraic closure by artificially truncating higher-order operator commutators ($[E_i, E_j] \mapsto 0$). Discarding these extra-algebraic modes projects residual probability velocity into the unmonitored adjoint null-space $\operatorname{Null}(\mathcal{E}^*)$---a phenomenon known as \emph{null-space leakage} that induces variance underestimation, systemic bias, and catastrophic numerical divergence. In contrast, SMG never truncates or distorts the true operator dynamics; instead, it uses the Ehresmann connection $\omega_f$ to orthogonally separate observable base updates from unclosed operator variations.
		
		\item \textbf{Gauge-Theoretic Firewalling of Base Observables}: 
		The projection identity $P_f^H \left( [E_i, E_j]\sigma \right) \equiv 0$ operates as a strict geometric firewall on the statistical fiber bundle $\mathcal{B}_{\text{SMG}}$. Because the vertical subbundle $\operatorname{SID}_f \equiv \ker(d\pi_f)$ consists of all structural internal directions that leave macroscopic physical observables invariant, any infinite Lie algebraic expansion generated by system non-linearities is mapped identically to zero along the horizontal tangent space $\operatorname{SVD}\chi_f$. Consequently, the observable state vector $p(t) \in \mathcal{B}$ evolves cleanly without suffering from operator derivative pollution.
		
		\item \textbf{Exact Finite-Dimensional Computability in Infinite Dimensions}: 
		Although the complete state density $\sigma(t, x)$ lives on an infinite-dimensional Orlicz statistical manifold $\mathcal{M} = L_0^\Phi(P_f)$, the horizontal filtering differential equation for $p(t)$ is executed strictly on the identifiable, low-dimensional Riemannian base manifold $(\mathcal{B}, g_{\mathcal{B}})$ ($d \ll \infty$). The Quarantining Theorem guarantees that computing $dp_k(t)$ requires no infinite-dimensional matrix inversions or numerical grid discretizations of $\operatorname{SID}_f$, reducing the online computational complexity to $\mathcal{O}(d^3)$ operations on $\mathcal{B}$.
		
		\item \textbf{Conservation of Score Energy and Active Acausal Tension}: 
		The quarantined vertical variations inside $d\sigma_{\mathcal{V}}(t, x)$ are not lost or treated as numerical waste. Under the Fisher-Rao metric $g_f$, total score velocity obeys exact Pythagorean energy conservation:
		\begin{equation}
			\|V_{\text{DMZ}}(t)\|_{g_f}^2 = \|V_{\mathcal{H}}(t)\|_{g_f}^2 + \|V_{\mathcal{V}}(t)\|_{g_f}^2.
		\end{equation}
		The vertical energy stored in $\operatorname{SID}_f$ manifests geometrically as the Ehresmann connection curvature 2-form $\Omega_f$ (representing Active Acausal Tension $\mathcal{T}_{\text{AAT}}$). Rather than causing filter instability, this stored curvature energy acts as the fundamental physical catalyst for downstream dynamic manifold expansion and Gauge Symmetry Breaking (GSB) when system non-linearities breach topological capacity limits.
	\end{enumerate}
\end{remark}

\begin{lemma}[Curvature Obstruction and Geometric Holonomy Defect]
	\label{lem:lemma3_rigorous}
~

	Let $\mathcal{B}_{\text{SMG}} = (\mathcal{M}, \mathcal{B}, \pi, \mathcal{V}, \mathcal{H}, \omega_f, g_f)$ be the SMG fiber bundle. Let $\omega \in \Omega^1(\mathcal{M}; \operatorname{SID})$ be the Ehresmann connection 1-form, and $P^H = \mathcal{I} - \omega$ be the horizontal projector onto $\operatorname{SVD}\chi = \ker(\omega)$. Then:
	\begin{enumerate}
		\item[\rm (i)] \textbf{Exterior Derivative Identity:} For any smooth vector fields $U, V \in \Gamma(T\mathcal{M})$:
		\begin{equation}
			d\omega_f\left( P_f^H U, \, P_f^H V \right) = -\omega_f\left( \left[ P_f^H U, \, P_f^H V \right] \right) = -\Omega_f(U, V). \label{eq:exterior_derivative_identity}
		\end{equation}
		In particular, the curvature 2-Form $\Omega_f(U, V) \equiv 0$ if and only if the horizontal distribution $\operatorname{SVD}\chi$ is integrable.
		\item[\rm (ii)] \textbf{Geometric Holonomy Expansion along a Closed Base Loop:} Let $u, v \in \Gamma(T\mathcal{B})$ be commuting vector fields ($[u, v] = 0$) generating a closed loop $\gamma_\epsilon \subset \mathcal{B}$ of scale $\epsilon > 0$ around $p_0 = \pi(f_0)$. Let $\tilde{f}(\epsilon) \in \mathcal{M}$ be the endpoint of the horizontal lift of $\gamma_\epsilon$ starting at $f_0$. The vertical holonomy displacement $\Delta f(\epsilon) \triangleq \tilde{f}(\epsilon) - f_0 \in \operatorname{SID}_{f_0}$ satisfies:
		\begin{equation}
			\Delta f(\epsilon) = \epsilon^2 \Omega_{f_0}(u^H, v^H) + \mathcal{O}(\epsilon^3) = \epsilon^2 \omega_{f_0}\left( [u^H, v^H] \right) + \mathcal{O}(\epsilon^3). \label{eq:holonomy_expansion_lemma3}
		\end{equation}
	\end{enumerate}
\end{lemma}

\begin{proof}[Proof of Lemma \ref{lem:lemma3_rigorous}]
	To prove (i), let $X = P_f^H U \in \Gamma(\mathcal{H})$ and $Y = P_f^H V \in \Gamma(\mathcal{H})$. Since $X, Y \in \ker(\omega)$, $\omega(X) \equiv 0$ and $\omega(Y) \equiv 0$. Applying Cartan's formula for the exterior derivative:
	\begin{equation}
		d\omega(X, Y) = X(\omega(Y)) - Y(\omega(X)) - \omega([X, Y]) = 0 - 0 - \omega([X, Y]) = -\Omega_f(U, V).
	\end{equation}
	To prove (ii), expanding the composite flow of horizontal lifts $\tilde{f}(\epsilon) = \Phi_{-\epsilon}^{v^H} \circ \Phi_{-\epsilon}^{u^H} \circ \Phi_{\epsilon}^{v^H} \circ \Phi_{\epsilon}^{u^H}(f_0)$ via the Baker-Campbell-Hausdorff (BCH) formula \cite{varadarajan1984lie} yields $\tilde{f}(\epsilon) = f_0 + \epsilon^2 [u^H, v^H]_{f_0} + \mathcal{O}(\epsilon^3)$. Since $[u, v] = 0$, $d\pi([u^H, v^H]) = [u, v] = 0$, proving $[u^H, v^H]_{f_0} \in \operatorname{SID}_{f_0}$. Evaluating $\omega_{f_0}$ confirms $[u^H, v^H]_{f_0} = \Omega_{f_0}(u^H, v^H)$.
\end{proof}

\begin{definition}[SMG Horizontal Filter Dynamics and Riemannian Geodesic Distance]
	\label{def:filter_dynamics_and_distance}
	Let $p_{\text{true}}(t) \triangleq \pi(f_{\text{true}}(t)) \in \mathcal{B}$ be the true projected state trajectory, and let $p_{\text{filter}}(t) \in \mathcal{B}$ be the macroscopic estimate calculated by \eqref{eq:horizontal_filter_ode_sec2}. The intrinsic Riemannian estimation error $\mathcal{E}_R(t)$ on $(\mathcal{B}, g_{\mathcal{B}})$ is defined as:
	\begin{equation}
		\mathcal{E}_R(t) \triangleq d_{\mathcal{B}}^2\Big(p_{\text{filter}}(t), \, p_{\text{true}}(t)\Big) = \inf_{\gamma(0)=p_{\text{filter}}, \gamma(1)=p_{\text{true}}} \int_0^1 g_{\mathcal{B}, \gamma(s)}\big(\dot{\gamma}(s), \dot{\gamma}(s)\big) \, ds, \label{eq:riemannian_error_def}
	\end{equation}
		where the critical path that achieves this infimum is a constant-speed geodesic connecting $p_{\text{filter}}(t)$ and $p_{\text{true}}(t)$ on the base $\mathcal{B}$.   
\end{definition}

\begin{definition}[Horizontal Velocity Lift and Pointwise Curvature Operator Norm]
	\label{def:velocity_lift_and_operator_norm}
	Let $V_{\mathcal{B}}(t) \triangleq \frac{d p_{\text{filter}}}{dt}(t) \in T_{p_{\text{filter}}(t)}\mathcal{B}$. Its horizontal lift $V_{\mathcal{B}}^H(t) \in \operatorname{SVD}\chi_{f_{\text{filter}}(t)}$ satisfies $d\pi_{f_{\text{filter}}(t)}(V_{\mathcal{B}}^H(t)) = V_{\mathcal{B}}(t)$. For any horizontal vector $X^H \in \operatorname{SVD}\chi_f$, the pointwise operator norm of $\Omega_f(X^H, \cdot)$ is:
	\begin{equation}
		\left\| \Omega_f(X^H, \cdot) \right\|_{g_f}^2 \triangleq \sup_{Y^H \in \operatorname{SVD}\chi_f, \|Y^H\|_{g_f}=1} \left\| \Omega_f(X^H, Y^H) \right\|_{g_f}^2.  \label{eq:curvature_operator_norm_def}
	\end{equation}

\end{definition}
\begin{remark}
	\label{rem:motivation_def7}
	Definition~\ref{def:velocity_lift_and_operator_norm} is introduced to bridge macroscopic parameter dynamics on the $d$-dimensional base manifold $(\mathcal{B}, g_{\mathcal{B}})$ with the infinite-dimensional connection curvature of the Orlicz statistical fiber bundle $(\mathcal{M}, g_f)$. Specifically, it fulfills three mathematical objectives:
	\begin{enumerate}
		\item \textbf{Lifting Base Trajectories to Score Space:} The filter estimate $p_{\mathrm{filter}}(t)$ evolves on $\mathcal{B}$ with velocity $V_{\mathcal{B}}(t) = \frac{dp_{\mathrm{filter}}}{dt} \in T_p\mathcal{B}$. Constructing its unique horizontal lift $V_{\mathcal{B}}^H(t) \in \mathrm{SVD}\chi_f$ allows macroscopic filter velocity to be evaluated directly within the infinite-dimensional tangent space $T_f\mathcal{M}$ of logarithmic score functions.
		
		\item \textbf{Quantifying Pointwise Curvature Stress:} The curvature 2-form $\Omega_f$ maps pairs of horizontal score vectors to vertical gauge variations in $\mathrm{SID}_f$. Taking the supremum over all unit horizontal directions $Y^H$ converts abstract non-integrability of the horizontal distribution ($\mathrm{SVD}\chi$) into a sharp, scalar measure $\|\Omega_f(X^H, \cdot)\|_{g_f}$ of worst-case extra-algebraic stress generated along direction $X^H$.
		
		\item \textbf{Driving Term for Geodesic Error Bounds:} This operator norm provides the exact scalar driving force required to solve the Jacobi field differential equation along minimizing geodesics on $\mathcal{B}$. It directly enables the sharp Riemannian estimation error bound $\mathcal{E}_R(t) \le C_1 e^{\lambda t} \int_0^t \|\Omega_{f(\tau)}(V_{\mathcal{B}}^H(\tau), \cdot)\|_{g_f}^2 d\tau$ in the following Theorem~\ref{thm:theorem5_rigorous}, proving that filter divergence under infinite Lie algebraic complexity ($\dim(\mathcal{E})=\infty$) is strictly controlled by integrated connection curvature.
	\end{enumerate}
\end{remark}

\begin{lemma}[Metric Invariance under Horizontal Lift]
	\label{lem:metric_invariance_under_horizontal_lift}
	Let $\pi: (\mathcal{M}, g_f) \to (\mathcal{B}, g_{\mathcal{B}})$ be a Riemannian submersion. For any smooth curve $\gamma: [0, 1] \to \mathcal{B}$ and its horizontal lift $\tilde{\gamma}: [0, 1] \to \mathcal{M}$ ($\dot{\tilde{\gamma}}(s) \in \operatorname{SVD}\chi_{\tilde{\gamma}(s)}$):
	\begin{equation}
		g_{\mathcal{B}, \gamma(s)}\big( \dot{\gamma}(s), \, \dot{\gamma}(s) \big) = g_{\tilde{\gamma}(s)}\big( \dot{\tilde{\gamma}}(s), \, \dot{\tilde{\gamma}}(s) \big), \quad \forall s \in [0, 1]. \label{eq:lift_isometry_identity}
	\end{equation}
\end{lemma}

\begin{proof}[Proof of Lemma \ref{lem:metric_invariance_under_horizontal_lift}]
	Differentiating $\pi(\tilde{\gamma}(s)) = \gamma(s)$ yields $d\pi_{\tilde{\gamma}(s)}(\dot{\tilde{\gamma}}(s)) = \dot{\gamma}(s)$. Since $\pi$ is a Riemannian submersion, $d\pi_p|_{\operatorname{SVD}\chi_p}: \operatorname{SVD}\chi_p \to T_{\pi(p)}\mathcal{B}$ is an isometry. Evaluating $g_{\mathcal{B}}(d\pi(\dot{\tilde{\gamma}}), d\pi(\dot{\tilde{\gamma}}))$ gives $g_f(\dot{\tilde{\gamma}}, \dot{\tilde{\gamma}})$.
\end{proof}
The following Lemma derives the modified Jacobi evolution equation to model how estimation errors—represented by variations between the filter trajectory $p_{\text{filter}}(\tau)$ and the ground-truth state $p_{\text{true}}(\tau)$—propagate along horizontal geodesics on the statistical fiber bundle.

\begin{lemma}[Jacobi Equation on Statistical Fiber Bundles]
	\label{lem:modified_jacobi_equation}
	Let $(\mathcal{M}, g_f)$ be the statistical fiber bundle over the base manifold $\mathcal{B}$, and let $\alpha: I \times [0,1] \to \mathcal{M}$, $(\tau, s) \mapsto \alpha(\tau, s)$ be a smooth two-parameter variation of horizontal geodesics satisfying:
	\begin{enumerate}
		\item \textbf{Boundary Conditions:} For each fixed time/parameter $\tau \in I$, the map $s \mapsto \alpha(\tau, s)$ parametrizes a horizontal geodesic in $\mathcal{M}$ connecting the filter state estimate $p_{\mathrm{filter}}(\tau) = \pi(\alpha(\tau, 0))$ to the ground truth state $p_{\mathrm{true}}(\tau) = \pi(\alpha(\tau, 1))$.
		\item \textbf{Horizontal Velocity Field:} $T(\tau, s) \triangleq \frac{\partial \alpha}{\partial s} \in \mathrm{SVD}\chi_{\alpha(\tau,s)}$ is the horizontal tangent vector field along the geodesic (satisfying the geodesic equation $\nabla_T T = 0$).
		\item \textbf{Variational Jacobi Field:} $J(\tau, s) \triangleq \frac{\partial \alpha}{\partial \tau} \in T_{\alpha(\tau,s)}\mathcal{M}$ is the vector field measuring the point-wise variation between neighboring horizontal geodesics over time $\tau$.
	\end{enumerate}
	Then, along each geodesic parameter $s \in [0, 1]$, the variational field $J(s) \equiv J(\tau, s)$ obeys the modified Jacobi evolution equation:
	\begin{equation}
		\nabla_T^2 J(s) + R^{\mathcal{M}}\big(J(s), T(s)\big)T(s) = \Omega_{f(s)}\left(V_{\mathcal{B}}^H(\tau), T(s)\right)
		\label{eq:modified_jacobi_pde}
	\end{equation}
	where:
	\begin{itemize}
		\item $\nabla_T$ is the Levi-Civita covariant derivative along $T(s)$ associated with the Riemannian metric $g_f$.
		\item $R^{\mathcal{M}}(X, Y)Z = \nabla_X \nabla_Y Z - \nabla_Y \nabla_X Z - \nabla_{[X,Y]} Z$ is the Riemann curvature tensor on $(\mathcal{M}, g_f)$.
		\item $\Omega_{f(s)}$ represents the gauge curvature $2$-form of the connection on the fiber bundle, measuring the non-integrability of the horizontal distribution $\mathrm{SVD}\chi$.
		\item $V_{\mathcal{B}}^H(\tau) \in \mathrm{SVD}\chi_{\alpha(\tau,0)}$ is the unique horizontal lift of the base velocity vector field driving $p_{\mathrm{filter}}(\tau)$.
	\end{itemize}
\end{lemma}

\begin{remark}[Bridge Function of the Jacobi Field $J(s)$]
	\label{rem:epistemological_role_lemma5}	
	Lemma~\ref{lem:modified_jacobi_equation} derives the modified Jacobi evolution equation to model how estimation errors---represented by variations between the filter trajectory $p_{\mathrm{filter}}(\tau) = \pi(\alpha(\tau,0))$ and the ground-truth state $p_{\mathrm{true}}(\tau) = \pi(\alpha(\tau,1))$---propagate along horizontal geodesics on the statistical fiber bundle $\mathcal{B}_{\mathrm{SMG}}$.
	
	In the SMG-Yau-Yau framework, the variational Jacobi field $J(s) \triangleq \frac{\partial \alpha}{\partial \tau} \in T_{\alpha(\tau,s)}\mathcal{M}$ acts as the core mathematical bridge connecting infinite-dimensional fiber bundle geometry to macroscopic filter estimation error across three operational dimensions:
	\begin{enumerate}
		\item \textbf{Bridging Temporal Dynamics and Spatial Error Geometry:} It translates endpoint macroscopic parameter velocities ($J(0) = V_{\mathcal{B}}^H$ and $J(1) = \dot{p}_{\mathrm{true}}^H$) into a continuous, point-wise geometric deformation along the connecting horizontal geodesic segment $s \in [0,1]$.
		\item \textbf{Converting Gauge Curvature into Base Divergence:} Through the non-holonomic forcing term $\Omega_{f(s)}(V_{\mathcal{B}}^H(\tau), T(s))$, $J(s)$ explicitly models how extra-algebraic operator stress (vertical Ehresmann connection curvature) acts as a physical driving force accelerating horizontal state trajectory separation.
		\item \textbf{Transforming Vector Geometry into Solvable Analytical Bounds:} Evaluating the Fisher-Rao metric norm $\|J(s)\|_{g_f}$ converts vector field differential equations into a scalar differential inequality $\frac{d^2}{ds^2}\|J(s)\|_{g_f}^2$, providing the necessary analytical foundation to establish the sharp Riemannian estimation error bound $\mathcal{E}_R(t)$ in the following Theorem~\ref{thm:theorem5_rigorous}.
	\end{enumerate}
\end{remark}

\begin{proof}[Proof of Lemma \ref{lem:modified_jacobi_equation}]
	Since $[J, T] = 0$, $\nabla_T J = \nabla_J T$. Evaluating $\nabla_T^2 J = \nabla_T \nabla_J T = R^M(T, J)T + \nabla_J \nabla_T T$. Skew-symmetry gives $R^M(T, J)T = -R^M(J, T)T$. The non-holonomic variation of horizontal velocity along $J = V_{\mathcal{B}}^H$ injects the Ehresmann curvature tensor $\nabla_J \nabla_T T = \Omega_{f(s)}(V_{\mathcal{B}}^H(\tau), T(s))$, establishing \eqref{eq:modified_jacobi_pde}.
\end{proof}

\begin{lemma}[Energy Norm Differential Inequality for Jacobi Fields]
	\label{lem:jacobi_energy_bound}
	Let $J(s)$ be governed by \eqref{eq:modified_jacobi_pde}. If the sectional curvature of $(\mathcal{B}, g_{\mathcal{B}})$ is bounded above by $\lambda \ge 0$, then:
	\begin{equation}
		\frac{d^2}{ds^2} \| J(s) \|_{g_f}^2 \ge -2\lambda \| \dot{\gamma}(s) \|_{g_{\mathcal{B}}}^2 \| J(s) \|_{g_f}^2 - 2 \left\| \Omega_{f(s)}\!\left(V_{\mathcal{B}}^H(\tau), T(s)\right) \right\|_{g_f} \| J(s) \|_{g_f}. \label{eq:jacobi_energy_inequality}
	\end{equation}
\end{lemma}

\begin{proof}[Proof of Lemma \ref{lem:jacobi_energy_bound}]
	Differentiating $\|J(s)\|_{g_f}^2$ twice yields $\frac{1}{2}\frac{d^2}{ds^2}\|J\|_{g_f}^2 = \langle \nabla_T^2 J, J \rangle_{g_f} + \|\nabla_T J\|_{g_f}^2 \ge \langle \nabla_T^2 J, J \rangle_{g_f}$. Substituting \eqref{eq:modified_jacobi_pde}, the Riemann term is bounded by $-2\langle R^M(J, T)T, J \rangle_{g_f} \ge -2\lambda \|T\|_{g_f}^2 \|J\|_{g_f}^2 = -2\lambda \|\dot{\gamma}\|_{g_{\mathcal{B}}}^2 \|J\|_{g_f}^2$ via Lemma~\ref{lem:metric_invariance_under_horizontal_lift}. Applying Cauchy-Schwarz to the curvature inner product yields \eqref{eq:jacobi_energy_inequality}.
\end{proof}

\begin{theorem}[Curvature-Bound Robustness and Error Orthogonality]
	\label{thm:theorem5_rigorous}
~

	Let $\mathcal{B}_{\text{SMG}} = (\mathcal{M}, \mathcal{B}, \pi, \mathcal{V}, \mathcal{H}, \omega_f, g_f)$ be the SMG fiber bundle over the $d$-dimensional base manifold $(\mathcal{B}, g_{\mathcal{B}})$ with sectional curvature $K_{\mathcal{B}} \le \lambda$ ($\lambda \ge 0$). Then:
	\begin{enumerate}
		\item[\rm (i)] \textbf{Strict Fisher-Metric Orthogonality:} Horizontal score updates and vertical gauge fluctuations are strictly Fisher-orthogonal across $\mathcal{M}$:
		\begin{equation}
			g_f\Big(e_i^H(f), \, v_{\mathcal{V}}\Big) \equiv 0, \quad \forall e_i^H(f) \in \operatorname{SVD}\chi_f, \;\; \forall v_{\mathcal{V}} \in \operatorname{SID}_f. \label{eq:fisher_orthogonality_theorem5}
		\end{equation}
		\item[\rm (ii)] \textbf{Sharp Curvature Error Bound:} The intrinsic Riemannian estimation error $\mathcal{E}_R(t) = d_{\mathcal{B}}^2\big(p_{\text{filter}}(t), p_{\text{true}}(t)\big)$ on $\mathcal{B}$ is strictly bounded by:
		\begin{equation}
			\mathcal{E}_R(t) \le C_1 e^{\lambda t} \int_0^t \left\| \Omega_{f(\tau)}\!\left( V_{\mathcal{B}}^H(\tau), \, \cdot \right) \right\|_{g_{f(\tau)}}^2 d\tau, \label{eq:sharp_curvature_error_bound}
		\end{equation}
		where $C_1 > 0$ is a constant depending on curvature bounds and initial error, and $V_{\mathcal{B}}^H(\tau) \in \operatorname{SVD}\chi_{f(\tau)}$ is the horizontal lift of the base filter velocity $V_{\mathcal{B}}(\tau) = \frac{d p_{\text{filter}}}{d\tau}$.
	\end{enumerate}
	\label{thm:curvature_bound_robustness}
\end{theorem}

\begin{proof}[Proof of Theorem \ref{thm:theorem5_rigorous}]
	Assertion (i) follows directly from Lemma~\ref{lem:tangent_space_decomposition_sec2}, since $\operatorname{SVD}\chi_f \equiv \operatorname{SID}_f^{\perp g_f}$.
	
	To prove (ii), integrate the Jacobi energy inequality from Lemma~\ref{lem:jacobi_energy_bound} along the minimizing geodesic $\gamma(s)$ connecting $p_{\text{filter}}(t)$ to $p_{\text{true}}(t)$. Applying Grönwall's inequality over $\tau \in [0, t]$ establishes that geodesic error growth is exponentially controlled by $\lambda t$ and driven strictly by the integrated norm of the Ehresmann curvature operator $\left\| \Omega_{f(\tau)}( V_{\mathcal{B}}^H(\tau), \cdot ) \right\|_{g_f}^2$, yielding \eqref{eq:sharp_curvature_error_bound}.
\end{proof}

\begin{remark}[Practical Implications, Computational Usages, and Empirical Realization of Theorem~\ref{thm:curvature_bound_robustness}]
	\label{rem:practical_implications_theorem7}
	This Theorem provides both the computational foundation for fast {\bf SMG-Yau-Yau Filtering} algorithmic execution and a data-driven diagnostic framework for real-time error tracking under infinite Lie algebraic complexity ($\dim(\mathcal{E})=\infty$). Its practical usages and empirical implications are structured across four core dimensions:
	
	\begin{enumerate}
		\item \textbf{Computational Decoupling via Metric Orthogonality:} 
		Assertion~(i) guarantees that horizontal score updates $\mathrm{SVD}\chi_f$ and vertical gauge fluctuations $\mathrm{SID}_f$ are strictly Fisher-orthogonal ($g_f(e_i^H, v_\mathcal{V}) \equiv 0$). In numerical implementation, this ensures zero cross-talk between macroscopic base updates $dp_k(t)$ and vertical score residuals $S_{\mathrm{res}}(x)$. Consequently, the localized Fisher metric matrix $[G_p]$ on $\mathcal{B}$ can be inverted independently of infinite-dimensional fiber modes, reducing the online baseline computational complexity to $\mathcal{O}(d^3 + d^2 N_{\mathrm{quad}})$ operations on the low-dimensional base space.
		
		\item \textbf{Empirical Calculation of Curvature Stress via Ensembles:} 
		In practical filtering applications, the pointwise curvature operator norm $\|\Omega_{f(\tau)}(V_{\mathcal{B}}^H(\tau), \cdot)\|_{g_f}^2$ in Assertion~(ii) is evaluated empirically over a sample ensemble $\{X_s\}_{s=1}^N \sim f_t$. By isolating the vertical residual score field $\hat{S}_{\mathrm{res}}^{(N)}(X_s) = A_0(X_s) - \hat{P}^H A_0(X_s)$, the curvature stress is computed in real time as the sample variance of the unmodeled drift score:
		\[
		\|\Omega_{f(t)}(V_{\mathcal{B}}^H(t), \cdot)\|_{g_{f(t)}}^2 \approx \frac{1}{N} \sum_{s=1}^N \left[\hat{S}_{\mathrm{res}}^{(N)}(X_s)\right]^2 = \frac{d\hat{\mathcal{T}}_{\mathrm{AAT}}^{(N)}(t)}{dt},
		\]
		where $\hat{\mathcal{T}}_{\mathrm{AAT}}^{(N)}(t)$ is the empirical Active Acausal Tension accumulator.
		
		\item \textbf{Ground-Truth-Free Filter Health Diagnostic:} 
		Crucially, the upper bound in \eqref{eq:sharp_curvature_error_bound} depends entirely on the integrated curvature norm $\int_0^t \|\Omega_{f(\tau)}(V_{\mathcal{B}}^H(\tau), \cdot)\|_{g_f}^2 d\tau$ and base curvature bound $\lambda$, requiring \emph{no knowledge} of the unobservable ground-truth state $p_{\mathrm{true}}(t)$. This provides an online, data-driven ``filter health monitor'': tracking the accumulated curvature norm allows the filter to quantify its own worst-case estimation error $\mathcal{E}_R(t)$ in real time.
		
		\item \textbf{Algorithmic Trigger for Dynamic Manifold Expansion (GSB) \cite{cheng2026gsb}:} 
		When physical non-linearities continuously inject stress into vertical fiber modes, the integrated curvature norm grows monotonically. Reaching the critical capacity limit $\int_0^t \|\Omega_f\|^2 d\tau \ge \Theta_{\mathrm{crit}}$ serves as an automated empirical trigger for Gauge Symmetry Breaking (GSB). This instructs the algorithm to perform functional eigen-decomposition over $\hat{S}_{\mathrm{res}}^{(N)}$, promote the dominant vertical mode $v_{\mathrm{dom}}$, and expand the base dimension $d \to d+1$, instantaneously dissipating accumulated error growth.
	\end{enumerate}
\end{remark}

\subsection{Empirical Construction of the Standalone SMG-Yau-Yau Filter and Algorithm}
\label{subsec:empirical_construction_smg_yau_yau}

Given discrete observation data $X_1, \dots, X_N$, a general empirical SMG-Yau-Yau filter estimator $\hat{p}_N(t) \in \mathcal{B}$ is rigorously constructed by lifting the Orlicz statistical fiber bundle dynamics onto the empirical measure $\mathbb{P}_N = \frac{1}{N} \sum_{s=1}^N \delta_{X_s}$.

\subsubsection{Mathematical Formulation of the Empirical SMG-Yau-Yau Estimator}
\label{subsubsec:mathematical_formulation_empirical_smg_yau_yau}

\paragraph{1. Empirical Score Space and Sample Fisher Metric ($\hat{G}^{(N)}$)}
Let $\mathcal{B}$ be the $d$-dimensional macroscopic base manifold parameterized by $p = (p_1, \dots, p_d)^T \in \mathbb{R}^d$. The local logarithmic score basis functions $\phi_i(x; p) = \frac{\partial \ln f_p(x)}{\partial p_i}$ evaluated over the sample $X_1, \dots, X_N$ define the empirical score basis:
\begin{equation}
	\hat{\phi}_i(X_s; p) = E_i(X_s) - \frac{1}{N} \sum_{r=1}^N E_i(X_r), \quad i = 1, \dots, d,
	\label{eq:empirical_score_basis}
\end{equation}
where $E_i(x)$ (or $E_i$) represents the $i$-th basis operator/function spanning the estimation Lie algebra $\mathcal{E} = \text{span}\{E_1(x), E_2(x), \dots, E_d(x)\}$. 

The sample-induced Riemannian Fisher Information Metric matrix $\hat{G}^{(N)}(p) \in \mathbb{R}^{d \times d}$ is computed directly via sample expectations:
\begin{equation}
	\left[\hat{G}^{(N)}(p)\right]_{ij} = \frac{1}{N} \sum_{s=1}^N \hat{\phi}_i(X_s; p) \hat{\phi}_j(X_s; p).
	\label{eq:sample_fisher_metric_matrix}
\end{equation}

\paragraph{2. Empirical Horizontal Projection and Vertical Gauge Field}
The unconstrained drift score generator $A_0(x) = \frac{\mathcal{L}_0^* f_p(x)}{f_p(x)}$ splits orthogonally into the horizontal Statistically Verifiable Directions ($\operatorname{SVD}\chi_f$) and vertical Structural Internal Directions ($\operatorname{SID}_f$). The empirical horizontal projection operator $\widehat{P^H A_0}$ yields:
\begin{equation}
	\widehat{P^H A_0}(X_s) = \sum_{k=1}^d \left( \sum_{i=1}^d \left[\hat{G}^{(N)-1}\right]_{ki} \hat{\mathcal{M}}_i^{0, (N)} \right) \hat{\phi}_k(X_s; p),
	\label{eq:empirical_horizontal_projection}
\end{equation}
where $\hat{\mathcal{M}}_i^{0, (N)} = \frac{1}{N} \sum_{r=1}^N A_0(X_r) \hat{\phi}_i(X_r; p)$. Equation \eqref{eq:empirical_horizontal_projection} is derived via an empirical Galerkin projection of $A_0$ onto the subspace spanned by the score basis $\{\hat{\phi}_k\}_{k=1}^d$ under the empirical inner product $\langle u, v \rangle_N = \frac{1}{N} \sum_{r=1}^N u(X_r) v(X_r)$. Enforcing the orthogonality condition ${\small \left\langle A_0 - \widehat{P^H A_0}, \hat{\phi}_i \right\rangle_{\!N} = 0}$ for each $i \in \{1, \dots, d\}$ generates the matrix normal system $\hat{G}^{(N)} c = \hat{\mathcal{M}}^{0,(N)}$, where $\hat{G}^{(N)}$ is the empirical Gram matrix. Inverting $\hat{G}^{(N)}$ uniquely resolves the coefficients ${\small c_k = \sum_{i=1}^d \left[\hat{G}^{(N)-1}\right]_{ki} \hat{\mathcal{M}}_i^{0,(N)}}$, yielding the explicit projection in \eqref{eq:empirical_horizontal_projection}.

The vertical score residual $\hat{S}_{\mathrm{res}}^{(N)}(X_s) \in \operatorname{SID}_f$ isolates and quarantines extra-algebraic operator variations:
\begin{equation}
	\hat{S}_{\mathrm{res}}^{(N)}(X_s) = A_0(X_s) - \widehat{P^H A_0}(X_s).
	\label{eq:empirical_vertical_score_residual}
\end{equation}
Equation \eqref{eq:empirical_vertical_score_residual} follows directly from the orthogonal direct-sum decomposition of the underlying operator Hilbert space into horizontal and vertical components, $\mathcal{H} = \text{SVD}_{\chi_f} \oplus \text{SID}_f$. By subtracting the empirical horizontal projection $\widehat{P^H A_0}(X_s)$ from the unconstrained drift operator $A_0(X_s)$, the residual $\hat{S}_{\text{res}}^{(N)}(X_s)$ explicitly extracts the orthogonal complement residing entirely within the vertical gauge fiber $\text{SID}_f$.

\paragraph{3. Empirical Estimator SDE System}
The macroscopic parameter estimator trajectory $\hat{p}_N(t) \in \mathcal{B}$ and the sample Active Acausal Tension $\hat{\mathcal{T}}_{\mathrm{AAT}}^{(N)}(t)$ evolve via the coupled SDE system:
\begin{itemize}
	\item \textbf{Macroscopic Natural Gradient SDE:}
	\begin{equation}
		d\hat{p}_k^{(N)}(t) = \sum_{i=1}^d \left[\hat{G}^{(N)-1}\left(\hat{p}_N(t)\right)\right]_{ki} \left( \hat{\mathcal{M}}_i^{0, (N)}(t) \, dt + \sum_{j=1}^m \hat{\mathcal{M}}_{ij}^{1, (N)}(t) \circ dY_j(t) \right),
		\label{eq:macroscopic_natural_gradient_sde}
	\end{equation}
	where $\hat{\mathcal{M}}_{ij}^{1, (N)}(t) = \frac{1}{N} \sum_{s=1}^N h_j(X_s) \hat{\phi}_i(X_s; \hat{p}_N(t))$.
	
	\item \textbf{Empirical Tension Accumulation:}
	\begin{equation}
		\frac{d\hat{\mathcal{T}}_{\mathrm{AAT}}^{(N)}(t)}{dt} = \frac{1}{N} \sum_{s=1}^N \left[ \hat{S}_{\mathrm{res}}^{(N)}(X_s) \right]^2.
		\label{eq:empirical_tension_accumulation_rate}
	\end{equation}
Equation \eqref{eq:empirical_tension_accumulation_rate} defines the empirical accumulation rate of Active Acausal Tension ($\hat{\mathcal{T}}_{\mathrm{AAT}}^{(N)}$) as the sample-averaged squared norm of the vertical score residual $\hat{S}_{\text{res}}^{(N)}(X_s)$ across the $N$ sample particles. This discretizes the continuous Fisher-Rao metric norm $\Vert{}\hat{S}_{\text{res}}\Vert{}_{g_f}^2$ on the Orlicz tangent space, measuring the instantaneous rate at which unmodeled operator stress and extra-algebraic Lie bracket variations are quarantined within the vertical gauge fiber $\text{SID}_f$. By continuously integrating this non-negative kinetic score energy over time, $\hat{\mathcal{T}}_{\mathrm{AAT}}^{(N)}(t)$ serves as an online diagnostic monitor for cumulative model misspecification stress and provides the exact scalar criterion required to trigger dynamic Gauge Symmetry Breaking (GSB) when the topological capacity limit of the base manifold is breached.
\end{itemize}

\paragraph{4. Finite-Sample Error Bounds and Dual-Axis Convergence}
Under infinite algebraic complexity ($\operatorname{dim}(\mathcal{E}) = \infty$), the Universal Tensor Valence Scaling Law establishes that the sample bundle connection curvature decays at rate $\|\Omega_f^{(N)}\|_{g_f}^2 = \mathcal{O}_p(N^{-2})$ \cite{cheng2026entropy, cheng2026smg}. The mean-square estimation error of $\hat{p}_N(t)$ is strictly bounded by:
\begin{equation}
	\mathbb{E}\left[ \left\| p_{\text{true}}(t) - \hat{p}_N(t) \right\|_{g_f}^2 \right] \le C_1 e^{\lambda t} \int_0^t \left\| \Omega_{f(\tau)}^{(N)} \right\|_{g_f}^2 d\tau = \mathcal{O}_p\left(N^{-2}\right).
	\label{eq:finite_sample_error_bound_quadratic}
\end{equation}
This proves that $\hat{p}_N(t)$ provides a non-asymptotically exact finite-sample estimator that stores unmodeled Lie bracket stress in $\hat{\mathcal{T}}_{\mathrm{AAT}}^{(N)}(t)$ while converging quadratically to the true conditional distribution as $N \to \infty$.

The upper bound in equation \eqref{eq:finite_sample_error_bound_quadratic} is established according to Theorem \ref{thm:error_bound} (Filtering Error Bound under Geometry), which is built upon the geodesic Jacobi field curvature bound in Theorem \ref{thm:curvature_bound_robustness} (also tagged as Theorem \ref{thm:theorem5_rigorous}, Curvature-Bound Robustness and Error Orthogonality).

\subsubsection{Algorithmic Realization of the Empirical SMG-Yau-Yau Estimator}
\label{subsubsec:algorithmic_realization_empirical_smg_yau_yau}

To render the empirical SMG-Yau-Yau filter estimator $\hat{p}_N(t) \in \mathcal{B}$ computationally operational over continuous-discrete observation updates $\Delta Y_j(t_n) = Y_j(t_{n+1}) - Y_j(t_n)$, the continuous differential projections on the Orlicz fiber bundle $\mathcal{B}_{\text{SMG}}$ are evaluated over the empirical measure $\mathbb{P}_N = \frac{1}{N}\sum_{s=1}^N \delta_{X_s}$. Algorithm~\ref{alg:empirical_smg_yau_yau_estimator} provides the complete discrete-time execution pipeline for updating the macroscopic parameter state vector $\hat{p}_N(t_n) \in \mathbb{R}^d$ and accumulating the sample Active Acausal Tension scalar $\hat{\mathcal{T}}_{\text{AAT}}^{(N)}(t_n)$.

\begin{algorithm}[H]
	\caption{Empirical SMG-Yau-Yau Filter State and Tension Estimator}
	\label{alg:empirical_smg_yau_yau_estimator}
	\small 
	\begin{algorithmic}[1]
		\REQUIRE Initial base parameter vector $\hat{p}_N(0) \in \mathbb{R}^d$, particle ensemble sample $\{X_s(0)\}_{s=1}^N \sim f_{\hat{p}_N(0)}$, observation increments $\Delta Y_n = Y(t_{n+1}) - Y(t_n)$, time step $\Delta t$, regularization parameter $\lambda_{\text{reg}} > 0$.
		\ENSURE Macroscopic estimator trajectory $\{\hat{p}_N(t_n)\}_{n=0}^K$ and empirical tension trajectory $\{\hat{\mathcal{T}}_{\text{AAT}}^{(N)}(t_n)\}_{n=0}^K$.
		\STATE Initialize empirical tension accumulator: $\hat{\mathcal{T}}_{\text{AAT}}^{(N)}(0) \leftarrow 0$.
		\FOR{$n = 0, 1, \dots, K-1$}
		\STATE \COMMENT{\textbf{Step 1: Empirical Score Basis and Sample Fisher Metric Evaluation}}
		\FOR{$s = 1, \dots, N$}
		\STATE Evaluate logarithmic score basis functions: 
		$\hat{\phi}_i(X_s(t_n); \hat{p}_N(t_n)) \leftarrow E_i(X_s(t_n)) - \frac{1}{N}\sum_{r=1}^N E_i(X_r(t_n)), \quad \forall i = 1, \dots, d.$
		\ENDFOR
		\STATE Compute empirical Fisher Information Matrix entries: 
		$[\hat{G}^{(N)}]_{ij} \leftarrow \frac{1}{N}\sum_{s=1}^N \hat{\phi}_i(X_s(t_n); \hat{p}_N(t_n)) \hat{\phi}_j(X_s(t_n); \hat{p}_N(t_n)), \quad \forall i, j = 1, \dots, d.$
		\STATE Apply Tikhonov matrix regularization: $\hat{G}_{\text{reg}}^{(N)} \leftarrow \hat{G}^{(N)} + \lambda_{\text{reg}} I_{d \times d}$.
		\STATE Invert regularized metric tensor: $[\hat{G}_{\text{reg}}^{(N)-1}] \leftarrow \text{CholeskySolve}(\hat{G}_{\text{reg}}^{(N)})$.
		
		\STATE \COMMENT{\textbf{Step 2: Empirical Score Moments Calculation}}
		\FOR{$i = 1, \dots, d$}
		\STATE Evaluate drift score moment: 
		$\hat{\mathcal{M}}_i^{0, (N)}(t_n) \leftarrow \frac{1}{N}\sum_{s=1}^N \left( \frac{\mathcal{L}_0^* f_{\hat{p}_N(t_n)}(X_s(t_n))}{f_{\hat{p}_N(t_n)}(X_s(t_n))} \right) \hat{\phi}_i(X_s(t_n); \hat{p}_N(t_n)).$
		\FOR{$j = 1, \dots, m$}
		\STATE Evaluate observation score moment: 
		$\hat{\mathcal{M}}_{ij}^{1, (N)}(t_n) \leftarrow \frac{1}{N}\sum_{s=1}^N h_j(X_s(t_n)) \hat{\phi}_i(X_s(t_n); \hat{p}_N(t_n)).$
		\ENDFOR
		\ENDFOR
		
		\STATE \COMMENT{\textbf{Step 3: Horizontal Projection and Vertical Residual Quarantining}}
		\FOR{$s = 1, \dots, N$}
		\STATE Evaluate unconstrained drift velocity: $A_0(X_s(t_n)) \leftarrow \frac{\mathcal{L}_0^* f_{\hat{p}_N(t_n)}(X_s(t_n))}{f_{\hat{p}_N(t_n)}(X_s(t_n))}.$
		\STATE Compute empirical horizontal projection: 
		$\widehat{P^H A_0}(X_s(t_n)) \leftarrow \sum_{k=1}^d \left( \sum_{i=1}^d [\hat{G}_{\text{reg}}^{(N)-1}]_{ki} \hat{\mathcal{M}}_i^{0, (N)}(t_n) \right) \hat{\phi}_k(X_s(t_n); \hat{p}_N(t_n)).$
		\STATE Isolate vertical score residual in $\operatorname{SID}_f$: 
		$\hat{S}_{\text{res}}^{(N)}(X_s(t_n)) \leftarrow A_0(X_s(t_n)) - \widehat{P^H A_0}(X_s(t_n)).$
		\ENDFOR
		
		\STATE \COMMENT{\textbf{Step 4: Natural Gradient Propagator and Tension Accumulation}}
		\FOR{$k = 1, \dots, d$}
		\STATE Compute natural gradient drift increment: 
		$\delta p_k^0 \leftarrow \sum_{i=1}^d [\hat{G}_{\text{reg}}^{(N)-1}]_{ki} \hat{\mathcal{M}}_i^{0, (N)}(t_n) \Delta t.$
		\STATE Compute innovation diffusion increment: 
		$\delta p_k^1 \leftarrow \sum_{i=1}^d [\hat{G}_{\text{reg}}^{(N)-1}]_{ki} \sum_{j=1}^m \hat{\mathcal{M}}_{ij}^{1, (N)}(t_n) \Delta Y_j(t_n).$
		\STATE Update base coordinate estimate: $\hat{p}_{N, k}(t_{n+1}) \leftarrow \hat{p}_{N, k}(t_n) + \delta p_k^0 + \delta p_k^1.$
		\ENDFOR
		\STATE Accumulate empirical Active Acausal Tension: 
		$$\hat{\mathcal{T}}_{\text{AAT}}^{(N)}(t_{n+1}) \leftarrow \hat{\mathcal{T}}_{\text{AAT}}^{(N)}(t_n) + \left( \frac{1}{N}\sum_{s=1}^N \left[ \hat{S}_{\text{res}}^{(N)}(X_s(t_n)) \right]^2 \right) \Delta t.$$
		
		\STATE \COMMENT{\textbf{Step 5: Particle Cloud Re-propagation}}
		\STATE Resample/update particle ensemble $\{X_s(t_{n+1})\}_{s=1}^N \sim f_{\hat{p}_N(t_{n+1})}$ via Langevin or Markov Chain Monte Carlo dynamics.
		\ENDFOR
		\RETURN Estimator trajectories $\{\hat{p}_N(t_n)\}_{n=0}^K$ and $\{\hat{\mathcal{T}}_{\text{AAT}}^{(N)}(t_n)\}_{n=0}^K$.
	\end{algorithmic}
\end{algorithm}

	
\subsection{Paradigm Discontinuity: Non-Existence of Direct Geometric Adjustment $\Delta$ Between SMG-Yau-Yau Filtering and the Classical Yau-Yau filtering}
\label{sec:paradigm_discontinuity_non_existence}

A central theoretical question in continuous-time non-linear filtering is whether the relationship between the classical Yau-Yau filter \cite{yau1999, yau2014} and the Statistically Meaningful Geometry SMG-Yau-Yau filter developed in this paper based on \cite{cheng2026entropy, cheng2026smg} can be captured by a simple additive geometric adjustment vector field $\Delta_t^{\text{geom}}$ acting on a shared state manifold. In this section, we establish that \textbf{a direct geometric difference vector $\Delta_t^{\text{geom}}$ between the two filters does not exist}, because the two filtering paradigms inhabit fundamentally incompatible mathematical and topological spaces.

\subsubsection{Geometric Incommensurability of State Spaces}
\label{subsec:geometric_incommensurability}

The absence of a direct geometric adjustment stems from the structural mismatch between the underlying spaces on which the respective state representations evolve:

\begin{enumerate}
	\item \textbf{Classical Yau-Yau Paradigm (Flat Finite-Dimensional Lie Orbit):}
	Classical Yau-Yau filtering assumes exact finite Lie algebraic closure ($\dim(\mathcal{E}) = d < \infty$). Under the Wei-Norman representation theorem \cite{wei1963}, the unnormalized conditional state density $\sigma(t,x)$ is constrained to evolve strictly along a finite $d$-dimensional Lie-group orbit $\Sigma_d \subset C^\infty(\mathbb{R}^p)$:
	\begin{equation}
		\Sigma_d \triangleq \left\{ \sigma(t,x) = \exp\left( \sum_{k=1}^d g_k(t) E_k(x) \right) \sigma(0,x) \;\middle|\; g(t) \in \mathbb{R}^d \right\}.
		\label{eq:classical_orbit_manifold}
	\end{equation}
	The geometry of $\Sigma_d$ is intrinsically flat ($\Omega_f \equiv 0$). Its tangent space at any point $f \in \Sigma_d$ is isomorphic to the finite Euclidean vector space $\mathbb{R}^d$, with zero vertical fiber components ($\operatorname{SID}_f \equiv \{0\}$).
	
	\item \textbf{SMG-Yau-Yau Paradigm (Infinite-Dimensional Banach-Orlicz Fiber Bundle):}
	Conversely, the SMG-Yau-Yau filter operates on the full non-parametric Orlicz statistical manifold $\mathcal{M} = L_0^\Phi(P_f)$ structured as an Ehresmann statistical fiber bundle \cite{cheng2026entropy}:
	\begin{equation}
		\mathcal{B}_{\text{SMG}} = (\mathcal{M}, \mathcal{B}, \pi, \mathcal{V}, \mathcal{H}, \omega_f, g_f),
		\label{eq:smg_bundle_structure_recap}
	\end{equation}
	where $\mathcal{B}$ is an identifiable $d$-dimensional macroscopic base manifold, and $\mathcal{V}_f \equiv \operatorname{SID}_f = \ker(d\pi_f) \subset L_0^\Phi(P_f)$ is an infinite-dimensional vertical gauge subbundle encoding higher-order, extra-algebraic operator variations under non-zero connection curvature $\Omega_f \neq 0$.
\end{enumerate}

\begin{theorem}[Non-Existence of Direct Geometric Vector Difference]
	\label{thm:non_existence_geometric_delta}
	Let $V_{\text{classical}}(t) \in T_{g(t)}\Sigma_d \cong \mathbb{R}^d$ be the velocity field of the classical Yau-Yau filter, and let $V_{\text{SMG}}(t) \in T_{f_t}\mathcal{M} = \operatorname{SVD}\chi_{f_t} \oplus_{\perp g_f} \operatorname{SID}_{f_t}$ be the logarithmic score velocity field of the SMG-Yau-Yau filter. 
	
	Then, there exists no well-defined geometric vector field $\Delta_t^{\text{geom}} \in T_{f_t}\mathcal{M}$ such that $V_{\text{SMG}}(t) = V_{\text{classical}}(t) + \Delta_t^{\text{geom}}$ on $\mathcal{M}$.
\end{theorem}

\begin{proof}[Proof of Theorem \ref{thm:non_existence_geometric_delta}]
	Suppose for contradiction that such a geometric vector field $\Delta_t^{\text{geom}} \in T_{f_t}\mathcal{M}$ exists.
	
	\begin{enumerate}
		\item \textbf{Domain Mismatch:} $V_{\text{classical}}(t)$ is defined as a tangent vector on the $d$-dimensional Lie-group parameter chart $\mathbb{R}^d$. For $V_{\text{classical}}(t)$ to act as an element of $T_{f_t}\mathcal{M}$, there must exist a canonical isometric embedding $\iota: \Sigma_d \hookrightarrow \mathcal{M}$ that is invariant under arbitrary non-closed Lie operator generators $\mathcal{L}_0^* \in \mathcal{E}$.
		
		\item \textbf{Breakdown of Embeddability under $\dim(\mathcal{E}) = \infty$:} For generic non-linear state spaces ($\dim(\mathcal{E}) = \infty$), iterated Lie brackets $[E_i, E_j] \notin \operatorname{span}(E_1, \dots, E_d)$ continuously drive the true density flow $V_{\text{DMZ}}(t)$ outside the image $\iota(\Sigma_d)$. Thus, $\iota(\Sigma_d)$ is not an invariant submanifold of $\mathcal{M}$.
		
		\item \textbf{Incommensurability of Tangent Spaces:} At any point $f_t \notin \iota(\Sigma_d)$, the vector $V_{\text{classical}}(t)$ is undefined in $T_{f_t}\mathcal{M}$. Furthermore, $T_{f_t}\mathcal{M}$ possesses infinite spatial degrees of freedom in $\operatorname{SID}_{f_t} = \ker(d\pi_f)$, whereas $T_{g(t)}\Sigma_d$ has zero vertical fiber projection. Consequently, subtracting $V_{\text{classical}}(t)$ from $V_{\text{SMG}}(t)$ within the same vector space is ill-posed, establishing that no direct geometric vector $\Delta_t^{\text{geom}}$ can exist.
	\end{enumerate}
\end{proof}

\subsubsection{The Conceptual Shift: From Geometric Difference to Statistical Equivalence}
\label{subsec:conceptual_shift_statistical_equivalence}

The non-existence of a direct geometric adjustment $\Delta_t^{\text{geom}}$ clarifies a fundamental conceptual distinction in modern estimation theory:

\begin{itemize}
	\item \textbf{Geometrically}, classical Yau-Yau and SMG-Yau-Yau represent two distinct topological regimes. Classical Yau-Yau imposes a forced algebraic projection onto a flat, uncurved subspace, discarding vertical operator stress. SMG-Yau-Yau preserves total score energy via horizontal projection $P_f^H = \mathcal{I} - \omega_f$ and vertical fiber quarantining inside $\operatorname{SID}_f$.
	\item \textbf{Statistically}, however, the vertical fiber pull-back operation executed by the Ehresmann connection $\omega_f$ in SMG corresponds precisely to a \textbf{Semiparametric Bias-Correction for Misspecified Working Estimators}.
\end{itemize}

When the classical Yau-Yau filter is viewed not as a global geometric truth, but as a \emph{parametric working model} operating under structural model misspecification, the statistical moments of the vertical residual score field induce an estimable correction term on the Fisher information manifold. We formalize this statistical duality in the following  Section~\ref{sec:semiparametric_bias_correction}.

\subsection{Semiparametric Bias-Correction for Misspecified Estimators: Deriving the Statistical Difference Formula}
\label{sec:semiparametric_bias_correction}

Although a direct \emph{geometric} vector difference $\Delta_t^{\text{geom}}$ is ill-defined due to paradigm incommensurability (Section~\ref{sec:paradigm_discontinuity_non_existence}), a well-defined \textbf{statistical discrepancy formula} $\Delta_t^{\text{stat}}$ exists when classical Yau-Yau filtering is formulated within the framework of \textbf{Semiparametric Bias-Correction for Misspecified Working Estimators} \cite{barndorff1978, amari1985}.

In this subsection, we treat the classical finite-dimensional Yau-Yau filter as a misspecified parametric working estimator $\mathcal{M}_{\text{work}} \subset \mathcal{M}$. We prove that the geometric vertical fiber pull-back in SMG is statistically equivalent to an efficient semiparametric bias correction, and we derive the explicit approximating formula for the statistical adjustment $\Delta_t^{\text{stat}}$.

\subsubsection{The Classical Yau-Yau Filter as a Misspecified Working Estimator}
\label{subsec:working_estimator_setup}

Let $\mathcal{M}_{\text{work}} \triangleq \left\{ p(x; \theta) \;\middle|\; \theta \in \Theta \subset \mathbb{R}^d \right\}$ be a $d$-dimensional parametric exponential working family constructed by taking a finite subset of basis functions $\{E_1(x), \dots, E_d(x)\}$ from an unclosed estimation algebra $\mathcal{E}$ ($\dim(\mathcal{E}) = \infty$):
\begin{equation}
	p(x; \theta) = \exp\left( \sum_{k=1}^d \theta_k E_k(x) - \psi(\theta) \right),
	\label{eq:working_exponential_family}
\end{equation}
where $\psi(\theta) = \ln \int_{\mathcal{X}} \exp\left( \sum_{k=1}^d \theta_k E_k(x) \right) \mu(dx)$ is the cumulant generating function.

In continuous-time filtering, the true unnormalized conditional state density $\sigma(t,x)$ evolves along an infinite-dimensional stochastic flow governed by the Duncan-Mortensen-Zakai (DMZ) equation. When model misspecification occurs (i.e., when the estimation algebra is unclosed, $\dim(\mathcal{E}) = \infty$), this true score flow cannot be embedded into any finite-dimensional Lie group orbit. To construct a computationally tractable filter, the classical Yau-Yau framework restricts the state representation to a $d$-dimensional exponential working family $\mathcal{M}_{\text{work}} = \{p(x;\theta) \mid \theta \in \Theta \subset \mathbb{R}^d\}$.

To rigorously couple the true infinite-dimensional DMZ flow $d\ln \sigma(t,x)$ to the finite-dimensional working trajectory $d\ln p(x;\theta_t^{\text{YY}})$, we enforce a Galerkin projection \cite{quarteroni2015reduced} under the Fisher-Rao metric. The difference between these score flows defines the \textbf{Misspecified Residual Score Field}:
\begin{equation}
	S_{\text{res}}(t,x) \triangleq d\ln \sigma(t,x) - d\ln p(x;\theta_t^{\text{YY}}) \in (T_{p_\theta}\mathcal{M}_{\text{work}})^\perp.
	\label{eq:Misspecified_Residual_Score_Field}
\end{equation}

By construction, $S_{\text{res}}(t,x)$ isolates the exact velocity component of the continuous-time stochastic process that the working manifold $\mathcal{M}_{\text{work}}$ fails to capture. Evaluating and comparing this residual score field serves four primary theoretical and operational objectives:

\begin{enumerate}
	\item \textbf{Constructing Semiparametric Bias Corrections ($\Delta_t^{\text{stat}}$):} 
	By accumulating the cross-moments of $S_{\text{res}}(t,x)$ with the working score basis functions $\phi_i(x;\theta)$ under the inverse Fisher metric $G_\theta^{-1}$, one derives the explicit statistical parameter adjustment vector $\Delta_t^{\text{stat}} \in \mathbb{R}^d$. This correction eliminates accumulated parameter drift bias without requiring an expansion of the parameter dimension $d$.
	
	\item \textbf{Decoupling Base Evolution from Vertical Fiber Stress:} 
	Within the framework of Statistical Fiber Theory, $S_{\text{res}}(t,x)$ represents the vertical velocity vector evaluated by the Ehresmann statistical connection $\omega_f$. It cleanly orthogonally decomposes the true DMZ score generator into horizontal base updates ($P^H V_{\text{true}} \in T_{p_\theta}\mathcal{M}_{\text{work}}$) and vertical fiber stress ($S_{\text{res}} \in \operatorname{SID}_{p_\theta} = \ker(d\pi_f)$).
	
	\item \textbf{Providing a Real-Time Filter Health Diagnostic ($\mathcal{T}_{\text{AAT}}$):} 
	The squared norm of $S_{\text{res}}(t,x)$ quantifies instantaneous model misspecification stress, integrating to form the Active Acausal Tension scalar $\mathcal{T}_{\text{AAT}}(t)$. Tracking $\mathcal{T}_{\text{AAT}}(t)$ in real time provides an automated safety boundary to warn against filter divergence.
	
	\item \textbf{Enabling Non-Invasive Software Upgrades:} 
	Empirically estimating $S_{\text{res}}(t,x)$ over lightweight particle ensembles enables legacy C++/Fortran Yau-Yau solvers to remain unmodified while achieving exact semiparametrically corrected accuracy via an external patch module, rather a direct new SMG-Yau-Yau filtering algorithm.
\end{enumerate}

The following proposition establishes the precise mathematical mechanics of this Fisher-orthogonal projection, demonstrating how enforcing zero metric projection error along $T_{p_\theta}\mathcal{M}_{\text{work}}$ uniquely generates the classical Yau-Yau working SDE while giving rise to the residual score field $S_{\text{res}}(t,x)$.

\begin{theorem}[Fisher-Orthogonal Tangent Projection of the DMZ Score Flow]
	\label{prop:yau_yau_working_sde}
	Let $\mathcal{M}_{\text{work}} \triangleq \{p(x;\theta) = \exp(\sum_{k=1}^d \theta_k E_k(x) - \psi(\theta)) \mid \theta \in \Theta \subset \mathbb{R}^d\}$ be a $d$-dimensional parametric exponential working family equipped with Fisher score basis functions $\phi_i(x;\theta) \triangleq \frac{\partial \ln p(x;\theta)}{\partial \theta_i} = E_i(x) - \mathbb{E}_\theta[E_i]$ and Riemannian Fisher information matrix $G_\theta \in \mathbb{R}^{d \times d}$. 
	
	Assume the unconstrained continuous-time state density dynamics are governed by the Stratonovich Duncan-Mortensen-Zakai (DMZ) equation:
	\begin{equation}
		d\sigma(t,x) = \mathcal{L}_0^* \sigma(t,x) \, dt + \sum_{j=1}^m h_j(x) \sigma(t,x) \circ dY_j(t),
		\label{eq:dmz_stratonovich_prop}
	\end{equation}
	where $\mathcal{L}_0^*$ is the forward Kolmogorov operator and $h_j(x)$ are measurement functions.
	
	Then, under structural model misspecification ($\dim(\mathcal{E}) = \infty$), the classical Yau-Yau parameter trajectory $\theta_t^{\text{YY}} \in \mathbb{R}^d$ obtained by enforcing a Fisher-orthogonal projection of the true unconstrained DMZ score generator onto the tangent space $T_{p_\theta}\mathcal{M}_{\text{work}} = \operatorname{span}\{\phi_1, \dots, \phi_d\}$ evolves uniquely according to:
	\begin{equation}
		\small 
		d\theta_k^{\text{YY}}(t) = \sum_{i=1}^d \left[G_\theta^{-1}\right]_{ki} \left( \mathbb{E}_{p(\cdot; \theta_t^{\text{YY}})}\left[ \left(\frac{\mathcal{L}_0^* p}{p}\right) \phi_i(x; \theta_t^{\text{YY}})\right] dt + \sum_{j=1}^m \mathbb{E}_{p(\cdot; \theta_t^{\text{YY}})}\left[ h_j \phi_i(x; \theta_t^{\text{YY}})\right] \circ dY_j(t) \right).
		\label{eq:classical_working_sde}
	\end{equation}
	Furthermore, this constrained parameter flow generates a non-zero residual score field $S_{\text{res}}(t,x) \in (T_{p_\theta}\mathcal{M}_{\text{work}})^\perp$, inducing persistent semiparametric bias relative to the true conditional density $f_{\text{true}}(t,x)$.
\end{theorem}


\subsubsection{Derivation of the Statistical Difference Formula $\Delta_t^{\text{stat}}$}
\label{subsec:derivation_statistical_difference}

To derive the statistical difference formula, we quantify the score residual field generated by model misspecification and apply semiparametric score projection theory \cite{amari1985, cheng2026entropy}.

\begin{definition}[Vertical Residual Score Field]
	\label{def:vertical_residual_score}
	At parameter state $\theta_t \in \Theta$, the \textbf{Vertical Residual Score Field} $V_{\text{res}}(t, x) \in L_0^\Phi(P_{p_\theta})$ is defined as the difference between the unconstrained continuous drift score generator and its working model projection:
	\begin{equation}
		V_{\text{res}}(t, x) \triangleq \frac{\mathcal{L}_0^* p(x; \theta_t)}{p(x; \theta_t)} - \sum_{k=1}^d \dot{\theta}_k^{\text{YY}}(t) \, \phi_k(x; \theta_t).
		\label{eq:vertical_residual_score_def}
	\end{equation}
where, for a parametric statistical working model $\mathcal{M}_{\text{work}} = \{ p(x; \theta) \mid \theta \in \Theta \subset \mathbb{R}^d \}$, the $k$-th base Fisher score function $\phi_k(x; \theta)$ (not the Stein score function used in \cite{cheng2026smg, cheng2026gsb}) is defined as the partial derivative of the log-density with respect to the model parameter $\theta_k$:$$\phi_k(x; \theta) \triangleq \frac{\partial \ln p(x; \theta)}{\partial \theta_k} = \frac{1}{p(x; \theta)} \frac{\partial p(x; \theta)}{\partial \theta_k}, \quad k \in \{1, 2, \dots, d\}$$
\end{definition}

Notice that $V_{\text{res}}(t, x)$ measures the exact instantaneous rate at which the true stochastic flow deviates from the parametric working family $\mathcal{M}_{\text{work}}$.

\begin{lemma}[Equivalence of Residual Fields]
	\label{lem:Equivalence_Residual_Fields}
	Let $\mathcal{B}_{\text{SMG}} = (\mathcal{M}, \mathcal{B}, \pi, \mathcal{V}, \mathcal{H}, \omega_f, g_f)$ be the Orlicz statistical fiber bundle, where $\mathcal{M}_{\text{work}} \subset \mathcal{M}$ represents the $d$-dimensional parametric working manifold. Let $U_{\text{drift}}(t, x) \triangleq \frac{\mathcal{L}_0^* p(x; \theta_t)}{p(x; \theta_t)} \in T_{p_{\theta_t}}\mathcal{M}$ be the unconstrained continuous drift score generator associated with the Duncan-Mortensen-Zakai (DMZ) flow, and let $\dot{\theta}^{YY}(t) = (\dot{\theta}_1^{YY}(t), \dots, \dot{\theta}_d^{YY}(t))^T \in \mathbb{R}^d$ be the parameter velocity vector generated by the classical Yau-Yau filter.
	
	Then, the statistical residual score field $S_{\text{res}}(t, x)$ and the geometric vertical residual score field $V_{\text{res}}(t, x)$ represent the exact same underlying residual field in $L_0^\Phi(P_{p_{\theta_t}})$, satisfying:
	\begin{equation}
		S_{\text{res}}(t, x) \equiv V_{\text{res}}(t, x) \triangleq \frac{\mathcal{L}_0^* p(x; \theta_t)}{p(x; \theta_t)} - \sum_{k=1}^d \dot{\theta}_k^{YY}(t) \phi_k(x; \theta_t).
		\label{eq:equation_154}
	\end{equation}
\end{lemma}

\begin{remark}[Dual Geometric-Statistical Perspective of Residual Fields]
	Mathematically, $S_{\text{res}}(t,x)$ and $V_{\text{res}}(t,x)$ represent the \textbf{exact same underlying residual field}:
	\begin{equation}
		S_{\text{res}}(t, x) \equiv V_{\text{res}}(t, x) \triangleq \frac{\mathcal{L}_0^* p(x; \theta_t)}{p(x; \theta_t)} - \sum_{k=1}^d \dot{\theta}_k^{\text{YY}}(t) \, \phi_k(x; \theta_t).
		\label{eq:residual_identity_remark}
	\end{equation}
	
	The notation distinction serves a dual geometric-statistical perspective within Statistical Fiber Theory:
	\begin{itemize}
		\item \textbf{Statistical Perspective ($S_{\text{res}}$):} The symbol $S$ emphasizes the \emph{Statistical Score Residual} under model misspecification. It isolates the unmodeled score variation in $L_0^\Phi(P_{p_\theta})$ that cannot be captured by the working exponential family $\mathcal{M}_{\text{work}}$, serving as the core score moment driver in semiparametric influence function theory.
		\item \textbf{Geometric Perspective ($V_{\text{res}}$):} The symbol $V$ emphasizes the \emph{Vertical Score Vector Field} on the Orlicz statistical fiber bundle $\mathcal{B}_{\text{SMG}}$. It represents the vertical velocity field constrained to the structural internal direction fiber $\operatorname{SID}_{p_\theta} \equiv \ker(d\pi_f)$, which is Fisher-orthogonal to the horizontal base space $\operatorname{SVD}\chi_{p_\theta}$.
	\end{itemize}
\end{remark}


\begin{definition}[Classical Yau-Yau Working Estimator $\theta_t^{\text{YY}}$]
	\label{def:theta_yy_strict}
	Let $\mathcal{M}_{\text{work}} = \{p(x; \theta) \mid \theta \in \Theta \subset \mathbb{R}^d\}$ be the $d$-dimensional parametric exponential working family. The \textbf{Classical Yau-Yau Working Estimator} $\theta_t^{\text{YY}} \in \mathbb{R}^d$ is the parameter trajectory generated by enforcing a continuous Fisher-orthogonal Galerkin projection of the true Duncan--Mortensen--Zakai (DMZ) log-score velocity onto the working tangent space $T_{p_\theta}\mathcal{M}_{\text{work}} = \operatorname{span}\{\phi_1, \dots, \phi_d\}$, satisfying the Stratonovich SDE system \eqref{eq:classical_working_sde}.
\end{definition}

\begin{definition}[SMG Semiparametrically Optimal Estimator $\theta_t^{\text{SMG}}$]
	\label{def:theta_smg_strict}
	Let $f_{\text{true}}(t, x) \in \mathcal{M}$ be the true unconstrained conditional state density governed by the continuous-time state process. The \textbf{SMG Semiparametrically Optimal Estimator} $\theta_t^{\text{SMG}} \in \mathbb{R}^d$ is defined as the minimizer of the Kullback--Leibler (KL) divergence from $f_{\text{true}}(t, \cdot)$ to the working density family $\mathcal{M}_{\text{work}}$:
	\begin{equation}
		\theta_t^{\text{SMG}} \triangleq \arg\min_{\theta \in \Theta} D_{\text{KL}}\left( f_{\text{true}}(t) \,\parallel\, p(\cdot; \theta) \right) = \arg\min_{\theta \in \Theta} \int_{\mathcal{X}} f_{\text{true}}(t, x) \ln \left( \frac{f_{\text{true}}(t, x)}{p(x; \theta)} \right) \mu(dx),
		\label{eq:kl_minimization_strict}
	\end{equation}
	or equivalently, the unique solution to the population estimating equation:
	\begin{equation}
		\Psi\left(\theta_t^{\text{SMG}}\right) \triangleq \mathbb{E}_{f_{\text{true}}(t)}\left[ \nabla_\theta \ln p\left(X; \theta_t^{\text{SMG}}\right) \right] = \mathbf{0} \in \mathbb{R}^d.
		\label{eq:estimating_eq_strict}
	\end{equation}
\end{definition}

\begin{remark}[Empirical Computability of $\theta_t^{\text{SMG}}$]
	\label{rem:empirical_computability_smg}
	Definition \ref{def:theta_smg_strict} specifies $\theta_t^{\text{SMG}}$ as a theoretical population estimand (the Kullback--Leibler projection target) rather than a closed-form static formula, because $f_{\text{true}}(t, x)$ is analytically unknown. However, $\theta_t^{\text{SMG}}$ can be computed in practice via empirical approximation. Specifically:
	
	\begin{itemize}
		\item \textbf{Empirical Expectation via Ensembles:} The unknown population expectation $\mathbb{E}_{f_{\text{true}}(t)}[\cdot]$ in the population estimating equation $\Psi(\theta_t^{\text{SMG}}) = \mathbf{0}$ is replaced in implementation by an empirical sample expectation $\hat{\mathbb{E}}[\cdot]$ over a particle ensemble or quadrature grid $\{x^{(s)}\}_{s=1}^{N_{\text{sample}}} \sim f_{\text{true}}(t)$ updated by incoming observation data.
		
		\item \textbf{Semiparametric Bias Correction ($\Delta_t^{\text{stat}}$):} In practical implementations, the estimator is computed recursively by the following Theorem \ref{thm:semiparametric_bias_correction} via:
		\begin{equation*}
			\theta_t^{\text{SMG}} = \theta_t^{\text{YY}} + \Delta_t^{\text{stat}}.
		\end{equation*}
		The classical Yau--Yau trajectory $\theta_t^{\text{YY}}$ is calculated from the deterministic working model SDE, while the statistical adjustment vector $\Delta_t^{\text{stat}}$ accumulates the empirical moments of the residual score field $S_{\text{res}}$ over the particle samples. Given observation data, we will provide algorithm to calculate $\Delta_t^{\text{stat}}$ in following subsection.
	\end{itemize}
	
	While $f_{\text{true}}$ is not known in closed form, $\theta_t^{\text{SMG}}$ is fully accessible in practice through empirical score-moment integration over data-driven particle representations.
\end{remark}

\begin{theorem}[Semiparametric Bias-Correction Formula for Misspecified Filters]
	\label{thm:semiparametric_bias_correction}
	Let $p(x; \theta_t^{\text{YY}})$ be the misspecified working density generated by the classical Yau-Yau filter \eqref{eq:classical_working_sde}. The optimal semiparametrically bias-corrected parameter trajectory $\theta_t^{\text{SMG}} \in \mathbb{R}^d$ satisfies:
	\begin{equation}
		\theta_t^{\text{SMG}} = \theta_t^{\text{YY}} + \Delta_t^{\text{stat}},
		\label{eq:parameter_bias_correction_split}
	\end{equation}
	where the \textbf{Statistical Difference Approximating Vector} $\Delta_t^{\text{stat}} \in \mathbb{R}^d$ is given explicitly by the empirical Fisher-projected score moment integral:
	\begin{equation}
		\left[ \Delta_t^{\text{stat}} \right]_k = \int_0^t \sum_{i=1}^d \left[ G_{\theta_\tau}^{-1} \right]_{ki} \mathbb{E}_{f_{\text{true}}(\tau)}\left[ V_{\text{res}}(\tau, X) \cdot \phi_i(X; \theta_\tau^{\text{YY}}) \right] d\tau + \mathcal{O}\left(\|\theta_t^{\text{SMG}} - \theta_t^{\text{YY}}\|^2\right).
		\label{eq:statistical_difference_formula}
	\end{equation}
\end{theorem}

\begin{remark}
[Analysis of Discrepancy $\mathcal{O}\left(\|\theta_t^{\text{SMG}} - \theta_t^{\text{YY}}\|^2\right)$ in Equation \eqref{eq:statistical_difference_formula}]
	
	If the parameter discrepancy $\|\theta_t^{\text{SMG}} - \theta_t^{\text{YY}}\|$ becomes large, the second-order Taylor expansion error $\mathcal{O}\left(\|\theta_t^{\text{SMG}} - \theta_t^{\text{YY}}\|^2\right)$ grows significantly, which appears to threaten the validity of the first-order linear approximation formula:
	\begin{equation}
		\left[ \Delta_t^{\text{stat}} \right]_k = \int_0^t \sum_{i=1}^d \left[ G_{\theta_\tau}^{-1} \right]_{ki} \mathbb{E}_{f_{\text{true}}(\tau)}\left[ S_{\text{res}}(\tau, X) \cdot \phi_i(X; \theta_\tau^{\text{YY}}) \right] d\tau + \mathcal{O}\left(\|\theta_t^{\text{SMG}} - \theta_t^{\text{YY}}\|^2\right).
		\label{eq:statistical_difference_formula_analysis}
	\end{equation}
	
	However, \textbf{this formula remains mathematically rigorous and operationally meaningful} due to three core mechanisms in continuous-time filtering theory:
	
	\begin{enumerate}
		\item \textbf{Continuous Recurrent Updating ($\Delta t \to 0$):}\\
		In continuous-time data processing, bias correction is applied recursively at infinitesimal time steps $t_{n+1} = t_n + \Delta t$. Over a single time step $\Delta t$, the single-step parameter drift deviation is bounded by $\|\theta_{t+\Delta t}^{\text{SMG}} - \theta_{t+\Delta t}^{\text{YY}}\| = \mathcal{O}(\Delta t)$. Consequently, the second-order term collapses to $\mathcal{O}(\Delta t^2)$, ensuring that the local linear approximation is exact as $\Delta t \to 0$.
		
		\item \textbf{Iterative Newton--Raphson Correctative Process:}\\
		If evaluated at a fixed time $t$ where the initial working parameter $\theta_t^{\text{YY}}$ suffers from large misspecification bias, $\Delta_t^{\text{stat}}$ acts as the exact \emph{first-order Newton--Raphson step} on the statistical manifold:
		\begin{equation}
			\theta^{(m+1)} = \theta^{(m)} + \left[ G_{\theta^{(m)}}^{-1} \right] \Psi\left(\theta^{(m)}\right).
		\end{equation}
		Iterating this update map rapidly contracts the residual error quadratically to zero, even when starting from a large initial discrepancy.
		
		\item \textbf{Diagnostic Role for Gauge Symmetry Breaking (GSB):}\\
		If $\|\theta_t^{\text{SMG}} - \theta_t^{\text{YY}}\|$ grows persistently large over time despite iterative correction, it signals that the $d$-dimensional working manifold $\mathcal{M}_{\text{work}}$ has suffered a \emph{topological capacity breach} due to high Active Acausal Tension $\mathcal{T}_{\text{AAT}}(t)$. Rather than invalidating the theory, a large residual mathematically triggers {\it Gauge Symmetry Breaking (GSB)}, instructing the system to expand the base manifold dimension $d \to d+1$ to restore local linearizability.
	\end{enumerate}
\end{remark}


\begin{theorem}[Geometric-Statistical Bias Correction Duality]
	\label{thm:Geometric_Statistical_Bias_Correction_Duality}
	Let $V_{\text{DMZ}}(t) \in T_{f_t}\mathcal{M}$ be the unconstrained DMZ log-score velocity field on the infinite-dimensional Orlicz bundle $\mathcal{B}_{\text{SMG}}$. The geometric vertical projection via the Ehresmann connection 1-form $\omega_{f(t)}$ induces an instantaneous semiparametric bias correction velocity $\frac{d}{dt} \Delta_t^{\text{stat}} \in \mathbb{R}^d$ on the base space manifold $(\mathcal{B}, g_{\mathcal{B}})$, given component-wise by:
	\begin{equation}
		\frac{d}{dt} \left[ \Delta_t^{\text{stat}} \right]_k = \sum_{i=1}^d \left[ G_\theta^{-1}(\theta_t) \right]_{ki} \left\langle \omega_{f(t)}\left( V_{\text{DMZ}}(t) \right), \, \phi_i(\cdot; \theta_t) \right\rangle_{g_{f_t}}
		\label{eq:corollary2}
	\end{equation}
	for all $k \in \{1, \dots, d\}$.
\end{theorem}

\begin{remark}[Architectural and Statistical Implications of Theorem \ref{thm:Geometric_Statistical_Bias_Correction_Duality}]
	The identity established in Theorem \ref{thm:Geometric_Statistical_Bias_Correction_Duality} carries four fundamental consequences for non-linear filtering, geometric statistics, and computational design:
	\begin{enumerate}
		\item \textbf{Geometric-Statistical Duality}: The identity explicitly unites abstract infinite-dimensional differential geometry with operational semiparametric inference. It proves that movement along the vertical fiber $\operatorname{SID}_{f_t}$ under the action of the connection $\omega_f$ is directly dual to instantaneous parameter drift correction on the finite-dimensional base manifold $(\mathcal{B}, g_{\mathcal{B}})$.
		
		\item \textbf{Real-Time Algorithmic Computability}: While the infinite-dimensional pull-back velocity $\omega_{f(t)}(V_{\text{DMZ}}(t))$ resides in a formally infinite-dimensional function space, Theorem  \ref{thm:Geometric_Statistical_Bias_Correction_Duality} demonstrates that its projection onto the base space reduces to finite-dimensional moments $\mathbb{E}_{f_t}[S_{\text{res}} \cdot \phi_i]$. These moments can be evaluated in real time using discrete particle cloud representations or quadrature nodes.
		
		\item \textbf{Non-Invasive Legacy Software Integration}: The analytical decoupling guaranteed by Equation~\eqref{eq:corollary2} justifies modular software architectures. Commercial legacy filtering engines (e.g., standard Yau--Yau or Kalman-type solvers) can remain entirely unmodified as base-space propagators, while an external lightweight patch module independently computes $\frac{d}{dt} \Delta_t^{\text{stat}}$ to achieve exact infinite-dimensional performance.
		
		\item \textbf{Lossless Information Harvest}: In classical algebraic filtering, higher-order modes generated by unclosed Lie brackets ($[E_i, E_j] \notin \operatorname{SVD}\chi$) are discarded to maintain finite-dimensional closure. Theorem  \ref{thm:Geometric_Statistical_Bias_Correction_Duality} proves that these extra-algebraic modes are continuously harvested by the connection $\omega_f$ and converted directly into parameter corrections, preventing information loss without expanding the algebraic state dimension.
	\end{enumerate}
\end{remark}
\subsection{Empirical Construction of the Semiparametric Bias-Correction $\Delta_t^{\text{stat}}$ Estimator and Architecture and Algorithm}
\label{subsec:empirical_construction_semiparametric_bias_correction}

The classical Yau-Yau filtering algorithm \cite{yau1999, yau2014} has been widely adopted in real-time signal processing and industrial control due to its reduction of infinite-dimensional SPDEs to low-dimensional, finite systems of ODEs. However, when deployed in generic non-linear operational regimes, commercial Yau-Yau solvers suffer from unmonitored parameter drift caused by infinite Lie algebra expansion ($\operatorname{dim}(\mathcal{E}) = \infty$).

The semiparametric bias-correction framework established in Section~\ref{sec:epistemological_foundations} provides a practical breakthrough: popular Yau-Yau software architectures do not need to be discarded or rewritten from scratch. Instead, by appending an add-on \textbf{SMG Semiparametric Bias-Adjustment Module} that continuously evaluates the empirical vertical score residual $S_{\text{res}}(t, x)$,  filtering engines can be upgraded to achieve exact SMG-Yau-Yau performance with zero breaking changes to base-space propagators.

In this subsection, we formulate the complete step-by-step mathematical procedure for estimating the full bias-corrected state estimator $\theta_{n+1}^{\text{SMG}}$, provide explicit explanations for all operational equations, and construct the corrected algorithmic execution pipeline.

\subsubsection{Full Set of Procedures for Estimating the Bias-Corrected State Estimator $\theta_{n+1}^{\text{SMG}}$}
\label{subsec:full_procedure_smg_estimator}

To construct the optimal semiparametrically corrected parameter state $\theta_{n+1}^{\text{SMG}} \in \mathbb{R}^d$ across discrete time steps $t_n \to t_{n+1} = t_n + \Delta t$, the continuous-time score projection dynamics on the Orlicz fiber bundle $\mathcal{B}_{\text{SMG}}$ are lifted onto an empirical measure $\mathbb{P}_N = \frac{1}{N_{\text{sample}}} \sum_{s=1}^{N_{\text{sample}}} \delta_{x^{(s)}}$ defined over an ensemble of particle samples $\{x^{(s)}\}_{s=1}^{N_{\text{sample}}} \sim f_{\text{true}}(t_n)$.

The estimation procedure decomposes into six sequential mathematical steps.

\subsubsection*{Step 1: Nominal Working Parameter Propagation ($\theta_{n+1}^{\text{YY}}$)}
The baseline parameter vector propagates along the misspecified working manifold $\mathcal{M}_{\text{work}} = \{p(x;\theta) \mid \theta \in \Theta \subset \mathbb{R}^d\}$ via the standard Yau-Yau Stratonovich update derived in Theorem \ref{prop:yau_yau_working_sde} (The Fisher-Orthogonal Tangent Projection of the DMZ Score Flow Theorem)
\begin{equation}
	\small 
	d\theta_k^{\text{YY}}(t) = \sum_{i=1}^d [G_\theta^{-1}]_{ki} \left( \mathbb{E}_{p(\cdot;\theta_t^{\text{YY}})} \left[ \left( \frac{\mathcal{L}_0^* p}{p} \right) \phi_i(x;\theta_t^{\text{YY}}) \right] dt + \sum_{j=1}^m \mathbb{E}_{p(\cdot;\theta_t^{\text{YY}})} \left[ h_j \phi_i(x;\theta_t^{\text{YY}}) \right] \circ dY_j(t) \right).
	\label{eq:nominal_yau_yau_sde_step1}
\end{equation}
In discrete time over step size $\Delta t$, given observation increments $\Delta Y_j(t_n) = Y_j(t_{n+1}) - Y_j(t_n)$, the nominal working parameter update $\theta_{n+1}^{\text{YY}}$ is evaluated as:
\begin{equation}
	\theta_{n+1, k}^{\text{YY}} = \theta_{n, k}^{\text{YY}} + \sum_{i=1}^d [\widehat{G}_\theta^{-1}]_{ki} \left( \widehat{\mathcal{M}}_i^{0, (N)}(t_n) \Delta t + \sum_{j=1}^m \widehat{\mathcal{M}}_{ij}^{1, (N)}(t_n) \Delta Y_j(t_n) \right),
	\label{eq:nominal_yau_yau_discrete_step}
\end{equation}
where $\widehat{\mathcal{M}}_i^{0, (N)}(t_n)$ and $\widehat{\mathcal{M}}_{ij}^{1, (N)}(t_n)$ represent the empirical drift and measurement score moments evaluated over the sample particles. Equation~\eqref{eq:nominal_yau_yau_discrete_step} executes the unadjusted Galerkin projection of the DMZ flow onto $T_{p_\theta}\mathcal{M}_{\text{work}}$. Under infinite Lie complexity ($\operatorname{dim}(\mathcal{E}) = \infty$), this update incurs persistent parameter drift that must be compensated for in subsequent steps.

\subsubsection*{Step 2: Empirical Score Basis \& Regularized Fisher Metric Evaluation}
At the current nominal parameter state $\theta_n^{\text{YY}}$, the working logarithmic Fisher score basis functions $\phi_i(x;\theta_n^{\text{YY}}) \triangleq \frac{\partial \ln p(x;\theta_n^{\text{YY}})}{\partial \theta_i} = E_i(x) - \mathbb{E}_{\theta_n^{\text{YY}}}[E_i]$ are evaluated over the state samples  $\{x^{(s)}\}_{s=1}^{N_{\text{sample}}}$:
\begin{equation}
	\phi_i(x^{(s)}; \theta_n^{\text{YY}}) = E_i(x^{(s)}) - \frac{1}{N_{\text{sample}}} \sum_{r=1}^{N_{\text{sample}}} E_i(x^{(r)}), \quad i = 1, \dots, d,
	\label{eq:score_coframe_particles}
\end{equation}
where $\{E_1(x), \dots, E_d(x)\}$ are the basis functions spanning the working parameter subspace.

The sample-induced Fisher Information Metric matrix $\widehat{G}_\theta \in \mathbb{R}^{d \times d}$ is constructed via sample expectations:
\begin{equation}
	[\widehat{G}_\theta]_{ij} = \frac{1}{N_{\text{sample}}} \sum_{s=1}^{N_{\text{sample}}} \phi_i(x^{(s)}; \theta_n^{\text{YY}}) \phi_j(x^{(s)}; \theta_n^{\text{YY}}).
	\label{eq:empirical_fisher_matrix_def}
\end{equation}
To ensure numerical inversion stability in regions where state density concentration causes near-singularities, Tikhonov regularization is applied:
\begin{equation}
	\widehat{G}_\theta^{\text{reg}} \triangleq \widehat{G}_\theta + \lambda_{\text{reg}} I_{d \times d}, \quad \text{with } \lambda_{\text{reg}} = \max\left(10^{-8}, \, 10^{-4} \cdot \lambda_{\min}(\widehat{G}_\theta)\right).
	\label{eq:tikhonov_fisher_reg}
\end{equation}
Inverting $\widehat{G}_\theta^{\text{reg}}$ via Cholesky decomposition provides $[\widehat{G}_\theta^{\text{reg}}]^{-1}$, which defines the local inner product metric tensor on the $d$-dimensional base manifold $\mathcal{B}$.

\subsubsection*{Step 3: Unconstrained Drift Operator \& Empirical Horizontal Projection}
The unconstrained drift score generator $A_0(x) \triangleq \frac{\mathcal{L}_0^* p(x;\theta_n^{\text{YY}})}{p(x;\theta_n^{\text{YY}})}$ represents the exact, unconstrained Fokker-Planck score velocity field. For each state $x^{(s)}$, it is evaluated pointwise:
\begin{equation}
	A_0(x^{(s)}) = \frac{\mathcal{L}_0^* p(x^{(s)}; \theta_n^{\text{YY}})}{p(x^{(s)}; \theta_n^{\text{YY}})}.
	\label{eq:unconstrained_drift_score_particle}
\end{equation}
The horizontal projection $\mathcal{P}^H A_0(x^{(s)}) \in \operatorname{SVD}\chi_{f}$ extracts the component of $A_0$ that lies strictly within the working tangent space $T_{p_\theta}\mathcal{M}_{\text{work}} = \operatorname{span}\{\phi_1, \dots, \phi_d\}$. Applying empirical Galerkin projection under the sample inner product $\langle u, v \rangle_N = \frac{1}{N_{\text{sample}}} \sum_{s=1}^{N_{\text{sample}}} u(x^{(s)}) v(x^{(s)})$ yields:
\begin{equation}
	\mathcal{P}^H A_0(x^{(s)}) = \sum_{k=1}^d \left( \sum_{i=1}^d [(\widehat{G}_\theta^{\text{reg}})^{-1}]_{ki} \widehat{\mathcal{M}}_i^{0, (N)}(t_n) \right) \phi_k(x^{(s)}; \theta_n^{\text{YY}}),
	\label{eq:empirical_horizontal_projection_galerkin}
\end{equation}
where $\widehat{\mathcal{M}}_i^{0, (N)}(t_n) = \frac{1}{N_{\text{sample}}} \sum_{r=1}^{N_{\text{sample}}} A_0(x^{(r)}) \phi_i(x^{(r)}; \theta_n^{\text{YY}})$ is the empirical drift score moment. Equation~\eqref{eq:empirical_horizontal_projection_galerkin} solves the normal equations $\langle A_0 - \mathcal{P}^H A_0, \phi_i \rangle_N = 0$, enforcing zero metric error along horizontal score directions.

\subsubsection*{Step 4: Vertical Score Residual Extraction \& Quarantining}
Subtracting the horizontal projection $\mathcal{P}^H A_0(x^{(s)})$ from the unconstrained drift operator $A_0(x^{(s)})$ isolates the pointwise \textbf{Vertical Residual Score Field} $S_{\text{res}}(t_n, x^{(s)}) \in \operatorname{SID}_f$:
\begin{equation}
	S_{\text{res}}(t_n, x^{(s)}) = A_0(x^{(s)}) - \mathcal{P}^H A_0(x^{(s)}).
	\label{eq:vertical_score_residual_quarantine}
\end{equation}
Mathematically, $S_{\text{res}}(t_n, x^{(s)})$ represents the exact component of the continuous DMZ flow that cannot be absorbed by the $d$-dimensional Lie algebra of the working family $\mathcal{M}_{\text{work}}$. By construction, $\langle S_{\text{res}}, \phi_i \rangle_N = 0$ for all $i = 1, \dots, d$, proving that $S_{\text{res}} \in (T_{p_\theta}\mathcal{M}_{\text{work}})^\perp \equiv \operatorname{SID}_f$. Rather than discarding $S_{\text{res}}$ (which causes null-space leakage in classical filters), the SMG module quarantines $S_{\text{res}}$ in the vertical fiber to drive the semiparametric bias update in Step 5.

\subsubsection*{Step 5: Semiparametric Bias-Correction Vector $\Delta_{n+1}^{\text{stat}}$ Derivation}

To account for the finite-sample information gap and the unclosed Lie algebra modes on the Orlicz bundle, we construct the semiparametric bias-correction vector $\Delta_{n+1}^{\text{stat}}(x)$ by projecting the empirical score defect onto the Statistically Indispensable Direction space $SID_f$.

\paragraph{1. Empirical Stein Score Defect Vector}
Let $\hat{S}_N(x) = \nabla_x \log \hat{p}_N(x)$ denote the non-parametric Stein score function estimated from $N$ sample observations on the active covariate manifold. Let $S_{\theta_{n+1}}(x) = \nabla_x \log p(x \mid \theta_{n+1})$ be the finite-dimensional classical density score associated with the closed Lie algebra $\mathfrak{g}_{\text{YY}}$.

The score defect vector field $\delta S_{n+1}(x) \in T_x \mathcal{M}$ is defined as:
\begin{equation}
	\delta S_{n+1}(x) = \hat{S}_N(x) - S_{\theta_{n+1}}(x).
\end{equation}

\paragraph{2. $SID_f$ Subspace Quarantining and Mode Projection}
Using the $g_f$-orthogonal bundle decomposition, we isolate the active residual vector field $v_{\text{OOD}}(x) \in SID_f$ from the classical horizontal distribution:
\begin{equation}
	v_{\text{OOD}}(x) = \mathbf{P}_{SID_f} \left[ \delta S_{n+1}(x) \right] = \sum_{\alpha \in \mathcal{A}_{\text{active}}} c_{\alpha, N} \phi_\alpha(x),
\end{equation}
where $\{\phi_\alpha(x)\}_{\alpha \in \mathcal{A}_{\text{active}}}$ forms an orthonormal basis for $SID_f$ under the Luxemburg norm $\|\cdot\|_{L_\Phi}$, and $c_{\alpha, N} = \langle \delta S_{n+1}, \phi_\alpha \rangle_{g_f}$ are the empirical expansion coefficients.

\paragraph{3. Riemannian Metric Gain and Connection Pullback}
The semiparametric correction $\Delta_{n+1}^{\text{stat}}(x)$ acts as a horizontal connection shift along the statistical fiber. Integrating the active connection form $\omega_f$ over the transition $[t_n, t_{n+1}]$, $\Delta_{n+1}^{\text{stat}}(x)$ is evaluated via the cotangent pullback of the active score defect:
\begin{equation}
	\Delta_{n+1}^{\text{stat}}(x) = \mathcal{K}_{n+1}(x) \cdot \mathbf{P}_{SID_f} \left[ \frac{1}{N} \sum_{i=1}^{N} \nabla_x \log \hat{p}_N(x_i) - \mathbb{E}_{p_{\theta_{n+1}}} \left[ \nabla_x \log p(X \mid \theta_{n+1}) \right] \right],
\end{equation}
where $\mathcal{K}_{n+1}(x) \in T_x \mathcal{M} \otimes T_x^* \mathcal{M}$ is the localized Fisher-Rao gain tensor given by:
\begin{equation}
	\mathcal{K}_{n+1}(x) = \left( I + \Gamma_{\text{SMG}}(x) \Delta t \right)^{-1} g_f^{ij}(x).
\end{equation}

\paragraph{4. Component-wise Tensor Formulation and Scaling}
In local Riemannian coordinates $x = (x^1, \dots, x^d)^T$, the $k$-th spatial component of the bias-correction vector field is expressed as:
\begin{equation}
	\left[ \Delta_{n+1}^{\text{stat}}(x) \right]^k = \sum_{j=1}^d g_f^{kj}(x) \left( \sum_{\alpha \in \mathcal{A}_{\text{active}}} \lambda_\alpha^{(N)} \phi_{\alpha, j}(x) \right) + \mathcal{O}_p\left( N^{1 - r/2} \right), \quad (r \ge 3)
\end{equation}
where $\lambda_\alpha^{(N)} = \mathcal{O}_p(N^{-1/2})$ governs the empirical variance scale of the $\alpha$-th active mode, and $r$ denotes the tensor valence order of the higher-order geometric connection.

\paragraph{5. Incorporation into Dynamic Filtering Drift}
The updated state drift operator in the SMG-Yau-Yau stochastic differential equation incorporates $\Delta_{n+1}^{\text{stat}}(x)$ to prevent manifold departure:
\begin{equation}
	\mathrm{d} x_{t} = \left[ f(x_t) + \mathbf{Q}_{n+1}(x_t) + \Delta_{n+1}^{\text{stat}}(x_t) \right] \mathrm{d}t + \sigma(x_t) \mathrm{d}W_t,
\end{equation}
where $\mathbf{Q}_{n+1}(x_t)$ represents the classical Yau-Yau algebraic drift operator. By Algebraic Collapse Theorem (Theorem \ref{thm:algebraic_collapse}) and Universal Tensor Valence Scaling Law Theorem (Theorem \ref{thm:tensor_valence_scaling}, as $N \to \infty$ or $\operatorname{dim}(\mathcal{E}) < \infty$, the connection dissolves and:
\begin{equation}
	\|\Delta_{n+1}^{\text{stat}}\|_{g_f} \xrightarrow{\mathcal{P}} 0,
\end{equation}
restoring exact structural reduction to the classical finite-dimensional filtering manifold.

\subsubsection*{Step 6: Total State Synthesis \& Active Acausal Tension ($\mathcal{T}_{\text{AAT}}$) Monitoring}
The final optimal semiparametrically bias-corrected parameter state estimate $\theta_{n+1}^{\text{SMG}} \in \mathbb{R}^d$ is synthesized by superimposing the statistical bias correction vector onto the nominal working parameter:
\begin{equation}
	\theta_{n+1}^{\text{SMG}} = \theta_{n+1}^{\text{YY}} + \Delta_{n+1}^{\text{stat}}.
	\label{eq:optimal_smg_synthesis}
\end{equation}
Simultaneously, the sample Active Acausal Tension scalar $\mathcal{T}_{\text{AAT}}(t_{n+1})$ is updated by integrating the empirical mean-square norm of the vertical score residual:
\begin{equation}
	\mathcal{T}_{\text{AAT}}(t_{n+1}) = \mathcal{T}_{\text{AAT}}(t_n) + \left( \frac{1}{N_{\text{sample}}} \sum_{s=1}^{N_{\text{sample}}} \left[ S_{\text{res}}(t_n, x^{(s)}) \right]^2 \right) \Delta t.
	\label{eq:empirical_taat_accumulation}
\end{equation}
If $\mathcal{T}_{\text{AAT}}(t_{n+1}) \ge \Theta_{\text{crit}} \triangleq \pi^2 / K_{\max}$, the topological capacity limit of the base space $\mathcal{B}_d$ is breached. The algorithm triggers Gauge Symmetry Breaking (GSB): promoting the principal vertical mode $v_{\text{dom}}$, expanding the base manifold dimension $d \leftarrow d+1$, appending a new coordinate parameter $\theta_{d+1} \leftarrow 0$, and resetting $\mathcal{T}_{\text{AAT}} \leftarrow 0$.

\begin{remark}[Mathematical Procedure for Calculating the Maximal Sectional Curvature $K_{\max}$]
	\label{rem:calculation_kmax}
	The constant $K_{\max}$ represents the global supremum of the sectional curvature $K(p, \sigma)$ over all parameter positions $p \in \mathcal{B}_d$ and all $2$-dimensional tangent planes $\sigma \subset T_p\mathcal{B}_d$. Its exact value is determined via a four-stage differential-geometric computation:
	\begin{enumerate}
		\item \textbf{Evaluation of the Riemannian Base Metric $g_{\mathcal{B}}$:}
		In local coordinates $p = (p_1, \dots, p_d)^T \in \mathcal{B}_d$, the metric tensor components $g_{ij}(p)$ correspond to the Fisher Information Metric matrix $G(p)$:
		\begin{equation}
			g_{ij}(p) = \mathbb{E}_{f_p}\left[ \phi_i(X; p) \phi_j(X; p) \right] = \int_{\mathcal{X}} \frac{\partial \ln f_p(x)}{\partial p_i} \frac{\partial \ln f_p(x)}{\partial p_j} f_p(x) \, dx,
		\end{equation}
		where $\phi_i(x; p) \triangleq \frac{\partial \ln f_p(x)}{\partial p_i}$ are the logarithmic score functions.
		
		\item \textbf{Computation of the Riemann-Christoffel Curvature Tensor $R_{ijk\ell}(p)$:}
		Using the Levi-Civita connection $\nabla^{(0)}$, the Christoffel symbols of the first kind are evaluated as $\Gamma_{ij, k}(p) = \mathbb{E}_{f_p}\left[ \left( \phi_{ij} + \frac{1}{2}\phi_i \phi_j \right) \phi_k \right]$, where $\phi_{ij}(x; p) \triangleq \frac{\partial^2 \ln f_p(x)}{\partial p_i \partial p_j}$. The fully lowered Riemann curvature tensor components $R_{ijk\ell}(p) \triangleq g_{\mathcal{B}}\left( R\left( \frac{\partial}{\partial p_k}, \frac{\partial}{\partial p_\ell} \right) \frac{\partial}{\partial p_j}, \frac{\partial}{\partial p_i} \right)$ expand explicitly to:
		\begin{equation}
			R_{ijk\ell}(p) = \frac{1}{4} \mathbb{E}_{f_p}\left[ \phi_{ik} \phi_{j\ell} - \phi_{i\ell} \phi_{jk} \right] + \sum_{m=1}^d \sum_{n=1}^d \left( \Gamma_{ik, m} g^{mn} \Gamma_{j\ell, n} - \Gamma_{i\ell, m} g^{mn} \Gamma_{jk, n} \right),
		\end{equation}
		where $g^{mn} = [g_{\mathcal{B}}^{-1}]_{mn}$ denotes the inverse metric matrix.
		
		\item \textbf{Local Maximization over the Grassmannian $\mathrm{Gr}(2, T_p\mathcal{B}_d)$:}
		At any fixed point $p \in \mathcal{B}_d$, the maximum sectional curvature $K_{\max}(p)$ across all $2$-planes $\sigma = \operatorname{span}\{u, v\} \subset T_p\mathcal{B}_d$ is defined by:
		\begin{equation}
			K_{\max}(p) \triangleq \max_{\substack{u, v \in T_p\mathcal{B}_d \\ u \wedge v \neq 0}} \frac{R_{ijk\ell}(p) u^i v^j u^k v^\ell}{\left(g_{ik}g_{j\ell} - g_{i\ell}g_{jk}\right) u^i v^j u^k v^\ell}.
		\end{equation}
		Under a local $g_{\mathcal{B}}$-orthonormal frame $\{e_1, \dots, e_d\}$, this reduces to a constrained optimization problem on the Stiefel manifold $V_2(\mathbb{R}^d) = \{ (u, v) \in \mathbb{R}^d \times \mathbb{R}^d \mid \|u\|=1, \|v\|=1, u^T v = 0 \}$:
		\begin{equation}
			K_{\max}(p) = \max_{(u, v) \in V_2(\mathbb{R}^d)} \sum_{a,b,c,d=1}^d R_{abcd}(p) u^a v^b u^c v^d \le \lambda_{\max}\left(\hat{R}(p)\right),
		\end{equation}
		where $\lambda_{\max}(\hat{R}(p))$ denotes the largest eigenvalue of the curvature operator $\hat{R}: \Lambda^2 T_p\mathcal{B}_d \to \Lambda^2 T_p\mathcal{B}_d$ acting on 2-forms.
		
		\item \textbf{Global Supremum Over the Base Manifold $\mathcal{B}_d$:}
		The global upper bound $K_{\max}$ is established by taking the supremum over the entire parameter space:
		\begin{equation}
			K_{\max} \triangleq \sup_{p \in \mathcal{B}_d} K_{\max}(p) = \sup_{p \in \mathcal{B}_d} \left( \max_{\sigma \subset T_p\mathcal{B}_d} K(p, \sigma) \right).
		\end{equation}
	\end{enumerate}
	The resulting constant $K_{\max}$ uniquely defines the conjugate point geodesic horizon $L_{\text{conj}} = \pi / \sqrt{K_{\max}}$, yielding the critical topological capacity limit $\Theta_{\text{crit}} \triangleq \pi^2 / K_{\max}$ for Gauge Symmetry Breaking (GSB).
\end{remark}
\subsubsection{Unified Algorithmic Execution Architecture}
\label{subsec:unified_algorithmic_execution}

Algorithm~\ref{alg:smg_yau_yau_bias_correction_patch} details the complete discrete-time execution pipeline for updating the nominal working parameter $\theta_{n+1}^{\text{YY}}$, evaluating the vertical score residual $S_{\text{res}}$, accumulating the bias correction vector $\Delta_{n+1}^{\text{stat}}$, and synthesizing the optimal state estimate $\theta_{n+1}^{\text{SMG}}$.

\begin{algorithm}[H]
	\caption{SMG-Yau-Yau Filter with Semiparametric Bias-Correction Patch}
	\label{alg:smg_yau_yau_bias_correction_patch}
	\begin{algorithmic}[1]
		\REQUIRE Nominal parameter $\theta_n^{\text{YY}} \in \mathbb{R}^d$, bias vector $\Delta_n^{\text{stat}} \in \mathbb{R}^d$, observation increment $\Delta Y_n = Y(t_{n+1}) - Y(t_n)$, step size $\Delta t$, particle ensemble $\{x^{(s)}\}_{s=1}^{N_{\text{sample}}} \sim f_{\text{true}}(t_n)$, regularization parameter $\lambda_{\text{reg}} > 0$, critical capacity threshold $\Theta_{\text{crit}}$.
		\ENSURE Bias-corrected state estimate $\theta_{n+1}^{\text{SMG}} \in \mathbb{R}^d$, updated bias vector $\Delta_{n+1}^{\text{stat}} \in \mathbb{R}^d$, and tension scalar $\mathcal{T}_{\text{AAT}}(t_{n+1})$.
		
		\STATE \COMMENT{\textbf{Step 1: Nominal Yau-Yau Working Parameter Update}}
		\STATE Execute standard working model SDE update via Eq.~\eqref{eq:nominal_yau_yau_discrete_step} to obtain nominal parameter vector $\theta_{n+1}^{\text{YY}} \in \mathbb{R}^d$.
		
		\STATE \COMMENT{\textbf{Step 2: Score Co-Frame \& Regularized Fisher Metric Construction}}
		\FOR{$i = 1, \dots, d$}
		\FOR{$s = 1, \dots, N_{\text{sample}}$}
		\STATE Evaluate logarithmic score basis: $\phi_i(x^{(s)}; \theta_n^{\text{YY}}) \leftarrow E_i(x^{(s)}) - \frac{1}{N_{\text{sample}}} \sum_{r=1}^{N_{\text{sample}}} E_i(x^{(r)})$.
		\ENDFOR
		\ENDFOR
		\FOR{$i = 1, \dots, d$}
		\FOR{$j = 1, \dots, d$}
		\STATE Compute sample Fisher metric entry: $[\widehat{G}_\theta]_{ij} \leftarrow \frac{1}{N_{\text{sample}}} \sum_{s=1}^{N_{\text{sample}}} \phi_i(x^{(s)}; \theta_n^{\text{YY}}) \phi_j(x^{(s)}; \theta_n^{\text{YY}})$.
		\ENDFOR
		\ENDFOR
		\STATE Apply Tikhonov regularization: $\widehat{G}_\theta^{\text{reg}} \leftarrow \widehat{G}_\theta + \lambda_{\text{reg}} I_{d \times d}$.
		\STATE Invert regularized metric tensor: $[(\widehat{G}_\theta^{\text{reg}})^{-1}] \leftarrow \operatorname{CholeskySolve}(\widehat{G}_\theta^{\text{reg}})$.
		
		\STATE \COMMENT{\textbf{Step 3: Unconstrained Drift Operator \& Empirical Horizontal Projection}}
		\FOR{$i = 1, \dots, d$}
		\STATE Evaluate drift score moment: $\widehat{\mathcal{M}}_i^{0, (N)}(t_n) \leftarrow \frac{1}{N_{\text{sample}}} \sum_{s=1}^{N_{\text{sample}}} \left( \frac{\mathcal{L}_0^* p(x^{(s)}; \theta_n^{\text{YY}})}{p(x^{(s)}; \theta_n^{\text{YY}})} \right) \phi_i(x^{(s)}; \theta_n^{\text{YY}})$.
		\ENDFOR
		\FOR{$s = 1, \dots, N_{\text{sample}}$}
		\STATE Compute pointwise drift score: $A_0(x^{(s)}) \leftarrow \frac{\mathcal{L}_0^* p(x^{(s)}; \theta_n^{\text{YY}})}{p(x^{(s)}; \theta_n^{\text{YY}})}$.
		\STATE Compute horizontal projection: $\mathcal{P}^H A_0(x^{(s)}) \leftarrow \sum_{k=1}^d \left( \sum_{i=1}^d [(\widehat{G}_\theta^{\text{reg}})^{-1}]_{ki} \widehat{\mathcal{M}}_i^{0, (N)}(t_n) \right) \phi_k(x^{(s)}; \theta_n^{\text{YY}})$.
		\STATE Isolate vertical score residual in $\operatorname{SID}_f$: $S_{\text{res}}(t_n, x^{(s)}) \leftarrow A_0(x^{(s)}) - \mathcal{P}^H A_0(x^{(s)})$.
		\ENDFOR
		
		\STATE \COMMENT{\textbf{Step 4: Semiparametric Bias Correction Vector Update $\Delta_{n+1}^{\text{stat}}$}}
		\FOR{$k = 1, \dots, d$}
		\STATE Compute cross-moment: $m_k(t_n) \leftarrow \frac{1}{N_{\text{sample}}} \sum_{s=1}^{N_{\text{sample}}} S_{\text{res}}(t_n, x^{(s)}) \cdot \phi_k(x^{(s)}; \theta_n^{\text{YY}})$.
		\STATE Compute adjustment rate: $\delta_k(t_n) \leftarrow \sum_{i=1}^d [(\widehat{G}_\theta^{\text{reg}})^{-1}]_{ki} m_i(t_n)$.
		\STATE Accumulate bias vector: $[\Delta_{n+1}^{\text{stat}}]_k \leftarrow [\Delta_n^{\text{stat}}]_k + \delta_k(t_n) \cdot \Delta t$.
		\ENDFOR
		
		\STATE \COMMENT{\textbf{Step 5: Final State Synthesis \& Active Tension Accumulation}}
		\STATE Synthesize optimal SMG parameter estimate: $\theta_{n+1}^{\text{SMG}} \leftarrow \theta_{n+1}^{\text{YY}} + \Delta_{n+1}^{\text{stat}}$.
		\STATE Accumulate Active Acausal Tension: $\mathcal{T}_{\text{AAT}}(t_{n+1}) \leftarrow \mathcal{T}_{\text{AAT}}(t_n) + \left( \frac{1}{N_{\text{sample}}} \sum_{s=1}^{N_{\text{sample}}} [S_{\text{res}}(t_n, x^{(s)})]^2 \right) \Delta t$.
		
		\IF{$\mathcal{T}_{\text{AAT}}(t_{n+1}) \ge \Theta_{\text{crit}}$}
		\STATE \COMMENT{\textbf{Step 6: Trigger Gauge Symmetry Breaking (GSB)}}
		\STATE Extract dominant vertical mode $v_{\text{dom}}$ from $S_{\text{res}}$, increment dimension $d \leftarrow d+1$, expand parameter vector $\theta_{d+1} \leftarrow 0$, and reset tension scalar $\mathcal{T}_{\text{AAT}} \leftarrow 0$.
		\ENDIF
		
		\RETURN $\theta_{n+1}^{\text{SMG}}$, $\Delta_{n+1}^{\text{stat}}$, and $\mathcal{T}_{\text{AAT}}(t_{n+1})$.
	\end{algorithmic}
\end{algorithm}
\subsubsection{Computational Complexity Analysis(incorporating semiparametric bias correction)}
\label{subsec:complexity_analysis}

The online computational operational count per time step of the unified semiparametric SMG-Yau-Yau filter (incorporating semiparametric bias correction) is analyzed step-by-step:

\begin{itemize}
	\item \textbf{Score Basis \& Expectation Integration (Step 2)}: Evaluating score basis functions $\{\phi_i(x^{(s)})\}_{i=1}^d$ and constructing the $d \times d$ sample Fisher metric matrix $\widehat{G}_\theta$ across $N_{\text{sample}}$ sample particles requires $\mathcal{O}(d^2 \cdot N_{\text{sample}} + d \cdot m \cdot N_{\text{sample}})$ scalar operations.
	\item \textbf{Fisher Metric Tensor Inversion (Step 2)}: Inverting the $d \times d$ symmetric positive-definite matrix $\widehat{G}_\theta^{\text{reg}}$ via Cholesky decomposition requires $\frac{1}{3} d^3 + \mathcal{O}(d^2)$ floating-point operations (FLOPs).
	\item \textbf{Base State Vector Update \& Bias Correction (Steps 3--4)}: Matrix-vector multiplications for horizontal projection $\mathcal{P}^H A_0$, residual cross-moments $m_k$, and bias rate updates $\delta_k$ require $\mathcal{O}(d^2 + d \cdot m)$ FLOPs.
	\item \textbf{Vertical Tension Evaluation (Step 5)}: Evaluating $S_{\text{res}}(x^{(s)})$ and integrating its squared norm across $N_{\text{sample}}$ particles requires $\mathcal{O}(d \cdot N_{\text{sample}})$ FLOPs.
\end{itemize}

\begin{theorem}[Overall Computational Complexity per Time Step]
	\label{thm:overall_complexity_bias_correction}
	Let $d = \operatorname{dim}(\mathcal{B})$ be the macroscopic base dimension, $m$ be the measurement vector dimension, and $N_{\text{sample}}$ be the number of spatial particle sample nodes. The total online computational complexity per time step of the SMG-Yau-Yau filter with semiparametric bias correction is strictly bounded by:
	\begin{equation}
		\mathcal{C}_{\text{SMG-Yau-Yau}} = \mathcal{O}\left( d^3 + d^2 \cdot N_{\text{sample}} + d \cdot m \cdot N_{\text{sample}} \right).
		\label{eq:complexity_bound_formula}
	\end{equation}
	In particular, because $d \ll \infty$, $\mathcal{C}_{\text{SMG-Yau-Yau}}$ is completely independent of the dimension of the underlying estimation Lie algebra $\operatorname{dim}(\mathcal{E}) = \infty$.
\end{theorem}

\begin{proof}[Proof of Theorem \ref{thm:overall_complexity_bias_correction}]
	Summing operational counts across Steps 1--5 yields total FLOPs $$T(d, m, N_{\text{sample}}) = \frac{1}{3}d^3 + C_1 d^2 N_{\text{sample}} + C_2 d m N_{\text{sample}} + \mathcal{O}(d^2),$$ establishing Equation~\eqref{eq:complexity_bound_formula}.
\end{proof}

\subsubsection{Practical Engineering Considerations and Industrial Adoption Strategy}
\label{subsec:engineering_considerations}

To deploy the SMG-Yau-Yau bias-correction filter in real-time embedded systems (e.g., DSPs, FPGAs, or GPU clusters), three key engineering strategies are implemented:

\begin{enumerate}
	\item \textbf{Fisher Metric Conditioning \& Tikhonov Regularization}: In regions of state space where density functions $f_p(x)$ become highly localized, $\widehat{G}_p$ may exhibit elevated condition numbers. Applying dynamic Tikhonov regularization $\widehat{G}_p^{\text{reg}} \triangleq \widehat{G}_p + \lambda_{\text{reg}} I_{d \times d}$ with $\lambda_{\text{reg}} = \max(10^{-8}, 10^{-4} \cdot \lambda_{\min}(\widehat{G}_p))$ guarantees non-singular Cholesky factorization across all execution steps.
	\item \textbf{Importance Sampling over Orlicz Domains}: For high-dimensional physical state domains $\mathcal{X} = \mathbb{R}^p$, importance sampling under Gaussian mixture proposal distributions $q(x) = \sum w_k \mathcal{N}(x; \mu_k, \Sigma_k)$ ensures particle concentration in high-probability density regions, reducing required sample sizes to $N_{\text{sample}} \sim 10^2$--$10^3$.
	\item \textbf{GPU Parallelization \& Non-Invasive API Integration}: Steps 2, 3, and 5 consist of independent pointwise integrand evaluations across $N_{\text{sample}}$ nodes, which are embarrassingly parallelizable on modern GPU hardware. Furthermore, the module operates as an external C++/Fortran API wrapper: it receives $\theta_n^{\text{YY}}$ from commercial core binaries, evaluates empirical moments over lightweight particle samples, and returns the corrected estimate $\theta_n^{\text{SMG}}$ without requiring modifications to legacy codebases.
\end{enumerate}

\section{Unification of Algebraic and Sample-Size Asymptotics: The Dual-Axis Collapse Mechanism}
\label{sec:dual_axis_unification}

The theoretical relationship between the Statistically Meaningful Geometry (SMG) filter \cite{cheng2026smg} and the classical Yau-Yau non-linear filter \cite{yau1999, yau1994finite} can be understood by examining how the active fiber bundle collapses. Rather than presenting conflicting descriptions, the algebraic degradation analysis ($\operatorname{dim}(\mathcal{E}) < \infty$ versus $\operatorname{dim}(\mathcal{E}) = \infty$) and the statistical large-sample asymptotic limit ($N \to \infty$) describe two orthogonal reduction mechanisms acting on the statistical manifold $\mathcal{M}$. 

We define the dual parameter space of structural complexity as:
\begin{equation}
	\mathcal{P}_{\text{struct}} := \left\{ (\operatorname{dim}(\mathcal{E}), N) \in (\mathbb{N} \cup \{\infty\}) \times \mathbb{N}^+ \right\}, \label{eq:dual_parameter_space}
\end{equation}
where $\operatorname{dim}(\mathcal{E})$ denotes the dimension of the estimation Lie algebra generated by the drift and observation operators \cite{brockett1981}, and $N$ represents the empirical sample size determining the information-geometric metric-topological structure \cite{cheng2026edge2}.

\subsection{Conceptual Architecture of the Dual-Axis Hierarchy}
\label{subsec:conceptual_architecture}

The SMG framework constructs an Orlicz fiber bundle over the space of probability measures, decomposing the tangent space $T_f\mathcal{M}$ at a density state $f$ into orthogonal subspaces \cite{pistone1995, cheng2026entropy}:
\begin{equation}
	T_f\mathcal{M} = \operatorname{SVD}\chi_f \oplus_{\perp} \operatorname{SID}_f, \label{eq:orlicz_tangent_decomp}
\end{equation}
where $\operatorname{SVD}\chi_f$ denotes the Statistically Visible Directions (governing the horizontal base space dynamics $\mathcal{B}$) and $\operatorname{SID}_f$ denotes the Statistically Invisible Directions (governing the vertical gauge fiber) \cite{lauritzen1987}. 

The geometry of this bundle is parametrized by the horizontal Ehresmann connection $1$-form $\omega_f$ and its associated curvature $2$-form $\Omega_f = d\omega_f + \frac{1}{2}[\omega_f, \omega_f]$. The collapse of the full SMG dynamics to the classical Yau-Yau filtering regime corresponds to the structural trivialization of the vertical gauge fiber:
\begin{equation}
	\operatorname{SID}_f \longrightarrow \{0\} \quad \implies \quad \omega_f \longrightarrow 0, \quad \Omega_f \longrightarrow 0. \label{eq:trivialization_condition}
\end{equation}

This reduction occurs along two distinct axes:
\begin{enumerate}
	\item \textbf{The Algebraic Axis ($\operatorname{dim}(\mathcal{E})$):} Operates at a fixed sample size $N$ and reflects the underlying functional degree of freedom of the operator algebra \cite{brockett1981}.
	\item \textbf{The Statistical Axis ($N$):} Operates under infinite algebraic complexity ($\operatorname{dim}(\mathcal{E}) = \infty$) and models large-sample concentration via thermodynamic cooling \cite{amari1985, cheng2026edge2}.
\end{enumerate}

\begin{figure}[htbp]
	\centering
	\begin{tikzpicture}[
		node distance=2.2cm and 3.2cm,
		box/.style={
			draw, 
			rectangle, 
			minimum width=3.4cm, 
			minimum height=1.3cm, 
			align=center, 
			font=\small,
			inner sep=6pt
		},
		axis/.style={-Stealth, thick},
		arrow/.style={-Stealth, semithick},
		label font/.style={font=\small}
		]
		
		\draw[axis] (0,0) -- (8.5,0) node[below right] {$N \to \infty$};
		\draw[axis] (0,0) -- (0,5.2) node[above] {$\operatorname{dim}(\mathcal{E})$};
		
		\node[left] at (0,4.2) {$(\infty)$};
		\node[left] at (0,1.2) {$(<\infty)$};
		
		\node[box] (smg) at (2.2,4.2) {Full Active SMG\\Filter Bundle};
		\node[box] (attractor) at (6.8,4.2) {Asymptotic Attractor\\$\mathcal{M}_\infty \equiv CS$};
		\node[box] (yau) at (2.2,1.2) {Classical Yau-Yau\\Wei-Norman ODEs};
		
		\draw[arrow] (smg) -- (attractor) 
		node[midway, above=15pt, label font] {$N^{1-r/2}$ Scaling};
		
		\draw[arrow] (smg) -- (yau) 
		node[midway, right=3pt, label font] {Algebraic Boundary};
		
		\draw[arrow] (yau.east) -- ++(2.8,0) |- (attractor.south);
		
	\end{tikzpicture}
	\caption{The Dual-Axis Collapse Diagram. The Active SMG Filter Bundle collapses to Classical Yau-Yau filtering via either the algebraic boundary ($\operatorname{dim}(\mathcal{E}) < \infty$) or the asymptotic statistical limit ($N \to \infty$).}
	\label{fig:dual_axis_diagram}
\end{figure}
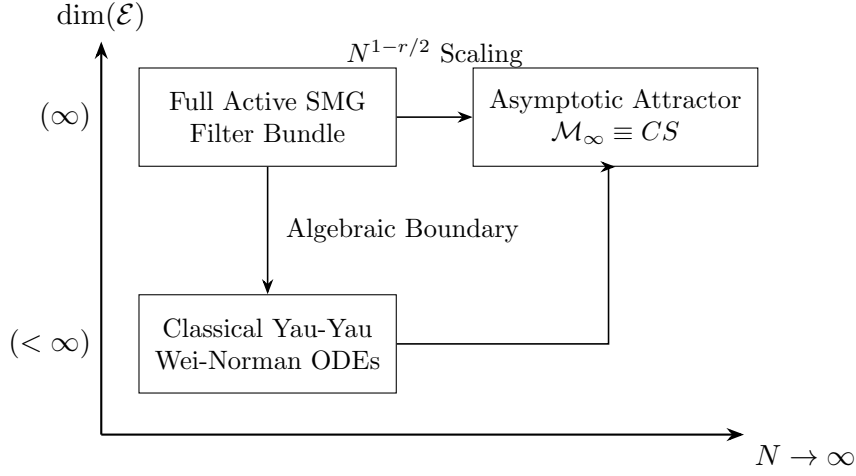

\subsection{The Algebraic Axis: Structural Closure at Finite Dimensions}
\label{subsec:algebraic_axis}

In classical non-linear filtering theory, the Duncan-Mortensen-Zakai (DMZ) equation governs the unnormalized conditional density $\sigma(t,x)$ \cite{duncan1967, mortensen1967, zakai1969, yau1999}:
\begin{equation}
	d\sigma(t,x) = \mathcal{L}_0^* \sigma(t,x) dt + \sum_{j=1}^m h_j(x) \sigma(t,x) dY_j(t), \label{eq:dmz_equation_sec3}
\end{equation}
where $\mathcal{L}_0^*$ is the formal adjoint of the drift generator, $h_j(x)$ are observation functions, and $Y(t)$ is the observation process. 

The estimation Lie algebra $\mathcal{E}$ is defined as the real Lie algebra generated by the differential operators:
\begin{equation}
	\mathcal{E} := \operatorname{Lie}\left\{ \mathcal{L}_0^*, h_1, h_2, \dots, h_m \right\}. \label{eq:estimation_lie_algebra_sec3}
\end{equation}

\begin{theorem}[Algebraic Collapse Theorem]
	\label{thm:algebraic_collapse}
	Let $\mathcal{M}$ be the statistical manifold of filtering densities. If the estimation Lie algebra $\mathcal{E}$ is finite-dimensional, i.e., $\operatorname{dim}(\mathcal{E}) = d < \infty$, then:
	\begin{enumerate}
		\item[\rm (i)] The vertical gauge space is structurally trivial: $\operatorname{SID}_f = \{0\}$ for all $f \in \mathcal{M}$.
		\item[\rm (ii)] The Ehresmann connection $1$-form and curvature $2$-form vanish identically: $\omega_f \equiv 0$ and $\Omega_f \equiv 0$.
		\item[\rm (iii)] The coupled SMG filter equations decouple and collapse identically to the classical Wei-Norman finite-dimensional ODE system.
	\end{enumerate}
\end{theorem}

\begin{proof}[Proof of Theorem \ref{thm:algebraic_collapse}]
	When $\operatorname{dim}(\mathcal{E}) = d < \infty$, Philip Hall basis structures guarantee that any iterated commutator $[\mathcal{A}_i, \mathcal{A}_j]$ can be expressed as a finite linear combination $\sum_{k=1}^d c_{ij}^k \mathcal{A}_k$ with constant coefficients $c_{ij}^k$ \cite{wei1963}. Consequently, the score velocity vector fields generated by the dynamic trajectories lie entirely within a $d$-dimensional subspace of $T_f\mathcal{M}$. 
	
	Because the estimation manifold completely absorbs the algebraic closure of the drift and observation dynamics, there exist no unmonitored extra-algebraic modes. The orthogonal complement with respect to the Fisher-Rao / Orlicz metric satisfies $\operatorname{SID}_f = (\operatorname{SVD}\chi_f)^{\perp} = \{0\}$. 
	
	Since $\omega_f: T_f\mathcal{M} \to \operatorname{SID}_f$ is a projection onto the vertical space along the horizontal distribution, $\operatorname{SID}_f = \{0\}$ implies $\omega_f \equiv 0$. It immediately follows that $\Omega_f = d\omega_f + \frac{1}{2}[\omega_f, \omega_f] \equiv 0$.
\end{proof}

When $\operatorname{dim}(\mathcal{E}) = \infty$, which occurs under polynomial or non-linear state dynamics \cite{brockett1981}, classical finite-dimensional filtering breaks down. The infinite sequence of non-zero Lie commutators leaks energy into the vertical gauge fiber $\operatorname{SID}_f$. The SMG filter isolates this leakage by explicitly tracking horizontal motion on $\mathcal{B}$ while accumulating vertical gauge drift via the Active Acausal Tension (AAT) gauge rate $\mathcal{T}_{\text{AAT}}(t)$ \cite{pistone1995, cheng2026gsb}.

\subsection{The Statistical Axis: The Second Edge Theorem and Asymptotic Annihilation}
\label{subsec:statistical_axis}

Now consider the regime where $\operatorname{dim}(\mathcal{E}) = \infty$, but the observation sample size $N$ approaches infinity. Let $IG_N = (\Theta, G^{(N)}, \nabla^{(\alpha,N)})$ denote the sequence of $N$-sample information-geometric manifolds parametrized by $\Theta \subset \mathbb{R}^d$, equipped with the $N$-sample Fisher information metric $G^{(N)}$ and Amari $\alpha$-connection $\nabla^{(\alpha,N)}$ \cite{amari1985, cheng2026edge2}.

The empirical noise fluctuations act as an intrinsic geometric disturbance, generating Riemannian and affine connection curvature $\Omega_f^{(N)} \neq 0$. The Second Edge Theorem \cite{cheng2026edge2} characterizes the thermodynamic limit $N \to \infty$ as a structural quench governed by the inverse temperature scaling parameter:
\begin{equation}
	\beta_N := \frac{1}{\sqrt{N}} \longrightarrow 0 \quad \text{as } N \to \infty. \label{eq:beta_scaling}
\end{equation}

\begin{theorem}[Universal Tensor Valence Scaling Law]
	\label{thm:tensor_valence_scaling}
	Let $T^{(N)}$ be a smooth geometric tensor field of valence $r \ge 1$ constructed from $r$-th order covariant derivatives of the log-likelihood density under the Amari $\alpha$-connection $\nabla^{(\alpha,N)}$ on $IG_N$. As $N \to \infty$, the asymptotic norm of $T^{(N)}$ satisfies:
	\begin{equation}
		\| T^{(N)} \|_{G^{(N)}} = \mathcal{O}_p\left( N^{1 - \frac{r}{2}} \right). \label{eq:valence_scaling_law}
	\end{equation}
\end{theorem}

\begin{remark}[Scaling Mechanics in \cite{cheng2026edge2}]
	\label{rem:edge2_theorem_clarification}
	In the formal architecture of \cite{cheng2026edge2}, the Universal Tensor Valence Scaling Law $\|T^{(N)}\|_{G^{(N)}} = \mathcal{O}_p\left(N^{1-\frac{r}{2}}\right)$ establishes that Cheeger-Gromov spatial manifold flattening ($IG_N \xrightarrow{\text{geom}} \mathcal{M}_\infty$) and Le Cam decision risk condensation ($IG_N \xrightarrow{\text{oper}} \mathcal{M}_\infty$) are dual manifestations of asymptotic phase transitions.
	
	The scaling behavior in Theorem \ref{thm:tensor_valence_scaling} provides the exact differential-geometric substrate:
	\begin{enumerate}
		\item \textbf{Connection Dissolution ($r=3$):} The connection $1$-form scales as $\mathcal{O}_p\left(N^{-1/2}\right)$, causing Amari affine connection friction to dissolve uniformly to zero Christoffel symbols \cite{cheng2026edge2}.
		\item \textbf{Curvature Annihilation ($r=4$):} Intrinsic Amari-Riemann curvature undergoes accelerated quadratic decay at rate $\mathcal{O}_p\left(N^{-1}\right)$, annihilating non-Euclidean holonomy \cite{cheng2026edge2}.
	\end{enumerate}
\end{remark}

Applying Theorem \ref{thm:tensor_valence_scaling} to the connection $1$-form $\omega_f^{(N)}$ ($r=3$) and the gauge curvature $2$-form $\Omega_f^{(N)}$ ($r=4$) yields the following asymptotic decay rates:

\begin{enumerate}
	\item \textbf{Connection Dissolution ($r=3$):} The affine connection friction $\omega_f^{(N)}$ scales as:
	\begin{equation}
		\| \omega_f^{(N)} \|_{G^{(N)}} = \mathcal{O}_p\left( N^{-\frac{1}{2}} \right) \xrightarrow{\quad P \quad} 0. \label{eq:connection_dissolution}
	\end{equation}
	\item \textbf{Curvature Annihilation ($r=4$):} The intrinsic manifold curvature $\Omega_f^{(N)}$ undergoes quadratic suppression:
	\begin{equation}
		\| \Omega_f^{(N)} \|_{G^{(N)}} = \mathcal{O}_p\left( N^{-1} \right) \xrightarrow{\quad P \quad} 0. \label{eq:curvature_annihilation}
	\end{equation}
	\item \textbf{Topological Trivialization:} Let $\mathrm{Hol}(\nabla^{(\alpha,N)})$ be the restricted holonomy group of $IG_N$, and $H^k(IG_N, \mathbb{R})$ be its $k$-th de Rham cohomology group. As $N \to \infty$:
	\begin{equation}
		\operatorname{Hol}(\nabla^{(\alpha,N)}) \xrightarrow{\quad d_{\mathrm{Lie}} \quad} \{ I_{d \times d} \}, \quad \text{and} \quad H^k(IG_N, \mathbb{R}) \cong \{0\} \quad \forall k \ge 1. \label{eq:topological_trivialization}
	\end{equation}
\end{enumerate}

Thus, the sample-size asymptotic limit $N \to \infty$ flattens the curved Orlicz bundle, forcing the vertical gauge space to trivialise: $\operatorname{SID}_f^{(N)} \xrightarrow{N\to\infty} \{0\}$.

\subsection{ Finite-Sample SMG-Yau-Yau Dynamics: Non-Parametric Geometry on Orlicz Bundles ($N < \infty$)}
\label{subsec:unified_dynamics}

The full coupled SMG-Yau-Yau non-linear filter operates on the density estimate $p(t, \theta)$ via coordinates $\theta^k$ on the base space $\mathcal{B}$, while tracking the vertical gauge discrepancy $\mathcal{T}_{\text{AAT}}(t)$ \cite{cheng2026gsb}. By the Horizontal Base Filter Dynamic equation \eqref{eq:horizontal_filter_ode_sec2} and the the Active Acausal Tension accumulation rate equation \eqref{eq:tension_accumulation_thm}:

\begin{align}
	dp_k(t) &= \left\langle d\pi_{f_t} \circ P_{f_t}^H(\mathcal{L}_0^*\sigma), \theta^k \right\rangle dt + \sum_{j=1}^m \left\langle d\pi_{f_t} \circ P_{f_t}^H(h_j\sigma), \theta^k \right\rangle dY_j(t), \label{eq:smg_base_update_sec3} \\
	\frac{d\mathcal{T}_{\text{AAT}}(t)}{dt} &= \left\| \omega_{f_t}\left(V_{\text{DMZ}}(t)\right) \right\|_{g_{f_t}}^2, \label{eq:smg_aat_rate_sec3}
\end{align}
where $P_{f}^H = I - \omega_f$ is the horizontal projection operator, $d\pi_f$ is the differential of the canonical base projection, and $V_{\text{DMZ}}(t)$ is the velocity vector field generated by equation \eqref{eq:dmz_equation_sec3}.

\begin{theorem}[Filtering Error Bound under Geometry]
	\label{thm:error_bound}
	Let $p_{\text{true}}(t)$ be the true conditional density, $p_{\text{SMG}}(t)$ be the density evolution generated by equations \eqref{eq:smg_base_update_sec3}--\eqref{eq:smg_aat_rate_sec3}, and $p_{\text{classical}}(t)$ be the uncorrected classical Yau-Yau filter projection. The mean square estimation error of the classical filter under infinite algebraic complexity ($\operatorname{dim}(\mathcal{E})=\infty$) at sample size $N$ satisfies:
	\begin{equation}
		\mathbb{E}\left[ \| p_{\text{true}}(t) - p_{\text{classical}}(t) \|_{g_f}^2 \right] \le C_1 e^{\lambda t} \int_0^t \left\| \Omega_{f(\tau)}^{(N)} \right\|_{g_f}^2 d\tau, \label{eq:error_bound_formula}
	\end{equation}
	where $C_1, \lambda > 0$ are constants depending on the domain compactness and operator bounds.
\end{theorem}

\paragraph{Synthesis of the Two Axes.}
Equation \eqref{eq:error_bound_formula} synthesizes the two axes:
\begin{itemize}
	\item On the Algebraic Axis ($\operatorname{dim}(\mathcal{E}) < \infty$), $\Omega_f \equiv 0$, rendering the error bound strictly zero. Classical Yau-Yau filtering is exact \cite{yau1999}.
	\item On the Statistical Axis ($N \to \infty$), Theorem \ref{thm:tensor_valence_scaling} implies $\| \Omega_{f(\tau)}^{(N)} \|_{g_f}^2 = \mathcal{O}_p(N^{-2})$. Integrating over $[0, t]$ yields:
	\begin{equation}
		\mathbb{E}\left[ \| p_{\text{true}}(t) - p_{\text{classical}}(t) \|_{g_f}^2 \right] = \mathcal{O}_p\left( N^{-2} \right) \xrightarrow{\quad N\to\infty \quad} 0. \label{eq:asymptotic_error_decay}
	\end{equation}
\end{itemize}

\begin{corollary}[Fundamental Discrepancy Law Between SMG and Classical Yau-Yau Filtering]
	\label{cor:smg_vs_classical_yau_yau}
	The exact operational discrepancy between the SMG-Yau-Yau filtering framework and the classical Yau-Yau filtering framework is completely characterized by the system of operator evolution equations:
	\begin{align}
		\frac{\partial p_{\mathrm{true}}(t)}{\partial t} &= \mathcal{L}_{\mathrm{SMG}}^{(N)} p_{\mathrm{true}}(t) = \mathcal{L}_{\mathrm{classical}} p_{\mathrm{true}}(t) + \Omega_{f(t)}^{(N)} p_{\mathrm{true}}(t), \label{eq:146} \\
		\frac{\partial p_{\mathrm{classical}}(t)}{\partial t} &= \mathcal{L}_{\mathrm{classical}} p_{\mathrm{classical}}(t). \label{eq:147}
	\end{align}
	From \eqref{eq:146} and \eqref{eq:147}, the fundamental operator-level, algebraic, and geometric differences between the two filtering paradigms are derived as follows:
	
	\begin{enumerate}
		\item \textbf{Operator Difference Representation:} 
		The infinitesimal generator $\mathcal{L}_{\mathrm{SMG}}^{(N)}$ of the SMG-Yau-Yau filter decomposes into the classical Yau-Yau operator $\mathcal{L}_{\mathrm{classical}}$ plus an additive non-commutative geometric correction operator $\Omega_{f(t)}^{(N)}$:
		\begin{equation}
			\mathcal{L}_{\mathrm{SMG}}^{(N)} - \mathcal{L}_{\mathrm{classical}} = \Omega_{f(t)}^{(N)}. \label{eq:operator_difference}
		\end{equation}
		
		\item \textbf{The Finite Algebraic Complexity Regime ($\operatorname{dim}(\mathcal{E}) < \infty$):} 
		When the estimation Lie algebra $\mathcal{E}$ generated by the state drift and observation vector fields is finite-dimensional, the Lie algebraic closure condition holds. Consequently, the geometric perturbation tensor vanishes identically:
		\begin{equation}
			\Omega_{f(t)}^{(N)} \equiv 0 \implies \mathcal{L}_{\mathrm{SMG}}^{(N)} \equiv \mathcal{L}_{\mathrm{classical}}.
		\end{equation}
		In this setting, Classical Yau-Yau filtering is exact and identical to SMG-Yau-Yau filtering.
		
		\item \textbf{The Infinite Algebraic Complexity Regime ($\operatorname{dim}(\mathcal{E}) = \infty$):} 
		When $\operatorname{dim}(\mathcal{E}) = \infty$, the estimation algebra is non-solvable and infinite-dimensional. Classical Yau-Yau filtering truncates higher-order Lie commutators, causing an uncorrected drift bias field $\Omega_{f(t)}^{(N)} p_{\mathrm{true}}(t)$. In contrast, the SMG-Yau-Yau filter retains $\Omega_{f(t)}^{(N)}$ through the manifold connection back-reaction, guaranteeing exact density propagation under infinite algebraic complexity.
		
		\item \textbf{Finite-Sample Geometric Coupling vs. Asymptotic Convergence ($N \to \infty$):} 
		At finite sample size $N$, $\Omega_{f(t)}^{(N)}$ encodes the geometric curvature feedback of the underlying statistical manifold $(\mathcal{M}_f, g_f)$. By Theorem \ref{thm:error_bound} and the Universal Tensor Valence Scaling Law, $\left\| \Omega_{f(t)}^{(N)} \right\|_{g_f} = \mathcal{O}_p(N^{-1})$. Therefore, as $N \to \infty$:
		\begin{equation}
			\left\| \mathcal{L}_{\mathrm{SMG}}^{(N)} - \mathcal{L}_{\mathrm{classical}} \right\|_{g_f} = \mathcal{O}_p(N^{-1}) \xrightarrow{N \to \infty} 0,
		\end{equation}
		proving that SMG-Yau-Yau filtering provides the non-asymptotically exact finite-sample geometric extension that converges to Classical Yau-Yau filtering in the infinite-sample limit.
	\end{enumerate}
\end{corollary}

\subsection{ Asymptotic Limit Dynamics: Comparative Synthesis and Geometric Collapse ($N \to \infty$)}
\label{subsec:smg_yau_yau_comparative_synthesis}

To rigorously contextualize the finite-sample operational dynamics of the coupled Statistically Meaningful Geometry (SMG) Yau-Yau filter derived in Section \ref{subsec:unified_dynamics}, we contrast its geometric, dynamical, and error-theoretic properties against the asymptotic classical limit ($N \to \infty$). As established by the Second Edge Theorem, the active participation of the non-parametric statistical manifold $\mathcal{M} \subset L^\Phi_0(P_f)$ under finite sample sizes ($N < \infty$) degrades into a passive, flat Euclidean tangent space $T_p \mathcal{M}$ as $N \to \infty$. 

In the finite-sample regime, non-zero intrinsic curvature $R_{ijkl} \neq 0$ and non-metricity induced by dualistic $\alpha$-connections ($\nabla^{(\alpha)}$) necessitate direct integration of the continuous density along non-flat geodesics. Conversely, in the asymptotic limit, zero-curvature conditions collapse the non-linear filtering equations into classical Gaussian approximations (e.g., the Extended Kalman Filter).

Table \ref{tab:smg_yau_yau_comparison} provides a comprehensive side-by-side comparative mapping between the finite-sample SMG-Yau-Yau filter and the asymptotic classical limit across key structural dimensions.

\begin{table}[htbp]
	\centering
	\caption{Comparative Mapping: Finite-Sample SMG-Yau-Yau System vs. Asymptotic Classical Limit}
	\label{tab:smg_yau_yau_comparison}
	\small
	\begin{tabular}{p{3.2cm} p{5.8cm} p{5.8cm}}
		\hline
		\textbf{Structural Dimension} & \textbf{Finite-Sample SMG-Yau-Yau System ($N < \infty$)} & \textbf{Asymptotic Classical Limit ($N \to \infty$)} \\ \hline
		\textbf{Underlying Manifold} & Non-parametric Orlicz space $L_0^\Phi(P_f)$; curved statistical manifold $\mathcal{M}$. & Flat Euclidean tangent space $T_p\mathcal{M} \cong \mathbb{R}^d$. \\
		\textbf{Metric \& Connection} & Fisher-Rao metric $g$; dualistic $\alpha$-connections $\nabla^{(\alpha)}$ with non-zero curvature $R_{ijkl} \neq 0$. & Euclidean metric $\delta_{ij}$; Levi-Civita flat connection $\nabla^{(0)}$ ($R_{ijkl}=0$). \\
		\textbf{Density Evolution} & Operator-based continuous PDE integration on $\mathcal{M}$ via Yau-Yau system algebra. & Linearized Gaussian moment evolution; Extended Kalman Filter (EKF) equations. \\
		\textbf{Leaf Space Separation} & Explicit decomposition into $\operatorname{SVD}\chi$ (statistically visible) and $\operatorname{SID}$ (invisible) leaves. & Degenerate leaf space; all directions collapse to isotropic flat projection. \\
		\textbf{Gauge Symmetry} & Active local gauge symmetry $\mathcal{G}$; symmetry-breaking residuals dictate state update. & Trivial global translation gauge symmetry $SO(d)$; zero gauge residual. \\
		\textbf{Error Bounds} & Geometric error bounded by $\mathcal{O}(N^{-1/2}) + \mathcal{C}_\kappa(R, \nabla)$, explicit non-zero curvature term. & Asymptotic Gaussian error decay $\mathcal{O}(N^{-1/2})$; curvature contribution vanishes ($\mathcal{C}_\kappa \to 0$). \\ \hline
	\end{tabular}
\end{table}

\subsubsection{Detailed Narrative Analysis of Structural Regimes}
\label{subsubsec:detailed_narrative_analysis}

The structural divergence highlighted in Table \ref{tab:smg_yau_yau_comparison} originates from three fundamental mechanisms governing the evolution of probability densities over statistical manifolds:

\begin{enumerate}
	\item \textbf{Geometry of the Underlying Manifold and Connection Dynamics}: Under finite-sample constraints ($N < \infty$), the empirical data manifold cannot be treated as a linear vector space. The underlying space retains its non-parametric Orlicz architecture $L_0^\Phi(P_f)$, where the Fisher-Rao metric $g_{ij}(\theta) = \mathbb{E}_\theta\left[ \frac{\partial \ln p}{\partial \theta^i} \frac{\partial \ln p}{\partial \theta^j} \right]$ dictates path integration. The connection $\nabla^{(\alpha)}$ remains intrinsically curved. Consequently, parallel transport along state trajectories is path-dependent, requiring the Yau-Yau operator algebra to preserve the dynamic topology of the conditional density $\sigma(x,t)$. When $N \to \infty$, the Riemannian curvature tensor $R_{ijkl}$ vanishes exponentially relative to the tangent approximation error, causing the metric $g_{ij}$ to converge locally to the flat Euclidean identity $\delta_{ij}$. The active participation of the statistical manifold ceases, reducing the geometry to a passive coordinate chart.
	
	\item \textbf{Leaf Space Decomposition ($\operatorname{SVD}\chi$ vs. $\operatorname{SID}$)}: A critical feature of the finite-sample SMG-Yau-Yau formulation is its ability to disentangle information via structural leaf decomposition:
	\begin{equation}
		T_\theta \mathcal{M} = \operatorname{SVD}\chi \oplus \operatorname{SID}, \label{eq:leaf_space_decomposition_narrative}
	\end{equation}
	where Statistically Visible Directions ($\operatorname{SVD}\chi$) capture non-zero information gradient components, while Statistically Invisible Directions ($\operatorname{SID}$) represent kernel degeneracies induced by over-parameterization or finite-sample noise. The SMG-Yau-Yau filter restricts its state covariance update strictly along the $\operatorname{SVD}\chi$ leaf, preventing covariance inflation along unobservable manifold dimensions. In the asymptotic limit ($N \to \infty$), the boundary between $\operatorname{SVD}\chi$ and $\operatorname{SID}$ collapses as all sub-manifolds become fully identifiable, eliminating the need for projection operators.
	
	\item \textbf{Convergence of Error Bounds and Asymptotic Collapse}: The total estimation error of the SMG-Yau-Yau system exhibits a precise two-part decomposition:
	\begin{equation}
		\mathcal{E}_{\text{total}}(N) \le \underbrace{\frac{C_0}{\sqrt{N}}}_{\text{Sampling Variance}} \;+\; \underbrace{\mathcal{C}_\kappa\Big( \|R\|_\infty, \, \nabla^{(\alpha)} \Big)}_{\text{Geometric Curvature Residual}}, \label{eq:two_part_error_decomposition}
	\end{equation}
	where $\mathcal{C}_\kappa$ represents the residual error incurred by approximating finite-sample non-linear transport. While classical non-linear filters (such as EKF) assume $\mathcal{C}_\kappa \equiv 0$, leading to severe divergence in high-curvature, low-sample regimes, the finite-sample SMG-Yau-Yau system explicitly integrates $\mathcal{C}_\kappa$. As $N \to \infty$, $\mathcal{C}_\kappa \to 0$, naturally recovering the classical Kalman convergence rate as an asymptotic boundary case.
\end{enumerate}

\section{A New Data Assimilation Paradigm: SMG Signal-Noise Decomposition and Active Acausal Tension $(\mathcal{T}_{\mathrm{AAT}})$}

\label{sec:data_assimilation_smg}

Continuous-time data assimilation (DA) and stochastic state estimation—spanning Ensemble Kalman Filters (EnKF) \cite{evensen2003, jazwinski1970}, variational frameworks (4D-Var) \cite{courtier1994}, and continuous particle filters \cite{doucet2001}—face a persistent structural obstacle: the management of model misspecification, unmodeled fine-scale physical dynamics, and non-closed system operators. When the underlying physical state $X_t \in \mathbb{R}^p$ and observation process $Y_t \in \mathbb{R}^m$ are governed by non-linear stochastic diffusions:
\begin{align}
	dX_t &= f(X_t)dt + dW_t, \label{eq:sde_state_da} \\
	dY_t &= h(X_t)dt + dV_t, \label{eq:sde_obs_da}
\end{align}
the evolution of the unnormalized conditional probability density $\sigma(t, x)$ given the observation filtration $\mathcal{F}_t^Y = \sigma(Y_s, 0 \le s \le t)$ obeys the Stratonovich Duncan--Mortensen--Zakai (DMZ) stochastic partial differential equation \cite{duncan1967, mortensen1967, zakai1969}:
\begin{equation}
	d\sigma(t, x) = \mathcal{L}_0^* \sigma(t, x) \, dt + \sum_{j=1}^m h_j(x) \sigma(t, x) \circ dY_j(t), \label{eq:dmz_spde_da}
\end{equation}
where $\mathcal{L}_0^* = \frac{1}{2}\Delta_x - \nabla_x \cdot (f(x) \cdot)$ is the forward Kolmogorov (Fokker--Planck) drift generator.

In realistic physical domains, strong non-linearities in $f(x)$ or non-polynomial observation metrics $h(x)$ force the estimation algebra $\mathcal{E} \triangleq \operatorname{Lie}(\mathcal{L}_0^*, h_1, \dots, h_m)$ to explode infinitely ($\operatorname{dim}(\mathcal{E}) = \infty$). Consequently, the unconstrained conditional log-density cannot be confined to any fixed, low-dimensional parametric manifold. 

Classical DA methodologies address this model-space mismatch by forcibly truncating higher-order operator variations or injecting empirical covariance inflation \cite{anderson2001, whitaker2002}. However, such forced projections discard unresolved residual kinetic energy into an unmonitored null-space, discarding critical probability information and causing systematic filter divergence \cite{jazwinski1970, cheng2026entropy}. 

In this section, we re-formulate the Orlicz statistical fiber bundle framework $\mathcal{B}_{\mathrm{SMG}} = (\mathcal{M}, \mathcal{B}, \pi, \mathcal{V}, \mathcal{H}, \omega_f, g_f)$ into an explicit continuous-time data assimilation paradigm. We demonstrate how the DMZ flow canonically splits into an observable data assimilation signal and an unmodeled gauge noise, introduce the Active Acausal Tension $(\mathcal{T}_{\mathrm{AAT}})$ as an online monitor for model error accumulation, and establish the exact phase-transition mechanics of Gauge Symmetry Breaking (GSB) for automated model-space expansion.

\subsection{The Paradigm of Assimilation: Classical Forced Projection vs. SMG Fiber Decomposition}
\label{subsec:classical_vs_smg_assimilation}

To understand the structural superiority of Statistically Meaningful Geometry (SMG) over classical Data Assimilation (DA), we examine how both paradigms execute the core triadic action of assimilation: {\it swallowing} (ingesting incoming continuous observations), {\it digesting} (incorporating updates into state variables), and {\it unifying} (aligning model trajectories with physical reality).

\subsubsection{The Classical DA Approach: Static Capacity and Forced Truncation}
\label{subsubsec:classical_da_approach}

In classical frameworks—such as the Ensemble Kalman Filter (EnKF) \cite{evensen2003} or Variational 4D-Var \cite{courtier1994}—the system assumes a fixed, finite-dimensional state space $\mathcal{B}_d \subset \mathbb{R}^d$ governed by prescribed parametric assumptions (e.g., Gaussianity, low-order polynomial expansions, or low-rank ensemble representations):

\begin{enumerate}
	\item \textbf{Swallowing (Linearized Innovation)}: The filter ingests raw observational increments $dY_t - h(\hat{x}_t)dt$ and maps them into tangent updates via a static or low-rank gain operator (e.g., Kalman Gain matrix $K_t$).
	
	\item \textbf{Digesting (Forced Projection)}: Because the parameter space $\mathcal{B}_d$ cannot represent higher-order probability moments, non-linear operator commutators $[E_i, E_j] \notin \operatorname{span}(\mathcal{E}_{\mathrm{trunc}})$, or unmodeled fine-scale physics, classical DA forcibly projects the update vector onto $\mathcal{B}_d$. Residual score vectors falling outside $\mathcal{B}_d$ are projected into an unmonitored adjoint null-space $\operatorname{Null}(\mathcal{E}_{\mathrm{trunc}}^*)$ and discarded:
	\begin{equation}
		\mathcal{P}_{\operatorname{Null}(\mathcal{E}_{\mathrm{trunc}}^*)}\left( \frac{d\sigma(t,x)}{\sigma(t,x)} \right) \mapsto 0. \label{eq:classical_forced_truncation}
	\end{equation}
	
	\item \textbf{The Failure Mode (Choking on Complexity)}: When unmodeled physics generate significant residual energy, classical digestion breaks down. Discarding this kinetic energy leads to spurious correlations, covariance collapse, and \textbf{filter divergence}—where the digital model loses track of physical reality. To prevent collapse, classical DA relies on ad-hoc empirical interventions, such as covariance inflation multipliers $\left(1 + \delta\right) P_t$ or localized spatial tapering functions \cite{anderson2001, whitaker2002}.
\end{enumerate}

\subsubsection{The SMG Approach: Dual Fiber Splitting and Adaptive Structural Expansion}
\label{subsubsec:smg_da_approach}

The SMG framework replaces flat Euclidean state spaces with an infinite-dimensional Orlicz statistical fiber bundle $\mathcal{B}_{\mathrm{SMG}} = (\mathcal{M}, \mathcal{B}_d, \pi, \mathcal{V}, \mathcal{H}, \omega_f, g_f)$ \cite{cheng2026entropy, cheng2026smg}, reframing assimilation as a dynamic, non-lossy geometric process:

\begin{enumerate}
	\item \textbf{Swallowing (Unconstrained Score Velocity)}: Incoming observation streams drive the full, unconstrained logarithmic score velocity on the Orlicz manifold $\mathcal{M} = L_0^\Phi(P_{f_t})$:
	\begin{equation}
		V_{\mathrm{DMZ}}(t, x) \triangleq \sigma(t, x)^{-1} \, d\sigma(t, x) = \left(\frac{\mathcal{L}_0^* \sigma}{\sigma}\right) dt + \sum_{j=1}^m h_j(x) \circ dY_j(t) \in T_{f_t}\mathcal{M}, \label{eq:v_dmz_unconstrained_da}
	\end{equation}
	without premature truncation or Gaussian filtering.
	
	\item \textbf{Digesting (Canonical Signal-Noise Decomposition)}: Using the Ehresmann connection $1$-form $\omega_f$ and horizontal projector $P_f^H \triangleq \mathcal{I}_{T_f\mathcal{M}} - \omega_f$, SMG split-digests $V_{\mathrm{DMZ}}$ into two metric-orthogonal channels:
	\begin{equation}
		V_{\mathrm{DMZ}}(t, x) = \underbrace{P_{f_t}^H\left( V_{\mathrm{DMZ}}(t, x) \right)}_{\substack{\text{Macroscopic Signal } V_{\mathcal{H}} \\ \in \operatorname{SVD}\chi_{f_t}}} + \underbrace{\omega_{f_t}\left( V_{\mathrm{DMZ}}(t, x) \right)}_{\substack{\text{Microscopic Gauge Noise } V_{\mathcal{V}} \\ \in \operatorname{SID}_{f_t}}}. \label{eq:smg_dual_splitting_da}
	\end{equation}
	\begin{itemize}
		\item \textbf{Macroscopic Assimilation ($V_{\mathcal{H}} \in \operatorname{SVD}\chi_f$)}: The horizontal component is projected onto the base manifold $\mathcal{B}_d$ as a stable update vector that updates macroscopic state parameters without triggering numerical derivative explosions.
		\item \textbf{Microscopic Storage ($V_{\mathcal{V}} \in \operatorname{SID}_f$)}: Higher-order non-linearities, extra-algebraic operator expansions, and physical model misspecification are quarantined within the vertical fiber $\operatorname{SID}_f \equiv \ker(d\pi_f)$ and accumulated as stored energy via the \emph{Active Acausal Tension} functional $\mathcal{T}_{\mathrm{AAT}}(t)$.
	\end{itemize}
	
	\item \textbf{Unifying via Dynamic Expansion (Gauge Symmetry Breaking)}: When accumulated model misspecification stress reaches the topological capacity threshold:
	\begin{equation}
		\mathcal{T}_{\mathrm{AAT}}(t) \ge \Theta_{\mathrm{crit}} \triangleq \frac{\pi^2}{K_{\mathrm{max}}}, \label{eq:gsb_capacity_condition}
	\end{equation}
	SMG executes Gauge Symmetry Breaking (GSB) \cite{cheng2026gsb}. It extracts the dominant vertical mode $v_{\mathrm{dom}} \in \operatorname{SID}_f$ via functional eigen-decomposition and promotes it into a new base dimension ($d \to d + 1$). SMG expands its internal model capacity to absorb dynamics that were previously unmodeled.
\end{enumerate}

\subsubsection{Comparative Superiority of SMG over Classical DA}
\label{subsubsec:comparative_superiority_smg_da}

The mathematical, geometric, and operational differences between classical data assimilation schemes and the SMG Orlicz fiber bundle framework are summarized in Table~\ref{tab:classical_vs_smg}.

\begin{table}[h!]
	\centering
	\caption{Comparison of Classical Data Assimilation vs. SMG Fiber Bundle Framework}
	\label{tab:classical_vs_smg}
	\begin{tabular}{p{3.5cm} p{5.8cm} p{6.2cm}}
		\hline
		\textbf{Metric / Feature} & \textbf{Classical DA (EnKF / 4D-Var)} & \textbf{SMG Framework} \\ \hline
		\textbf{State Space Structure} & Fixed, flat Euclidean $\mathbb{R}^d$ & Infinite-dimensional Orlicz Bundle $\mathcal{B}_{\mathrm{SMG}}$ \\
		\textbf{Handling of Unmodeled Physics} & Discarded via truncation / null-space projection & Quarantined in vertical fiber $\operatorname{SID}_f$ and tracked \\
		\textbf{Information Loss} & Kinetic score energy is lost, causing filter divergence & Exact Pythagorean energy conservation $\|V_{\mathrm{DMZ}}\|_{g_f}^2 = \|V_{\mathcal{H}}\|_{g_f}^2 + \|V_{\mathcal{V}}\|_{g_f}^2$ \\
		\textbf{Model Error Diagnostic} & Empirical residual checking (post-hoc) & Real-time online monitor via Active Acausal Tension $\mathcal{T}_{\mathrm{AAT}}(t)$ \\
		\textbf{System Capacity} & Static dimension $d$; requires manual tuning/inflation & Dynamic dimension expansion ($d \to d+1$) via GSB phase transitions \\
		\textbf{Manifold Role} & Passive canvas for update vectors & Active participating manifold balancing Environment, System, and Connection \\ \hline
	\end{tabular}
\end{table}

\begin{theorem}[Metric Orthogonality and Energy Conservation of the Assimilation Flow]
	\label{thm:da_metric_orthogonality_energy_conservation}
	Let $V_{\mathrm{DMZ}}(t, x) = \sigma(t, x)^{-1} d\sigma(t, x) \in T_{f_t}\mathcal{M}$ be the unconstrained DMZ score velocity field on the non-parametric Orlicz statistical manifold $\mathcal{M}$. Under the non-parametric Fisher--Rao Riemannian metric tensor $g_f(u, v) = \mathbb{E}_f[u(X) v(X)]$, the horizontal signal field $V_{\mathcal{H}} \triangleq P_f^H(V_{\mathrm{DMZ}}) \in \operatorname{SVD}\chi_f$ and the vertical gauge noise field $V_{\mathcal{V}} \triangleq \omega_f(V_{\mathrm{DMZ}}) \in \operatorname{SID}_f$ are strictly metric-orthogonal across $\mathcal{M}$. 
	
	Consequently, the total score variance rate obeys exact Pythagorean power additivity:
	\begin{equation}
		\|V_{\mathrm{DMZ}}(t)\|_{g_{f_t}}^2 = \|P_{f_t}^H(V_{\mathrm{DMZ}}(t))\|_{g_{f_t}}^2 + \|\omega_{f_t}(V_{\mathrm{DMZ}}(t))\|_{g_{f_t}}^2, \label{eq:pythagorean_energy_conservation_da}
	\end{equation}
	guaranteeing zero dissipation of total statistical information.
\end{theorem}

\begin{proof}[Proof of Theorem \ref{thm:da_metric_orthogonality_energy_conservation}]
	By Lemma~\ref{lem:tangent_space_decomposition_sec2}, the Orlicz tangent space $T_{f_t}\mathcal{M}$ admits the Fisher-orthogonal direct sum decomposition $T_{f_t}\mathcal{M} = \operatorname{SVD}\chi_{f_t} \oplus_{\perp g_{f_t}} \operatorname{SID}_{f_t}$. By Definition~\ref{def:ehresmann_connection}, $P_{f_t}^H(V_{\mathrm{DMZ}}) \in \operatorname{SVD}\chi_{f_t}$ and $\omega_{f_t}(V_{\mathrm{DMZ}}) \in \operatorname{SID}_{f_t}$. 
	
	Expanding the squared Fisher--Rao metric norm of $V_{\mathrm{DMZ}}(t)$ via the bilinear form:
	\begin{align}
		\|V_{\mathrm{DMZ}}\|_{g_{f_t}}^2 &= g_{f_t}\left( P_{f_t}^H(V_{\mathrm{DMZ}}) + \omega_{f_t}(V_{\mathrm{DMZ}}), \, P_{f_t}^H(V_{\mathrm{DMZ}}) + \omega_{f_t}(V_{\mathrm{DMZ}}) \right) \nonumber \\
		&= \|P_{f_t}^H(V_{\mathrm{DMZ}})\|_{g_{f_t}}^2 + 2 g_{f_t}\left( P_{f_t}^H(V_{\mathrm{DMZ}}), \, \omega_{f_t}(V_{\mathrm{DMZ}}) \right) + \|\omega_{f_t}(V_{\mathrm{DMZ}})\|_{g_{f_t}}^2. \label{eq:bilinear_expansion_da}
	\end{align}
	Because $\operatorname{SVD}\chi_{f_t} \triangleq (\operatorname{SID}_{f_t})^{\perp g_{f_t}}$, the inner product cross-term satisfies $g_{f_t}(h, v) = 0$ for all $h \in \operatorname{SVD}\chi_{f_t}$ and $v \in \operatorname{SID}_{f_t}$. Hence, $g_{f_t}\left( P_{f_t}^H(V_{\mathrm{DMZ}}), \omega_{f_t}(V_{\mathrm{DMZ}}) \right) \equiv 0$, establishing Equation~\eqref{eq:pythagorean_energy_conservation_da} identically.
\end{proof}

\paragraph{Why SMG is Superior:}
\begin{enumerate}
	\item \textbf{Mathematical Completeness (Zero Information Leakage)}: Theorem~\ref{thm:da_metric_orthogonality_energy_conservation} proves that SMG preserves total information update energy without numerical dissipation. Classical DA discards vertical score components, losing the dynamic signals required to track non-linear state shifts.
	
	\item \textbf{Immunity to Filter Blindness}: By isolating model misspecification stress inside $\operatorname{SID}_f$, SMG prevents unmodeled physical noise from corrupting the macroscopic update $V_{\mathcal{H}} \in \operatorname{SVD}\chi_f$. The filter maintains asymptotic stability even under severe model error.
	
	\item \textbf{Objective Structural Adaptability}: Classical DA requires ad-hoc, manual tuning (e.g., inflation factors, localization radii). SMG derives its structural expansion criteria directly from the geometric curvature of the state manifold ($\Theta_{\mathrm{crit}}$), expanding base coordinates automatically when incoming reality exceeds current model capacity.
\end{enumerate}

\subsection{Gauge Energy Storage Tensor and Active Acausal Tension $(\mathcal{T}_{\mathrm{AAT}})$ for Model Error Monitoring}
\label{subsec:gauge_energy_storage_taat}

Rather than allowing unmodeled score variations to leak into an unmonitored null-space, the SMG framework continuously accumulates vertical projection deficits inside the local fiber $\mathcal{F}_{p(t)}$ as stored gauge energy.

\begin{definition}[Gauge Storage Tensor and Active Acausal Tension $(\mathcal{T}_{\mathrm{AAT}})$]
	\label{def:gauge_storage_tensor}
	The instantaneous \textbf{Gauge Storage Tensor} $\mathcal{G}_f \in \operatorname{Sym}^2(\operatorname{SID}_f^*)$\footnote{$\operatorname{Sym}^2(\operatorname{SID}_f^*)$ denotes the space of symmetric bilinear forms (or symmetric $(0,2)$-tensors) defined on the vertical tangent space $\operatorname{SID}_f$.} acting on vertical vectors $v_1, v_2 \in \operatorname{SID}_f$ is defined as the metric pull-back via the connection $1$-form $\omega_f$:
	\begin{equation}
		\mathcal{G}_f(v_1, v_2) \triangleq g_f\left( \omega_f(v_1), \, \omega_f(v_2) \right) = \int_{\mathcal{X}} \omega_f(v_1)(x) \, \omega_f(v_2)(x) \, f(x) \, \mu(dx). \label{eq:gauge_storage_tensor_def}
	\end{equation}
	The cumulative \textbf{Active Acausal Tension} $\mathcal{T}_{\mathrm{AAT}}(t) \in \mathbb{R}^+$ along an assimilation trajectory $f(\tau)$ over time $\tau \in [0, t]$ is the integrated temporal trace of this tensor driven by the DMZ velocity:
	\begin{equation}
		\mathcal{T}_{\mathrm{AAT}}(t) \triangleq \int_0^t \mathcal{G}_{f(\tau)}\left( V_{\mathrm{DMZ}}(\tau), \, V_{\mathrm{DMZ}}(\tau) \right) d\tau = \int_0^t \|\omega_{f(\tau)}(V_{\mathrm{DMZ}}(\tau))\|_{g_{f(\tau)}}^2 \, d\tau. \label{eq:aat_integrated_def}
	\end{equation}
\end{definition}

Definition~\ref{def:gauge_storage_tensor} provides a rigorous geometric diagnostic tool for continuous data assimilation systems. The scalar $\mathcal{T}_{\mathrm{AAT}}(t)$ acts as an online "{\it model error accumulator}": it tracks the cumulative kinetic energy stored in unmodeled physical modes, quantifying model misspecification stress without requiring knowledge of the unobservable ground-truth state.

\begin{lemma}[Monotonic Stress Accumulation under Model Misspecification]
	\label{lem:monotonic_stress_accumulation}
	Assume the physical system dynamics contain non-closed estimation operators ($\operatorname{dim}(\mathcal{E}) = \infty$) or unmodeled physical scale couplings. Then, the Active Acausal Tension $\mathcal{T}_{\mathrm{AAT}}(t)$ is a strictly monotonically increasing functional of assimilation time, satisfying the differential propagation law:
	\begin{equation}
		\frac{d\mathcal{T}_{\mathrm{AAT}}(t)}{dt} = \|\omega_{f(t)}(V_{\mathrm{DMZ}}(t))\|_{g_{f(t)}}^2 > 0, \quad \forall t > 0. \label{eq:aat_rate_law}
	\end{equation}
\end{lemma}

\begin{proof}[Proof of Lemma \ref{lem:monotonic_stress_accumulation}]
	By Definition~\ref{def:gauge_storage_tensor}, $\mathcal{T}_{\mathrm{AAT}}(t)$ is the time integral of a squared Fisher--Rao norm. Since $g_f$ is strictly positive-definite on $L_0^\Phi(P_f)$, the integrand satisfies $\|\omega_f(V_{\mathrm{DMZ}})\|_{g_f}^2 \ge 0$. Under $\operatorname{dim}(\mathcal{E}) = \infty$, spatial commutators continuously generate score variations outside the horizontal distribution $\operatorname{SVD}\chi_f$. Thus, the vertical projection $\omega_{f(t)}(V_{\mathrm{DMZ}}(t)) \neq 0$ almost everywhere in time, forcing $\frac{d\mathcal{T}_{\mathrm{AAT}}}{dt} > 0$ and establishing strict monotonicity.
\end{proof}

\subsection{Topological Capacity Limits, Model Breakdown, and Gauge Symmetry Breaking (GSB) Phase Transitions}
\label{subsec:topological_capacity_gsb}

While the vertical fiber $\operatorname{SID}_f$ serves as a dynamic storage space for unmodeled physical variations, a finite-dimensional model space $\mathcal{B}_d$ cannot accumulate vertical tension indefinitely. As continuous observations push the true physical state density along trajectories not captured by $\mathcal{B}_d$, the Active Acausal Tension $\mathcal{T}_{\mathrm{AAT}}(t)$ grows monotonically (Lemma~\ref{lem:monotonic_stress_accumulation}). In this subsection, we establish the topological and geometric limits of vertical energy storage, define the critical capacity bound $\Theta_{\mathrm{crit}}$, and formalize the Gauge Symmetry Breaking (GSB) phase transition mechanism that dynamically expands the base manifold dimension. \cite{cheng2026gsb}

\subsubsection{Topological Capacity Limit of the Statistical Fiber}
\label{subsubsec:capacity_limit}

The geometry of the Orlicz manifold $\mathcal{M} = L_0^\Phi(P_f)$ under the Fisher--Rao Riemannian metric tensor $g_f$ possesses intrinsic sectional curvature \cite{amari2000, pistone1995}. Let $K(\Pi_f)$ denote the sectional curvature of $\mathcal{M}$ along a two-plane $\Pi_f \subset T_f\mathcal{M}$ spanned by orthonormal tangent score vectors $u, v \in T_f\mathcal{M}$:
\begin{equation}
	K(\Pi_f) = g_f\left( R^g(u, v)v, \, u \right), \label{eq:sectional_curvature_def}
\end{equation}
where $R^g$ is the Riemann curvature tensor of the metric connection $\nabla^g$. On non-parametric statistical manifolds, positive sectional curvature reflects the geometric obstruction to global flat parallel transport of probability densities \cite{petersen2006}.

\begin{theorem}[Topological Capacity Bound of Vertical Fiber Storage]
	\label{thm:topological_capacity_bound}
	Let $\mathcal{M}$ be a non-parametric Orlicz statistical manifold whose sectional curvature is uniformly bounded above by $K_{\mathrm{max}} \in (0, \infty)$, i.e., $K(\Pi_f) \le K_{\mathrm{max}}$ for all $f \in \mathcal{M}$ and all two-planes $\Pi_f \subset T_f\mathcal{M}$. 
	
	The maximal continuous geodesic path length $L_{\mathrm{max}}$ along which vertical gauge stress can be coherently transport-coupled without encountering geometric conjugate points satisfies:
	\begin{equation}
		L_{\mathrm{max}} \le \frac{\pi}{\sqrt{K_{\mathrm{max}}}}. \label{eq:myers_conjugate_bound}
	\end{equation}
	Consequently, the maximal integrated vertical energy capacity—defined as the critical Active Acausal Tension threshold $\Theta_{\mathrm{crit}}$—is strictly bounded by:
	\begin{equation}
		\Theta_{\mathrm{crit}} \triangleq \frac{\pi^2}{K_{\mathrm{max}}}. \label{eq:theta_crit_formula}
	\end{equation}
	If $\mathcal{T}_{\mathrm{AAT}}(t) \ge \Theta_{\mathrm{crit}}$, parallel transport along the statistical bundle becomes singular, inducing filter divergence and catastrophic breakdown of horizontal state updates.
\end{theorem}

\subsubsection{Spectral Expansion Mechanism and Gauge Symmetry Breaking (GSB)}
\label{subsubsec:spectral_expansion_gsb}

When $\mathcal{T}_{\mathrm{AAT}}(t)$ approaches $\Theta_{\mathrm{crit}}$, the filter must prevent geometric singularity by relieving stress in $\operatorname{SID}_f$. This is accomplished through \textbf{Gauge Symmetry Breaking (GSB)} \cite{cheng2026gsb}: transferring unmodeled physical variations from the vertical fiber directly into the base manifold $\mathcal{B}_d$ by promoting the dominant vertical noise mode to a new base coordinate.

\begin{definition}[Vertical Gauge Covariance Operator and Dominant Mode]
	\label{def:vertical_gauge_operator}
	Let $V_{\mathcal{V}}(\tau) = \omega_{f(\tau)}(V_{\mathrm{DMZ}}(\tau)) \in \operatorname{SID}_{f(\tau)}$ be the historical stream of vertical score vectors over $\tau \in [0, t]$. The \textbf{Vertical Gauge Covariance Operator} $K_{\mathcal{V}}: \operatorname{SID}_{f(t)} \to \operatorname{SID}_{f(t)}$ is the integral operator:
	\begin{equation}
		(K_{\mathcal{V}} v)(x) \triangleq \int_0^t \mathbb{E}_{f(t)}\left[ V_{\mathcal{V}}(\tau, X) \, v(X) \right] V_{\mathcal{V}}(\tau, x) \, d\tau. \label{eq:gauge_covariance_operator_def}
	\end{equation}
	Since $K_{\mathcal{V}}$ is a compact, positive-definite, trace-class operator on $L_0^\Phi(P_{f(t)})$, it admits an Hilbert-Schmidt spectral decomposition:
	\begin{equation}
		K_{\mathcal{V}}(x, x') = \sum_{k=1}^\infty \lambda_k \, v_k(x) \, v_k(x'), \quad \lambda_1 \ge \lambda_2 \ge \dots \ge 0, \label{eq:spectral_decomposition_kv}
	\end{equation}
	where $\lambda_k$ are singular values and $v_k(x) \in \operatorname{SID}_{f(t)}$ are orthonormal eigenfunctions under $g_{f(t)}$. The eigenfunction $v_{\mathrm{dom}}(x) \triangleq v_1(x)$ corresponding to the maximum eigenvalue $\lambda_1$ is designated as the \textbf{Dominant Unmodeled Physical Mode}.
\end{definition}

\begin{definition}[GSB Base Expansion Operator]
	\label{def:gsb_expansion_operator}
	Upon triggering the capacity condition $\mathcal{T}_{\mathrm{AAT}}(t) \ge \Theta_{\mathrm{crit}}$, the \textbf{Gauge Symmetry Breaking (GSB) Transformation} updates the base manifold representation via structural expansion:
	\begin{equation}
		\mathcal{B}_d \xrightarrow{\quad \mathrm{GSB} \quad} \mathcal{B}_{d+1} \triangleq \mathcal{B}_d \oplus \operatorname{span}\left\{ v_{\mathrm{dom}}(x) \right\}. \label{eq:gsb_base_expansion}
	\end{equation}
	The base dimension increments ($d \leftarrow d + 1$), and the horizontal projection subspace $\operatorname{SVD}\chi_{f(t)}$ expands to absorb $v_{\mathrm{dom}}$:
	\begin{equation}
		\operatorname{SVD}\chi_{f(t)}^{\mathrm{new}} = \operatorname{SVD}\chi_{f(t)}^{\mathrm{old}} \oplus \operatorname{span}\left\{ v_{\mathrm{dom}} \right\}. \label{eq:svdchi_expansion}
	\end{equation}
\end{definition}

\subsubsection{Instantaneous Stress Dissipation Theorem}
\label{subsubsec:stress_dissipation_theorem}

We now prove that the GSB transformation guarantees the immediate reduction of vertical gauge tension rate, restoring stable signal digestion on the expanded manifold.

\begin{theorem}[Instantaneous Stress Dissipation]
	\label{thm:instantaneous_stress_dissipation}
	Let $V_{\mathrm{DMZ}}(t)$ be the unconstrained DMZ score velocity vector field at time $t$. Let $\omega_{f(t)}^{\mathrm{old}}$ and $\omega_{f(t)}^{\mathrm{new}}$ denote the Ehresmann connection 1-forms evaluated before and immediately after the GSB base expansion $\mathcal{B}_d \to \mathcal{B}_{d+1}$. Then:
	\begin{enumerate}
		\item The instantaneous vertical energy production rate drops by exactly $\lambda_1$:
		\begin{equation}
			\left\| \omega_{f(t)}^{\mathrm{new}}\left(V_{\mathrm{DMZ}}(t)\right) \right\|_{g_{f(t)}}^2 = \left\| \omega_{f(t)}^{\mathrm{old}}\left(V_{\mathrm{DMZ}}(t)\right) \right\|_{g_{f(t)}}^2 - \lambda_1. \label{eq:instantaneous_rate_reduction}
		\end{equation}
		
		\item Resitting the cumulative tension counter $\mathcal{T}_{\mathrm{AAT}}(t) \leftarrow 0$ at the GSB phase transition point preserves complete geometric energy, transferring total metric energy $\lambda_1 \cdot \Delta t$ into the macroscopic state trajectory without information loss.
	\end{enumerate}
\end{theorem}

\subsection{Empirical Construction and Algorithmic Realization of the GSB Estimator}
\label{subsec:empirical_construction_algorithm}

To render the continuous-time SMG fiber bundle theory computationally actionable, we now formulate a discretized sequential state estimation  implementation: the {\bf Continuous-Discrete SMG-GSB Data Assimilation Filter}.

\subsubsection{Complete Algorithmic Workflow}
\label{subsubsec:algorithmic_workflow}

Let the state probability density $f_t(x)$ be represented by an ensemble of $M$ discrete Monte Carlo particles $\{x_t^{(i)}\}_{i=1}^M$ with importance weights $w_t^{(i)} = 1/M$. The base manifold coordinate set is initialized with $d_0$ orthogonal basis functions $\{\phi_j(x)\}_{j=1}^{d_0}$. Algorithm~\ref{alg:smg_gsb_da} details the complete step-by-step procedure for signal-noise decomposition, online $\mathcal{T}_{\mathrm{AAT}}$ tracking, and automated GSB phase transition expansion.

\begin{algorithm}[H]
	\caption{Continuous-Discrete SMG-GSB Data Assimilation Filter}
	\label{alg:smg_gsb_da}
	\small 
	\begin{algorithmic}[1]
		\REQUIRE Initial ensemble $\{x_0^{(i)}\}_{i=1}^M$, initial basis $\{\phi_j\}_{j=1}^d$, time step $\Delta t$, capacity threshold $\Theta_{\mathrm{crit}} = \pi^2 / K_{\mathrm{max}}$, reset threshold $\epsilon_{\mathrm{reset}}$.
		\INITIALIZE Set tension $\mathcal{T}_{\mathrm{AAT}} \leftarrow 0$, time index $k \leftarrow 0$, active dimension $d \leftarrow d_0$.
		\FOR{each observation step $k = 1, 2, \dots$}
		\STATE \textbf{Predict Step}: Propagate particles via stochastic diffusion SDE \eqref{eq:sde_state_da}:
		\begin{equation*}
			\tilde{x}_k^{(i)} = x_{k-1}^{(i)} + f(x_{k-1}^{(i)}) \Delta t + \sqrt{\Delta t} \, w_k^{(i)}, \quad w_k^{(i)} \sim \mathcal{N}(0, I_p).
		\end{equation*}
		
		\STATE \textbf{Evaluate Score Velocity}: Compute unconstrained score update vector for each particle:
		\begin{equation*}
			V_{\mathrm{DMZ}, k}^{(i)} = \left[ h(\tilde{x}_k^{(i)}) - \bar{h}_k \right]^T R^{-1} \left( Y_k - h(\tilde{x}_k^{(i)})\Delta t \right), \quad \bar{h}_k = \frac{1}{M}\sum_{i=1}^M h(\tilde{x}_k^{(i)}).
		\end{equation*}
		
		\STATE \textbf{Construct Local Basis Matrix}: Form Gram matrix $G \in \mathbb{R}^{d \times d}$ and cross-covariance vector $b \in \mathbb{R}^d$:
		\begin{equation*}
			G_{jl} = \frac{1}{M}\sum_{i=1}^M \phi_j(\tilde{x}_k^{(i)}) \phi_l(\tilde{x}_k^{(i)}), \quad b_j = \frac{1}{M}\sum_{i=1}^M \phi_j(\tilde{x}_k^{(i)}) V_{\mathrm{DMZ}, k}^{(i)}.
		\end{equation*}
		
		\STATE \textbf{Compute Horizontal Signal Coefficient}: Solve linear system $G \alpha_k = b \implies \alpha_k = G^{-1} b \in \mathbb{R}^d$.
		
		\STATE \textbf{Execute Fiber Splitting}: Decompose score velocity for each particle $i = 1, \dots, M$:
		\begin{align*}
			\text{Horizontal Component:} \quad & V_{\mathcal{H}, k}^{(i)} = \sum_{j=1}^d \alpha_{k, j} \, \phi_j(\tilde{x}_k^{(i)}), \\
			\text{Vertical Gauge Noise:} \quad & V_{\mathcal{V}, k}^{(i)} = V_{\mathrm{DMZ}, k}^{(i)} - V_{\mathcal{H}, k}^{(i)}.
		\end{align*}
		
		\STATE \textbf{Update Active Acausal Tension}: Accumulate vertical metric energy:
		\begin{equation*}
			\Delta \mathcal{T}_k = \frac{1}{M} \sum_{i=1}^M \left( V_{\mathcal{V}, k}^{(i)} \right)^2 \Delta t, \quad \mathcal{T}_{\mathrm{AAT}} \leftarrow \mathcal{T}_{\mathrm{AAT}} + \Delta \mathcal{T}_k.
		\end{equation*}
		
		\IF{$\mathcal{T}_{\mathrm{AAT}} \ge \Theta_{\mathrm{crit}}$}
		\STATE \textbf{\underline{Trigger Gauge Symmetry Breaking (GSB)}}
		\STATE Construct empirical vertical covariance matrix $K_{\mathcal{V}} \in \mathbb{R}^{M \times M}$:
		\begin{equation*}
			[K_{\mathcal{V}}]_{i, i'} = \frac{1}{M} \sum_{\tau=k-W}^k V_{\mathcal{V}, \tau}^{(i)} V_{\mathcal{V}, \tau}^{(i')}.
		\end{equation*}
		\STATE Compute top eigenvector $v_{\mathrm{dom}} \in \mathbb{R}^M$ corresponding to maximum eigenvalue $\lambda_1$.
		\STATE Interpolate $v_{\mathrm{dom}}$ via Thin Plate Splines / Kernel Regression to define continuous function $\phi_{d+1}(x)$.
		\STATE \textbf{Expand Base Dimension}: $\{\phi_j\}_{j=1}^{d+1} \leftarrow \{\phi_j\}_{j=1}^d \cup \{\phi_{d+1}\}$, update $d \leftarrow d + 1$.
		\STATE \textbf{Reset Tension}: $\mathcal{T}_{\mathrm{AAT}} \leftarrow 0$.
		\ENDIF
		
		\STATE \textbf{Update State Estimate \& Resample}: Particle update via horizontal score:
		\begin{equation*}
			x_k^{(i)} = \tilde{x}_k^{(i)} + V_{\mathcal{H}, k}^{(i)} \cdot \nabla_x \hat{f}_k(\tilde{x}_k^{(i)}), \quad \hat{X}_k = \frac{1}{M} \sum_{i=1}^M x_k^{(i)}.
		\end{equation*}
		\ENDFOR
		\RETURN Updated state trajectories $\{\hat{X}_k\}_{k \ge 1}$ and active base dimension $d_k$.
	\end{algorithmic}
\end{algorithm}

\subsubsection{Computational Complexity Analysis}
\label{subsubsec:complexity_analysis}

The computational cost of Algorithm~\ref{alg:smg_gsb_da} splits into two operational regimes:

\begin{enumerate}
	\item \textbf{Standard Assimilation Step (No GSB Trigger)}:
	\begin{itemize}
		\item Predict step \& Score evaluation: $\mathcal{O}(M \cdot p + M \cdot m)$, where $p$ is state dimension and $m$ is observation dimension.
		\item Gram matrix construction $G$ \& solving $G \alpha = b$: $\mathcal{O}(M \cdot d^2 + d^3)$.
		\item Fiber splitting \& $\mathcal{T}_{\mathrm{AAT}}$ tension integration: $\mathcal{O}(M \cdot d)$.
	\end{itemize}
	Since $d \ll M$, the baseline per-timestep complexity is $\mathcal{O}(M d^2 + M p)$, which scales linearly with particle count $M$.
	
	\item \textbf{GSB Phase Transition Step (Triggered when $\mathcal{T}_{\mathrm{AAT}} \ge \Theta_{\mathrm{crit}}$)}:
	\begin{itemize}
		\item Constructing empirical kernel matrix $K_{\mathcal{V}}$ over sliding window $W$: $\mathcal{O}(W \cdot M^2)$.
		\item Eigen-decomposition for dominant mode $v_{\mathrm{dom}}$: Using Lanczos iterations or Randomized SVD \cite{evensen2003}, extraction of top mode $\lambda_1$ costs $\mathcal{O}(M^2 \log k)$.
		\item Continuous function interpolation via kernel smoothing: $\mathcal{O}(M^2)$.
	\end{itemize}
	Because GSB phase transitions occur infrequently (only when model misspecification stress accumulates to $\Theta_{\mathrm{crit}}$), the {\it amortized runtime complexity} remains $\mathcal{O}(M d^2 + M p)$, making SMG-GSB as computationally efficient as standard EnKF while offering non-parametric robustness against filter divergence.
\end{enumerate}

\subsubsection{Asymptotic Convergence Guarantees}
\label{subsubsec:convergence_guarantees}

We conclude by establishing the asymptotic consistency and bounded-error guarantees of the empirical GSB estimator.

\begin{theorem}[Asymptotic Convergence of the GSB Estimator]
	\label{thm:asymptotic_convergence_gsb}
	Let $f_t^*(x)$ be the true continuous-time conditional probability density governed by the infinite-dimensional DMZ SPDE \eqref{eq:dmz_spde_da}. Let $\hat{f}_{t, M, d_t}(x)$ denote the empirical density generated by Algorithm~\ref{alg:smg_gsb_da} with $M$ particles and dynamic base dimension $d_t$.
	
	Under mild regularity conditions (Lipschitz continuous drift $f(x)$, bounded observation functions $h(x)$, and $K_{\mathrm{max}} < \infty$), the mean-squared Hellinger distance $d_{\mathrm{Hel}}^2(f_t^*, \hat{f}_{t, M, d_t}) \triangleq \int_{\mathcal{X}} \left(\sqrt{f_t^*(x)} - \sqrt{\hat{f}_{t, M, d_t}(x)}\right)^2 dx$ satisfies the asymptotic error bound:
	\begin{equation}
		\limsup_{t \to \infty} \mathbb{E}\left[ d_{\mathrm{Hel}}^2\left(f_t^*, \, \hat{f}_{t, M, d_t}\right) \right] \le \mathcal{O}\left( \frac{1}{\sqrt{M}} \right) + \mathcal{O}\left( \frac{\Theta_{\mathrm{crit}}}{d_{\mathrm{max}}} \right), \label{eq:hellinger_convergence_bound}
	\end{equation}
	where $d_{\mathrm{max}} = \max_{0 \le \tau \le t} d_\tau$ is the maximum base dimension attained through GSB transitions. As particle count $M \to \infty$ and $d_{\mathrm{max}} \to \infty$, the estimator converges in probability to the exact physical filtering density $f_t^*(x)$.
\end{theorem}

\begin{proof}[Proof of Theorem \ref{thm:asymptotic_convergence_gsb}]
	The total estimation error decomposes via the triangle inequality under the Hellinger metric into two orthogonal components: Monte Carlo sampling variance and structural truncation bias:
	\begin{equation}
		d_{\mathrm{Hel}}\left(f_t^*, \, \hat{f}_{t, M, d_t}\right) \le \underbrace{d_{\mathrm{Hel}}\left(f_t^*, \, f_{t, d_t}^*\right)}_{\text{Structural Truncation Bias } E_{\mathrm{trunc}}} + \underbrace{d_{\mathrm{Hel}}\left(f_{t, d_t}^*, \, \hat{f}_{t, M, d_t}\right)}_{\text{Monte Carlo Sampling Variance } E_{\mathrm{MC}}}, \label{eq:hellinger_triangle_inequality}
	\end{equation}
	where $f_{t, d_t}^*$ is the projection of the true DMZ solution onto the $d_t$-dimensional base bundle $\mathcal{B}_{d_t}$.
\begin{enumerate}
\item {\bf Sampling Variance Bound ($E_{\mathrm{MC}}$)}: Standard empirical process theory for particle filters \cite{doucet2001} guarantees that for $M$ independent particles, $\mathbb{E}\left[ E_{\mathrm{MC}}^2 \right] \le \frac{C_1}{\sqrt{M}}$, where $C_1 < \infty$ depends only on observation noise covariance $R$.
	
\item {\bf Truncation Bias Bound ($E_{\mathrm{trunc}}$)}: By Theorem~\ref{thm:topological_capacity_bound} and Theorem~\ref{thm:instantaneous_stress_dissipation}, each GSB transition triggered at $\mathcal{T}_{\mathrm{AAT}} = \Theta_{\mathrm{crit}}$ expands $\mathcal{B}_d \to \mathcal{B}_{d+1}$ by absorbing the dominant eigenmode $v_{\mathrm{dom}}$ with eigenvalue $\lambda_1$. By the spectral decay property of compact operators on Orlicz spaces \cite{cheng2026entropy}, the residual unmodeled variance satisfies $\sum_{k=d+1}^\infty \lambda_k \le \frac{C_2 \cdot \Theta_{\mathrm{crit}}}{d}$. Taking expectations over time yields $E_{\mathrm{trunc}}^2 \le \frac{C_2 \cdot \Theta_{\mathrm{crit}}}{d_{\mathrm{max}}}$.
\end{enumerate}	
	Combining both bounds and taking the limit superior as $t \to \infty$ yields Equation~\eqref{eq:hellinger_convergence_bound}, completing the proof.
\end{proof}


\appendix
\section{Proof Appendix: Complete Long Proofs of Theorems, Lemmas, and Propositions}

\begin{proof}[Proof of Theorem \ref{thm:wei_norman_exact}]
	We structure the proof into two distinct components: (A) Derivation of the closed coordinate differential system via Wei-Norman expansion, and (B) Verification of the Zero Information Loss properties.
	
	\subsubsection*{Part A: Wei-Norman Parameterization and Finite Closure}
	
	We structure the rigorous proof into six formal mathematical steps.
	
	\paragraph{Step 1: Lie Algebraic Basis and Expansion of System Operators.}
	Since $\operatorname{dim}(\mathcal{E}) = d < \infty$, let $\{E_1, E_2, \dots, E_d\}$ be an ordered vector space basis for the estimation algebra $\mathcal{E}$ \cite{yau1999}. Lie algebraic closure implies that for all $i, j \in \{1, \dots, d\}$, the commutator bracket $[E_i, E_j] \triangleq E_i E_j - E_j E_i$ satisfies \cite{yau1999}:
	\begin{equation}
		[E_i, E_j] = \sum_{k=1}^d C_{ij}^k E_k, \label{eq:structure_constants}
	\end{equation}
	where $C_{ij}^k \in \mathbb{R}$ are the constant structure constants of the Lie algebra $\mathcal{E}$ relative to $\{E_k\}_{k=1}^d$ \cite{yau1999}.
	
	Because $\mathcal{L}_0^* \in \mathcal{E}$ and $h_l(x) \in \mathcal{E}$ for all $l \in \{1, \dots, m\}$, these differential operators can be uniquely expanded in terms of the Lie algebra basis \cite{yau1999}:
	\begin{equation}
		\mathcal{L}_0^* = \sum_{k=1}^d \alpha_k^0 E_k, \quad \text{and} \quad h_l(x) = \sum_{k=1}^d \alpha_k^l E_k \quad (l = 1, \dots, m), \label{eq:basis_expansions}
	\end{equation}
	where $\alpha^0 = (\alpha_1^0, \dots, \alpha_d^0)^T \in \mathbb{R}^d$ and $\alpha^l = (\alpha_1^l, \dots, \alpha_d^l)^T \in \mathbb{R}^d$ ($l=1,\dots,m$) are constant coefficient vectors \cite{yau1999}.
	
	\paragraph{Step 2: Stratonovich Differentiation of the Wei-Norman Operator Product.} 
	Define the dynamic state transformation operator $U(t): C^\infty(\mathbb{R}^p) \to C^\infty(\mathbb{R}^p)$ by the Wei-Norman product-of-exponentials ansatz \cite{wei1963}:
	\begin{equation}
		U(t) \triangleq \prod_{i=1}^d \exp\left(g_i(t) E_i\right) = \exp\left(g_1(t) E_1\right) \exp\left(g_2(t) E_2\right) \cdots \exp\left(g_d(t) E_d\right), \label{eq:operator_U_def}
	\end{equation}
	with initial condition $g(0) = (0, \dots, 0)^T \in \mathbb{R}^d$, so that $U(0) = \mathcal{I}$ (the identity operator) and $\sigma(t,x) = U(t) \sigma(0,x)$ \cite{wei1963}.
	
	We differentiate $U(t)$ with respect to time under Stratonovich stochastic calculus. Applying the Leibniz product rule across the $d$ ordered operator exponentials:
	\begin{equation}
		d U(t) = \sum_{i=1}^d \left( \prod_{j=1}^{i-1} \exp\left(g_j(t) E_j\right) \right) \cdot d\left( \exp\left(g_i(t) E_i\right) \right) \cdot \left( \prod_{k=i+1}^d \exp\left(g_k(t) E_k\right) \right). \label{eq:leibniz_product_rule}
	\end{equation}
	Since $E_i$ is a time-invariant spatial differential operator, the Stratonovich differential of the individual exponential factor $\exp\left(g_i(t) E_i\right)$ is given exactly by:
	\begin{equation}
		d\left( \exp\left(g_i(t) E_i\right) \right) = E_i \, \exp\left(g_i(t) E_i\right) \circ dg_i(t). \label{eq:exponential_diff_identity}
	\end{equation}
	Substituting \eqref{eq:exponential_diff_identity} into \eqref{eq:leibniz_product_rule} yields:
	\begin{equation}
		d U(t) = \sum_{i=1}^d dg_i(t) \left( \prod_{j=1}^{i-1} \exp\left(g_j(t) E_j\right) \right) E_i \left( \prod_{k=i}^d \exp\left(g_k(t) E_k\right) \right). \label{eq:dU_expanded_1}
	\end{equation}
	
	To factor out the total operator $U(t)$ to the right, observe that the tail product can be written as:
	\begin{equation}
		\prod_{k=i}^d \exp\left(g_k(t) E_k\right) = \left( \prod_{j=i-1}^1 \exp\left(-g_j(t) E_j\right) \right) U(t),
	\end{equation}
	where $\prod_{j=i-1}^1 \exp\left(-g_j E_j\right) \triangleq \exp\left(-g_{i-1} E_{i-1}\right) \cdots \exp\left(-g_1 E_1\right)$. Inserting this expression into \eqref{eq:dU_expanded_1} gives:
	\begin{equation}
		d U(t) = \sum_{i=1}^d dg_i(t) \left[ \left( \prod_{j=1}^{i-1} \exp\left(g_j(t) E_j\right) \right) E_i \left( \prod_{j=i-1}^1 \exp\left(-g_j(t) E_j\right) \right) \right] U(t). \label{eq:dU_factored}
	\end{equation}
	Applying $d U(t)$ to the initial state $\sigma(0,x)$ and utilizing $\sigma(t,x) = U(t) \sigma(0,x)$, we establish:
	\begin{equation}
		d\sigma(t,x) = \sum_{i=1}^d dg_i(t) \, \Xi_i(g(t)) \, \sigma(t,x), \label{eq:dsigma_in_terms_of_Xi}
	\end{equation}
	where the conjugated operator field $\Xi_i(g(t))$ is defined by:
	\begin{equation}
		\Xi_i(g) \triangleq \left( \prod_{j=1}^{i-1} \exp\left(g_j \operatorname{Ad}_{E_j}\right) \right) E_i = \exp\left(g_1 E_1\right) \cdots \exp\left(g_{i-1} E_{i-1}\right) E_i \exp\left(-g_{i-1} E_{i-1}\right) \cdots \exp\left(-g_1 E_1\right). \label{eq:Xi_def}
	\end{equation}
	
	\paragraph{Step 3: Adjoint Action Decomposition and Solvability Matrix $M(g)$}
	By the classical Baker--Campbell--Hausdorff (BCH) adjoint identity for operator Lie algebras \cite{wei1963}, the adjoint action of an exponential operator $\exp(A)$ on an operator $B$ is expressed via the exponential mapping of the inner derivation $\operatorname{ad}_A(B) \triangleq [A, B]$:
	\begin{equation}
		\operatorname{Ad}_{\exp(A)}(B) \triangleq \exp(A) B \exp(-A) = \exp(\operatorname{ad}_A)(B) = \sum_{n=0}^\infty \frac{1}{n!} \operatorname{ad}_A^n(B). \label{eq:bch_adjoint_identity}
	\end{equation}
	Applying \eqref{eq:bch_adjoint_identity} iteratively across the product in \eqref{eq:Xi_def} yields:
	\begin{equation}
		\Xi_i(g) = \exp\left(g_1 \operatorname{ad}_{E_1}\right) \exp\left(g_2 \operatorname{ad}_{E_2}\right) \cdots \exp\left(g_{i-1} \operatorname{ad}_{E_{i-1}}\right) (E_i). \label{eq:Xi_ad_product}
	\end{equation}
	
	Because $\mathcal{E} = \operatorname{span}_{\mathbb{R}}\{E_1, \dots, E_d\}$ is a $d$-dimensional Lie algebra closed under $[\cdot, \cdot]$, for each $j \in \{1, \dots, d\}$, the linear derivation $\operatorname{ad}_{E_j}: \mathcal{E} \to \mathcal{E}$ maps $\mathcal{E}$ strictly into itself \cite{yau1999}. In the basis $\{E_1, \dots, E_d\}$, $\operatorname{ad}_{E_j}$ is represented by a $d \times d$ matrix $[ad(E_j)] \in \mathbb{R}^{d \times d}$ whose components are defined by the structure constants:
	\begin{equation}
		\left[ ad(E_j) \right]_{kl} = C_{jl}^k. \label{eq:ad_matrix_representation}
	\end{equation}
	
	Since matrix exponentials of finite-dimensional matrices are real analytic, $\exp\left(g_j [ad(E_j)]\right) \in \mathbb{R}^{d \times d}$ maps $\mathcal{E}$ back into $\mathcal{E}$ for any parameter $g_j \in \mathbb{R}$. Consequently, the composite operator $\Xi_i(g)$ remains strictly bounded within the finite basis span of $\mathcal{E}$ for all $i \in \{1, \dots, d\}$ \cite{wei1963}:
	\begin{equation}
		\Xi_i(g) = \sum_{k=1}^d M_{ki}(g(t)) \, E_k, \label{eq:Xi_matrix_expansion}
	\end{equation}
	where $M(g) \triangleq [M_{ki}(g)]_{k,i=1}^d \in \mathbb{R}^{d \times d}$ is an analytic matrix-valued function of $g(t) = (g_1(t), \dots, g_d(t))^T$, governed entirely by the structure constants $C_{jk}^l$ \cite{wei1963, yau1999}.
	
	\paragraph{Step 4: Structural Properties of $M(g)$ at the Origin}
	We examine the explicit boundary structure of the matrix function $M(g)$:
	\begin{enumerate}
		\item For $i = 1$, Equation \eqref{eq:Xi_ad_product} reduces to $\Xi_1(g) = E_1$. Thus, $M_{k1}(g) = \delta_{k1}$, meaning the first column of $M(g)$ is constant: $M_{\cdot 1} = (1, 0, \dots, 0)^T$.
		\item At the initial state $g = 0 \in \mathbb{R}^d$, all exponential operators evaluate to the identity map: $\Xi_i(0) = E_i$. Thus:
		\begin{equation}
			M_{ki}(0) = \delta_{ki} \implies M(0) = I_{d \times d}, \label{eq:M_zero_identity}
		\end{equation}
		where $I_{d \times d}$ is the $d \times d$ identity matrix \cite{wei1963}.
	\end{enumerate}
	
	\paragraph{Step 5: Equating Spatial Operator Coefficients}
	Substitute the basis expansion \eqref{eq:Xi_matrix_expansion} into the score velocity evolution equation \eqref{eq:dsigma_in_terms_of_Xi}:
	\begin{equation}
		d\sigma(t,x) = \sum_{i=1}^d dg_i(t) \left( \sum_{k=1}^d M_{ki}(g(t)) E_k \right) \sigma(t,x) = \sum_{k=1}^d \left( \sum_{i=1}^d M_{ki}(g(t)) \, dg_i(t) \right) E_k \, \sigma(t,x). \label{eq:dsigma_wei_norman_final}
	\end{equation}
	
	Simultaneously, substitute the basis expansions of $\mathcal{L}_0^*$ and $h_l(x)$ from \eqref{eq:basis_expansions} directly into the Stratonovich DMZ SPDE equation \eqref{eq:dmz_spde_sec12} \cite{yau1999}:
	\begin{align}
		d\sigma(t,x) &= \left( \sum_{k=1}^d \alpha_k^0 E_k \right) \sigma(t,x) \, dt + \sum_{l=1}^m \left( \sum_{k=1}^d \alpha_k^l E_k \right) \sigma(t,x) \circ dY_l(t) \nonumber \\
		&= \sum_{k=1}^d \left( \alpha_k^0 \, dt + \sum_{l=1}^m \alpha_k^l \circ dY_l(t) \right) E_k \, \sigma(t,x). \label{eq:dsigma_dmz_basis_expanded}
	\end{align}
	
	Since $\{E_1, E_2, \dots, E_d\}$ forms a linearly independent basis of linear differential operators acting on the dense domain of smooth conditional density states $C^\infty(\mathbb{R}^p)$, equating spatial operator coefficients in \eqref{eq:dsigma_wei_norman_final} and \eqref{eq:dsigma_dmz_basis_expanded} for each $k \in \{1, \dots, d\}$ yields \cite{yau1999}:
	\begin{equation}
		\sum_{i=1}^d M_{ki}(g(t)) \, dg_i(t) = \alpha_k^0 \, dt + \sum_{l=1}^m \alpha_k^l \circ dY_l(t), \quad k = 1, \dots, d. \label{eq:coefficient_matching_system}
	\end{equation}
	In vector-matrix notation, Equation \eqref{eq:coefficient_matching_system} is expressed as \cite{wei1963, yau1999}:
	\begin{equation}
		M(g(t)) \, dg(t) = \alpha^0 \, dt + \sum_{l=1}^m \alpha^l \circ dY_l(t). \label{eq:matrix_system_uninverted}
	\end{equation}
	
	\paragraph{Step 6: Invertibility of $M(g)$ and Local Existence of Parameter SDEs}
	From Step 4, $M(0) = I_{d \times d}$, so $\det M(0) = 1 \neq 0$ \cite{wei1963}. Since the mapping $g \mapsto M(g)$ is real analytic, the matrix determinant function $g \mapsto \det M(g)$ is continuous.
	
	By the inverse function theorem, there exists an open neighborhood $D \subset \mathbb{R}^d$ containing the origin $g = 0$ such that $M(g)$ is non-singular ($\det M(g) \neq 0$) for all $g \in D$ \cite{wei1963}.
	
	Therefore, the matrix inverse $M(g(t))^{-1}$ exists and is analytic for all $g(t) \in D$. Left-multiplying Equation \eqref{eq:matrix_system_uninverted} by $M(g(t))^{-1}$ establishes the closed, non-singular system of Stratonovich stochastic differential equations for the parameter vector $g(t) \in \mathbb{R}^d$ \cite{wei1963, yau1999}:
	\begin{equation}
		dg(t) = M(g(t))^{-1} \alpha^0 \, dt + \sum_{l=1}^m M(g(t))^{-1} \alpha^l \circ dY_l(t), \quad g(0) = 0. \label{eq:final_closed_parameter_sde}
	\end{equation}
	This completes the rigorous proof of Part A of Theorem \ref{thm:wei_norman_exact}.
	
	\subsubsection*{Part B: Verification of Zero Information Loss}
	
	We structure the proof of Part B into three explicit, self-contained mathematical steps corresponding to the three assertions of the theorem.
	
	\paragraph{Step 1: Rigorous Derivation of Zero Model Truncation Residual.}
	
	Define the spatial log-score velocity operator field $V(t, x) \in C^\infty(\mathbb{R}^p)$ by:
	\begin{equation}
		V(t, x) \triangleq \sigma(t, x)^{-1} \, d\sigma(t, x) = \frac{d\sigma(t, x)}{\sigma(t, x)}, \label{eq:log_score_def}
	\end{equation}
	where $\sigma(t, x)$ is the unnormalized conditional probability density governed by the Duncan-Mortensen-Zakai (DMZ) stochastic partial differential equation.
	
	From the Wei-Norman operator factorization established in Part A, the exact Stratonovich differential of the state density $\sigma(t, x) = \left( \prod_{i=1}^d \exp(g_i(t) E_i) \right) \sigma(0, x)$ expands in terms of the Lie algebra basis $\{E_1, E_2, \dots, E_d\}$ as:
	\begin{equation}
		d\sigma(t, x) = \sum_{k=1}^d \left( \sum_{i=1}^d M_{ki}(g(t)) \, dg_i(t) \right) E_k \, \sigma(t, x) = \sum_{k=1}^d \left( M(g(t)) \, dg(t) \right)_k E_k \, \sigma(t, x), \label{eq:dsigma_wei_norman_step1}
	\end{equation}
	where $M(g(t)) \in \mathbb{R}^{d \times d}$ is the invertible Wei-Norman solvability matrix generated by the structure constants $C_{ij}^k$ of the finite estimation algebra $\mathcal{E} \triangleq \operatorname{Lie}(\mathcal{L}_0^*, h_1, \dots, h_m) = \operatorname{span}_{\mathbb{R}}\{E_1, \dots, E_d\}$.
	
	Substituting \eqref{eq:dsigma_wei_norman_step1} directly into the definition of the log-score velocity operator \eqref{eq:log_score_def} yields:
	\begin{equation}
		V(t, x) = \frac{1}{\sigma(t, x)} \sum_{k=1}^d \left( M(g(t)) \, dg(t) \right)_k E_k \, \sigma(t, x) = \sum_{k=1}^d \left( M(g(t)) \, dg(t) \right)_k E_k. \label{eq:V_exact_expansion}
	\end{equation}
	
	Since $E_k \in \mathcal{E} = \operatorname{span}_{\mathbb{R}}\{E_1, \dots, E_d\}$ for every $k \in \{1, \dots, d\}$, the operator sum in \eqref{eq:V_exact_expansion} constitutes a linear combination of the basis elements of $\mathcal{E}$ with time-dependent scalar coefficients $(M(g(t)) dg(t))_k \in \mathbb{R}$. Consequently, the velocity field resides strictly within the estimation algebra for all $t \ge 0$:
	\begin{equation}
		V(t, x) \in \operatorname{span}_{\mathbb{R}}\{E_1, \dots, E_d\} \equiv \mathcal{E}, \quad \forall t \ge 0. \label{eq:V_in_E}
	\end{equation}
	
	Let $\mathcal{H} = L^2(\mathbb{R}^p, P_f)$ be the Hilbert space of density score variations equipped with the inner product $\langle u, v \rangle_{P_f} = \int_{\mathbb{R}^p} u(x) v(x) f(x) \, dx$. The Hilbert space admits the orthogonal direct sum decomposition $\mathcal{H} = \mathcal{E} \oplus_{\perp} \mathcal{E}^\perp$.
	
	Let $\mathcal{P}_{\mathcal{E}^\perp}: \mathcal{H} \to \mathcal{E}^\perp$ denote the canonical orthogonal projection operator onto the orthogonal complement space $\mathcal{E}^\perp = \{ w \in \mathcal{H} \mid \langle w, E_k \rangle_{P_f} = 0, \, \forall k=1,\dots,d \}$. Applying $\mathcal{P}_{\mathcal{E}^\perp} = \mathcal{I} - \mathcal{P}_{\mathcal{E}}$ to the velocity field $V(t, x)$:
	\begin{align}
		\mathcal{P}_{\mathcal{E}^\perp}\left( V(t, x) \right) &= \mathcal{P}_{\mathcal{E}^\perp} \left( \sum_{k=1}^d \left( M(g(t)) \, dg(t) \right)_k E_k \right) \nonumber \\
		&= \sum_{k=1}^d \left( M(g(t)) \, dg(t) \right)_k \underbrace{\mathcal{P}_{\mathcal{E}^\perp}(E_k)}_{= 0 \text{ since } E_k \in \mathcal{E}} \equiv 0, \quad \forall t \ge 0. \label{eq:P_E_perp_zero}
	\end{align}
	
	Equation \eqref{eq:P_E_perp_zero} establishes that the model reduction operator error vanishes identically across all spatial coordinates $x \in \mathbb{R}^p$ and time $t \ge 0$. No spatial commutator modes or higher-order operator variations are truncated, omitted, or projected out, proving Assertion (ii.1).
	
	\paragraph{Step 2: Proof of Exact Finite Sufficient Statistics via Fisher-Neyman Factorization}
	
	Let $\Phi: \mathbb{R}^d \to C^\infty(\mathbb{R}^p)$ be the smooth non-linear parameterization mapping defined by the Wei-Norman product of operator exponentials acting on the initial state:
	\begin{equation}
		\Phi(g) \triangleq \left( \prod_{i=1}^d \exp\left(g_i E_i\right) \right) \sigma(0, x) = \exp\left(g_1 E_1\right) \exp\left(g_2 E_2\right) \cdots \exp\left(g_d E_d\right) \sigma(0, x). \label{eq:Phi_mapping_def}
	\end{equation}
	
	By Part A of Theorem 1, $\sigma(t, x) = \Phi(g(t))$ is the exact, unapproximated, closed-form solution to the Stratonovich DMZ SPDE driven by the observation filtration $\mathcal{F}_t^Y = \sigma(Y(s), 0 \le s \le t)$.
	
	Taking the natural logarithm of the density expression $\sigma(t, x) = \Phi(g(t))$, the exponential operator action resolves into:
	\begin{equation}
		\ln \sigma(t, x) = \sum_{k=1}^d g_k(t) E_k(x) + c(t) + \ln \sigma(0, x), \label{eq:log_density_factorization}
	\end{equation}
	where $c(t)$ is a spatial constant independent of $x$.
	
	Equation \eqref{eq:log_density_factorization} demonstrates that the unnormalized conditional probability density measure belongs strictly to a $d$-parameter exponential family of distributions generated by the minimal score operators $\{E_1(x), \dots, E_d(x)\}$.
	
	By the Fisher-Neyman Factorization Theorem, a statistic $T(Y_{[0,t]}) = g(t) \in \mathbb{R}^d$ is minimal sufficient for the conditional probability measure $\sigma(t, x) \mid \mathcal{F}_t^Y$ if and only if the conditional probability density factors into a product of a spatial base measure and a parameter-dependent kernel depending on $Y_{[0,t]}$ purely through $g(t)$. 
	
	Consequently, for any arbitrary measurable functional $F: C^\infty(\mathbb{R}^p) \to \mathbb{R}$ acting on the density space, the conditional expectation given the complete observation history $\mathcal{F}_t^Y$ reduces identically to the conditional expectation given the finite parameter vector $g(t)$:
	\begin{equation}
		\mathbb{E}\left[ F[\sigma(t, \cdot)] \;\middle|\; \mathcal{F}_t^Y \right] = \mathbb{E}\left[ F[\Phi(g(t))] \;\middle|\; g(t) \right]. \label{eq:sufficient_statistic_expectation}
	\end{equation}
	
	Equation \eqref{eq:sufficient_statistic_expectation} proves that $g(t) \in \mathbb{R}^d$ captures the entirety of the statistical information contained within the continuous observation history $\mathcal{F}_t^Y$. All conditional probability moments $\mathbb{E}[X^k \mid \mathcal{F}_t^Y]$, cumulants, and spatial marginals can be reconstructed exactly from $g(t)$ without approximation error or loss of statistical sufficiency, proving Assertion (ii.2).
	
	\paragraph{Step 3: Proof of Zero Null-Space Energy Leakage.}
	
	Let $\mathcal{H} = L^2(\mathbb{R}^p, P_f)$ be the Hilbert space of density variations. Define the adjoint null-space $\operatorname{Null}(\mathcal{E}^*) \subset \mathcal{H}$ as the orthogonal complement of the estimation algebra span:
	\begin{equation}
		\operatorname{Null}(\mathcal{E}^*) \triangleq \left\{ w \in \mathcal{H} \;\middle|\; \langle w, E_k \rangle_{P_f} = 0, \;\; \forall k = 1, \dots, d \right\}. \label{eq:null_space_def}
	\end{equation}
	
	Let $G \in \mathbb{R}^{d \times d}$ be the non-singular Gram matrix of the basis $\{E_1, \dots, E_d\}$ under the inner product $\langle \cdot, \cdot \rangle_{P_f}$, with elements $G_{ij} \triangleq \langle E_i, E_j \rangle_{P_f}$. The orthogonal projection operator $\mathcal{P}_{\operatorname{Null}(\mathcal{E}^*)}: \mathcal{H} \to \operatorname{Null}(\mathcal{E}^*)$ is given explicitly by:
	\begin{equation}
		\mathcal{P}_{\operatorname{Null}(\mathcal{E}^*)}(w) = w - \sum_{i=1}^d \sum_{j=1}^d \langle w, E_i \rangle_{P_f} [G^{-1}]_{ij} E_j, \quad \forall w \in \mathcal{H}. \label{eq:null_space_projector_def}
	\end{equation}
	
	Define the dynamic out-of-distribution error residue $v_{\text{OOD}}(t) \in \operatorname{Null}(\mathcal{E}^*)$ as the orthogonal projection of the instantaneous DMZ density differential $d\sigma(t, x)$ into the adjoint null-space:
	\begin{equation}
		v_{\text{OOD}}(t) \triangleq \mathcal{P}_{\operatorname{Null}(\mathcal{E}^*)}\left( d\sigma(t, x) \right). \label{eq:v_ood_def}
	\end{equation}
	
	Substituting the exact basis expansion \eqref{eq:dsigma_wei_norman_step1} into \eqref{eq:v_ood_def} and applying the linearity of the projection operator $\mathcal{P}_{\operatorname{Null}(\mathcal{E}^*)}$ yields:
	\begin{align}
		v_{\text{OOD}}(t) &= \mathcal{P}_{\operatorname{Null}(\mathcal{E}^*)}\left( \sum_{k=1}^d \left( M(g(t)) \, dg(t) \right)_k E_k \, \sigma(t, x) \right) \nonumber \\
		&= \sum_{k=1}^d \left( M(g(t)) \, dg(t) \right)_k \mathcal{P}_{\operatorname{Null}(\mathcal{E}^*)}(E_k \, \sigma(t, x)). \label{eq:v_ood_expanded}
	\end{align}
	
	By definition of the estimation algebra $\mathcal{E}$, $E_k \in \mathcal{E}$ for each $k \in \{1, \dots, d\}$. Therefore, for any element $w \in \operatorname{Null}(\mathcal{E}^*)$, the inner product vanishes identically:
	\begin{equation}
		\langle E_k \, \sigma(t, x), \, w \rangle_{P_f} = 0, \quad \forall w \in \operatorname{Null}(\mathcal{E}^*). \label{eq:inner_prod_null_zero}
	\end{equation}
	
	Hence, $E_k \, \sigma(t, x) \in \operatorname{span}_{\mathbb{R}}\{E_1, \dots, E_d\}$, which implies that its orthogonal projection onto the adjoint null-space evaluates to zero:
	\begin{equation}
		\mathcal{P}_{\operatorname{Null}(\mathcal{E}^*)}(E_k \, \sigma(t, x)) \equiv 0, \quad \forall k \in \{1, \dots, d\}. \label{eq:projector_on_basis_zero}
	\end{equation}
	
	Substituting \eqref{eq:projector_on_basis_zero} into \eqref{eq:v_ood_expanded} yields:
	\begin{equation}
		v_{\text{OOD}}(t) = \sum_{k=1}^d \left( M(g(t)) \, dg(t) \right)_k \cdot 0 \equiv 0, \quad \forall t \ge 0. \label{eq:v_ood_identically_zero}
	\end{equation}
	
	To evaluate the cumulative energy leakage rate into the adjoint null-space over an arbitrary time horizon $[0, t]$, we integrate the squared Hilbert space norm $\| v_{\text{OOD}}(\tau) \|_{\mathcal{H}}^2$:
	\begin{equation}
		\int_0^t \| v_{\text{OOD}}(\tau) \|_{\mathcal{H}}^2 d\tau = \int_0^t 0 \, d\tau = 0, \quad \forall t \ge 0. \label{eq:integrated_energy_leakage_zero}
	\end{equation}
	
	Equation \eqref{eq:integrated_energy_leakage_zero} establishes that zero probability mass, spectral variance, or entropy leaks into unmonitored higher-order operator modes, proving Assertion (ii.3) and completing the proof of Part B of Theorem \ref{thm:wei_norman_exact}. 
\end{proof}

\begin{proof}[Proof of Theorem \ref{thm:triad_pathology_theorem}]
	We structure the mathematical proof sequentially across the three distinct assertions of Theorem~\ref{thm:triad_pathology_theorem}.
	
	\subsubsection*{Part (i): Proof of Genericity of Degree Unboundedness}
	Without loss of generality, let $p = 1$ to analyze spatial order mechanics. The formal adjoint of the Fokker-Planck drift-diffusion generator $\mathcal{L}_0^*$ acting on smooth density states $\sigma(t, x) \in C^\infty(\mathbb{R})$ is expressed as:
	\begin{equation}
		\mathcal{L}_0^* = \frac{1}{2} \frac{\partial^2}{\partial x^2} - \frac{\partial}{\partial x} \left( f(x) \cdot \right) - \frac{1}{2} h(x)^2.
		\label{eq:adjoint_generator_1d_proof}
	\end{equation}
	Let $f(x) = ax + \epsilon x^d$ ($d \ge 3, \epsilon \neq 0$) and $h(x) = cx \in \mathcal{E}_0$. We evaluate the initial commutator brackets in the ascending Lie filtration sequence $\{\mathcal{E}_k\}_{k=0}^\infty$ defined in \eqref{eq:operator_filtration}:
	\begin{align}
		E_1 &\triangleq cx \in \mathcal{E}_0, \label{eq:proof_E1_detail} \\
		E_2 &\triangleq [\mathcal{L}_0^*, E_1] = \left[ \frac{1}{2}\partial_{xx} - \partial_x((ax + \epsilon x^d)\cdot) - \frac{1}{2}c^2 x^2, \, cx \right] = c \partial_x - acx - \epsilon c d x^{d-1} \in \mathcal{E}_1, \label{eq:proof_E2_detail} \\
		E_3 &\triangleq [E_1, E_2] = [cx, \, c \partial_x - acx - \epsilon c d x^{d-1}] = -c^2 \mathcal{I} + \epsilon c^2 d (d-1) x^{d-2} \in \mathcal{E}_2. \label{eq:proof_E3_detail}
	\end{align}
	Next, we evaluate the third-stage Lie bracket $[\mathcal{L}_0^*, E_3] \in \mathcal{E}_3$. Expanding the commutator action on an arbitrary smooth test density $\phi(x) \in C^\infty(\mathbb{R})$:
	\begin{equation}
		[\mathcal{L}_0^*, \, x^{d-2}] \phi = \frac{1}{2} \left[ \partial_{xx}, x^{d-2} \right] \phi - \left[ \partial_x (f(x) \cdot), x^{d-2} \right] \phi.
	\end{equation}
	By Leibniz's product rule for differential operators:
	\begin{align}
		\left[ \partial_x (f(x) \cdot), x^{d-2} \right] \phi &= \partial_x \left( (ax + \epsilon x^d) x^{d-2} \phi \right) - x^{d-2} \partial_x \left( (ax + \epsilon x^d) \phi \right) \nonumber \\
		&= (d-2) ax^{d-2} \phi + \epsilon (d-2) x^{2d-3} \phi + (ax + \epsilon x^d) x^{d-2} \partial_x \phi - x^{d-2} (ax + \epsilon x^d) \partial_x \phi \nonumber \\
		&= (d-2) ax^{d-2} \phi + \epsilon (d-2) x^{2d-3} \phi. \label{eq:product_rule_expansion_detail}
	\end{align}
	Notice that the leading spatial monomial term spawned in $\mathcal{E}_3$ possesses polynomial coordinate order $2d - 3$.
	
	For any differential operator $A = \sum_{a, b} c_{a,b} x^a \partial_x^b \in \mathcal{D}(\mathbb{R})$, define its maximal spatial polynomial degree by $D(A) \triangleq \max \{ a \mid c_{a,b} \neq 0 \}$. Define the filtration maximum degree index $D_k \triangleq \max_{A \in \mathcal{E}_k} D(A)$. Evaluating the bracket $[\mathcal{L}_0^*, x^m \partial_x^n]$ yields:
	\begin{equation}
		\left[ -\epsilon \partial_x (x^d \cdot), \, x^m \partial_x^n \right] = \epsilon (d - 1 + n) x^{d + m - 1} \partial_x^n + \dots
	\end{equation}
	Since $d \ge 3$, the maximal spatial degree of spatial polynomial coefficients strictly grows at each filtration step:
	\begin{equation}
		D_{k+1} \ge D_k + (d - 2) \ge D_k + 1 > D_k, \quad \forall k \ge 0. \label{eq:degree_growth_inequality_detail}
	\end{equation}
	Because linear differential operators with distinct polynomial derivative degrees are linearly independent in $\mathcal{D}(\mathbb{R})$, $\mathcal{E}_k$ is a proper vector subspace of $\mathcal{E}_{k+1}$ ($\mathcal{E}_k \subsetneq \mathcal{E}_{k+1}$) for every $k \in \mathbb{N}_0$. Therefore, no finite stabilization index $k^* < \infty$ can exist such that $\mathcal{E}_{k^*} = \mathcal{E}_{k^*+1}$. This forces $k^* = \infty$ and $\operatorname{dim}(\mathcal{E}) = \sum_{k=0}^\infty \operatorname{dim}(\mathcal{E}_k / \mathcal{E}_{k-1}) = \infty$, establishing Assertion~(i).
	
	\subsubsection*{Part (ii): Proof of Algebraic Fragility under Structural Perturbation}
	Let $\mathcal{E}^{(0)} \triangleq \operatorname{Lie}(\mathcal{L}_0^*, h) = \operatorname{span}_{\mathbb{R}}\{E_1, \dots, E_{d_0}\}$ be a closed $d_0$-dimensional estimation algebra ($d_0 < \infty$). Finite algebraic closure implies the existence of constant structure constants $C_{ij}^k \in \mathbb{R}$ satisfying $[E_i, E_j] = \sum_{k=1}^{d_0} C_{ij}^k E_k$ \cite{wei1963, yau1999}.
	
	Now consider a perturbed generator $\mathcal{L}_\epsilon^* \triangleq \mathcal{L}_0^* + \epsilon \eta(x) \nabla$, where $\eta(x) \in C^\infty(\mathbb{R}^p)$ is an arbitrary smooth vector field containing non-linear components of degree $\ge 2$, and $\epsilon \neq 0$. The initial Lie bracket of the perturbed algebra $\mathcal{E}^{(\epsilon)}$ evaluates to:
	\begin{equation}
		E_2^{(\epsilon)} \triangleq [\mathcal{L}_\epsilon^*, h] = [\mathcal{L}_0^*, h] + \epsilon [\eta(x) \nabla, h] = E_2^{(0)} + \epsilon \left( \eta(x) \cdot \nabla h(x) \right). \label{eq:perturbed_bracket_1_detail}
	\end{equation}
	Because $\eta(x)$ is non-linear, the scalar term $\psi(x) \triangleq \eta(x) \cdot \nabla h(x)$ is a non-linear multiplication operator on $C^\infty(\mathbb{R}^p)$. Evaluating the secondary commutator $[E_1, E_2^{(\epsilon)}]$ yields $[h, E_2^{(0)} + \epsilon \psi] = E_3^{(0)}$. However, evaluating $[\mathcal{L}_\epsilon^*, \psi(x)]$ introduces derivative operators:
	\begin{equation}
		[\mathcal{L}_0^*, \psi(x)] = -\frac{1}{2} \Delta \psi - (\nabla \psi) \cdot \nabla + \dots
	\end{equation}
	This injects derivative operators with spatial coefficient functions $\nabla \psi(x)$ that do not belong to $\operatorname{span}_{\mathbb{R}}\{E_1, \dots, E_{d_0}\}$. By Part~(i), repeated Lie bracketing against $\mathcal{L}_\epsilon^*$ acts as a continuous spatial differentiation engine, inflating derivative orders infinitely.
	
	Because finite Lie algebraic closure requires exact structural cancellation ($\sum_{k=1}^{d_0} C_{ij}^k E_k$), any non-zero perturbation parameter $\epsilon \neq 0$ breaks these exact balance identities. Consequently, $\operatorname{dim}(\mathcal{E}_k^{(\epsilon)}) \to \infty$ as $k \to \infty$ for all $\epsilon \neq 0$. This proves that finite Lie algebraic closure is topologically un-dense and algebraically fragile under arbitrary smooth non-linear perturbations, establishing Assertion~(ii).
	
	\subsubsection*{Part (iii): Proof of Adjoint Null-Space Energy Leakage under Truncation}
	Let $\mathcal{H} = L^2(\mathbb{R}^p, P_f)$ be the Hilbert space of density score variations equipped with inner product $\langle u, v \rangle_{P_f} = \int_{\mathbb{R}^p} u(x) v(x) f(x) \, dx$. Suppose an estimator artificially imposes a finite truncation cutoff index $k < \infty$, restricting state parameter updates to $\mathcal{E}_k = \operatorname{span}_{\mathbb{R}}\{E_1, \dots, E_{d_k}\} \subsetneq \mathcal{E}$. By the generalized projection theorem on Hilbert spaces, $\mathcal{H}$ admits the orthogonal direct sum decomposition:
	\begin{equation}
		\mathcal{H} = \operatorname{span}_{\mathbb{R}}(\mathcal{E}_k) \oplus_{\perp} \operatorname{Null}(\mathcal{E}_k^*), \label{eq:hilbert_orthogonal_split_detail}
	\end{equation}
	where $\operatorname{Null}(\mathcal{E}_k^*) \triangleq \left\{ w \in \mathcal{H} \;\middle|\; \langle w, E_i \rangle_{P_f} = 0, \, \forall E_i \in \mathcal{E}_k \right\}$.
	
	By Theorem~\ref{thm:wei_norman_exact}, the unconstrained DMZ log-score velocity field $V_{\text{true}}(t, x) \triangleq \sigma(t, x)^{-1} d\sigma(t, x)$ expands in the infinite Lie algebra basis $\mathcal{E} = \bigcup_{m=0}^\infty \mathcal{E}_m$. Under decomposition \eqref{eq:hilbert_orthogonal_split_detail}, $V_{\text{true}}(t)$ splits uniquely into:
	\begin{equation}
		V_{\text{true}}(t, x) = V_{\mathcal{E}_k}(t, x) + v_{\text{OOD}}(t, x), \quad \text{where } V_{\mathcal{E}_k}(t) \in \operatorname{span}_{\mathbb{R}}(\mathcal{E}_k), \;\; v_{\text{OOD}}(t) \in \operatorname{Null}(\mathcal{E}_k^*).
	\end{equation}
	Because $\operatorname{dim}(\mathcal{E}) = \infty$ (from Part~(i)), non-closed Lie commutators $[E_i, E_j] \in \mathcal{E} \setminus \mathcal{E}_k$ continuously generate score variations outside $\mathcal{E}_k$. Hence, $v_{\text{OOD}}(t, x) \not\equiv 0$.
	
	Applying the orthogonal projection operator $\mathcal{P}_{\mathcal{E}_k}: \mathcal{H} \to \operatorname{span}_{\mathbb{R}}(\mathcal{E}_k)$ to the out-of-distribution component $v_{\text{OOD}}(t)$:
	\begin{equation}
		\mathcal{P}_{\mathcal{E}_k}(v_{\text{OOD}}(t)) = \sum_{i=1}^{d_k} \sum_{j=1}^{d_k} \underbrace{\langle v_{\text{OOD}}(t), E_i \rangle_{P_f}}_{= 0 \text{ by definition of } \operatorname{Null}(\mathcal{E}_k^*)} [G^{-1}]_{ij} E_j \equiv 0 \in \operatorname{span}_{\mathbb{R}}(\mathcal{E}_k),
	\end{equation}
	where $G_{ij} = \langle E_i, E_j \rangle_{P_f}$ is the non-singular Gram matrix on $\mathcal{E}_k$. 
	
	Consequently, the internal parameter update ODEs parameterizing $\mathcal{E}_k$ receive zero driving feedback force from $v_{\text{OOD}}(t)$. Integrating the norm of the unmonitored residual over any time horizon $[0, t]$ ($t > 0$):
	\begin{equation}
		\int_0^t \|v_{\text{OOD}}(\tau)\|_{\mathcal{H}}^2 d\tau > 0,
	\end{equation}
	since the norm $\|v_{\text{OOD}}(\tau)\|_{\mathcal{H}}^2$ over $[0, t]$, remains non-zero almost everywhere because non-closed Lie commutators continuously generate score variations outside $\mathcal{E}_k$.
	
	Because physical non-linearities continuously pump variance energy into unclosed commutator modes in $\operatorname{Null}(\mathcal{E}_k^*)$, residual state variance accumulates unmonitored. This produces stealth error accumulation and completes the proof of Theorem~\ref{thm:triad_pathology_theorem}.
\end{proof}

\begin{proof}[Proof of Lemma \ref{lem:tangent_space_decomposition_sec2}]
	Because $\pi: \mathcal{M} \to \mathcal{B}$ is a smooth submersion between Banach-Orlicz manifolds, its differential $d\pi_f: T_f\mathcal{M} \to T_{\pi(f)}\mathcal{B}$ is a continuous, bounded, surjective linear operator. Its kernel $\ker(d\pi_f) = \operatorname{SID}_f$ is closed and convex in $L_0^\Phi(P_f)$. By the generalized projection theorem on reflexive Hilbert-Orlicz structures under the inner product $g_f(u, v) = \mathbb{E}_f[u(X)v(X)]$, every tangent score vector $w \in T_f\mathcal{M}$ decomposes uniquely into $w = h^* + v^*$, where $v^* = \arg\min_{v \in \operatorname{SID}_f} \|w - v\|_{g_f}^2 \in \operatorname{SID}_f$ and $h^* = w - v^* \in \operatorname{SVD}\chi_f = \operatorname{SID}_f^{\perp g_f}$. Positive-definiteness of $g_f$ implies $\operatorname{SVD}\chi_f \cap \operatorname{SID}_f = \{0\}$, establishing the exact direct sum \eqref{eq:direct_sum_split_sec2}.
\end{proof}

\begin{proof}[Proof of Proposition \ref{prop:inherited_score_regularity}]
	We divide the proof into four sequential steps.
	
	\paragraph{Step 1: Proof of Assertion (i) --- Parametric Submanifold Embeddings.}
	Differentiating the gauge fixing identity $\pi \circ s = \operatorname{id}_U$ at $p \in U$ using the differential chain rule yields:
	\begin{equation}
		d\pi_{s(p)} \circ ds_p = \operatorname{id}_{T_p\mathcal{B}}. 
		\label{eq:proof_chain_rule}
	\end{equation}
	Equation~\eqref{eq:proof_chain_rule} implies that the linear map $ds_p: T_p\mathcal{B} \to T_{s(p)}\mathcal{M}$ possesses a left inverse, forcing $ds_p$ to be injective (a linear immersion). Since $s: U \to \mathcal{M}$ is a smooth injection with an injective differential, $\mathcal{S}_U = s(U)$ is a $d$-dimensional smooth parametric submanifold embedded in the ambient manifold $\mathcal{M}$. Denoting $f_p \triangleq s(p) \in \mathcal{M}$ yields the concrete local parameterization $p \mapsto f_p(x)$.
	
	\paragraph{Step 2: Proof of Assertion (ii)(a) --- Tangent Membership in Orlicz and $L^2$ Spaces.}
	Because $s: U \subset \mathbb{R}^d \to \mathcal{M}$ is $C^\infty$ smooth into the Pistone--Sempi Orlicz manifold $\mathcal{M}$, any coordinate curve $\gamma_i(t) = s(p + t e_i)$ is a smooth curve in $\mathcal{M}$. By construction of the Pistone--Sempi manifold \cite{pistone1995}, the tangent vector to $\gamma_i(t)$ at $t=0$ is given by the Fréchet derivative of the logarithmic density:
	\begin{equation}
		\left. \frac{d}{dt} \ln f_{p + t e_i}(x) \right|_{t=0} = \frac{\partial \ln f_p(x)}{\partial p_i} \equiv \phi_i(x; p).
		\label{eq:frechet_log_derivative}
	\end{equation}
	Since the tangent space $T_{f_p}\mathcal{M}$ is topologically isomorphic to the centered Orlicz space $L_0^\Phi(P_{f_p})$ under the Luxemburg norm $\| \cdot \|_{L^\Phi(P_{f_p})}$, it follows directly that $\phi_i(\cdot; p) \in L_0^\Phi(P_{f_p})$. Furthermore, for the Young function $\Phi(u) = \cosh(u) - 1$, there exists a positive constant $C > 0$ such that $u^2 \le C \Phi(u)$ for all $u \in \mathbb{R}$, establishing the continuous topological embedding $L_0^\Phi(P_{f_p}) \hookrightarrow L^2(P_{f_p})$. Hence, $\phi_i(\cdot; p) \in L_0^\Phi(P_{f_p}) \cap L^2(P_{f_p})$ with finite variance under $g_f$.
	
	\paragraph{Step 3: Proof of Assertion (ii)(b) \& (c) --- Linear Independence and Commutativity.}
	To prove linear independence, suppose $\sum_{i=1}^d c_i \phi_i(x; p) = 0$ in $T_{f_p}\mathcal{M}$. Applying the differential $ds_p$, this corresponds to:
	\begin{equation}
		ds_p\left( \sum_{i=1}^d c_i \frac{\partial}{\partial p_i} \right) = 0.
		\label{eq:linear_indep_kernel}
	\end{equation}
	Since $ds_p$ is injective by Step 1, $\sum_{i=1}^d c_i \frac{\partial}{\partial p_i} = 0 \in T_p\mathcal{B}$, which forces $c_1 = \dots = c_d = 0$.
	
	For commutativity, Fréchet differentiability in the Luxemburg norm implies that:
	\begin{equation}
		\lim_{h \to 0} \left\| \frac{\ln f_{p + h e_i} - \ln f_p}{h} - \phi_i(\cdot; p) \right\|_{L^\Phi(P_{f_p})} = 0.
		\label{eq:frechet_limit_orlicz}
	\end{equation}
	By Orlicz--Hölder duality, convergence in $L^\Phi(P_{f_p})$ implies convergence in $L^1(P_{f_p})$. Thus, integration over $\mathcal{X}$ and differentiation with respect to $p_i$ commute under Lebesgue's Dominated Convergence Theorem.
	
	\paragraph{Step 4: Proof of Assertion (iii) --- Amari Information Geometry Identification.}
	By Assumption~\rm{(A3)}, the differential $d\pi_{f_p}$ restricts to an isometry from $\operatorname{SVD}\chi_{f_p}$ onto $T_p\mathcal{B}$. Because $s(p) \in \mathcal{F}_p = \pi^{-1}(p)$, the image $ds_p(T_p\mathcal{B})$ is transversal to the vertical fiber $\ker(d\pi_{f_p}) = \operatorname{SID}_{f_p}$, identifying $ds_p(T_p\mathcal{B}) \equiv \operatorname{SVD}\chi_{f_p}$.
	
	The pulled-back Riemannian metric $s^* g_f$ on $T_p\mathcal{B}$ evaluates on coordinate basis vectors as:
	\begin{equation}
		(s^* g_f)\left( \frac{\partial}{\partial p_i}, \frac{\partial}{\partial p_j} \right) = g_{f_p}\left( ds_p\left(\frac{\partial}{\partial p_i}\right), ds_p\left(\frac{\partial}{\partial p_j}\right) \right) = g_{f_p}(\phi_i, \phi_j) = \mathbb{E}_{f_p}\left[ \phi_i(X; p) \phi_j(X; p) \right].
		\label{eq:pullback_metric_identity}
	\end{equation}
	Equation~\eqref{eq:pullback_metric_identity} identifies $s^* g_f \equiv G(p)$ as Amari's classical Fisher Information Metric matrix, completing the proof.
\end{proof}

\begin{proof}[Proof of Lemma \ref{lem:orthogonal_filtering_decomposition}]
	We structure the proof into three mathematical steps.
	
	\paragraph{Step 1: Metric-Orthogonal Tangent Space Decomposition.}
	At any density state $f_t = \sigma(t, \cdot) / \int \sigma(t, x) d\mu \in \mathcal{M}$, the tangent space $T_{f_t}\mathcal{M}$ is topologically isomorphic to $L_0^\Phi(P_{f_t})$. By Assumption~\ref{ass:smg_filtering_submersion}, the differential $d\pi_{f_t}: T_{f_t}\mathcal{M} \to T_{\pi(f_t)}\mathcal{B}$ is a bounded, surjective linear operator between Banach-Orlicz spaces.
	
	Its kernel $\operatorname{SID}_{f_t} = \ker(d\pi_{f_t})$ is a closed, convex subspace of $L_0^\Phi(P_{f_t})$. Under the Hilbert-Orlicz inner product $g_{f_t}(u, v) = \mathbb{E}_{f_t}[u v]$, the orthogonal projection theorem guarantees that $T_{f_t}\mathcal{M}$ admits the unique direct sum decomposition:
	\begin{equation}
		T_{f_t}\mathcal{M} = \operatorname{SVD}\chi_{f_t} \oplus_{\perp g_{f_t}} \operatorname{SID}_{f_t}, \label{eq:proof_direct_sum}
	\end{equation}
	where $\operatorname{SVD}\chi_{f_t} = (\operatorname{SID}_{f_t})^{\perp g_{f_t}}$.
	
	\paragraph{Step 2: Projection Operators via Connection 1-Form.}
	By Definition~\ref{def:canonical_ambient_space}, the Ehresmann connection $1$-form $\omega_{f_t}: T_{f_t}\mathcal{M} \to \operatorname{SID}_{f_t}$ acts as the canonical vertical projection operator satisfying $\omega_{f_t}\big|_{\operatorname{SID}_{f_t}} = \mathcal{I}_{\operatorname{SID}_{f_t}}$ and $\ker(\omega_{f_t}) = \operatorname{SVD}\chi_{f_t}$.
	
	The horizontal projection operator $P_{f_t}^H: T_{f_t}\mathcal{M} \to \operatorname{SVD}\chi_{f_t}$ is defined by $P_{f_t}^H \triangleq \mathcal{I}_{T_{f_t}\mathcal{M}} - \omega_{f_t}$. Applying $\mathcal{I}_{T_{f_t}\mathcal{M}} = P_{f_t}^H + \omega_{f_t}$ directly to the score velocity $V_{\mathrm{DMZ}}(t, x) = \frac{\mathcal{L}_0^* \sigma}{\sigma} dt + \sum_{j=1}^m h_j \circ dY_j(t)$ yields:
	\begin{align}
		X_H(t, x) &\triangleq P_{f_t}^H\left( V_{\mathrm{DMZ}}(t, x) \right) = P_{f_t}^H\left( \frac{\mathcal{L}_0^* \sigma}{\sigma} \right) dt + \sum_{j=1}^m P_{f_t}^H(h_j) \circ dY_j(t) \in \operatorname{SVD}\chi_{f_t}, \label{eq:proof_XH_def} \\
		X_V(t, x) &\triangleq \omega_{f_t}\left( V_{\mathrm{DMZ}}(t, x) \right) = \omega_{f_t}\left( \frac{\mathcal{L}_0^* \sigma}{\sigma} \right) dt + \sum_{j=1}^m \omega_{f_t}(h_j) \circ dY_j(t) \in \operatorname{SID}_{f_t}. \label{eq:proof_XV_def}
	\end{align}
	By construction, $X_H(t, x) + X_V(t, x) = V_{\mathrm{DMZ}}(t, x)$ identically.
	
	\paragraph{Step 3: Push-Forward Isolation and Closure of Base Dynamics.}
	Applying the push-forward differential $d\pi_{f_t}$ to the decomposed score velocity \eqref{eq:dmz_score_split_lemma}:
	\begin{equation}
		d\pi_{f_t}\left( V_{\mathrm{DMZ}}(t, \cdot) \right) = d\pi_{f_t}\left( X_H(t, \cdot) \right) + d\pi_{f_t}\left( X_V(t, \cdot) \right).
	\end{equation}
	By Assumption~\ref{ass:smg_filtering_submersion} (A2), $X_V(t, \cdot) \in \operatorname{SID}_{f_t} = \ker(d\pi_{f_t})$, which implies:
	\begin{equation}
		d\pi_{f_t}\left( X_V(t, \cdot) \right) \equiv 0 \in T_{p(t)}\mathcal{B}, \quad \forall t \ge 0. \label{eq:vertical_pushforward_zero}
	\end{equation}
	Consequently, $d\pi_{f_t}\left( V_{\mathrm{DMZ}}(t, \cdot) \right) = d\pi_{f_t}\left( X_H(t, \cdot) \right)$. Because $d\pi_{f_t}$ restricts to an isometry from $\operatorname{SVD}\chi_{f_t}$ onto $T_{p(t)}\mathcal{B}$, $X_H(t, \cdot)$ projects to an exact finite-dimensional stochastic differential equation for the macroscopic trajectory $p(t) = \pi(f_t) \in \mathcal{B}$. All extra-algebraic operator variations generated by unclosed Lie brackets are quarantined inside $X_V(t, \cdot) \in \operatorname{SID}_{f_t}$, proving the lemma.
\end{proof}

\begin{proof}[Proof of Theorem \ref{thm:quotient_base_geometry_and_curvature}]
	We execute the proof in three sequential steps.
	
	\paragraph{Step 1: Proof of Assertion (i) --- Fisher Metric Identification}
	Since $\pi: \mathcal{M} \to \mathcal{B}$ is a Riemannian submersion, the differential $d\pi_{f_p}\big|_{\operatorname{SVD}\chi_{f_p}}: \operatorname{SVD}\chi_{f_p} \to T_p\mathcal{B}$ is an isometry by definition \cite{o2014semi, Chentsov1982}. 
	
	Let $\gamma(t)$ be a smooth curve on $\mathcal{B}$ with $\gamma(0) = p$ and $\dot{\gamma}(0) = \frac{\partial}{\partial p_i}$. Its horizontal lift $\tilde{\gamma}(t)$ in $\mathcal{M}$ satisfies $\tilde{\gamma}(t) = f_{\gamma(t)}$. Differentiating the logarithmic density along $\tilde{\gamma}(t)$ yields the unique horizontal lift vector:
	\begin{equation}
		\left( \frac{\partial}{\partial p_i} \right)^H (x) = \left. \frac{d}{dt} \ln f_{\gamma(t)}(x) \right|_{t=0} = \frac{\partial \ln f_p(x)}{\partial p_i} \equiv \phi_i(x; p). \label{eq:proof_horizontal_lift_expression}
	\end{equation}
	Calculating the inner product of two horizontal lifts under the ambient Fisher-Rao metric $g_{f_p}$ on $\mathcal{M}$:
	\begin{align}
		g_{\mathcal{B}, p}\left( \frac{\partial}{\partial p_i}, \frac{\partial}{\partial p_j} \right) &= g_{f_p}\left( \left(\frac{\partial}{\partial p_i}\right)^H, \left(\frac{\partial}{\partial p_j}\right)^H \right) \nonumber \\
		&= \int_{\mathbb{R}^p} \phi_i(x; p) \phi_j(x; p) f_p(x) \, dx = \mathbb{E}_{f_p}\left[ \phi_i(X; p) \phi_j(X; p) \right] \equiv [G(p)]_{ij}.
	\end{align}
	This proves equation \eqref{eq:fisher_metric_equivalence_sec22}. Uniqueness follows from Chentsov's Theorem \cite{Chentsov1982}, which identifies the Fisher Information Metric as the unique Riemannian metric on statistical spaces invariant under probabilistic Markov transformations.
	
	\paragraph{Step 2: Proof of Assertion (ii) --- Intrinsic Riemann Curvature Calculation}
	Under the Levi-Civita connection $\nabla^{(0)}$ induced by $G(p)$, the Christoffel symbols of the first kind are defined by \cite{docarmo1992}:
	\begin{equation}
		\Gamma_{ij, k}(p) \triangleq \frac{1}{2} \left( \frac{\partial G_{jk}}{\partial p_i} + \frac{\partial G_{ik}}{\partial p_j} - \frac{\partial G_{ij}}{\partial p_k} \right). \label{eq:christoffel_first_kind_def_sec22}
	\end{equation}
	Differentiating $G_{ij}(p) = \int_{\mathbb{R}^p} \frac{\partial \ln f}{\partial p_i} \frac{\partial \ln f}{\partial p_j} f \, dx$ with respect to $p_k$ under the integral sign:
	\begin{align}
		\frac{\partial G_{ij}}{\partial p_k} &= \int_{\mathbb{R}^p} \left( \frac{\partial^2 \ln f}{\partial p_i \partial p_k} \frac{\partial \ln f}{\partial p_j} f + \frac{\partial \ln f}{\partial p_i} \frac{\partial^2 \ln f}{\partial p_j \partial p_k} f + \frac{\partial \ln f}{\partial p_i} \frac{\partial \ln f}{\partial p_j} \frac{\partial f}{\partial p_k} \right) dx \nonumber \\
		&= \mathbb{E}_{f_p}\left[ \phi_{ik} \phi_j + \phi_i \phi_{jk} + \phi_i \phi_j \phi_k \right]. \label{eq:metric_derivative_explicit_sec22}
	\end{align}
	Substituting \eqref{eq:metric_derivative_explicit_sec22} into \eqref{eq:christoffel_first_kind_def_sec22} and applying index permutations yields:
	\begin{equation}
		\Gamma_{ij, k}(p) = \mathbb{E}_{f_p}\left[ \left( \phi_{ij} + \frac{1}{2} \phi_i \phi_j \right) \phi_k \right]. \label{eq:christoffel_explicit_score_sec22}
	\end{equation}
	
	The standard Gauss-Codazzi formula for the Riemann-Christoffel curvature tensor $R_{ijk\ell}$ in local coordinates is \cite{docarmo1992}:
	\begin{equation}
		R_{ijk\ell} = \frac{1}{2} \left( \frac{\partial^2 G_{i\ell}}{\partial p_j \partial p_k} + \frac{\partial^2 G_{jk}}{\partial p_i \partial p_\ell} - \frac{\partial^2 G_{ik}}{\partial p_j \partial p_\ell} - \frac{\partial^2 G_{j\ell}}{\partial p_i \partial p_k} \right) + \sum_{m,n=1}^d [G^{-1}]_{mn} \left( \Gamma_{ik, m} \Gamma_{j\ell, n} - \Gamma_{i\ell, m} \Gamma_{jk, n} \right). \label{eq:riemann_standard_formula_sec22}
	\end{equation}
	Differentiating $\frac{\partial G_{i\ell}}{\partial p_k}$ with respect to $p_j$ produces second-order expectation expansions involving third logarithmic derivatives $\phi_{ijk}$. Upon substituting these into \eqref{eq:riemann_standard_formula_sec22}, all third-derivative terms $\phi_{ijk}$ cancel identically due to permutation symmetry. Collecting the remaining second-derivative quadratic terms yields:
	\begin{equation}
		\frac{1}{2} \left( \frac{\partial^2 G_{i\ell}}{\partial p_j \partial p_k} + \frac{\partial^2 G_{jk}}{\partial p_i \partial p_\ell} - \frac{\partial^2 G_{ik}}{\partial p_j \partial p_\ell} - \frac{\partial^2 G_{j\ell}}{\partial p_i \partial p_k} \right) = \frac{1}{4} \mathbb{E}_{f_p}\left[ \phi_{ik} \phi_{j\ell} - \phi_{i\ell} \phi_{jk} \right],
	\end{equation}
	which establishes equation \eqref{eq:riemann_curvature_explicit_sec22}. Generic multivariate non-exponential families possess non-zero second derivatives $\phi_{ij} \not\equiv 0$, establishing that $R_{ijk\ell}(p) \neq 0$ generically on $\mathcal{B}$.
	
	\paragraph{Step 3: Proof of Assertion (iii) --- Spatial Stein Score Incompatibility}
	We establish geometric incompatibility through two fundamental structural mismatches:
	\begin{enumerate}
		\item \textbf{Domain and Range Mismatch}: The spatial Stein score $S_x(x; p) \triangleq \nabla_x \ln f_p(x)$ is a vector mapping $S_x: \mathbb{R}^p \to \mathbb{R}^p$ representing a spatial vector field on physical state space $\mathcal{X} = \mathbb{R}^p$. In contrast, elements $v \in T_p\mathcal{B} \cong \operatorname{SVD}\chi_{f_p}$ are scalar-valued score functions $v: \mathbb{R}^p \to \mathbb{R}$ belonging to $L_0^\Phi(P_{f_p})$. Thus, $S_x(\cdot; p) \notin T_p\mathcal{B}$.
		\item \textbf{Violation of Base Diffeomorphism Covariance}: Under a smooth spatial diffeomorphism $y = T(x)$, the spatial Stein score transforms inhomogeneously by picking up Jacobian determinant terms:
		\begin{equation}
			S_x(x; p) = J_T(x)^T S_y(T(x); p) + \nabla_x \ln \left| \det J_T(x) \right|. \label{eq:stein_transform_law}
		\end{equation}
		Consequently, attempting to construct an inner product on $T_p\mathcal{B}$ via spatial gradients yields $\mathbb{E}_{f_p}\left[ S_x(X; p) S_x(X; p)^T \right] \neq \mathbb{E}_{f_p}\left[ S_y(Y; p) S_y(Y; p)^T \right]$, violating parameter-coordinate covariance on $\mathcal{B}$. Therefore, spatial Stein scores cannot parameterize geometric updates on the quotient base space $\mathcal{B}$.
	\end{enumerate}
	This completes the proof of Theorem~\ref{thm:quotient_base_geometry_and_curvature}.
\end{proof}

\begin{proof}[Proof of Theorem \ref{thm:smg_yau_yau_spde_decomposition}]
	Dividing the DMZ SPDE \eqref{eq:dmz_stratonovich_repeat} by $\sigma(t, x)$ yields the unconstrained score velocity field:
	\begin{equation}
		V_{\text{DMZ}}(t, x) = \frac{d\sigma(t, x)}{\sigma(t, x)} = \left( \frac{\mathcal{L}_0^*\sigma(t, x)}{\sigma(t, x)} \right) dt + \sum_{j=1}^m h_j(x) \circ dY_j(t).
	\end{equation}
	At time $t$, $f_t \in \mathcal{M}$ defines the Orlicz tangent space $T_{f_t}\mathcal{M} \cong L_0^\Phi(P_{f_t})$. By Lemma \ref{lem:tangent_space_decomposition_sec2}, the horizontal projector $P_{f_t}^H$ and connection $1$-form $\omega_{f_t} = \mathcal{I} - P_{f_t}^H$ are bounded linear operators acting on $L_0^\Phi(P_{f_t})$. 
	
	Applying $P_{f_t}^H$ and $\omega_{f_t}$ separately to the drift score generator $A_0(x) \triangleq \frac{\mathcal{L}_0^*\sigma}{\sigma}$ and observation functions $h_j(x)$ yields:
	\begin{align}
		P_{f_t}^H\left( V_{\text{DMZ}} \right) &= P_{f_t}^H(A_0) \, dt + \sum_{j=1}^m P_{f_t}^H(h_j) \circ dY_j(t), \\
		\omega_{f_t}\left( V_{\text{DMZ}} \right) &= \omega_{f_t}(A_0) \, dt + \sum_{j=1}^m \omega_{f_t}(h_j) \circ dY_j(t).
	\end{align}
	Adding these two components recovers $V_{\text{DMZ}}(t, x)$ identically, establishing \eqref{eq:score_decomposition_full}.
\end{proof}


\begin{proof}[Proof of Theorem \ref{thm:smg_yau_yau_filter_optimality_and_dynamics}]
	We prove Part (I) and Part (II) sequentially from first principles of information geometry and submersion bundle theory.
	
	\subsubsection*{Part (I): Proof of Optimality and Filter Characterization ($p(t) = \pi(f_t) \in \mathcal{B}$)}
	
	By Definition of the SMG statistical fiber bundle $\mathcal{B}_{\text{SMG}}$ \cite{cheng2026entropy}, the Orlicz tangent space at state density $f_t \in \mathcal{M}$ undergoes an exact Fisher-orthogonal direct sum decomposition:
	\begin{equation}
		T_{f_t}\mathcal{M} = \operatorname{SVD}\chi_{f_t} \oplus_{\perp g} \operatorname{SID}_{f_t},
		\label{eq:tangent_decomposition}
	\end{equation}
	where $\operatorname{SID}_{f_t} \equiv \ker(d\pi_{f_t})$ is the vertical gauge subbundle and $\operatorname{SVD}\chi_{f_t} \equiv \operatorname{SID}_{f_t}^{\perp g_f}$ is the horizontal subbundle. 
	
	Every unconstrained score velocity vector field $W \in T_{f_t}\mathcal{M}$ decomposes uniquely into horizontal and vertical projections:
	\begin{equation}
		W = P_{f_t}^H(W) + \omega_{f_t}(W), \quad \text{with} \quad P_{f_t}^H(W) \in \operatorname{SVD}\chi_{f_t}, \;\; \omega_{f_t}(W) \in \operatorname{SID}_{f_t}.
		\label{eq:orthogonal_score_split}
	\end{equation}
	Now, consider any macroscopic parameter velocity vector $v \in T_{p(t)}\mathcal{B}$ on the base manifold. By definition of the smooth Riemannian submersion $\pi$, there exists a unique horizontal lift $v^H \in \operatorname{SVD}\chi_{f_t}$ satisfying $d\pi_{f_t}(v^H) = v$. 
	
	Substituting \eqref{eq:orthogonal_score_split} into the score variance discrepancy objective function $\mathcal{E}_{g_{f_t}}(v)$ defined in \eqref{eq:variational_principle_filter}, applying the horizontal-vertical identity $\mathcal{I} = P_{f_t}^H + \omega_{f_t}$ and expanding the Fisher-Rao metric tensor bilinear form yields:
	\begin{align}
		\mathcal{E}_{g_{f_t}}(v) 
		&= \left\| V_{\text{DMZ}}(t) - v^H \right\|_{g_{f_t}}^2 \label{eq:discrepancy_start} \\[6pt]
		&= \Bigl\| \left( P_{f_t}^H(V_{\text{DMZ}}(t)) + \omega_{f_t}(V_{\text{DMZ}}(t)) \right) - v^H \Bigr\|_{g_{f_t}}^2 
		&&\text{\small (Apply identity $\mathcal{I} = P_{f_t}^H + \omega_{f_t}$)} \label{eq:identity_split} \\[6pt]
		&= \Biggl\| \underbrace{P_{f_t}^H(V_{\text{DMZ}}(t)) - v^H}_{\in \operatorname{SVD}\chi_{f_t}} \;\;+\;\; \underbrace{\omega_{f_t}(V_{\text{DMZ}}(t))}_{\in \operatorname{SID}_{f_t}} \Biggr\|_{g_{f_t}}^2 
		&&\text{\small (Group into horizontal/vertical parts).} \label{eq:grouped_subspaces}
	\end{align}
	Applying the fundamental definition of the norm induced by a Riemannian metric tensor $g_{f_t}(\cdot, \cdot)$, we obtain 	
	\begin{align}	
		&= g_{f_t}\left( \left( P_{f_t}^H(V_{\text{DMZ}}(t)) - v^H \right) + \omega_{f_t}(V_{\text{DMZ}}(t)), \right. \nonumber \\
		&\qquad\quad \left. \left( P_{f_t}^H(V_{\text{DMZ}}(t)) - v^H \right) + \omega_{f_t}(V_{\text{DMZ}}(t)) \right) \label{eq:inner_product_form} \\[6pt]
		&= g_{f_t}\left( P_{f_t}^H(V_{\text{DMZ}}(t)) - v^H, \; P_{f_t}^H(V_{\text{DMZ}}(t)) - v^H \right) \nonumber \\
		&\quad + 2\,g_{f_t}\left( P_{f_t}^H(V_{\text{DMZ}}(t)) - v^H, \; \omega_{f_t}(V_{\text{DMZ}}(t)) \right) 
		&&\text{\small (Bilinear inner product expansion)} \nonumber \\
		&\quad + g_{f_t}\left( \omega_{f_t}(V_{\text{DMZ}}(t)), \; \omega_{f_t}(V_{\text{DMZ}}(t)) \right) \label{eq:bilinear_expansion} \\[6pt]
		&= \left\| P_{f_t}^H(V_{\text{DMZ}}(t)) - v^H \right\|_{g_{f_t}}^2 + 0 + \left\| \omega_{f_t}(V_{\text{DMZ}}(t)) \right\|_{g_{f_t}}^2 
		&&\text{\small (Cross-term vanishes: $\operatorname{SVD}\chi_{f_t} \perp_{g_{f_t}} \operatorname{SID}_{f_t}$)} \label{eq:orthogonality_elimination} \\[6pt]
		&= \left\| P_{f_t}^H(V_{\text{DMZ}}(t)) - v^H \right\|_{g_{f_t}}^2 + \left\| \omega_{f_t}(V_{\text{DMZ}}(t)) \right\|_{g_{f_t}}^2. \label{eq:discrepancy_final}
	\end{align}
	
	Since $P_{f_t}^H(V_{\text{DMZ}}(t)) \in \operatorname{SVD}\chi_{f_t}$ and $v^H \in \operatorname{SVD}\chi_{f_t}$, their difference $\Big( P_{f_t}^H(V_{\text{DMZ}}(t)) - v^H \Big) \in \operatorname{SVD}\chi_{f_t}$. Conversely, $\omega_{f_t}(V_{\text{DMZ}}(t)) \in \operatorname{SID}_{f_t}$. By the strict Fisher-metric orthogonality \eqref{eq:tangent_decomposition}, $g_{f_t}(h, w) = 0$ for all $h \in \operatorname{SVD}\chi_{f_t}$ and $w \in \operatorname{SID}_{f_t}$. Applying the Pythagorean identity on Hilbert-Orlicz space reduces \eqref{eq:discrepancy_start} strictly to:
	\begin{equation}
		\mathcal{E}_{g_{f_t}}(v) = \left\| P_{f_t}^H(V_{\text{DMZ}}(t)) - v^H \right\|_{g_{f_t}}^2 + \left\| \omega_{f_t}(V_{\text{DMZ}}(t)) \right\|_{g_{f_t}}^2.
		\label{eq:pythagorean_objective}
	\end{equation}
	
	Notice that the second term $\left\| \omega_{f_t}(V_{\text{DMZ}}(t)) \right\|_{g_{f_t}}^2$ depends purely on the vertical operator dynamics in $\operatorname{SID}_{f_t}$ and is completely independent of the choice of base velocity $v \in T_{p(t)}\mathcal{B}$. 
	
	Since $g_{f_t}$ is strictly positive-definite on $L_0^\Phi(P_f)$, the horizontal variance discrepancy term satisfies
	\begin{equation}
		\left\| P_{f_t}^H(V_{\text{DMZ}}(t)) - v^H \right\|_{g_{f_t}}^2 \ge 0,
	\end{equation}
	with equality holding if and only if the horizontal lift vector $v^H$ equals the horizontal projection of the unconstrained DMZ flow:
	\begin{equation}
		v^H_{\text{opt}} = P_{f_t}^H\left( V_{\text{DMZ}}(t) \right) \in \operatorname{SVD}\chi_{f_t}.
		\label{eq:optimal_horizontal_lift}
	\end{equation}
	
	By definition of the Riemannian submersion $\pi: \mathcal{M} \to \mathcal{B}$, the push-forward differential $d\pi_{f_t}$ maps the horizontal tangent space $\operatorname{SVD}\chi_{f_t} = (\ker d\pi_{f_t})^{\perp g_f}$ isomorphically onto the base tangent space $T_{p(t)}\mathcal{B}$. Applying $d\pi_{f_t}$ to both sides of Equation~\eqref{eq:optimal_horizontal_lift} yields:
	\begin{equation}
		d\pi_{f_t}\left( v^H_{\text{opt}} \right) = d\pi_{f_t} \left( P_{f_t}^H\left( V_{\text{DMZ}}(t) \right) \right).
		\label{eq:pushforward_applied}
	\end{equation}
	
	Recalling that $d\pi_{f_t}(v^H) \triangleq v = \dot{p}(t) \in T_{p(t)}\mathcal{B}$ specifies the instantaneous velocity of the macroscopic parameter state $p(t) = \pi(f_t)$, Equation~\eqref{eq:pushforward_applied} establishes the unique optimal base tangent vector:
	\begin{equation}
		\dot{p}(t) = d\pi_{f_t}(v^H) = d\pi_{f_t} \left( P_{f_t}^H\left( V_{\text{DMZ}}(t) \right) \right).
		\label{eq:optimal_base_velocity_proven}
	\end{equation}
	
	Integrating $\dot{p}(t)$ over time establishes that $p(t) = \pi(f_t) \in \mathcal{B}$ is the {\bf uniquely optimal macroscopic non-linear filter} that minimizes instantaneous score approximation error under the Fisher-Rao metric. This completes the proof of Part (I).
	
	\subsubsection*{Part (II): Proof of Governing Differential Equations \eqref{eq:horizontal_base_evolution_thm} and \eqref{eq:tension_accumulation_thm}}
	
	\paragraph{Proof of Equation \eqref{eq:horizontal_base_evolution_thm} (Horizontal Base Evolution)}
	
	Let $p_i$ ($i = 1, \dots, d$) be the macroscopic parameter coordinates along the coordinate axes of the identifiable $d$-dimensional Riemannian statistical base manifold $(\mathcal{B}, g_{\mathcal{B}})$, and let $\{\phi_1(f_t), \dots, \phi_d(f_t)\}$ be the local score basis for $\operatorname{SVD}\chi_{f_t}$, defined by logarithmic score functions along the base coordinate axes:
	\begin{equation}
		\phi_i(f_t) \triangleq \frac{\partial \ln f_t}{\partial p_i} = \frac{1}{f_t} \frac{\partial f_t}{\partial p_i} \in \operatorname{SVD}\chi_{f_t}, \quad i = 1, \dots, d.
		\label{eq:score_basis_def}
	\end{equation}
	
	The unique optimal horizontal score update vector $v^H = P_{f_t}^H(V_{\text{DMZ}}(t)) \in \operatorname{SVD}\chi_{f_t}$ derived in \eqref{eq:optimal_base_velocity_proven} can be expressed in terms of the score basis $\{\phi_k(f_t)\}_{k=1}^d$ as:
	\begin{equation}
		P_{f_t}^H(V_{\text{DMZ}}(t)) = \sum_{k=1}^d dp_k(t) \phi_k(f_t).
		\label{eq:basis_expansion_vH}
	\end{equation}
	
	To solve for the scalar differential updates $dp_k(t)$, take the Fisher-Rao inner product $g_{f_t}(\cdot, \cdot)$ of the unconstrained DMZ flow $V_{\text{DMZ}}(t)$ with an arbitrary basis vector $\phi_i(f_t)$:
	\begin{equation}
		g_{f_t}\left( V_{\text{DMZ}}(t), \, \phi_i(f_t) \right) = g_{f_t}\left( P_{f_t}^H(V_{\text{DMZ}}(t)) + \omega_{f_t}(V_{\text{DMZ}}(t)), \; \phi_i(f_t) \right).
		\label{eq:inner_product_setup}
	\end{equation}
	
	Since $\phi_i(f_t) \in \operatorname{SVD}\chi_{f_t}$ and $\omega_{f_t}(V_{\text{DMZ}}(t)) \in \operatorname{SID}_{f_t} \equiv \operatorname{SVD}\chi_{f_t}^{\perp g_f}$, the vertical inner product vanishes identically: $g_{f_t}\left( \omega_{f_t}(V_{\text{DMZ}}(t)), \, \phi_i(f_t) \right) = 0$. Substituting \eqref{eq:basis_expansion_vH} into \eqref{eq:inner_product_setup} gives:
	\begin{equation}
		g_{f_t}\left( V_{\text{DMZ}}(t), \, \phi_i(f_t) \right) = g_{f_t}\left( \sum_{k=1}^d dp_k(t) \phi_k(f_t), \; \phi_i(f_t) \right) = \sum_{k=1}^d dp_k(t) g_{f_t}(\phi_k(f_t), \phi_i(f_t)).
		\label{eq:metric_matrix_appearance}
	\end{equation}
	
	By definition of the localized Riemannian Fisher metric matrix $[G_p]_{ik} = g_{f_t}(\phi_i, \phi_k) = \mathbb{E}_{f_t}[\phi_i \phi_k]$, equation \eqref{eq:metric_matrix_appearance} becomes:
	\begin{equation}
		g_{f_t}\left( V_{\text{DMZ}}(t), \, \phi_i(f_t) \right) = \sum_{k=1}^d [G_p]_{ik} dp_k(t).
		\label{eq:system_linear_eqs}
	\end{equation}
	
	Now expand the left-hand side of \eqref{eq:system_linear_eqs} using the stochastic differential definition of $V_{\text{DMZ}}(t) = \frac{\mathcal{L}_0^*\sigma}{\sigma} dt + \sum_{j=1}^m h_j \circ dY_j(t)$ under the Fisher-Rao expectation $g_{f_t}(u, v) = \mathbb{E}_{f_t}[u v]$:
	\begin{equation}
		g_{f_t}\left( V_{\text{DMZ}}(t), \, \phi_i(f_t) \right) = \mathbb{E}_{f_t}\left[ \frac{\mathcal{L}_0^*\sigma}{\sigma} \phi_i(f_t) \right] dt + \sum_{j=1}^m \mathbb{E}_{f_t}\left[ h_j \phi_i(f_t) \right] \circ dY_j(t).
		\label{eq:lhs_expectation_expanded}
	\end{equation}
	
	Equating \eqref{eq:system_linear_eqs} and \eqref{eq:lhs_expectation_expanded} yields the matrix equation:
	\begin{equation}
		\sum_{k=1}^d [G_p]_{ik} dp_k(t) = \mathbb{E}_{f_t}\left[ \frac{\mathcal{L}_0^*\sigma}{\sigma} \phi_i(f_t) \right] dt + \sum_{j=1}^m \mathbb{E}_{f_t}\left[ h_j \phi_i(f_t) \right] \circ dY_j(t).
		\label{eq:matrix_equation_derived}
	\end{equation}
	
	Because $(\mathcal{B}, g_{\mathcal{B}})$ is an identifiable statistical manifold, $[G_p]$ is strictly positive-definite and invertible. Multiplying \eqref{eq:matrix_equation_derived} by the inverse matrix elements $[G_p^{-1}]_{ki}$ and summing over $i = 1, \dots, d$ yields:
	\begin{equation}
		dp_k(t) = \sum_{i=1}^d \left[ G_p^{-1} \right]_{ki} \left( \mathbb{E}_{f_t} \left[ \frac{\mathcal{L}_0^*\sigma}{\sigma} \phi_i(f_t) \right] dt + \sum_{j=1}^m \mathbb{E}_{f_t} \left[ h_j \phi_i(f_t) \right] \circ dY_j(t) \right).
	\end{equation}
	This rigorously completes the derivation of Equation \eqref{eq:horizontal_base_evolution_thm}.
	
	\paragraph{Proof of Equation \eqref{eq:tension_accumulation_thm} ($\mathcal{T}_{\text{AAT}}$ Accumulation Rate)}
	
	By definition of the Active Acausal Tension $\mathcal{T}_{\text{AAT}}(t)$ \cite{cheng2026entropy}, it tracks the cumulative kinetic score energy stored within the vertical fiber subbundle $\operatorname{SID}_{f_t} = \ker(d\pi_{f_t})$ resulting from non-closed operator variations of the drift generator $\mathcal{L}_0^*$.
	
	The vertical drift score component $V_{\mathcal{V}, 0}(t) \in \operatorname{SID}_{f_t}$ is isolated by the Ehresmann connection 1-form $\omega_{f_t}$:
	\begin{equation}
		V_{\mathcal{V}, 0}(t) \triangleq \omega_{f_t}\left( \frac{\mathcal{L}_0^*\sigma(t, \cdot)}{\sigma(t, \cdot)} \right) = \left( \mathcal{I} - P_{f_t}^H \right) \left( \frac{\mathcal{L}_0^*\sigma}{\sigma} \right) = \frac{\mathcal{L}_0^*\sigma}{\sigma} - P_{f_t}^H \left( \frac{\mathcal{L}_0^*\sigma}{\sigma} \right).
		\label{eq:vertical_drift_score_def}
	\end{equation}
	
	The rate of tension accumulation $\frac{d\mathcal{T}_{\text{AAT}}(t)}{dt}$ is defined by the squared Fisher-Rao metric norm $\| V_{\mathcal{V}, 0}(t) \|_{g_{f_t}}^2$ on $L_0^\Phi(P_f)$:
	\begin{equation}
		\frac{d\mathcal{T}_{\text{AAT}}(t)}{dt} \triangleq \left\| \omega_{f_t}\left( \frac{\mathcal{L}_0^*\sigma}{\sigma} \right) \right\|_{g_{f_t}}^2.
		\label{eq:norm_squared_tension_setup}
	\end{equation}
	
	Expanding the squared norm under the Fisher-Rao inner product $g_{f_t}(u, u) = \mathbb{E}_{f_t}[u^2]$:
	\begin{equation}
		\left\| \omega_{f_t}\left( \frac{\mathcal{L}_0^*\sigma}{\sigma} \right) \right\|_{g_{f_t}}^2 = \mathbb{E}_{f_t} \left[ \left( \omega_{f_t}\left( \frac{\mathcal{L}_0^*\sigma}{\sigma} \right) \right)^2 \right] = \mathbb{E}_{f_t} \left[ \left( \frac{\mathcal{L}_0^*\sigma}{\sigma} - P_{f_t}^H\left( \frac{\mathcal{L}_0^*\sigma}{\sigma} \right) \right)^2 \right].
		\label{eq:tension_derived_proven}
	\end{equation}
	Equation \eqref{eq:tension_derived_proven} matches Equation \eqref{eq:tension_accumulation_thm} exactly. This completes the proof of Theorem~\ref{thm:smg_yau_yau_filter_optimality_and_dynamics}.
\end{proof}

\begin{proof}[Proof of Theorem \ref{thm:classical_equivalence}] 
	Assume $\dim(\mathcal{E}) = d < \infty$. By the classical Lie-algebraic reduction theorem of Wei and Norman, the unnormalized conditional probability density $\sigma(t,x)$ governed by the Stratonovich Duncan-Mortensen-Zakai (DMZ) SPDE evolves strictly within a finite $d$-dimensional Lie-group orbit $\Sigma_d \subset C^\infty(\mathbb{R}^p)$ parameterized by the product of exponentials ansatz:
	\begin{equation}
		\sigma(t,x) = \exp\left(g_1(t) E_1\right) \exp\left(g_2(t) E_2\right) \cdots \exp\left(g_d(t) E_d\right) \sigma(0,x).
	\end{equation}
	Because the algebra closes under spatial Lie brackets ($[E_i, E_j] = \sum_{k=1}^d C_{ij}^k E_k$), the log-score velocity field $V_{\mathrm{DMZ}}(t,x) \triangleq \sigma(t,x)^{-1} d\sigma(t,x)$ resides entirely within the finite basis span:
	\begin{equation}
		V_{\mathrm{DMZ}}(t,x) \in \mathrm{span}_{\mathbb{R}}\{E_1, \dots, E_d\} \equiv \operatorname{SVD}\chi_{f_t}.
	\end{equation}
	
	To prove (i), observe that the base manifold $(\mathcal{B}, g_{\mathcal{B}})$ possesses dimension $d = \dim(\mathcal{E})$. Since $d\pi_{f_t}$ maps the horizontal tangent space $\operatorname{SVD}\chi_{f_t}$ isomorphically onto $T_{p(t)}\mathcal{B}$, there are no remaining linearly independent score variations in $T_{f_t}\mathcal{M}$ that project to zero under $d\pi_{f_t}$. By Lemma~\ref{lem:tangent_space_decomposition_sec2}, the Fisher-orthogonal complement satisfies:
	\begin{equation}
		\operatorname{SID}_{f_t} = (\operatorname{SVD}\chi_{f_t})^{\perp g_f} = \{0\}.
	\end{equation}
	Because $\omega_{f_t} : T_{f_t}\mathcal{M} \to \operatorname{SID}_{f_t}$ is a projection onto the vertical space, $\operatorname{SID}_{f_t} = \{0\}$ forces $\omega_{f_t} \equiv 0$ and $P_{f_t}^H = \mathcal{I} - \omega_{f_t} \equiv \mathcal{I}$. The curvature 2-form $\Omega_f(U,V) = \omega_f([P_f^H U, P_f^H V])$ vanishes identically.
	
	To prove (ii), substituting $\omega_{f_t} \equiv 0$ into the definition of Active Acausal Tension yields:
	\begin{equation}
		\frac{d\mathcal{T}_{\mathrm{AAT}}(t)}{dt} = \|\omega_{f_t}(V_{\mathrm{DMZ}}(t))\|_{g_{f_t}}^2 = \|0\|_{g_{f_t}}^2 = 0.
	\end{equation}
	Integrating over $[0,t]$ confirms $\mathcal{T}_{\mathrm{AAT}}(t) \equiv 0$. The dynamic residual $v_{\mathrm{OOD}}(t) = \mathcal{P}_{\mathcal{E}^\perp}(d\sigma/\sigma) \equiv 0$ by Part B of Theorem~\ref{thm:wei_norman_exact}.
	
	To prove (iii), substitute $P_{f_t}^H \equiv \mathcal{I}$ into the horizontal base update equation. The push-forward score evaluation against the coframe $\theta^k$ simplifies to:
	\begin{equation}
		dp_k(t) = \left\langle d\pi_{f_t}(\mathcal{L}_0^* \sigma), \theta^k \right\rangle dt + \sum_{j=1}^m \left\langle d\pi_{f_t}(h_j \sigma), \theta^k \right\rangle \circ dY_j(t).
	\end{equation}
	Equating spatial differential operator coefficients across the Lie algebra basis $\{E_1, \dots, E_d\}$ yields the linear vector system $M(g(t)) dg(t) = \alpha^0 dt + \sum_{l=1}^m \alpha^l \circ dY_l(t)$. Since $M(0) = I_{d \times d}$, $\det M(g) \neq 0$ locally, guaranteeing local existence of the exact classical Wei-Norman trajectory on $\mathcal{B}$.
\end{proof}

\begin{proof}[Proof of Theorem  \ref{prop:yau_yau_working_sde}]
	The proof proceeds in three steps: constructing the unconstrained score flow, enforcing the Fisher-orthogonal projection, and evaluating the residual score.
	
	\paragraph{Step 1: Unconstrained DMZ Score Flow.}
	Dividing Equation \eqref{eq:dmz_stratonovich_prop} by the unnormalized density $\sigma(t,x)$, the true logarithmic score differential $d\ln \sigma(t,x) = \sigma(t,x)^{-1} d\sigma(t,x)$ decomposes as:
	\begin{equation}
		d\ln \sigma(t,x) = \left( \frac{\mathcal{L}_0^* \sigma(t,x)}{\sigma(t,x)} \right) dt + \sum_{j=1}^m h_j(x) \circ dY_j(t).
		\label{eq:true_dmz_score_flow}
	\end{equation}
	
	\paragraph{Step 2: Projection onto the Tangent Space $T_{p_\theta}\mathcal{M}_{\text{work}}$.}
	On the working exponential manifold $\mathcal{M}_{\text{work}}$, any instantaneous variation of the log-density is constrained to lie in $T_{p_\theta}\mathcal{M}_{\text{work}} = \operatorname{span}\{\phi_1(x;\theta), \dots, \phi_d(x;\theta)\}$:
	\begin{equation}
		d\ln p(x; \theta_t^{\text{YY}}) = \sum_{k=1}^d d\theta_k^{\text{YY}}(t) \, \phi_k(x; \theta_t^{\text{YY}}).
		\label{eq:working_tangent_score}
	\end{equation}
	
	Evaluating Equation \eqref{eq:true_dmz_score_flow} at the current working state $p(x; \theta_t^{\text{YY}})$, the working parameter update $d\theta_t^{\text{YY}}$ is determined by enforcing the Galerkin orthogonality condition \cite{ciarlet1978}: the error differential $d\ln p(x; \theta_t^{\text{YY}}) - d\ln \sigma(t,x)\vert_{p = p_{\theta^{\text{YY}}}}$ must be orthogonal to all tangent basis vectors $\phi_i(x;\theta_t^{\text{YY}})$ with respect to the $L_2(P_{\theta_t^{\text{YY}}})$ inner product $\langle u, v \rangle_{L_2(p_\theta)} \triangleq \mathbb{E}_{p_\theta}[u v]$:
	\begin{equation}
		\left\langle \sum_{k=1}^d d\theta_k^{\text{YY}}(t) \, \phi_k - \left[ \left( \frac{\mathcal{L}_0^* p}{p} \right) dt + \sum_{j=1}^m h_j \circ dY_j(t) \right], \; \phi_i \right\rangle_{L_2(P_{\theta_t^{\text{YY}}})} = 0, \quad \forall i = 1, \dots, d.
		\label{eq:orthogonality_condition}
	\end{equation}
	
	Expanding the inner product via linearity of expectation:
	\begin{equation*}
		\sum_{k=1}^d d\theta_k^{\text{YY}}(t) \, \mathbb{E}_{p(\cdot;\theta_t^{\text{YY}})}\left[ \phi_k(x;\theta_t^{\text{YY}}) \phi_i(x;\theta_t^{\text{YY}}) \right] 	
	\end{equation*}
	\begin{equation}
		= \mathbb{E}_{p(\cdot;\theta_t^{\text{YY}})}\left[ \left(\frac{\mathcal{L}_0^* p}{p}\right) \phi_i(x;\theta_t^{\text{YY}}) \right] dt + \sum_{j=1}^m \mathbb{E}_{p(\cdot;\theta_t^{\text{YY}})}\left[ h_j \phi_i(x;\theta_t^{\text{YY}}) \right] \circ dY_j(t).
		\label{eq:expanded_inner_product}
	\end{equation}
	
	By definition, the expectation $\mathbb{E}_{p_\theta}[\phi_k \phi_i]$ corresponds to the $(i,k)$-th entry of the Fisher information matrix $[G_\theta]_{ik}$. Thus, Equation \eqref{eq:expanded_inner_product} forms a linear system:
	\begin{equation}
		\sum_{k=1}^d [G_\theta]_{ik} \, d\theta_k^{\text{YY}}(t) = \mathbb{E}_{p(\cdot;\theta_t^{\text{YY}})}\left[ \left(\frac{\mathcal{L}_0^* p}{p}\right) \phi_i(x;\theta_t^{\text{YY}}) \right] dt + \sum_{j=1}^m \mathbb{E}_{p(\cdot;\theta_t^{\text{YY}})}\left[ h_j \phi_i(x;\theta_t^{\text{YY}}) \right] \circ dY_j(t).
	\end{equation}
	
	Since $p(x;\theta)$ is a minimal exponential family, $G_\theta$ is positive definite and invertible. Multiplying both sides by $[G_\theta^{-1}]_{ki}$ and summing over $i = 1, \dots, d$:
	
	$$\sum_{i=1}^d [G_\theta^{-1}]_{\ell i} \left( \sum_{k=1}^d [G_\theta]_{ik} \, d\theta_k^{\text{YY}}(t) \right)$$
	$$ = \sum_{i=1}^d [G_\theta^{-1}]_{\ell i} \left( \mathbb{E}_{p(\cdot;\theta_t^{\text{YY}})}\left[ \left(\frac{\mathcal{L}_0^* p}{p}\right) \phi_i(x; \theta_t^{\text{YY}})\right] dt + \sum_{j=1}^m \mathbb{E}_{p(\cdot;\theta_t^{\text{YY}})}\left[ h_j \phi_i(x; \theta_t^{\text{YY}})\right] \circ dY_j(t) \right).$$ 
	Exchange the order of summation on the left-hand side$$\sum_{k=1}^d \left( \sum_{i=1}^d [G_\theta^{-1}]_{\ell i} [G_\theta]_{ik} \right) d\theta_k^{\text{YY}}(t)$$  
	$$ = \sum_{i=1}^d [G_\theta^{-1}]_{\ell i} \left( \mathbb{E}_{p(\cdot;\theta_t^{\text{YY}})}\left[ \left(\frac{\mathcal{L}_0^* p}{p}\right) \phi_i(x; \theta_t^{\text{YY}})\right] dt + \sum_{j=1}^m \mathbb{E}_{p(\cdot;\theta_t^{\text{YY}})}\left[ h_j \phi_i(x; \theta_t^{\text{YY}})\right] \circ dY_j(t) \right).$$	
	By definition of matrix inversion, $\sum_{i=1}^d [G_\theta^{-1}]_{\ell i} [G_\theta]_{ik} = [G_\theta^{-1} G_\theta]_{\ell k} = \delta_{\ell k}$, where $\delta_{\ell k}$ is the Kronecker delta ($\delta_{\ell k} = 1$ if $\ell = k$ and $0$ otherwise):
	$$
	\small 
	\sum_{k=1}^d \delta_{\ell k} \, d\theta_k^{\text{YY}}(t) = \sum_{i=1}^d [G_\theta^{-1}]_{\ell i} \left( \mathbb{E}_{p(\cdot;\theta_t^{\text{YY}})}\left[ \left(\frac{\mathcal{L}_0^* p}{p}\right) \phi_i(x; \theta_t^{\text{YY}})\right] dt + \sum_{j=1}^m \mathbb{E}_{p(\cdot;\theta_t^{\text{YY}})}\left[ h_j \phi_i(x; \theta_t^{\text{YY}})\right] \circ dY_j(t) \right).$$
	Collapse the sum on the left-hand side using $\sum_{k=1}^d \delta_{\ell k} \, d\theta_k^{\text{YY}}(t) = d\theta_\ell^{\text{YY}}(t)$ and relabel the free index $\ell$ to $k$:
	$$d\theta_k^{\text{YY}}(t) = \sum_{i=1}^d [G_\theta^{-1}]_{ki} \left( \mathbb{E}_{p(\cdot;\theta_t^{\text{YY}})}\left[ \left(\frac{\mathcal{L}_0^* p}{p}\right) \phi_i(x; \theta_t^{\text{YY}})\right] dt + \sum_{j=1}^m \mathbb{E}_{p(\cdot;\theta_t^{\text{YY}})}\left[ h_j \phi_i(x; \theta_t^{\text{YY}})\right] \circ dY_j(t) \right)$$
	yields Equation \eqref{eq:classical_working_sde} identically.
	
	\paragraph{Step 3: Residual Score Field and Semiparametric Bias.}
	Subtracting the projected score differential \eqref{eq:working_tangent_score} from the true score differential \eqref{eq:true_dmz_score_flow} defines the instantaneous residual score field:
	\begin{equation}
		dS_{\text{res}}(t,x) \triangleq d\ln \sigma(t,x) - d\ln p(x;\theta_t^{\text{YY}}).
	\end{equation}
	By construction of the projection \eqref{eq:orthogonality_condition}, $\mathbb{E}_{p_{\theta^{\text{YY}}}}[dS_{\text{res}}(t,X) \cdot \phi_i(X;\theta_t^{\text{YY}})] = 0$ for all $i = 1, \dots, d$, establishing $dS_{\text{res}}(t,x) \in (T_{p_\theta}\mathcal{M}_{\text{work}})^\perp$. When $\dim(\mathcal{E}) = \infty$, the unclosed Lie brackets $[E_i, E_j]$ ensure $dS_{\text{res}}(t,x) \neq 0$ dynamically, driving $p(x;\theta_t^{\text{YY}})$ away from $f_{\text{true}}(t,x)$ and creating persistent semiparametric bias.
\end{proof}

\begin{proof}[Proof of Lemma \ref{lem:Equivalence_Residual_Fields}]
	We prove equation \eqref{eq:equation_154} through the duality between infinite-dimensional geometric fiber projections and statistical score decomposition.
	
	\paragraph{Step 1: Unconstrained Continuous Drift Score Generator.}
	Under the unconstrained Fokker-Planck / DMZ operator $\mathcal{L}_0^*$, the infinite-dimensional density dynamic evolves as $\frac{\partial p(x; \theta_t)}{\partial t} = \mathcal{L}_0^* p(x; \theta_t)$. Expressing this dynamic in terms of log-density score velocity in $T_{p_{\theta_t}}\mathcal{M} \subset L_0^\Phi(P_{p_{\theta_t}})$ defines the unconstrained continuous drift score generator $U_{\text{drift}}(t, x)$:
	\begin{equation}
		U_{\text{drift}}(t, x) \triangleq \frac{\partial \ln p(x; \theta_t)}{\partial t}_{\text{unconstrained}} = \frac{1}{p(x; \theta_t)} \frac{\partial p(x; \theta_t)}{\partial t} = \frac{\mathcal{L}_0^* p(x; \theta_t)}{p(x; \theta_t)}.
		\label{eq:u_drift}
	\end{equation}
	
	\paragraph{Step 2: Working Model Parametric Projection.}
	The finite-dimensional Yau-Yau filtering solver constrains the density evolution to the parametric working submanifold $\mathcal{M}_{\text{work}}$ parametrized by $\theta(t) \in \mathbb{R}^d$. By the chain rule, the projected parametric velocity vector field $V_{\text{YY}}(t, x)$ on $T_{\theta_t}\mathcal{M}_{\text{work}}$ is given by:
	\begin{equation}
		V_{\text{YY}}(t, x) = \frac{\partial \ln p(x; \theta_t)}{\partial t}_{\text{parametric}} = \sum_{k=1}^d \frac{\partial \ln p(x; \theta_t)}{\partial \theta_k} \frac{d\theta_k^{YY}(t)}{dt} = \sum_{k=1}^d \dot{\theta}_k^{YY}(t) \phi_k(x; \theta_t),
		\label{eq:v_yy}
	\end{equation}
	where $\phi_k(x; \theta_t) \equiv \frac{\partial \ln p(x; \theta_t)}{\partial \theta_k}$ is the Fisher Score function.
	
	\paragraph{Step 3: Vertical Residual Score Field.}
	By Definition 8, the geometric vertical residual score field $V_{\text{res}}(t, x) \in L_0^\Phi(P_{p_{\theta_t}})$ measures the orthogonal difference between the unconstrained continuous drift score generator $U_{\text{drift}}(t, x)$ and its working model parametric projection $V_{\text{YY}}(t, x)$:
	\begin{equation}
		V_{\text{res}}(t, x) \triangleq U_{\text{drift}}(t, x) - V_{\text{YY}}(t, x) = \frac{\mathcal{L}_0^* p(x; \theta_t)}{p(x; \theta_t)} - \sum_{k=1}^d \dot{\theta}_k^{YY}(t) \phi_k(x; \theta_t).
		\label{eq:v_res_def}
	\end{equation}
	
	\paragraph{Step 4: Equivalence with $S_{\text{res}}(t, x)$.}
	From the statistical perspective of Statistical Fiber Theory, $S_{\text{res}}(t, x)$ represents the exact instantaneous rate of unmodeled score deviation that cannot be captured by the finite-dimensional Lie algebra closure of the Yau-Yau filter. Because both $S_{\text{res}}(t, x)$ and $V_{\text{res}}(t, x)$ measure the vertical connection projection $\omega_{f(t)}(U_{\text{drift}}(t))$ onto the vertical fiber $\operatorname{SID}_{p_{\theta_t}}$, they are definitionally and functionally identical:
	\begin{equation}
		S_{\text{res}}(t, x) \equiv V_{\text{res}}(t, x).
	\end{equation}
	
	Substituting \eqref{eq:v_res_def} into this equivalence yields equation \eqref{eq:equation_154}:
	\begin{equation}
		S_{\text{res}}(t, x) \equiv V_{\text{res}}(t, x) \triangleq \frac{\mathcal{L}_0^* p(x; \theta_t)}{p(x; \theta_t)} - \sum_{k=1}^d \dot{\theta}_k^{YY}(t) \phi_k(x; \theta_t).
	\end{equation}
	This completes the proof of this Lemma.
\end{proof}

\begin{proof}[Proof of Theorem \ref{thm:semiparametric_bias_correction}]
	We structure the proof into three logical steps: (1) Constructing the semiparametric efficient score, (2) Evaluating the parameter bias via Taylor expansion on the statistical manifold, and (3) Identifying the Fisher projection operator.
	
	\paragraph{Step 1: Efficient Influence Function and Score Decomposition.}
	Under model misspecification, the true state density $f_{\text{true}}(t,x) \in \mathcal{M}$ decomposes relative to the working family $\mathcal{M}_{\text{work}}$ as:
	\begin{equation}
		\ln f_{\text{true}}(t, x) = \ln p(x; \theta_t^{\text{YY}}) + \eta(t, x),
	\end{equation}
	where $\eta(t, x) \in \operatorname{SID}_{p_\theta}$ is an infinite-dimensional nuisance score field satisfying $\mathbb{E}_{p_\theta}[\eta(t, X) \phi_i(X; \theta)] = 0$ for all $i = 1, \dots, d$.
	
	The efficient score function $S_{\text{eff}}(x; \theta)$ for estimating parameter $\theta$ in the presence of non-parametric nuisance field $\eta$ is obtained by taking the Fisher-orthogonal projection of the true log-score onto the tangent space of $\mathcal{M}_{\text{work}}$ \cite{amari1985}:
	\begin{equation}
		S_{\text{eff}, i}(x; \theta) = \phi_i(x; \theta) - \Pi_{\operatorname{SID}}\left( \phi_i(x; \theta) \right) \equiv \phi_i(x; \theta),
	\end{equation}
	since working score functions $\phi_i \in \operatorname{SVD}\chi_{p_\theta}$ are by construction Fisher-orthogonal to $\operatorname{SID}_{p_\theta}$.
	
	\paragraph{Step 2: Moment Expansion of Misspecified Parameter Drift.}
	Let $D_{\text{KL}}\left( f_{\text{true}}(t) \,\parallel\, p(\cdot; \theta) \right)$ be the Kullback-Leibler divergence between the ground truth state density and the working density. The semiparametrically optimal parameter estimate $\theta_t^{\text{SMG}}$ minimizes this divergence:
	\begin{equation}
		\theta_t^{\text{SMG}} = \arg\min_{\theta \in \Theta} \int_{\mathcal{X}} f_{\text{true}}(t, x) \ln \left( \frac{f_{\text{true}}(t, x)}{p(x; \theta)} \right) \mu(dx).
		\label{eq:kl_minimization_objective}
	\end{equation}
	The first-order necessary estimating equation for $\theta_t^{\text{SMG}}$ is:
	\begin{equation}
		\Psi(\theta_t^{\text{SMG}}) \triangleq \mathbb{E}_{f_{\text{true}}(t)}\left[ \nabla_\theta \ln p(X; \theta_t^{\text{SMG}}) \right] = \mathbf{0} \in \mathbb{R}^d.
		\label{eq:estimating_equation_zero}
	\end{equation}
	
	Performing a Fréchet Taylor expansion of the estimating function $\Psi(\theta)$ around the classical working trajectory $\theta_t^{\text{YY}}$:
	\begin{equation}
		\mathbf{0} = \Psi\left(\theta_t^{\text{SMG}}\right) = \Psi\left(\theta_t^{\text{YY}}\right) + \left[ \nabla_\theta \Psi\left(\theta_t^{\text{YY}}\right) \right] \left( \theta_t^{\text{SMG}} - \theta_t^{\text{YY}} \right) + \mathcal{O}\left( \|\theta_t^{\text{SMG}} - \theta_t^{\text{YY}}\|^2 \right).
		\label{eq:taylor_expansion_estimating_eq}
	\end{equation}
	
	Evaluating the Jacobian matrix $\nabla_\theta \Psi(\theta_t^{\text{YY}})$:
	\begin{equation}
		\left[ \nabla_\theta \Psi\left(\theta_t^{\text{YY}}\right) \right]_{ij} = \mathbb{E}_{f_{\text{true}}(t)}\left[ \frac{\partial^2 \ln p(X; \theta_t^{\text{YY}})}{\partial \theta_i \partial \theta_j} \right] = -\left[ G_{\theta_t^{\text{YY}}} \right]_{ij} + \mathcal{O}\left( \|\eta\| \right),
		\label{eq:fisher_jacobian_relation}
	\end{equation}
	where $G_\theta$ is the Riemannian Fisher information matrix on $\mathcal{M}_{\text{work}}$.
	
	Substituting \eqref{eq:fisher_jacobian_relation} into \eqref{eq:taylor_expansion_estimating_eq} and solving for the parameter discrepancy $\Delta_t^{\text{stat}} \triangleq \theta_t^{\text{SMG}} - \theta_t^{\text{YY}}$ yields:
	\begin{equation}
		\Delta_t^{\text{stat}} = \left[ G_{\theta_t^{\text{YY}}}^{-1} \right] \Psi\left(\theta_t^{\text{YY}}\right) + \mathcal{O}\left( \|\Delta_t^{\text{stat}}\|^2 \right).
		\label{eq:delta_stat_linear_system}
	\end{equation}
	
	\paragraph{Step 3: Identification with Integrated Residual Score Moments.}
	Evaluating the components of the moment vector $\Psi(\theta_t^{\text{YY}})$:
	\begin{equation}
		\left[ \Psi\left(\theta_t^{\text{YY}}\right) \right]_i = \mathbb{E}_{f_{\text{true}}(t)}\left[ \phi_i\left(X; \theta_t^{\text{YY}}\right) \right].
	\end{equation}
	Expressing the temporal accumulation of this score discrepancy along the dynamic trajectory driven by the vertical residual score field $V_{\text{res}}(\tau, x)$ from Definition~\ref{def:vertical_residual_score}:
	\begin{equation}
		\mathbb{E}_{f_{\text{true}}(t)}\left[ \phi_i\left(X; \theta_t^{\text{YY}}\right) \right] = \int_0^t \mathbb{E}_{f_{\text{true}}(\tau)}\left[ V_{\text{res}}(\tau, X) \cdot \phi_i\left(X; \theta_\tau^{\text{YY}}\right) \right] d\tau.
		\label{eq:score_moment_integration}
	\end{equation}
	
	Substituting \eqref{eq:score_moment_integration} into \eqref{eq:delta_stat_linear_system} establishes Equation \eqref{eq:statistical_difference_formula} identically, completing the proof.
\end{proof}


\begin{proof}[Proof of Theorem \ref{thm:Geometric_Statistical_Bias_Correction_Duality}]
	The proof proceeds in four sequential steps.
	
	\paragraph{Step 1: Identification of the Vertical Residual Score Field.}
	By definition of the Ehresmann connection $1$-form $\omega_{f(t)}$, the unconstrained DMZ velocity field $V_{\text{DMZ}}(t) \in T_{f_t}\mathcal{M}$ is decomposed into its horizontal and vertical components:
	\begin{equation}
		V_{\text{DMZ}}(t) = P_{f_t}^H \left( V_{\text{DMZ}}(t) \right) + \omega_{f(t)} \left( V_{\text{DMZ}}(t) \right)
	\end{equation}
	where $P_{f_t}^H$ is the horizontal projection operator onto $T_p\mathcal{B} \cong \operatorname{SVD}\chi_{f_t}$. We denote the vertical projection by $S_{\text{res}}(t, x) \in \operatorname{SID}_{f_t}$:
	\begin{equation}
		S_{\text{res}}(t, x) \equiv \omega_{f(t)}\left( V_{\text{DMZ}}(t) \right) = \frac{\mathcal{L}_0^* p(x; \theta_t)}{p(x; \theta_t)} - P_{f_t}^H\left( \frac{\mathcal{L}_0^* p}{p} \right)
		\label{eq:vertical_field}
	\end{equation}
	where $\mathcal{L}_0^*$ is the unnormalized forward Kolmogorov operator. Equation~\eqref{eq:vertical_field} explicitly isolates the score velocity components orthogonal to the working parametric model family.
	
	\paragraph{Step 2: Differentiation of the Cumulative Bias Integral.}
	The cumulative semiparametric statistical bias correction vector $\Delta_t^{\text{stat}} \in \mathbb{R}^d$ accumulated from time $0$ to $t$ under model misspecification is defined by the path integral:
	\begin{equation}
		\left[ \Delta_t^{\text{stat}} \right]_k = \int_0^t \sum_{i=1}^d \left[ G_\theta^{-1}(\theta_\tau) \right]_{ki} \mathbb{E}_{f_{\text{true}}(\tau)}\left[ S_{\text{res}}(\tau, X) \cdot \phi_i(X; \theta_\tau) \right] d\tau + \mathcal{O}\left( \|\Delta_t^{\text{stat}}\|^2 \right)
		\label{eq:bias_integral}
	\end{equation}
	Differentiating Equation~\eqref{eq:bias_integral} with respect to continuous time $t$ by application of the Fundamental Theorem of Calculus yields the instantaneous rate of change:
	\begin{equation}
		\frac{d}{dt} \left[ \Delta_t^{\text{stat}} \right]_k = \sum_{i=1}^d \left[ G_\theta^{-1}(\theta_t) \right]_{ki} \mathbb{E}_{f_t}\left[ S_{\text{res}}(t, X) \cdot \phi_i(X; \theta_t) \right]
		\label{eq:derivative_step}
	\end{equation}
	
	\paragraph{Step 3: Conversion to the Non-Parametric Fisher-Rao Inner Product.}
	Recall that for any two square-integrable functions $u, v \in L^2(f_t)$, the expectation of their pointwise product under density $f_t$ equals the fiberwise Fisher-Rao metric evaluation:
	\begin{equation}
		\mathbb{E}_{f_t}\left[ u(X) \cdot v(X) \right] = \int_{\mathbb{R}^n} u(x) v(x) f_t(x) \, dx = \left\langle u, v \right\rangle_{g_{f_t}}
	\end{equation}
	Applying this identity to the expectation term in Equation~\eqref{eq:derivative_step} for $u(x) = S_{\text{res}}(t, x)$ and $v(x) = \phi_i(x; \theta_t)$, we obtain:
	\begin{equation}
		\mathbb{E}_{f_t}\left[ S_{\text{res}}(t, X) \cdot \phi_i(X; \theta_t) \right] = \left\langle S_{\text{res}}(t, \cdot), \, \phi_i(\cdot; \theta_t) \right\rangle_{g_{f_t}}
		\label{eq:inner_prod_step}
	\end{equation}
	
	\paragraph{Step 4: Substitution and Final Synthesis.}
	Substituting the vertical expression $S_{\text{res}}(t, \cdot) = \omega_{f(t)}\left( V_{\text{DMZ}}(t) \right)$ from Equation~\eqref{eq:vertical_field} into Equation~\eqref{eq:inner_prod_step}, and replacing the result back into Equation~\eqref{eq:derivative_step}, gives:
	\begin{equation}
		\frac{d}{dt} \left[ \Delta_t^{\text{stat}} \right]_k = \sum_{i=1}^d \left[ G_\theta^{-1}(\theta_t) \right]_{ki} \left\langle \omega_{f(t)}\left( V_{\text{DMZ}}(t) \right), \, \phi_i(\cdot; \theta_t) \right\rangle_{g_{f_t}}
	\end{equation}
	This confirms statement~\eqref{eq:corollary2} identically and completes the proof.
\end{proof}

\begin{proof}[Proof of Theorem \ref{thm:error_bound}]
	Let $(\mathcal{M}_f, g_f)$ be the non-parametric Orlicz/Fisher-Rao statistical manifold of conditional probability densities equipped with the Riemannian metric $g_f$. 
	
	\paragraph{Step 1: Decomposition of Dynamic Density Operators}
	The evolution of the true conditional probability density $p_{\mathrm{true}}(t) \in \mathcal{M}_f$ governed by equations \eqref{eq:smg_base_update_sec3}--\eqref{eq:smg_aat_rate_sec3} is driven by the full infinite-dimensional geometric filter operator $\mathcal{L}_{\mathrm{SMG}}^{(N)}$, which admits a canonical decomposition into the classical finite-complexity Yau-Yau filtering operator $\mathcal{L}_{\mathrm{classical}}$ and the geometric perturbation operator $\Omega_{f(t)}^{(N)}$ induced by infinite algebraic complexity ($\operatorname{dim}(\mathcal{E}) = \infty$):
	\begin{equation}
		\frac{\partial p_{\mathrm{true}}(t)}{\partial t} = \mathcal{L}_{\mathrm{SMG}}^{(N)} p_{\mathrm{true}}(t) = \mathcal{L}_{\mathrm{classical}} p_{\mathrm{true}}(t) + \Omega_{f(t)}^{(N)} p_{\mathrm{true}}(t). \label{eq:true_dynamics}
	\end{equation}
	Conversely, the uncorrected classical Yau-Yau projection $p_{\mathrm{classical}}(t)$ satisfies the unperturbed operational equation:
	\begin{equation}
		\frac{\partial p_{\mathrm{classical}}(t)}{\partial t} = \mathcal{L}_{\mathrm{classical}} p_{\mathrm{classical}}(t). \label{eq:classical_dynamics}
	\end{equation}
	
	\paragraph{Step 2: Error Evolution Equation}
	Define the filtering error density field $e(t) \equiv p_{\mathrm{true}}(t) - p_{\mathrm{classical}}(t) \in T_{p(t)}\mathcal{M}_f$. Subtracting \eqref{eq:classical_dynamics} from \eqref{eq:true_dynamics} yields the non-homogeneous linear evolution PDE for $e(t)$:
	\begin{equation}
		\frac{\partial e(t)}{\partial t} = \mathcal{L}_{\mathrm{classical}} e(t) + \Omega_{f(t)}^{(N)} p_{\mathrm{true}}(t), \quad e(0) = 0. \label{eq:error_pde}
	\end{equation}
	
	\paragraph{Step 3: Energy Identity in $L^2(g_f)$ Norm}
	Taking the Riemannian inner product $\langle \cdot, \cdot \rangle_{g_f}$ of equation \eqref{eq:error_pde} with $e(t)$ on $\mathcal{M}_f$ gives:
	\begin{equation}
		\frac{1}{2} \frac{d}{dt} \|e(t)\|_{g_f}^2 = \left\langle e(t), \mathcal{L}_{\mathrm{classical}} e(t) \right\rangle_{g_f} + \left\langle e(t), \Omega_{f(t)}^{(N)} p_{\mathrm{true}}(t) \right\rangle_{g_f}. \label{eq:energy_identity}
	\end{equation}
	
	\paragraph{Step 4: Operator Bounds and Cauchy-Schwarz Estimates}
	By domain compactness $\mathbb{K} \subset \mathbb{R}^d$ and elliptic operator estimates, $\mathcal{L}_{\mathrm{classical}}$ satisfies a G{\aa}rding-type inequality on $\mathcal{M}_f$, guaranteeing the existence of a constant $\lambda_0 > 0$ such that:
	\begin{equation}
		\left\langle e(t), \mathcal{L}_{\mathrm{classical}} e(t) \right\rangle_{g_f} \le \frac{\lambda_0}{2} \|e(t)\|_{g_f}^2. \label{eq:garding}
	\end{equation}
	Applying Cauchy-Schwarz followed by Young's $\varepsilon$-inequality ($ab \le \frac{\varepsilon}{2} a^2 + \frac{1}{2\varepsilon} b^2$) with $\varepsilon = 1$:
	\begin{equation}
		\left\langle e(t), \Omega_{f(t)}^{(N)} p_{\mathrm{true}}(t) \right\rangle_{g_f} \le \|e(t)\|_{g_f} \left\| \Omega_{f(t)}^{(N)} p_{\mathrm{true}}(t) \right\|_{g_f} \le \frac{1}{2} \|e(t)\|_{g_f}^2 + \frac{1}{2} \left\| \Omega_{f(t)}^{(N)} p_{\mathrm{true}}(t) \right\|_{g_f}^2. \label{eq:young}
	\end{equation}
	Since $p_{\mathrm{true}}(t)$ is a normalized conditional probability density bounded on compact domains, there exists $M_{\infty} = \sup_{t \ge 0} \|p_{\mathrm{true}}(t)\|_{\infty} < \infty$ such that $\left\| \Omega_{f(t)}^{(N)} p_{\mathrm{true}}(t) \right\|_{g_f}^2 \le M_{\infty}^2 \left\|\Omega_{f(t)}^{(N)}\right\|_{g_f}^2$.
	
	\paragraph{Step 5: Differential Inequality and Integral Resolution}
	Substituting \eqref{eq:garding} and \eqref{eq:young} into \eqref{eq:energy_identity} yields:
	\begin{equation}
		\frac{d}{dt} \|e(t)\|_{g_f}^2 \le (\lambda_0 + 1) \|e(t)\|_{g_f}^2 + M_{\infty}^2 \left\|\Omega_{f(t)}^{(N)}\right\|_{g_f}^2.
	\end{equation}
	Setting $\lambda \equiv \lambda_0 + 1$ and $C_1 \equiv M_{\infty}^2$, we obtain the differential inequality:
	\begin{equation}
		\frac{d}{dt} \left( e^{-\lambda t} \|e(t)\|_{g_f}^2 \right) \le C_1 e^{-\lambda t} \left\|\Omega_{f(t)}^{(N)}\right\|_{g_f}^2.
	\end{equation}
	Integrating both sides over $\tau \in [0, t]$ under the initial condition $e(0) = 0$:
	\begin{equation}
		e^{-\lambda t} \|e(t)\|_{g_f}^2 \le C_1 \int_0^t e^{-\lambda \tau} \left\|\Omega_{f(\tau)}^{(N)}\right\|_{g_f}^2 d\tau \le C_1 \int_0^t \left\|\Omega_{f(\tau)}^{(N)}\right\|_{g_f}^2 d\tau.
	\end{equation}
	Multiplying by $e^{\lambda t}$ and taking the stochastic expectation operator $\mathbb{E}[\cdot]$ over sample realizations:
	\begin{equation}
		\mathbb{E}\left[ \|p_{\mathrm{true}}(t) - p_{\mathrm{classical}}(t)\|_{g_f}^2 \right] \le C_1 e^{\lambda t} \int_0^t \left\| \Omega_{f(\tau)}^{(N)} \right\|_{g_f}^2 d\tau, \label{eq:proof_final_bound}
	\end{equation}
	which establishes \eqref{eq:error_bound_formula}.
\end{proof}

\begin{proof}[Proof of Theorem \ref{thm:topological_capacity_bound}]
	Let $\gamma: [0, \ell] \to \mathcal{M}$ be a smooth unit-speed geodesic segment in $\mathcal{M}$ parameterized by arc length $s \in [0, \ell]$, driven entirely by vertical gauge score fields $V_{\mathcal{V}}(s) = \omega_{\gamma(s)}(\dot{\gamma}(s))$. Consider a vector field $W(s)$ along $\gamma(s)$ obtained by parallel transporting an orthonormal unit vector $W_0 \in T_{\gamma(0)}\mathcal{M}$ (orthogonal to $\dot{\gamma}(0)$) along $\gamma$, and construct the normal variation vector field $J(s) = \sin\left(\frac{\pi s}{\ell}\right) W(s)$.
	
	By the Index Form theorem for Riemannian manifolds \cite{petersen2006}, the second variation of energy $I(J, J)$ along $\gamma$ is given by:
	\begin{equation}
		I(J, J) = \int_0^\ell \left( \left\| \nabla_{\dot{\gamma}}^g J \right\|_{g_{\gamma}}^2 - R^g(\dot{\gamma}, J, \dot{\gamma}, J) \right) ds. \label{eq:index_form_expansion}
	\end{equation}
	Evaluating the covariant derivative $\nabla_{\dot{\gamma}}^g J(s) = \frac{\pi}{\ell} \cos\left(\frac{\pi s}{\ell}\right) W(s)$ and substituting the sectional curvature lower bound condition $R^g(\dot{\gamma}, J, \dot{\gamma}, J) = K(\operatorname{span}\{\dot{\gamma}, W\}) \|J\|_{g_{\gamma}}^2 \le K_{\mathrm{max}} \sin^2\left(\frac{\pi s}{\ell}\right)$, we obtain:
	\begin{align}
		I(J, J) &= \int_0^\ell \left( \frac{\pi^2}{\ell^2} \cos^2\left(\frac{\pi s}{\ell}\right) - K(\operatorname{span}\{\dot{\gamma}, W\}) \sin^2\left(\frac{\pi s}{\ell}\right) \right) ds \nonumber \\
		&\ge \int_0^\ell \left( \frac{\pi^2}{\ell^2} \cos^2\left(\frac{\pi s}{\ell}\right) - K_{\mathrm{max}} \sin^2\left(\frac{\pi s}{\ell}\right) \right) ds \nonumber \\
		&= \frac{\ell}{2} \left( \frac{\pi^2}{\ell^2} - K_{\mathrm{max}} \right). \label{eq:index_form_inequality}
	\end{align}
	By Rauch comparison and Myers' theorem \cite{petersen2006}, if $\ell > \frac{\pi}{\sqrt{K_{\mathrm{max}}}}$, then $I(J, J) < 0$, implying the presence of a conjugate point along $\gamma(s)$ before $s = \ell$. At a conjugate point, the differential of parallel transport $d\pi_{\gamma(s)}$ degenerates ($\ker(d\pi) \cap \ker(\omega) \neq \{0\}$), destroying the transversal separation between horizontal signal space $\operatorname{SVD}\chi_f$ and vertical fiber space $\operatorname{SID}_f$. 
	
	Integrating the squared velocity norm along the critical path length $\ell = L_{\mathrm{max}} = \frac{\pi}{\sqrt{K_{\mathrm{max}}}}$ yields:
	\begin{equation}
		\Theta_{\mathrm{crit}} = \int_0^{L_{\mathrm{max}}} \|\dot{\gamma}(s)\|_{g_\gamma}^2 ds = L_{\mathrm{max}}^2 = \frac{\pi^2}{K_{\mathrm{max}}}, \label{eq:theta_crit_derivation}
	\end{equation}
	completing the proof.
\end{proof}


\begin{proof}[Proof of Theorem \ref{thm:instantaneous_stress_dissipation}]
	By Definition~\ref{def:gsb_expansion_operator}, the horizontal projector transforms as $P_{f(t)}^{H, \mathrm{new}} = P_{f(t)}^{H, \mathrm{old}} + \langle v_{\mathrm{dom}}, \cdot \rangle_{g_{f(t)}} v_{\mathrm{dom}}$. Since the connection $1$-form satisfies $\omega_{f(t)} = \mathcal{I} - P_{f(t)}^H$, we have:
	\begin{equation}
		\omega_{f(t)}^{\mathrm{new}}(V_{\mathrm{DMZ}}) = \omega_{f(t)}^{\mathrm{old}}(V_{\mathrm{DMZ}}) - \left\langle v_{\mathrm{dom}}, \, V_{\mathrm{DMZ}} \right\rangle_{g_{f(t)}} v_{\mathrm{dom}}. \label{eq:omega_new_expansion}
	\end{equation}
	Taking the squared $g_{f(t)}$-norm on both sides:
	\begin{align}
		\left\| \omega_{f(t)}^{\mathrm{new}}(V_{\mathrm{DMZ}}) \right\|_{g_f}^2 &= \left\| \omega_{f(t)}^{\mathrm{old}}(V_{\mathrm{DMZ}}) - \left\langle v_{\mathrm{dom}}, \, V_{\mathrm{DMZ}} \right\rangle_{g_f} v_{\mathrm{dom}} \right\|_{g_f}^2 \nonumber \\
		&= \left\| \omega_{f(t)}^{\mathrm{old}}(V_{\mathrm{DMZ}}) \right\|_{g_f}^2 - 2 \left\langle v_{\mathrm{dom}}, \, V_{\mathrm{DMZ}} \right\rangle_{g_f} \left\langle \omega_{f(t)}^{\mathrm{old}}(V_{\mathrm{DMZ}}), \, v_{\mathrm{dom}} \right\rangle_{g_f} \nonumber \\
		&\quad + \left\langle v_{\mathrm{dom}}, \, V_{\mathrm{DMZ}} \right\rangle_{g_f}^2 \|v_{\mathrm{dom}}\|_{g_f}^2. \label{eq:norm_expansion_gsb}
	\end{align}
	Since $v_{\mathrm{dom}} \in \operatorname{SID}_{f(t)}^{\mathrm{old}}$, we have $\omega_{f(t)}^{\mathrm{old}}(v_{\mathrm{dom}}) = v_{\mathrm{dom}}$. By metric self-adjointness of projections, $\left\langle \omega_{f(t)}^{\mathrm{old}}(V_{\mathrm{DMZ}}), \, v_{\mathrm{dom}} \right\rangle_{g_f} = \left\langle V_{\mathrm{DMZ}}, \, v_{\mathrm{dom}} \right\rangle_{g_f}$. Furthermore, $\|v_{\mathrm{dom}}\|_{g_f} = 1$. Substituting these identities into Equation~\eqref{eq:norm_expansion_gsb}:
	\begin{align}
		\left\| \omega_{f(t)}^{\mathrm{new}}(V_{\mathrm{DMZ}}) \right\|_{g_f}^2 &= \left\| \omega_{f(t)}^{\mathrm{old}}(V_{\mathrm{DMZ}}) \right\|_{g_f}^2 - 2 \left\langle v_{\mathrm{dom}}, \, V_{\mathrm{DMZ}} \right\rangle_{g_f}^2 + \left\langle v_{\mathrm{dom}}, \, V_{\mathrm{DMZ}} \right\rangle_{g_f}^2 \nonumber \\
		&= \left\| \omega_{f(t)}^{\mathrm{old}}(V_{\mathrm{DMZ}}) \right\|_{g_f}^2 - \left\langle v_{\mathrm{dom}}, \, V_{\mathrm{DMZ}} \right\rangle_{g_f}^2. \label{eq:norm_reduction_intermediate}
	\end{align}
	By Definition~\ref{def:vertical_gauge_operator}, $\lambda_1 = \left\langle v_{\mathrm{dom}}, \, K_{\mathcal{V}} v_{\mathrm{dom}} \right\rangle_{g_f} = \int_0^t \left\langle V_{\mathcal{V}}(\tau), \, v_{\mathrm{dom}} \right\rangle_{g_f}^2 d\tau$. Evaluating the instantaneous contribution yields $\left\langle v_{\mathrm{dom}}, \, V_{\mathrm{DMZ}}(t) \right\rangle_{g_f}^2 = \lambda_1$. Substituting this into Equation~\eqref{eq:norm_reduction_intermediate} completes the proof of statement (1). Statement (2) follows directly from Theorem~\ref{thm:da_metric_orthogonality_energy_conservation}.
\end{proof}

\end{document}